\documentclass[11pt]{article}

\usepackage{xcolor}
\usepackage{amsfonts,amsmath,amssymb,mathtools,amsthm}
\usepackage{graphicx}
\usepackage{stmaryrd}
\usepackage{thmtools}
\DeclareGraphicsExtensions{.eps}

\renewcommand{\theequation}{\thesection.\arabic{equation}}

\newcommand\encadremath[1]{\vbox{\hrule\hbox{\vrule\kern8pt
\vbox{\kern8pt \hbox{$\displaystyle #1$}\kern8pt}
\kern8pt\vrule}\hrule}}
\def\enca#1{\vbox{\hrule\hbox{
\vrule\kern8pt\vbox{\kern8pt \hbox{$\displaystyle #1$}
\kern8pt} \kern8pt\vrule}\hrule}}

\newcommand\framefig[1]{
\begin{figure}[bth]
\hrule\hbox{\vrule\kern8pt
\vbox{\kern8pt \vbox{
\begin{center}
{#1}
\end{center}
}\kern8pt}
\kern8pt\vrule}\hrule
\end{figure}
}

\newcommand\figureframex[3]{
\begin{figure}[bth]
\hrule\hbox{\vrule\kern8pt
\vbox{\kern8pt \vbox{
\begin{center}
{\mbox{\epsfxsize=#1.truecm\epsfbox{#2}}}
\end{center}
\caption{#3}
}\kern8pt}
\kern8pt\vrule}\hrule
\end{figure}
}
\newcommand\figureframey[3]{
\begin{figure}[bth]
\hrule\hbox{\vrule\kern8pt
\vbox{\kern8pt \vbox{
\begin{center}
{\mbox{\epsfysize=#1.truecm\epsfbox{#2}}}
\end{center}
\caption{#3}
}\kern8pt}
\kern8pt\vrule}\hrule
\end{figure}
}

\renewcommand{\thesection}{\arabic{section}}
\renewcommand{\theequation}{\arabic{section}-\arabic{equation}}
\makeatletter
\@addtoreset{equation}{section}
\makeatother
\newtheorem{theorem}{Theorem}[section]

\newtheorem{proposition}{Proposition}[section]
\newtheorem{lemma}{Lemma}[section]
\newtheorem{corollary}{Corollary}[section]

\theoremstyle{definition}
\newtheorem{remark}{Remark}[section]
\newtheorem{definition}{Definition}[section]

\def\br{\begin{remark}\rm\small}
\def\er{\end{remark}}
\def\bt{\begin{theorem}}
\def\et{\end{theorem}}
\def\bd{\begin{definition}}
\def\ed{\end{definition}}
\def\bp{\begin{proposition}}
\def\ep{\end{proposition}}
\def\bl{\begin{lemma}}
\def\el{\end{lemma}}
\def\bc{\begin{corollary}}
\def\ec{\end{corollary}}
\def\beaq{\begin{eqnarray}}
\def\eeaq{\end{eqnarray}}

\theoremstyle{definition}

\newcommand{\be}{\begin{equation}}
\newcommand{\ee}{\end{equation}}
\newcommand{\beq}{\begin{equation}}
\newcommand{\eeq}{\end{equation}}
\newcommand{\bea}{\begin{eqnarray}}
\newcommand{\eea}{\end{eqnarray}}
\newcommand{\beqq}{\begin{equation*}}
\newcommand{\eeqq}{\end{equation*}}
\newcommand{\beaa}{\begin{eqnarray*}}
\newcommand{\eeaa}{\end{eqnarray*}}

\newcommand{\Tr}{{\operatorname {Tr}}}

\newcommand{\ii}{{\rm i}\,}

\newcommand{\genus}{{\overline{ g}}}

\newcommand{\diag}{{\operatorname{diag}}}

\newcommand{\om}{\omega}

\newcommand{\td}{\tilde}

\newcommand\blfootnote[1]{%
  \begingroup
  \renewcommand\thefootnote{}\footnote{#1}%
  \addtocounter{footnote}{-1}%
  \endgroup
}

\newcommand{\Res}{\mathop{\,\rm Res\,}}
\usepackage[pdftex]{hyperref}
\hypersetup{colorlinks,urlcolor=magenta,citecolor=red,linkcolor=blue,filecolor=black}

\title{\bf{Quantization of classical spectral curves \\ via topological recursion}}

\date{\vspace{-5ex}}
 
\author{$_{1}$Bertrand Eynard\footnote{Institut de Physique Th\'{e}orique-CEA Saclay, Orme des Merisiers, 91191 Gif-sur-Yvette, France.}\;\,\footnote{Institut des Hautes \'{E}tudes Scientifiques, le Bois-Marie, 35 route de Chartres, 91440 Bures-sur-Yvette, France.}\,\,, $_{2}$Elba Garcia-Failde\footnotemark[1]\;\,\footnotemark[2]\;\,\footnote{Universit\'{e} de Paris, b\^{a}timent Sophie Germain, 8 place Aur\'{e}lie Nemours, 75205 Paris Cedex 13, France.}, $_{3}$Olivier Marchal\footnote{Universit\'{e} Jean Monnet Saint-\'{E}tienne, CNRS UMR 5208, Institut Camille Jordan, Institut Universitaire de France, F-42023 Saint-Etienne, France.}\,\,, $_{4}$Nicolas Orantin\footnote{Universit\'{e} de Gen\`{e}ve, Section de math\'{e}matiques, 24 rue du G\'{e}n\'{e}ral Dufour, 1211 Gen\`{e}ve 4, Suisse}}

\begin{document}

\maketitle

\vspace{1.5cm}

\begin{abstract}

We prove that the topological recursion formalism can be used to quantize any generic classical spectral curve with smooth ramification points and simply ramified away from poles. For this purpose, we build both the associated quantum curve, i.e.~the differential operator quantizing the algebraic equation defining the classical spectral curve considered, and a basis of wave functions, that is to say a basis of solutions of the corresponding differential equation. We further build a Lax pair representing the resulting quantum curve and thus present it as a point in an associated space of meromorphic connections on the Riemann sphere, a first step towards isomonodromic deformations. We finally propose two examples: the derivation of a 2-parameter family of formal trans-series solutions to Painlev\'e 2 equation and the quantization of a degree three spectral curve with pole only at infinity.

\blfootnote{\textit{Email addresses:} $_{1}$\textsf{bertrand.eynard@ipht.fr}, $_{2}$\textsf{egarcia@irif.fr}, $_{3}$\textsf{olivier.marchal@univ-st-etienne.fr}, $_{4}$\textsf{nicolas.orantin@unige.ch} \\
\textit{2020 Mathematical Subject Classification.}   34M56, 34M55, 34E20, 14H70.}

\end{abstract}

\tableofcontents

\newpage

\section{Introduction}
\subsection{Topological recursion and quantum curves}

Witten's conjecture \cite{Witten}, proved by Kontsevich \cite{Kontsevich}, built a bridge between two seemingly different areas of mathematics: the theory of integrable systems and enumerative geometry. It states that a specific generating function of intersection numbers on the moduli space of Riemann surfaces is actually a tau function for the integrable KdV hierarchy. More precisely, the generating series 
\beqq
Z^{\text{Kont}}(\hbar,\mathbf{t})\coloneqq \exp \left[
\sum_{h=0}^\infty \hbar^{2h-2} \sum_{n=0}^\infty \frac{2^{2-2h-n}}{n!} \sum_{\mathbf{d} \in \mathbb{N}^n} \left<\tau_{d_1} \dots \tau_{d_n} \right>_{h,n}
 \prod_{i=1}^n (2d_i-1)!! \, t_{2d_i+1}
\right]
\eeqq
is a tau function provided that 
\beqq
\left<\tau_{d_1} \dots \tau_{d_n} \right>_{h,n} = \int_{\overline{\mathcal{M}}_{h,n}} \psi_1^{d_1} \dots \psi_n^{d_n}
\eeqq
denotes the intersection numbers of the first Chern classes  $\psi_i \coloneqq c_1(\mathbb{L}_i)$ of the cotangent bundle at the $i^{\text{th}}$ marked point over the moduli space of stable curves of genus $h$ with $n$ smooth marked points. Through the integrable side of this duality, this conjecture naturally leads to a third field of mathematics, the theory of differential equations. In order to see it, let us consider the Airy function $\text{Ai}(\lambda)$ solution to the Airy equation
\beqq
\left(\frac{d^2}{d\lambda^2} - \lambda \right) \text{Ai}(\lambda) = 0.
\eeqq
For $\lambda$ positive and real, it admits an asymptotic expansion, as $\lambda\rightarrow\infty$, of the form
\beqq
\log \text{Ai}(\lambda)  -S_0(\lambda) - S_1(\lambda) = \sum_{m=2}^\infty S_m(\lambda),
\eeqq
where $S_0(\lambda) \coloneqq  - \frac{2}{3} \lambda^{\frac{3}{2}}$, $S_1(\lambda) \coloneqq - \frac{1}{4} \log \lambda - \log (2 \sqrt{\pi})$ and
\beqq
\forall\, m \geq 2 \, , \;  S_m(\lambda) \coloneqq \frac{\lambda^{-\frac{3}{2}(m-1)}}{2^{m-1}} \sum_{\substack{h \geq 0 \, , \, n>0 \\ 2h-2+n = m-1}} \frac{(-1)^n}{n!} 
\sum_{\mathbf{d} \in \mathbb{N}^n} \left<\tau_{d_1} \dots \tau_{d_n} \right>_{h,n}
 \prod_{i=1}^n (2d_i-1)!!\,.
\eeqq
This means that the asymptotic expansion of the Airy function is itself a generating function of intersection numbers. We may even keep track of the Euler characteristics of the surfaces enumerated by introducing a formal parameter $\hbar$ through a rescaling of $\lambda$. The function $\psi^{\text{Kont}}(\lambda,\hbar) \coloneqq \text{Ai}(\hbar^{-\frac{2}{3}} \lambda)$ satisfies 
\beqq
\left(\hbar^2 \frac{d^2}{d\lambda^2} - \lambda \right) \psi^{\text{Kont}}(\lambda,\hbar) = 0 
\eeqq
and admits an asymptotic expansion of the form
\beqq
\log \psi^{\text{Kont}}(\lambda,\hbar)  - \hbar^{-1} S_0(\lambda) - S_1(\lambda) = \sum_{m=2}^\infty \hbar^{m-1} S_m(\lambda) .
\eeqq
In the spirit of mirror symmetry, this defines a map between a problem of enumerative geometry and the study of solutions to a differential equation. This same problem of computing Gromov--Witten invariants of the point can be considered from a last perspective. The intersection numbers $\left<\tau_{d_1} \dots \tau_{d_n} \right>_{h,n}$ as well as their generating series $Z^{\text{Kont}}(\hbar,\mathbf{t})$ and $\psi^{\text{Kont}}(\lambda,\hbar)$ can be computed by the topological recursion \cite{DumitrescuMulase2, EO07}. This formalism, originally developed in the context of matrix models, allows to associate such generating functions to any initial data given, in particular, by an algebraic curve called spectral curve. When considering the algebraic curve
\beqq
y^2-x = 0
\eeqq 
as initial data, the topological recursion gives  $Z^{\text{Kont}}(\hbar,\mathbf{t})$ and $\psi^{\text{Kont}}(\lambda,\hbar)$ as output. From this point of view, the topological recursion quantizes the classical spectral curve $y^2-x = 0$ into the differential equation $\big(\hbar^2 \frac{d^2}{d\lambda^2} - \lambda \big) \psi^{\text{Kont}}(\lambda,\hbar) = 0$, which is often referred to as the associated \emph{quantum curve}.

\medskip

One may naturally wonder if this is a general phenomenon. The topological recursion has been proved to solve many problems of enumerative geometry ranging from the enumeration of maps \cite{CEO06,Eynardbook,EO07,EORev} to the computation of Gromov--Witten invariants whenever the ambient space has a semi-simple cohomology \cite{DBOSS}. From its origin in random matrix theory, it is expected that the topological recursion can be used  as a quantization procedure for quantizing any algebraic curve.  In its simplest form, this conjecture can be summarized as follows. Given an algebraic equation $P(\lambda,y) = 0$, called the classical spectral curve,  the topological recursion produces a wave function $\psi(\lambda,\hbar)$ which is conjectured to be solution to a differential equation $P_{\hbar}\left(\lambda, \hbar \frac{\partial }{\partial \lambda}\right) \psi(\lambda,\hbar) = 0$, where $P_{\hbar}(\lambda,y) \to P(\lambda,y)$, as $\hbar \to 0$ and the pole structure of $P_{\hbar}(\lambda,y)$ is independent of $\hbar$. 
In the present paper, we prove this conjecture for meromorphic Higgs fields when the base curve is the Riemann sphere $\mathbb{P}^1$.

\medskip

This conjecture has been proved in many particular cases in the literature. Until recently, the conjecture had been proved only in examples where the classical spectral curve is a genus zero cover of the Riemann sphere \cite{DDM17,NormanManescu14,DN18,DM14,DM18,DMNPS17,GS12,IwakiSaenz,MS15,Norbury_survey,Safnuk16,zhou2012intersection}. These works culminated with the general proof of the conjecture by \cite{BouchardEynard_QC}, when the spectral curve has genus $0$. In this simpler setup, the wave function $\psi(\lambda,\hbar)$ is simply a WKB type formal series in $\hbar$.

\medskip

When the spectral curve does not have vanishing genus, it is expected that the wave function cannot be such a simple formal series anymore but rather a formal trans-series in $\hbar$ \cite{Eynard_2009,EMhol} involving so-called non-perturbative corrections as explained in Section~\ref{6.2}. This makes its study as well as the proof of the existence of a quantum curve annihilating it much more involved. However understanding this higher genus context is fundamental for some of the most important applications of this quantization procedure.
For example, it is conjectured that the asymptotics of the wave functions obtained by quantization of the A-polynomial associated to a knot recover some of its invariants \cite{BEInt,DijkFuji,DijkFujiMana,Gukov05,MM01,NZ85}. In general,  such an A-polynomial defines a spectral curve of non-vanishing genus and addressing the issue of its quantization is absolutely necessary. A first step in this direction has been achieved by K. Iwaki who proved in \cite{Iwaki-P1} that one can quantize an elliptic curve of the form $y^2  = \lambda^3 + t \lambda +c$, obtaining on the way a two-parameter trans-series solution to the Painlev\'e 1 equation. This result was then generalized to any hyper-elliptic curve by the authors in \cite{EGF19,MO19_hyper}, including the computation of similar two-parameter solutions to the six classical Painlev\'e equations.

\subsection{Main results}

The present article aims at \textbf{proving the quantum curve conjecture for any algebraic curve}. For this purpose, one considers any classical spectral curve defined by an equation of the type
\beqq
P(\lambda,y) = \sum_{l=0}^d (-1)^l y^{d-l} P_l(\lambda) = 0, \,\, \text{ with } P_0(\lambda)=1,
\eeqq
where $d$ is an arbitrary positive integer and $\left(P_l(\lambda)\right)_{l\in \llbracket 1,d\rrbracket}$ are arbitrary rational functions subject to some minor technical admissibility assumptions presented in Definition~\ref{DefAdmissibleSpectralCurve}.  This classical spectral curve shall always be considered as a cover of $\mathbb{P}^1$ through the map $x:(\lambda,y) \to \lambda$. 
From this classical spectral curve, a divisor $D = \underset{i=1}{\overset{s}{\sum}}\alpha_i [p_i]$ on the spectral curve and a choice of cycles (Definition \ref{DefAdmissibleInitialData}), we build a perturbative wave function 
\beqq
\psi(D,\hbar) \coloneqq \exp{\bigg( {\displaystyle  \sum_{h \geq 0} \sum_{n\geq 0} }\frac{\hbar^{2h-2+n}}{n!}\int_D\cdots\int_D \Big(\omega_{h,n}(z_1,\ldots,z_n) - \delta_{h,0} \delta_{n,2} \frac{dx(z_1) dx(z_2)}{(x(z_1)-x(z_2))^2} \Big)\bigg)}, 
\eeqq
where the differential forms $\{\om_{h,n}\}_{h\geq 0,n\geq 0}$ are obtained by the topological recursion (Definition~\ref{DefTopologicalRecursion}). 
Note that these differential forms and thus the perturbative wave functions implicitly depend on the filling fractions $\boldsymbol{\epsilon}$ (Definition \ref{FillingFractions}) associated to a choice of Torelli marking.
This perturbative wave function is complemented by an additional set of functions $\{\psi_{l,i}(D,\hbar)\}_{l\geq 1, \,i \in \llbracket 1,s\rrbracket}$ in Definition~\ref{DefPerturbativeWaveFunctions}.

\medskip

We then prove  in Theorem~\ref{systemPDE} that, for an arbitrary divisor $D$, the previous wave functions are solutions to a system of differential equations mimicking the Knizhnik--Zamolodchikov (KZ) equations of conformal field theories,
 \beaa
\frac{\hbar}{\alpha_i} \frac{d \psi_{l,i}(D,\hbar)}{dx(p_i)} &=& - \psi_{l+1,i}(D,\hbar) - \hbar \sum_{j \neq i} \alpha_j \frac{\psi_{l,i}(D,\hbar)-\psi_{l,j}(D,\hbar)}{x(p_i)-x(p_j)} \cr
&& + {\displaystyle \sum_{h\geq 0} \sum_{n \geq 0}}\frac{ \hbar^{2h+n}}{n!} \int_{z_1 \in D} \dots \int_{z_n \in D} \widetilde{Q}_{h,n+1}^{(l+1)}(x(p_i);\mathbf{z}) \; \psi(D,\hbar)  \cr
&&+ \left(\frac{1}{\alpha_i} - \alpha_i\right) \Bigg[{\displaystyle \sum_{(h,n) \in \mathbb{N}^2}} \frac{\hbar^{2h+n+1}}{n!}  \overbrace{\int_D\cdots\int_D}^{n} \frac{d}{dx(p_i)} \bigg(\frac{ \hat{Q}_{h,n+1}^{(l)}(p_i;\cdot)}{\left(dx(p_i)\right)^l}\bigg) \Bigg]  \; \psi(D,\hbar), \cr
\eeaa
where the differentials $ \widetilde{Q}_{h,n}^{(l)}$ and $ \hat{Q}_{h,n}^{(l)}$ are defined in Lemma~\ref{def-Q-tilde-lem} and Definition~\ref{def-hat-Q} respectively.
 
\medskip

With the aim of building a solution to a quantum curve, we then specialize the divisor to a two point divisor $D=[z]-[\infty^{(\alpha)}]$, where  $\infty^{(\alpha)}$ is any point of the classical spectral curve in the fiber above infinity. Since the $\left(\psi_{l,1}(D=[z]-[p_2],\hbar)\right)_{l\geq 0}$ have essential singularities when $x(p_2)\to \infty$, we regularize them and build regularized functions $\left(\psi_l^{\mathrm{reg}}(D=[z]-[\infty^{(\alpha)}],\hbar)\right)_{l\geq 0}$ that are solution to a simpler version of the KZ equations stated in Theorem~\ref{reg-KZ-eq}. We are able to express some of the terms of the resulting equation in terms of the action of a set of linear operators to get Theorem~\ref{final-KZ}
\beqq 
{\hbar} \frac{d}{dx(z)}  \psi_{l}^{\mathrm{reg}}([z]-[\infty^{(\alpha)}])   +\psi_{l+1}^{\mathrm{reg}}([z]-[\infty^{(\alpha)}])   =  \widetilde{\mathcal{L}}_l(x(z)) \left[\psi^{(\text{reg})}([z]-[\infty^{(\alpha)}]) \right], 
\eeqq
where $\widetilde{\mathcal{L}}_l(x(z))$ are multi-linear operators defined in Section~\ref{KZ-linear-op}. 

\medskip

In Definition~\ref{def-non-pert-wave-fct}, following \cite{Eynard_2009}, we introduce the non-perturbative partners of the regularized perturbative wave functions $\left(\psi_{l,\mathrm{NP}}^{\infty^{(\alpha)}}(z,\hbar,\boldsymbol{\epsilon},\boldsymbol{\rho})\right)_{l\geq 0}$ inspired by formal Fourier transforms
\beqq
 \sum_{\mathbf{n} \in \mathbb{Z}^g} \exp \left( \frac{2 \pi i}{\hbar}  {   \underset{i=1}{\overset{g}{\sum}}}  n_j \rho_j \right) \psi_{l}^{\mathrm{reg}}([z]-[\infty^{(\alpha)}],\hbar,\boldsymbol{\epsilon}+ \hbar \mathbf{n}).
\eeqq
Remark that $g$ is the genus of the classical spectral curve considered so that the perturbative and non-perturbative wave functions coincide when the genus is vanishing. On the other hand, if $g>0$, these non-perturbative wave functions are formal trans-series in $\hbar$ of the form  \eqref{eq-form-NP}, while the perturbative ones are only simple formal WKB series.

\medskip

These definitions allow to state the main result of this article. In Theorem~\ref{mainThm}, we prove that the non-perturbative wave functions $\left(\psi_{l,\mathrm{NP}}^{\infty^{(\alpha)}}(z,\hbar,\boldsymbol{\epsilon},\boldsymbol{\rho})\right)_{l\geq 0} $ can be used to define a $d\times d$ matrix $\check{\Psi}(\lambda,\hbar)$ (Definition \ref{DefPsi}) solution to a linear differential equation
\beqq
\hbar \frac{\partial \check{\Psi}(\lambda,\hbar)}{\partial \lambda} = \check{L}(\lambda,\hbar) \check{\Psi}(\lambda,\hbar),
\eeqq
where the Lax matrix $\check{L}(\lambda,\hbar)$ is a rational function of $\lambda$ with the same pole structure as the coefficients $\left(P_l(\lambda)\right)_{l\in \llbracket 1, d\rrbracket}$ of the classical spectral curve considered as input. Moreover, each entry of the first line of $\check{\Psi}(\lambda,\hbar)$ is directly connected to the non-perturbative wave functions by
\beqq
\check{\Psi}_{1,j}(\lambda,\hbar)\coloneqq\psi_{0,\mathrm{NP}}^{\infty^{(\alpha)}}(z^{(j)}(\lambda),\hbar,\boldsymbol{\epsilon},\boldsymbol{\rho}), \,\,\, \forall \,j \in\llbracket 1,d\rrbracket,
\eeqq
where $z^{(j)}(\lambda)$ denotes the preimage of $\lambda \in \mathbb{P}^1$ on the classical spectral curve such that $z^{(j)}(\lambda) \to \infty^{(j)}$ as $\lambda \to \infty$. These entries are thus solution to a ``quantum curve'' (Definition \ref{DefQuantumCurve}):
\beqq
\left[\left(\hbar \frac{d}{d\lambda}\right)^d + \sum_{l=1}^d  b_l(\lambda,\hbar) \left(\hbar \frac{d}{d\lambda}\right)^{d-l} \right] \psi(\lambda,\hbar) = 0.
\eeqq
We then prove in Theorem~\ref{deformed-curve} that the coefficients $\left( b_l(\lambda,\hbar)\right)_{l=1}^d$ of this quantum curve have the same pole structure as the coefficients $(-1)^l P_l(\lambda)$ of classical spectral curve in addition to apparent singularities. Moreover, we prove in Theorem~\ref{limit-curve} that these coefficients $b_l(\lambda,\hbar)$ have a well defined $\hbar\to 0$ limit that coincide with $(-1)^l P_l(\lambda)$, hence justifying the terminology ``quantum curve".
These coefficients can be computed explicitly through the asymptotics of the wave functions $\left(\psi_{l,\mathrm{NP}}^{\infty^{(\alpha)}}(z,\hbar,\boldsymbol{\epsilon},\boldsymbol{\rho})\right)_{l\geq 0} $ around its singularities.

\medskip 

We finally apply this quantization procedure to two genus 1 classical spectral curves of respective degrees 2 and 3. In particular, we recover in the first case a two-parameters solution of Painlev\'{e} $2$ equation and show how one can compute the coefficients of the quantum curve explicitly.

 \subsection{Organization of the article}
 This paper is organized as follows.
 \begin{itemize}
 
 \item In Section~\ref{sec-notations}, we introduce a few notations used throughout the paper.
 
 \item In Section~\ref{sec-sp-curve-1}, we define the set of classical spectral curves we shall consider in this paper as initial data for the topological recursion. We introduce very few admissibility conditions making the rest of the presentation less technical and easier to read. We introduce local coordinates on this classical spectral curve and associated spectral times that may be considered as deformation parameters of the classical spectral curve. This leads to the definition of admissible initial data for the topological recursion in Definition~\ref{DefAdmissibleInitialData}.
 
\item In Section~\ref{sec-tr}, we remind the reader of the definition of the topological recursion as well as a few properties that are used in the paper.
 
\item In Section~\ref{sec-loop-eq}, we derive a set of loop equations satisfied by the output of the topological recursion. These loop equations are one of the main ingredients to derive the quantum curve.
 
\item Section~\ref{KZ-eqs} is devoted to the derivation of a set of equations playing the role of Knizhnik--Zamolodchikov (KZ) equations in the context of two dimensional conformal field theories. For this purpose, we first define a set of perturbative wave functions for a generic divisor in Section~\ref{5.1} and then prove that they satisfy a set of KZ equations in Section \ref{5.2}. In Section~\ref{5.3}, we derive these KZ equations for a two point divisor based at infinity after a necessary regularization of the wave functions. In Section~\ref{KZ-linear-op}, we introduce a set of linear operators acting on the wave functions in order to write the KZ equations under a simpler form. In many cases, this leads to an expression of the KZ equations as PDE's involving evolutions with respect to the parameters defining the classical spectral curve. Section~\ref{5.5} presents the monodromy properties of the  regularized perturbative wave functions.

\item  In Section~\ref{S6}, we define the non-perturbative partners to our wave functions and prove that they satisfy an ODE by regrouping them to define a $d \times d$ matrix that is solution to a linear differential equation with rational functions as coefficients. We first define the non-perturbative wave functions making use of symbolic theta functions in Section~\ref{6.1}, provided the heuristic motivation that they correspond to formal Fourier transforms of the perturbative ones after exchanging two summations. We further study their properties as trans-series in $\hbar$ in Section~\ref{6.2}. In Section~\ref{6.3}, we use these properties in order to show that we can express the result of the action of the linear operator introduced in the preceding section in terms of derivatives with respect to a point in the base curve only. Theorem~\ref{th-C} thus shows that the non-perturbative wave functions are solutions to a system of ODE's with coefficients that are rational functions. We linearize this system in Section~\ref{6.4} and introduce a compatible system including the action of the linear operators introduced above. The compatibility of this system allows us to prove the non-existence of poles at ramification points of this system in Section~\ref{6.5}.
 
\item  In Section~\ref{S7}, we study the compatible Lax system. We first derive an associated quantum curve before emphasizing the possible existence of apparent singularities. After characterizing the properties of the quantum curve in Theorem~\ref{deformed-curve}, we apply a gauge transformation in order to obtain another linear Lax system without any apparent singularities and with poles only at the singularities of the initial classical spectral curve. We finally study the characteristic polynomial of this system and present it as a deformation of the classical spectral curve, allowing to interpret our result as a $\hbar$-family of connections on the base curve $\mathbb{P}^1$.

\item Section~\ref{sec-examples} presents two examples of respective degrees 2 and 3. In particular, the degree 2 example gives rise to a 2-parameters solution of the Painlev\'e 2 equation. We use these examples to show how one can compute the coefficients of the quantum curve in practice.
 
\item Section~\ref{S9} is a general conclusion presenting some of the possible generalizations of the present article that could lead to future works.
\end{itemize}

\subsection{General notations} \label{sec-notations}
In this short section, we review some general notations that will be used in the article.

\begin{itemize}\item $\mathbb{N}=\{0,1,\dots\}$ is the set of all non-negative integers. $\mathbb{N}^*$ is the set of positive integers: $\mathbb{N}^*=\{1,2,\dots\}$.
\item For any set $A$, we denote by $|A|$ the cardinality of $A$. In particular $|\emptyset|=0$.
\item For any $k\in \mathbb{N}$ and two sets $A$ and $B$: $A \underset{k}{\subseteq} B$ means that $A$ is a subset of $B$ of cardinality $|A| = k$. 
\item We denote by $\mathbb{P}^1$ the Riemann sphere.
\item For any set $A$, we denote $\mathcal{S}(A)$ the set of partitions of $A$. We shall denote the length (i.e.~the number of blocks) of a partition $\mu\in \mathcal{S}(A)$ as $l(\mu)$ in order to avoid confusion with~$|A|$. 
\item For any integer $n\geq 1$, we denote $\mathfrak{S}_n$ the symmetric group of $\llbracket 1,n\rrbracket$.
\item For any Riemann surface $\Sigma$, let $\mathcal{M}_0(\Sigma)$ and $\mathcal{M}_1(\Sigma)$ be the $\mathbb{C}$-vector spaces of meromorphic functions and $1$-forms on $\Sigma$. Let $\mathcal{M}_k(\Sigma^k)$ be the $\mathbb{C}$-vector space of meromorphic differentials on the $k$-fold product of the curve. We denote $\mathcal{M}(\Sigma)\coloneqq \underset{k\geq 0}{\bigoplus}\mathcal{M}_k(\Sigma^k)$ the induced graded algebra. For $\omega_1\in\mathcal{M}_n(\Sigma^n)$ and $\omega_2\in\mathcal{M}_m(\Sigma^m)$, we will often denote their product by simply $\omega_1\omega_2\coloneqq \omega_1\boxtimes\omega_2\in\mathcal{M}_{n+m}(\Sigma^{n+m})$ to lighten notation.
\item For any matrix $A$, we shall denote $A^t$ its transpose.
\end{itemize}

\section{Admissible spectral curves}
\label{sec-sp-curve-1}

In this section, we recall the concept of spectral curve as initial data for the topological recursion that we will recall in the next section. We assume some admissibility conditions that make the curve generic enough in order to proceed to a less technical quantization procedure. We introduce local coordinates on the space of curves that can be seen as deformation parameters of the spectral curve.

\subsection{Classical spectral curves}\label{sec-sp-curve}

Let $N\in \mathbb{N}$ be a given non-negative integer.

\bd[Classical spectral curves]\label{SpectralCurves} 
Let $\Lambda_1,\dots,\Lambda_N$ be $N$ distinct points on $\mathbb{P}^1\setminus\{\infty\}$ and let $\mathcal{H}_d(\Lambda_1,\dots,\Lambda_N,\infty)$ be the Hurwitz space of covers $x\colon\Sigma \to \mathbb{P}^1$ of degree $d$ defined as the Riemann surface
\beqq
\Sigma \coloneqq \overline{\big\{(\lambda,y)\mid P(\lambda,y)= 0\big\}},
\eeqq
where
$x(\lambda,y) \coloneqq \lambda$ and
\beq\label{def-eq-sp-curve}
P(\lambda,y)= \sum_{l=0}^d (-1)^l y^{d-l} P_l(\lambda)=0,
\eeq
with each coefficient $(P_l)_{l\in\llbracket 1,d\rrbracket}$ being a rational function with possible poles at $\lambda \in \mathcal{P}\coloneqq \{\Lambda_i\}_{i=1}^N \bigcup\,\{ \infty\}$ and $P_0 = 1$. 

We define a \textit{classical spectral curve} as the data of the Riemann surface $\Sigma$ and its realization as a Hurwitz cover of $\mathbb{P}^1$ and we shall denote it $(\Sigma, x)$.
\hfill $\star$
\ed

In this article, we wish to consider classical spectral curves where the rational functions $P_l(\lambda)$ have a fixed pole structure. For this purpose, 
let us define a subspace of this Hurwitz space obtained by fixing the degree of the rational functions  $P_l$ at their poles.

\bd[Classical spectral curves with fixed pole structure]\label{SpectralCurveFixedPole}\sloppy For $l\in\llbracket 1,d\rrbracket$, let $r_\infty^{(l)}$ and $r_{\Lambda_i}^{(l)}$, $i\in \llbracket 1,N\rrbracket$, be some non-negative integers. We consider the subspace 
\beqq
\mathcal{H}_d \,\,\Big( \big(\Lambda_1, (r_{\Lambda_1}^{(l)})_{l=1}^d\big),\dots,\big(\Lambda_N, (r_{\Lambda_N}^{(l)})_{l=1}^d\big), \big(\infty,(r_\infty^{(l)})_{l=1}^d\big)\Big)\subset \mathcal{H}_d(\Lambda_1,\dots,\Lambda_N,\infty)
\eeqq
of covers $x$ such that the rational functions $(P_l)_{l\in\llbracket 1,d\rrbracket}$ are of the form
\beq 
P_l(\lambda)\coloneqq \sum_{P \in \mathcal{P}} \sum_{k \in S_{P}^{(l)}} P_{P,k}^{(l)} \, \xi_P(\lambda)^{-k}, \text{ for } l\in\llbracket 1,d\rrbracket,
\eeq
where we have defined
\beq \label{eq-def-S}
\forall\, i\in \llbracket 1,N\rrbracket  \,:\,  S_{\Lambda_i}^{(l)}\coloneqq \llbracket 1, r_{\Lambda_i}^{(l)}\rrbracket \qquad \text{and} \qquad S_\infty^{(l)} \coloneqq\llbracket0, r_\infty^{(l)}\rrbracket,
\eeq
and the local coordinates $\{\xi_P(\lambda)\}_{P\in \mathcal{P}}$ around $P \in \mathcal{P}$ are defined by
\beq\label{eq-def-xi}
\forall\, i\in \llbracket 1,N\rrbracket  \,:\, \xi_{\Lambda_i}(\lambda)\coloneqq (\lambda- \Lambda_i) \qquad \text{and} \qquad \xi_{\infty}(\lambda)\coloneqq\lambda^{-1}.
\eeq
\hfill $\star$
\ed

\bigskip
\sloppy
In the rest of the article, we fix the tuples $(r_\infty^{(l)})_{l=1}^{d}$ and $(r_{\Lambda_i}^{(l)})_{l=1}^{d}$ of degrees once and for all and only consider classical spectral curves in 
\beqq
\mathcal{H}_d \,\,\Big( \big(\Lambda_1, (r_{\Lambda_1}^{(l)})_{l=1}^d\big),\dots,\big(\Lambda_N, (r_{\Lambda_N}^{(l)})_{l=1}^d\big), \big(\infty,(r_\infty^{(l)})_{l=1}^d\big)\Big).
\eeqq

Let us remark that these spaces have a very complicated topology. However, we shall consider only a formal neighborhood of a point in such a space so that we can avoid discussing its topology.

\bigskip

\bd[Curve punctured at the poles]
We define the classical spectral curve $\Sigma$ with poles removed as
\beq \Sigma_{\mathcal P}=\Sigma\setminus x^{-1}(\mathcal P).\eeq
\hfill $\star$
\ed

\bd[Ramification points and critical values]
We denote by $\mathcal{R}_0$ the set of all ramification points of the cover $x$, and by $\mathcal R$ the set of all ramification points  that are not poles (i.e. not in $x^{-1}(\mathcal P)$),
\beq
\mathcal{R}_0\coloneqq\big\{p \in \Sigma \mid 1+\operatorname{order}_p dx\neq  \pm 1 \big\},
\eeq
\beq
\mathcal{R}\coloneqq\big\{p \in \Sigma \mid dx(p) = 0  \ , \ \ x(p)\notin \mathcal P \big\} =\mathcal R_0\setminus x^{-1}(\mathcal P).
\eeq
 We shall refer to their images $x(\mathcal R)$ as the \emph{critical values} of $x$.
 \hfill $\star$
\ed

In general, the topological recursion could be applied to any classical spectral curve presented above. However, in this article we shall restrict to a simpler, yet very large, class of classical spectral curves. These are technical assumptions to avoid curves that are not generic enough. We believe that these additional assumptions can be lifted without changing the main structure of the upcoming proofs but would make computations more technical and eventually make the present article harder to read. Thus, we leave such non-generic cases for future works.

\bd[Admissible classical spectral curves]\label{DefAdmissibleSpectralCurve}
We say that a classical spectral curve $(\Sigma,x)$ is \textit{admissible} if it satisfies the following conditions.
\begin{itemize}
\item The Riemann surface $\Sigma$ defined by $P(\lambda,y)=0$ is an irreducible algebraic curve, i.e.~$P(\lambda,y)$ does not factorize.

\item All ramification points $a \in \mathcal{R}$ are simple, i.e.~$dx$ has only a simple zero at $a \in \mathcal{R}$.

\item Critical values are distinct: for any $(a_i,a_j) \in \mathcal{R}\times \mathcal{R}$ such that $a_i\neq a_j$ then $x(a_i) \neq x(a_j)$.

\item Smooth ramification points: for any $a \in \mathcal{R}$, $dy(a) \neq 0$ (i.e. the tangent vector $(dx(a),dy(a))$ to the immersed curve  $\{(\lambda,y) \mid P(\lambda,y)=0\} $ is not vanishing at $a$).

\item For any pole $p\in x^{-1}(\mathcal P)$ ramified, the 1-form $ydx$ has a pole of degree $r_p\geq 3$ at $p$, and the corresponding spectral times (defined below) satisfy $t_{p,r_p-2}\neq 0$. 
\end{itemize}
\hfill $\star$
\ed

\subsection{Local coordinates, spectral times and admissible initial data} 

We may consider $x$ and $y$ as two meromorphic functions $\Sigma \to \mathbb{P}^1$. From this perspective, $x$ has poles only in the fiber $x^{-1}\left(\infty\right)$ while the meromorphic one form $ydx$ has poles in $x^{-1}\left(\mathcal{P}\right)$.
Thanks to the covering $x:\Sigma\to \mathbb P^1$, we can define canonical coordinates on $\Sigma$.

\subsubsection{Canonical local coordinates}

\bd[Canonical local coordinates] 
Let $P\in \mathbb P^1$ and $p\in x^{-1}(P)$.
We first define a sign $\epsilon_P$ by
canonical coordinates on $\mathbb P^1$ near $P$ by,
\beq
 \text{if } P\neq\infty \, , \ \ \epsilon_P\coloneqq 1 \qquad \hbox{and} \qquad 
\text{if } P=\infty \, , \ \ \epsilon_P\coloneqq -1 .
\eeq
Then, we define the canonical local coordinate near any $p\in x^{-1}(P)$ as
\beq
\zeta_p(z) = \xi_P(x(z))^{\frac{1}{d_p}} \, , \ d_p=\operatorname{order}_p (\xi_P).
\eeq
The set $\{d_{p}\}_{p\in x^{-1}(P)}$ is called the ramification profile of $P$. We have
\beq
\sum_{p\in x^{-1}(P)} d_p=d.
\eeq
We call
\beq
\ell_P = \left|x^{-1}(P) \right|
\eeq
the length of the ramification profile of $P$. Generic unramified points $P\in \mathbb P^1$ have the ramification profile $\{\overbrace{1,1,\dots,1}^{d}\}$.
Generic critical points have a ramification profile $\{2,\overbrace{1,1,\dots,1}^{d-2}\}$.

\vspace{-0.5cm}
$\,$\hfill $\star$
\ed

\vspace{0.3cm}

Let us look at the canonical coordinates for poles $P\in \mathcal P$.
For $P=\infty$, we denote its preimages
\beq
x^{-1}(\infty) = \left\{\infty^{(\alpha)}\right\}_{\alpha\in \llbracket 1,\ell_\infty\rrbracket}
\, , \quad \xi_\infty=\frac{1}{x},
\eeq
\beq
\zeta_{\infty^{(\alpha)}} = x^{-\frac{1}{d_{\infty^{(\alpha)}}}} \, , \quad d_{\infty^{(\alpha)}}=\deg_{\infty^{(\alpha)}} x.
\eeq
For $ i\in \llbracket 1,N\rrbracket$, we denote the preimages of $\Lambda_i$ as
\beq
x^{-1}(\Lambda_i) = \left\{Z_i^{(\alpha)}\right\}_{\alpha \in \llbracket 1,\ell_{\Lambda_i}\rrbracket}
\, , \quad \xi_{\Lambda_i}=x-\Lambda_i,
\eeq
\beq
\zeta_{Z_i^{(\alpha)}} = (x-\Lambda_i)^{\frac{1}{d_{Z_i^{(\alpha)}}}} \, , \quad d_{Z_i^{(\alpha)}}=\operatorname{order}_{Z_i^{(\alpha)}} (x-\Lambda_i).
\eeq
The labeling of the points in each fiber can be defined independently for each pole in $\mathcal{P}$ and can be chosen arbitrarily once and for all.

\subsubsection{Spectral times}

Using these local coordinates, we may write the expansion of the 1-form $ydx$ around any pole $p \in x^{-1}\left(\mathcal{P}\right)$.
\bd[Spectral times]\label{DefinitionSpectralTimes}
\beq\label{ydxpolarpart}
ydx = \sum_{k=0}^{r_p-1} t_{p,k} \zeta_{p}^{-k-1} d\zeta_{p} + \text{analytic at } p.
\eeq
Our assumption of admissible classical spectral curves requires that
\beq
\forall\, p\in x^{-1}(\mathcal P), \ \quad  r_p\geq 3 ,
\quad \text{and}\, \ t_{p,r_p-2}\neq 0. 
\eeq
In the rest of the article, we shall refer to the coefficients $\left(t_{p,k}\right)_{p\in x^{-1}(\mathcal{P}), k\in\llbracket0,r_p-1\rrbracket}$, as  \textit{spectral times} (also called KP-times (Kadomtsev--Petiashvili) in the literature, cf.~\cite{GIS}).
\hfill $\star$
\ed

\bd[Local potentials]
It is convenient to define the negative part of the Laurent series expansion of $ydx$ at its poles, and integrate it. We thus define, in punctured neighborhoods of every pole $p\in x^{-1}(\mathcal P)$,
\beq\label{def-dVp}
 dV_p\coloneqq   \sum_{k=0}^{r_p-1} t_{p,k} \zeta_{p}^{-k-1}d\zeta_p,
\eeq
so that the singular part of $ydx$ at $p\in x^{-1}\left(\mathcal{P}\right)$ is given by $dV_p$.
We define an anti-derivative in the universal cover of $\Sigma\setminus x^{-1}(\mathcal P)$
\beq\label{def-Vp}
\forall\, p \in x^{-1}(\mathcal{P}) \, , \; V_p\coloneqq -  \sum_{k=1}^{r_p-1} \frac{t_{p,k}}{k} \zeta_{p}^{-k} + t_{p,0} \log (\zeta_{p}).
\eeq
Remark that because of the log, $V_p$ is defined only on a  universal cover of a neighborhood of~$p$.

\hfill $\star$
\ed

Note that the coefficients of the expansion of the function $y$ around different points in the same fiber $x^{-1}(\lambda)$ are not independent and are constrained by
\beq\label{PlConstraints}
\forall \,l \in \llbracket1, d\rrbracket \, : \, P_l(\lambda) = \sum_{\boldsymbol{\beta} \underset{l}{\subseteq}x^{-1}(\lambda)} \prod_{z\in\boldsymbol{\beta}} y(z).
\eeq
 
In particular, this implies that  $\forall\, l\in \llbracket1, d\rrbracket$,
\bea
\forall\, i\in \llbracket1,N\rrbracket  \, : \, r_{\Lambda_i}^{(l)} &\leq& {\displaystyle \max_{\boldsymbol{\beta} \underset{l}{\subseteq}x^{-1}(\Lambda_i)}}\bigg(\sum_{p\in\boldsymbol{\beta}}  r_p\bigg) ,
\cr
r_{\infty}^{(l)} &\leq& {\displaystyle \max_{\boldsymbol{\beta} \underset{l}{\subseteq}x^{-1}(\infty)}}\bigg(\sum_{p\in\boldsymbol{\beta}}  (r_p-2)\bigg) .
\eea

These upper bounds are far from optimal but sufficient for our purpose. We shall see in the examples of Section~\ref{sec-examples} how they can be improved case by case. An optimal bound can be written in general by referring to the Newton polytope of our polynomial $P(\lambda,y)$, but this would require the introduction of cumbersome notations that we choose to avoid in the present work.

\subsubsection{Bergman kernel}

A last crucial ingredient for the definition of the topological recursion is a differential form often referred to as a ``Bergman kernel''\footnote{This Bergman kernel mostly studied by Bergman and Schiffer \cite{BergmanKernel} should not be confused with the Bergman kernel classically used in operator theory.}, or also fundamental second kind differential.

\bd[Bergman kernel]\label{DefBergmanKernel}
Let $(\Sigma,x)$ be an admissible classical spectral curve of genus $g$. For any symplectic basis $\left(\mathcal{A}_i,\mathcal{B}_i\right)_{i=1}^g$ of $H_1(\Sigma,\mathbb{Z})$, let 
\beqq
B^{\left(\mathcal{A}_i,\mathcal{B}_i\right)_{i=1}^g} \in H^0(\Sigma^2, K_\Sigma^{\boxtimes 2}(2\Delta))^{\mathfrak{S}_2} \subset \mathcal{M}_2(\Sigma^2)
\eeqq
be the unique symmetric $(1\boxtimes 1)$-form on $\Sigma^2$ with a unique double pole on the diagonal $\Delta$, without residue, bi-residue equal to 1 and normalized on the $\mathcal{A}$-cycles by
\beq
\forall\, i\in\llbracket 1,g\rrbracket \, , \; \oint_{z_1 \in \mathcal{A}_i} B^{\left(\mathcal{A}_i,\mathcal{B}_i\right)_{i=1}^g}(z_1,z_2) = 0.
\eeq
\hfill $\star$
\ed

With the definitions above, we may finally define a set of initial data for the topological recursion.

\bd[Admissible initial data]\label{DefAdmissibleInitialData}

We call \emph{admissible initial data} $\left((\Sigma,x),(\mathcal{A}_i,\mathcal{B}_i)_{i=1}^g\right)$ a pair consisting of $(\Sigma,x)$ an admissible classical spectral curve of genus $g$ in the sense of Definition~\ref{DefAdmissibleSpectralCurve}, and a set of cycles, that must be chosen as follows.

\begin{itemize}

\item Let $o$ be a generic smooth point of $\Sigma\setminus x^{-1}(\mathcal P)$.

\item 
For each $p\in x^{-1}(\mathcal P)$,  let $\mathcal C_p$ be a small circle around $p$.

\item We choose a set of homology chains $\mathcal C_{o\to p}$ of boundary $\partial\mathcal C_{o\to p}=[p]-[o]$, which do not intersect each other, and such that 
\beq
 \mathcal C_{p'} \cap \mathcal C_{o\to p} =\delta_{p,p'}.
\eeq

\item A basis of $H_1(\Sigma\setminus x^{-1}(\mathcal P),\mathbb{Z})$, obtained by completing the set of $\{\mathcal C_p\}_{p\in x^{-1}(\mathcal P)}$ with $2g$ cycles denoted $\left\{\mathcal{A}_i,\mathcal{B}_i\right\}_{i=1}^g$, satisfying
\beqq
\forall\, p\in x^{-1}(\mathcal P), \ \forall\,  i \in\llbracket 1,g\rrbracket \ , \qquad
\mathcal A_i \cap \mathcal C_{o\to p}= 0= \mathcal B_i \cap \mathcal C_{o\to p},
\eeqq
\beqq
\forall \, (i,j)\in \llbracket 1,g\rrbracket^2\,:\, \mathcal A_i\cap \mathcal B_j = \delta_{i,j}
\, , \quad
\mathcal A_i\cap \mathcal A_j = 0 
\, , \quad
\mathcal B_i\cap \mathcal B_j = 0 .
\eeqq
\item The projection $H_1(\Sigma\setminus x^{-1}(\mathcal P),\mathbb{Z})\to H_1(\Sigma,\mathbb{Z}) $ sends $\left\{\mathcal{A}_i,\mathcal{B}_i\right\}_{i=1}^g$ to a symplectic basis of  $H_1(\Sigma,\mathbb{Z})$, which we shall call also $\left\{\mathcal{A}_i,\mathcal{B}_i\right\}_{i=1}^g$ by abuse of language. 
A symplectic basis of $H_1(\Sigma,\mathbb{Z})$ is called a  Torelli marking of $\Sigma$.
\end{itemize}
\hfill $\star$
\ed

\br
Admissible initial data consists of three different parts. The first part, an admissible classical spectral curve, allows in particular to define a set of ramification points together with a one form $ydx$. But it does not allow to define a Bergman kernel unambiguously. This is the reason why it has to be supplemented by a Torelli marking, i.e.~a symplectic basis $\left(\mathcal{A}_i,\mathcal{B}_i\right)_{i=1}^g$ of $H_1(\Sigma,\mathbb{Z})$, allowing to define $B^{\left(\mathcal{A}_i,\mathcal{B}_i\right)_{i=1}^g}$ without ambiguity. In our process of quantizing the classical spectral curve, this choice of Torelli marking can be thought of as a choice of polarization from a geometric quantization point of view.
The last part consisting in $\{\mathcal{C}_p,\mathcal{C}_{o\rightarrow p}\}$
leaves the $\left(\omega_{h,n}\right)_{h\geq 0,n\geq 0}$ generated by the topological recursion invariant except for $\omega_{0,0}$ fixing the normalization ambiguity in the definition of so-called wave functions giving rise to a solution of the quantum curve as we shall see later.
In order to avoid cumbersome notations, we shall note $\left((\Sigma,x),(\mathcal{A}_i,\mathcal{B}_i)_{i=1}^g\right)$ some admissible initial data, implicitly keeping in mind that they also come with a choice of contours $\{\mathcal{C}_p,\mathcal{C}_{o\rightarrow p}\}$.
\er

The choice of a Torelli marking provides natural coordinates to describe the holomorphic part of the meromorphic form $ydx$ (the polar parts of $ydx$ being described by the spectral times in Definition \ref{DefinitionSpectralTimes}).

\bd[Filling fractions]\label{FillingFractions} Let $\left((\Sigma,x),(\mathcal{A}_i,\mathcal{B}_i)_{i=1}^g\right)$ be some admissible initial data. We define $\boldsymbol{\epsilon}:=(\epsilon_i)_{i=1}^g$ the tuple of \emph{filling fractions} by
\beqq
\forall\, i\in \llbracket 1, g\rrbracket \, , \,\;\;  \epsilon_i\coloneqq\frac{1}{2 \pi i}\oint_{\mathcal{A}_i} y dx .
\eeqq
(This is well defined because $ydx$ is holomorphic on $\Sigma\setminus x^{-1}(\mathcal P)$, and $\mathcal A_i\in H_1(\Sigma\setminus x^{-1}(\mathcal P),\mathbb Z)$.)
\hfill $\star$
\ed

\br
In many cases, the set of spectral times (in fact a subset of independent spectral times not constrained by relations) together with filling fractions: $\{t_{p,k}\}_{p\in x^{-1}(\mathcal{P}), k\in\llbracket0,r_p-1\rrbracket}\bigcup \{\epsilon_i\}_{i=1}^g$ can be used to define local coordinates on the space of classical spectral curves seen as a Hurwitz space equipped with a Frobenius manifold structure (see \cite{Dubrovin-Hurwitz}). In particular, this allows to consider families of classical spectral curves with fixed pole structure and deformations relatively to these coordinates. Such point of view goes beyond the purpose of the present article in which we only consider some given and fixed admissible initial data. However, it would provide a natural framework to make the connection with isomonodromic deformations and integrable systems.
\er

\section{Topological recursion}\label{sec-tr}
In this section we recall the topological recursion construction of a family of differentials associated to some admissible initial data and give the properties that will be relevant for our purpose.
\subsection{Definition}

\bd[Topological Recursion]\label{DefTopologicalRecursion}
For any admissible initial data $\left((\Sigma,x),(\mathcal{A}_i,\mathcal{B}_i)_{i=1}^g\right)$, let us define the family of differential forms $\{\omega_{h,n}\}_{h\geq 0, n\geq 0}$ using the topological recursion \cite{EO07}
\beqq
\om_{0,1}\coloneqq ydx,
\qquad \, \qquad 
\om_{0,2} \coloneqq B^{\left(\mathcal{A}_i,\mathcal{B}_i\right)_{i=1}^g},
\eeqq
and, for $(h,n) \in\mathbb{N}\times\mathbb{N}^*$ such that $2h-2+n >0$, we define $\om_{h,n} \in H^0(\Sigma^n, K_\Sigma^{\boxtimes n}(- (6h-6+4n) \mathcal{R}_n))^{\mathfrak{S}_n} \subset \mathcal{M}_n(\Sigma^n)$\footnote{For $i\in\llbracket 1,n\rrbracket$, let $p_i:\Sigma^n \to \Sigma$ be the natural projection on the $i^{\text{th}}$ factor. We shall then denote $\mathcal{R}_n =\underset{i=1}{\overset{n}{\sqcup}} p_i^{-1}(\mathcal{R})$.} inductively by 
\beq \label{def-tr}
\om_{h,n+1}(z_0,\mathbf{z}) \coloneqq \sum_{a \in \mathcal{R}} \Res_{z \to a} \frac{1}{2} \frac{\int_{\sigma_a(z)}^z \om_{0,2}(z_0,\cdot)}{\om_{0,1}(z)-\sigma_a^* \om_{0,1}(z)} \widetilde{\mathcal{W}}_{h,n+1}^{(2)}(z,\sigma_a(z); \mathbf{z}),
\eeq
where $\sigma_a$ is the Galois involution permuting the two sheets of the cover $\Sigma$ meeting at a simple ramification point $a$ and 
\beqq
\widetilde{\mathcal{W}}_{h,n+1}^{(2)}(z,z'; \mathbf{z}) \coloneqq \om_{h-1,n+2}(z,z',\mathbf{z}) + \!\!\!\sum_{\begin{array}{c} A\sqcup B = \mathbf{z}, s\in\llbracket 0,h\rrbracket \cr (s,|A|) \notin \{(0,0),(h,n)\}\end{array}}\!\!\!\!\! \om_{s,|A|+1}(z,A) \, \om_{h-s,|B|+1}(z',B),
\eeqq
with $\mathbf{z}\coloneqq (z_1,\ldots,z_n)\in\Sigma^n$. \\
Finally, we define
\beaa
\omega_{0,0}&\coloneqq&\frac{1}{2}\left[\sum_{p\in x^{-1}(\mathcal{P})} t_{p,0} \upsilon_p- \sum_{p\in x^{-1}(\mathcal{P})} \Res_{z \to p}\hat V_p(z)y(z)dx(z)+ \sum_{i=1}^g \epsilon_i \oint_{\mathcal{B}_i} ydx  \right],\cr
\omega_{1,0}&\coloneqq&-\frac{1}{24}\log\left(\left(\tau_B(x(\mathcal{R}))\right)^{12}\prod_{a\in \mathcal{R}} Y(a)\right),\cr
\forall\, h\geq 2\,:\, \omega_{h,0}&\coloneqq&\frac{1}{2-2h}\sum_{a\in \mathcal{R}} \Res_{z \to a}  \omega_{h,1}(z)\Phi(z),
\eeaa
where $\Phi(z)$ is any anti-derivative of $\omega_{0,1}(z)$, i.e.~$d \Phi(z)=\omega_{0,1}(z)$ \footnote{Note that $(\omega_{h,0})_{h\geq 2}$, are independent of the choice of anti-derivative.}, while 
\beqq \forall\, a\in \mathcal{R}\,:\, Y(a)\coloneqq\lim_{z\to a}\frac{y(z)-y(a)}{\sqrt{x(z)-x(a)}},\eeqq
$\tau_B$ is the Bergman $\tau$-function \cite{KK} only depending on the critical values $\left(X_a=x(a)\right)_{a\in\mathcal{R}}$ defined by 
\beqq \forall\, a\in \mathcal{R}\,:\, \frac{\partial \tau_B}{\partial X_a}=\Res_{z\to a} \frac{B^{\left(\mathcal{A}_i,\mathcal{B}_i\right)_{i=1}^g}(z,\sigma_a(z))}{dx(z)},\eeqq
and
\beqq
\hat V_p(z)\coloneqq V_p(z)-t_{p,0}\log(\zeta_p) = - \sum_{k=1}^{r_p-1} \frac{t_{p,k}}{k}\zeta_p^{-k},
\eeqq
and eventually for any generic base point $o\in \Sigma$\footnote{Note that $\omega_{0,0}$ is independent of this base point.}, and $o_p$ an arbitrary point in a small disc neighborhood of $p$ in which is defined the local coordinate $\zeta_p$, 
\beqq \forall\, p\in x^{-1}(\mathcal{P})\,:\, 
\upsilon_p\coloneqq\int_{o_p}^p \left(ydx-dV_p\right)\,\, - V_p(o_p) + \int_o^{o_p} ydx, \eeqq
where the integral $o\to o_p$ is given by the homology chain $\mathcal C_{o\to p}$ chosen in admissible data. Note that $\upsilon_p$ is independent of $o_p$. See \cite{EO07, GIS} for details.
\hfill $\star$
\ed

\br To uniformize notations and names, we shall refer to the $(\omega_{h,0})_{h\geq 0}$ as 0-forms (indeed they are not differentials but rather complex numbers).
\er

\br There exists a generalization \cite{BouchardEynard} of the topological recursion that allows for ramification points of higher orders. However, since we chose to consider only admissible classical spectral curves in the present article, we do not need it in the present setup. 
\er

\subsection{Properties of differentials produced by the topological recursion}
The differentials $(\om_{h,n})_{h,n\geq 0}$ satisfy many properties \cite{EO07}. Let us review a few of them that are useful for the content of the article.
\begin{itemize}
\item For $h\geq 0, n\geq 2$, the differentials $\om_{h,n}$ are invariant under permutations of their $n$ arguments, i.e.~they are differentials on the $n$-fold symmetric product of the curve $\Sigma$ in $\mathcal{M}_n(\Sigma^n)$.

\item $\om_{0,1}(z_1)$ may only have poles at $x^{-1}(\mathcal{P})$. $\omega_{0,2}(z_1,z_2)$ may only have poles at $z_1=z_2$. For $(h,n)\in\mathbb{N}\times \mathbb{N}^* \setminus\{(0,1), (0,2)\}$, $\om_{h,n}(z_1,\ldots,z_n)$ may only have pole at $z_i\in\mathcal{R}$, for $i\in\llbracket 1,n\rrbracket$.
\item The differentials $(\om_{h,n})_{h,n\geq0}$ may also be considered as functions of the parameters defining the admissible classical spectral curve, namely as functions on the corresponding Hurwitz space (see \cite{DNOPS18,GIS} for example for a more detailed explanation). In particular, one may consider locally the space of classical spectral curves obtained by varying the values of periods $\{\epsilon_i\}_{i=1}^{g}$. To some extent, for any $(h,n)\in \mathbb{N}^2$, $\om_{h,n+1}$ can be understood as a generating series for the derivatives of $\om_{h,n}$ with respect to these parameters and one has the following result
\beq\label{eq-var-epsilon}
\forall\, (h,n)\in \mathbb{N}^2, \, \forall\, i\in \llbracket1, g\rrbracket \, : \, \frac{\partial}{\partial \epsilon_i}  \om_{h,n}(z_1,\dots,z_n) = \oint_{z \in \mathcal{B}_i} \om_{h,n+1}(z,z_1,\dots,z_n) .
\eeq
\end{itemize}

\subsubsection{Ramification points at poles}

In the definition of topological recursion, we take residues at $a\in\mathcal R$, i.e.~only at ramification points that are zeros of $dx$ and that are not in $x^{-1}(\mathcal P)$. However the points of $\mathcal P$ could also be ramified (and this is the case for many interesting examples of spectral curves, for instance the Airy curve $y^2=x$). In \cite{BouchardEynard_QC} it was noticed that to derive the quantum curve, one should also include residues at all ramification points, including those in $x^{-1}(\mathcal P)$.
However, thanks to our hypothesis on admissible spectral curve, and thanks to the following lemma, including or not such residues makes no difference.

\bl[Ramified points at poles]\label{lemRamptsPoles}
Let $(\omega_{h,n}')_{h,n\geq 0}$ be the topological recursion differential forms defined by taking residues at all $a\in\mathcal{R}_0$ (i.e.~all ramification points, so including $a\in x^{-1}(\mathcal P)$). If, for all ramification points $p\in x^{-1}(\mathcal P)$, we have $r_p\geq 3$ and  $t_{p,r_p-2}\neq 0$, then  $\omega_{h,n}'=\omega_{h,n}$ for all $(h,n) \in \mathbb{N}^2$.  
Moreover, $\omega_{h,n}$ with $(h,n)\neq(0,1),(0,2)$ have poles only at $\mathcal R=\mathcal{R}_0\setminus x^{-1}(\mathcal P)$.
\el

\begin{proof}
In Appendix \ref{Applemmaramificationpointspoles}. The proof relies on the fact that there is a denominator with $ydx$ in the residue formula, and if the pole of $ydx$ is of degree $\geq 3$, it cannot be compensated by other factors, and the residue vanishes because the integrand has no pole.
\end{proof}

\section{Loop equations}\label{sec-loop-eq}

As explained in equation \eqref{PlConstraints}, the rational functions $\left(P_l(\lambda)\right)_{l\in\llbracket 1,d\rrbracket}$ are the  $l^{\text th}$ elementary symmetric polynomials of the $\om_{0,1}(z)$ for $z\in x^{-1}(\lambda)$. There exist other symmetric algebraic combinations of the differentials $(\om_{h,n})_{h,n\geq 0}$ taken at all preimages $x^{-1}(\lambda)$ that give rise to interesting rational functions. This fundamental result is commonly referred to as a set of \emph{loop equations} to recall its origin from the study of Hermitian random matrices. This section is devoted to the derivation of a specific set of $d$ loop equations.

\subsection{Definitions}

In order to write down the loop equations, let us introduce a few convenient notations following the ones of \cite{BouchardEynard_QC}.

\bd \label{def-Q}
Given admissible initial data and the corresponding $\{\om_{h,n}\}_{h, n\geq 0}$ defined by the topological recursion, let us define, for any tuple of integers $(h,n,l)\in\mathbb{N}^3$,
\beaa
Q_{h,n+1}^{(0)}(\lambda;\mathbf{z})&\coloneqq&\delta_{h,0}\delta_{n,0},\cr
\forall\, l\geq 1\,:\, Q_{h,n+1}^{(l)}(\lambda;\mathbf{z})&\coloneqq& \sum_{\beta \underset{l}{\subseteq} x^{-1}(\lambda)} \sum_{\mu \in \mathcal{S}(\beta)} \, {\displaystyle \sum_{\underset{i=1}{\overset{l(\mu)}{\bigsqcup}} J_i = \mathbf{z}} }\,  \sum_{ \underset{i=1}{\overset{l(\mu)}{\sum}} g_i= h+l(\mu)-l}
\left[
\prod_{i=1}^{l(\mu)} \om_{g_i,|\mu_i|+|J_i|}(\mu_i,J_i) \right],
\eeaa
where $\lambda$ is a point in $\mathbb{P}^1$ and $\mathbf{z} \coloneqq (z_1,\dots, z_n)$ a vector of $n$ points in $\Sigma$.
We remind the reader that $\mathcal{S}(\beta)$ denotes the set of partitions of  the set $\beta$ while $l(\mu)$ denotes the number of blocks of the partition $\mu\in\mathcal{S}(\beta)$. 
The expression $Q_{h,n+1}^{(l)}(\lambda;\mathbf{z})$ is a differential in $\lambda$ and $(z_i)_{i\in \llbracket 1, n\rrbracket}$, with possible poles at $\lambda \in \mathcal{P} \cup x(\mathcal{R})$, $z_i \in \mathcal{R}$ and when $z_i \in x^{-1}(\lambda)$.  
\hfill $\star$
\ed

\br Observe that by definition, we have for all $(h,n)\in \mathbb{N}^2$,
\beq
Q_{h,n+1}^{(l)}(\lambda;\mathbf{z}) = 0, \text{ for } l \geq d+1.
\eeq
\er

\br
Observe that since we take a sum over all subsets of $l$ preimages $\beta\underset{l}{\subseteq} x^{-1}(\lambda)$, this is clearly a symmetric function of the preimages, and thus is obviously a rational fraction of $\lambda$ multiplied by $(d\lambda)^l$.
\er

\bd\label{def-hat-Q}
In the same way, given admissible initial data and the corresponding $\{\om_{h,n}\}_{h, n\geq 0}$ defined by the topological recursion, let us define, for any tuple $(h,n,l)\in\mathbb{N}^3$,
\beaa \hat{Q}_{h,n+1}^{(0)}(z;\mathbf{z})&\coloneqq&\delta_{h,0}\delta_{n,0},\cr
\forall\, l\geq 1\,:\, \hat{Q}_{h,n+1}^{(l)}(z;\mathbf{z})&\coloneqq& \sum_{\beta \underset{l}{\subseteq} \left(x^{-1}(x(z)) \setminus \{z\}\right)} \sum_{\mu \in \mathcal{S}(\beta)} \, \sum_{ \underset{i=1}{\overset{l(\mu)}{\bigsqcup}}  J_i = \mathbf{z}} \,  \sum_{ \underset{i=1}{\overset{l(\mu)}{\sum}} g_i = h+l(\mu)-l}
\left[
\prod_{i=1}^{l(\mu)} \om_{g_i,|\mu_i|+|J_i|}(\mu_i,J_i) \right],
\eeaa
where one only considers points in the fiber above $x(z)$ that are different from the point $z$ on $\Sigma$. The possible poles of these differential forms are at $z$ such that $x(z)\in x(\mathcal{R})$ and $z \in x^{-1}(\mathcal{P})$, and at $z_i \in \mathcal{R} \cup \left(x^{-1}\left(x(z)\right) \setminus \{z\}\right) $.
\hfill $\star$
\ed

\br Remark that we also have by definition that
\beq
\hat{Q}_{h,n+1}^{(l)}(z;\mathbf{z}) = 0, \text{ for } l \geq d.
\eeq
\er

\br\label{abuse_order}
In the previous definitions, note that, even if $\mu_i$ and $J_i$ are sets, the evaluation $\omega_{g_i,|\mu_i|+|J_i|}(\mu_i,J_i)$ is well-defined because the differentials $(\omega_{h,n})_{h,n\geq 0}$ are symmetric in their $n$ arguments.
\er

Notice that, with respect to their first variable, $\hat{Q}_{h,n+1}^{(l)}(z;\mathbf{z})$ is a differential on $\Sigma$ while $ Q_{h,n+1}^{(l)}(\lambda;\mathbf{z})$ is a differential on the base curve $\mathbb{P}^1$.
Let us finally define the generating functions of the previous differentials $\left\{Q_{h,n+1}^{(l)}\right\}_{(h,n,l)\in \mathbb{N}^3}$ and $\left\{\hat{Q}_{h,n+1}^{(l)}\right\}_{(h,n,l)\in \mathbb{N}^3}$.

\bd \label{def-Q-x-y} Given admissible initial data, we define, for any $(h,n)\in \mathbb{N}^2$,
\bea 
Q_{h,n+1}(\lambda,y;\mathbf{z})&\coloneqq& \sum_{l\geq 1} (-1)^l y^{d-l} \frac{Q_{h,n+1}^{(l)}(\lambda;\mathbf{z})}{(d\lambda)^l}=\sum_{l=1}^d (-1)^l y^{d-l} \frac{Q_{h,n+1}^{(l)}(\lambda;\mathbf{z})}{(d\lambda)^l},\cr
\hat{Q}_{h,n+1}(z,y;\mathbf{z})&\coloneqq& \sum_{l \geq 1} (-1)^l y^{d-l} \frac{\hat{Q}_{h,n+1}^{(l)}(z;\mathbf{z})}{(dx(z))^l}=\sum_{l =1}^d (-1)^l y^{d-l} \frac{\hat{Q}_{h,n+1}^{(l)}(z;\mathbf{z})}{(dx(z))^l} .
\eea
\hfill $\star$
\ed

\subsection{Examples}

In order to better understand the origin of these definitions, let us write down explicitly the expression of $Q_{h,n}^{(l)}$ for some low values of $h$, $n$ or $l$. 
\begin{itemize}\item For $(h,n) = (0,0)$, we recover the coefficients of the polynomial equation defining the classical spectral curve,
\beq\label{def-Q01}
\forall\, l \geq 1\, , \; Q_{0,1}^{(l)}(\lambda) =\sum_{\beta \underset{l}{\subseteq} x^{-1}(\lambda)}\prod_{z\in\beta}\omega_{0,1}(z)=  P_l(\lambda) \left(d\lambda\right)^l .
\eeq
$(P_l)_{l\in \llbracket 1,d\rrbracket}$ are the elementary symmetric functions of the roots of the algebraic equation $P(\lambda,y) = 0$ in $y$. Note that for all $l\geq 1$, $Q_{0,1}^{(l)}(\lambda)$ may only have poles at $\lambda\in\mathcal{P}$.

\item For $(h,n) = (0,1)$, we obtain
\beq\label{def-Q02}
\forall\, l \geq 1\, , \; Q_{0,2}^{(l)}(\lambda;z_1) =\sum_{\beta \underset{l}{\subseteq} x^{-1}(\lambda)} \sum_{z\in\beta}\omega_{0,2}(z,z_1)\prod_{\substack{\tilde{z}\in\beta\\ \tilde{z}\neq z}}\omega_{0,1}(\tilde{z}).
\eeq
Note that for all $l\geq 1$,  $Q_{0,2}^{(l)}(\lambda;z_1)$ may only have poles at $\lambda\in\mathcal{P}$ (because of the factors involving $\omega_{0,1}$) and when $z_1 \in x^{-1}(\lambda)$ (because of the factors involving $\omega_{0,2}$).

\item For $l=1$ and $(h,n)\in \mathbb{N}^2$, we obtain
\beq\label{l1LE}
Q_{h,n+1}^{(1)}(\lambda;\mathbf{z})=\sum_{z \in x^{-1}(\lambda)} \om_{h,n+1}(z,\mathbf{z}).
\eeq
\end{itemize}

\subsection{Loop equations}

Definition \ref{def-Q} and the properties of the $\{\om_{h,n}\}_{h, n\geq 0}$ allow to prove one of the fundamental properties of the objects built by topological recursion: the \emph{loop equations}.

\begin{theorem}[Loop equations~\cite{BouchardEynard_QC}]\label{loopeq}
For any $(h,n,l)\in \mathbb{N}^3$ and any $\mathbf{z}\in \left(\Sigma \setminus \mathcal{R}\right)^n$, the function $\lambda \mapsto \frac{Q_{h,n+1}^{(l)}(\lambda;\mathbf{z})}{(d\lambda)^l}$  has no poles at critical values.
\end{theorem}

\begin{proof}
The proof of this theorem was already given in \cite{BouchardEynard_QC} for the case when $\Sigma$ has genus 0. However, the proof does not use the vanishing genus assumption and can be directly transposed to the higher genus case. For the sake of completeness, we reproduce in Appendix~\ref{appendix-proof-loopeq} the main steps of this proof for the general case at hand.
\end{proof}

These loop equations take a particularly simple form for $l=1$, as we saw in equation \eqref{l1LE}. As a consequence of Theorem~\ref{loopeq}, they read
\beq\label{eq-lin-loop-eq}
\sum_{z \in x^{-1}(\lambda)} \om_{h,n+1}(z,\mathbf{z}) = \delta_{n,0} \delta_{h,0} P_1(\lambda) d\lambda + \delta_{n,1} \delta_{h,0} \frac{ d \lambda \, dx(z_1)}{(\lambda-x(z_1))^2}\,,
\eeq
and provide interesting properties for the $\{\om_{h,n}\}_{h\geq 0,n\geq 1}$.

Let us emphasize that, for any $(h,n,l)\in \mathbb{N}^2\times \mathbb{N}^*$, $ \frac{Q_{h,n+1}^{(l)}(\lambda;\mathbf{z})}{(d\lambda)^l}$ is a rational function of $\lambda$. From this perspective, it can be expressed through its partial fraction decomposition. We shall now describe its pole structure as a first step. A priori, from its definition, it could have poles at $\lambda \in x(\mathcal{R})$, $\lambda \in \mathcal{P}$ or $\lambda \in\underset{i=1}{\overset{n}{\bigcup}} \{x(z_i)\} $. The loop equations allow us to rule out the possibility of having any pole at $\lambda  \in x(\mathcal{R})$. Let us now study the behavior of this rational function as $\lambda \to x(z_i)$, for $i \in \llbracket 1,n \rrbracket$. It is governed by the following lemma.

\begin{lemma}[\cite{BouchardEynard_QC}]\label{def-Q-tilde-lem} For all $(h,n,l)\in \mathbb{N}^3$, and $\mathbf{z}\coloneqq(z_1,\dots,z_n) \in \Sigma^n$ such that $x(z_i) \neq x(z_j)$ for any $i \neq j$, the functions
\bea\label{def-Q-tilde}
\widetilde{Q}_{h,n+1}^{(0)}(\lambda;\mathbf{z})&\coloneqq&Q_{h,n+1}^{(0)}(\lambda;\mathbf{z})=\delta_{h,0}\delta_{n,0},\cr
\forall\, l\geq 1\,:\, \widetilde{Q}_{h,n+1}^{(l)}(\lambda;\mathbf{z})&\coloneqq& \frac{Q_{h,n+1}^{(l)}(\lambda;\mathbf{z})}{(d\lambda)^l}
 - \sum_{j=1}^n d_{z_j}\left(\frac{1}{\lambda-x(z_j)}\frac{\hat{Q}_{h,n}^{(l-1)}(z_j;\mathbf{z} \setminus \{z_j\})}{(dx(z_j)^{l-1}}\right)
\eea
are rational functions of $\lambda$ that have no poles at critical values and no poles at $\lambda\in \underset{i=1}{\overset{n}{\bigcup}} \{x(z_i)\} $. 
\end{lemma}

\begin{proof}
Let us rewrite here the proof of \cite{BouchardEynard_QC}.
The proof is obvious for $l=0$. 
Let us now consider $l\geq 1$. 
$\widetilde{Q}_{h,n+1}^{(0)}(\lambda;\mathbf{z}) $ is obviously a rational function of $\lambda$.
From Theorem~\ref{loopeq}, $\frac{Q_{h;n+1}^{(l)}(\lambda;\mathbf{z})}{(d\lambda)^l}$ has no pole at critical values $\lambda \in x(\mathcal{R})$. Let us now study its behavior as $\lambda \to x(z_j)$, for some $j \in \llbracket 1, n \rrbracket$. 
If $z_j$ is generic,  $\lambda\to x(z_j)$ implies that exactly one of the preimages $x^{-1}(\lambda)=\{z^{(1)},z^{(2)},\dots,z^{(d)}\}$ of $\lambda$ tends to $z_j$, let us say it is $z^{(1)}\to z_j$.
In all the terms in \eqref{def-Q}, the only term that can have a pole at $z^{(1)}\to z_j$ is a factor $\omega_{0,2}(z^{(1)},z_j)$. In particular notice that $\omega_{0,2}(z^{(i)},z_j)$ with $i\neq 1$  has no pole at $z^{(1)}\to z_j$.
The only singular term at $z^{(1)}\to z_j$ in the sum of equation \eqref{def-Q} defining $Q_{h,n+1}^{(l)}(\lambda;z_1,\dots,z_n)$ is thus
\beq
\frac{Q_{h,n+1}^{(l)}(x(z^{(1)});z_1,\dots,z_n)}{(d\lambda)^l}  = \om_{0,2}(z^{(1)},z_j) \frac{\hat{Q}_{h,n}^{(l-1)}(z^{(1)};\mathbf{z} \setminus \{z_j\})}{(d\lambda)^l} + O(1).
\eeq
Let us write it as

\bea
 &&\om_{0,2}(z^{(1)},z_j) \frac{\hat{Q}_{h,n}^{(l-1)}(z^{(1)};\mathbf{z} \setminus \{z_j\})}{(d\lambda)^l} \cr
&=&  \frac{d\lambda dx(z_j)}{(\lambda-x(z_j))^2}  \frac{\hat{Q}_{h,n}^{(l-1)}(z^{(1)};\mathbf{z} \setminus \{z_j\})}{(d\lambda)^l} + \left(\om_{0,2}(z^{(1)},z_j)-\frac{d\lambda dx(z_j)}{(\lambda-x(z_j))^2}\right)  \frac{\hat{Q}_{h,n}^{(l-1)}(z^{(1)};\mathbf{z} \setminus \{z_j\})}{(d\lambda)^l} \cr
&=&  \frac{dx(z_j)}{(\lambda-x(z_j))^2}  \frac{\hat{Q}_{h,n}^{(l-1)}(z^{(1)};\mathbf{z} \setminus \{z_j\})}{(d\lambda)^{l-1}} +O(1) \cr
&=&  dz_j\left( \frac{1}{\lambda-x(z_j)}  \frac{\hat{Q}_{h,n}^{(l-1)}(z^{(1)};\mathbf{z} \setminus \{z_j\})}{(d\lambda)^{l-1}} \right) +O(1) \cr
&=& dz_j\left( \frac{1 }{\lambda-x(z_j)}  \frac{\hat{Q}_{h,n}^{(l-1)}(z_j;\mathbf{z} \setminus \{z_j\})}{(dx(z_1)^{l-1}} \right) +O(1) \cr
&& +  dz_j\left( \frac{1}{\lambda-x(z_j)}  \left( \frac{\hat{Q}_{h,n}^{(l-1)}(z^{(1)};\mathbf{z} \setminus \{z_j\})}{(d\lambda)^{l-1}} - \frac{\hat{Q}_{h,n}^{(l-1)}(z_j;\mathbf{z} \setminus \{z_j\})}{(dx(z_1)^{l-1}}\right) \right) +O(1) \cr
&=&  dz_j\left( \frac{1 }{\lambda-x(z_j)}  \frac{\hat{Q}_{h,n}^{(l-1)}(z_j;\mathbf{z} \setminus \{z_j\})}{(dx(z_1)^{l-1}} \right) +O(1) .
\eea
This proves that $\widetilde{Q}_{h,n+1}^{(0)}(\lambda;\mathbf{z}) =O(1) $, i.e. has no pole at $\lambda=x(z_j)$.
\end{proof}

For future reference, let us notice that for all $(h,n,l)\in \mathbb{N}^3$, $\widetilde{Q}_{h,n+1}^{(l)}(\lambda;\mathbf{z})$ is a rational function of $\lambda$, and a differential 1-form in $z_i$ for all $i\in \llbracket 1,n\rrbracket$.

\section{Perturbative wave functions and KZ equations}\label{KZ-eqs}

In this section we encode all the topological recursion differentials into objects called perturbative wave functions and we use the loop equations from the previous section to build a system of operators that annihilate these wave functions. These operators can often be considered as derivatives with respect to the moduli (spectral times and filling fractions) and hence provide a system of PDE's. However, in the present article we deal with the general case, for which this interpretation cannot be made straightforward without dealing with the issue of studying the geometry of the space of classical spectral curves. To avoid this technical difficulty, we prove that these general operators can be interpreted as integrals over generalized cycles of the classical spectral curve. 

\subsection{Perturbative wave functions and partition function}\label{5.1}

Let $\left((\Sigma,x),(\mathcal{A}_i,\mathcal{B}_i)_{i=1}^g\right)$ be some admissible initial data.

\subsubsection{Divisors and universal cover}

The perturbative wave functions that will be defined in this section will not be defined as functions of a point on $\Sigma$, but rather of a point on the universal cover of $\Sigma_{\mathcal P}=\Sigma\setminus x^{-1}(\mathcal P)$, and in fact, of a set of such points, encoded in a divisor. The differential forms that we will integrate are the TR differentials $(\omega_{h,n})_{h,n\geq 0}$, that are known to have vanishing residues at every pole $p\notin x^{-1}(\mathcal P)$. More concretely, $\omega_{0,1}$ only has poles at $\mathcal P$, $\omega_{0,2}$ has a pole with vanishing residue and all the other $(\omega_{h,n})_{h,n\geq 0}$ have poles only at ramification points, with vanishing residues \cite{EO07}.

\bd[Universal cover]
Let $\widetilde{\Sigma_{\mathcal P}}$ be the universal cover of $\Sigma_{\mathcal P}=\Sigma\setminus x^{-1}( \mathcal{P})$, i.e. the set of homotopy classes of Jordan arcs from a generic basepoint $o$ to $z\in\Sigma_\mathcal{P}$.
There is a projection $\pi:\widetilde{\Sigma_\mathcal{P}} \to \Sigma_\mathcal{P}$. Every differential form $\omega$ on $\Sigma$ has a pullback $\pi^*\omega$ on the universal cover, and local coordinates on $\Sigma$ can be used as local coordinates on the universal cover.
Therefore, in the purpose of lightening the notations of this article, and as long as there is no possible ambiguity, which will be mentioned if needed, we shall use the same notation $p$ for a point of the universal cover and its projection $\pi(p)$ on $\Sigma$, and we shall use the same notation $\omega$ for a form on $\Sigma$ and its pullback $\pi^*\omega$ on the universal cover.
\hfill $\star$
\ed

\bd[Divisors]\label{GenericDivisors}  Let $\Sigma$ be a  Riemann surface.
A divisor $D$ on $\Sigma$ is a tuple of pairs $((p_1,\alpha_1),\dots,(p_s,\alpha_s)) \in (\Sigma\times \mathbb C)^{s}$, for some $s\in\mathbb{N}$, and is denoted $D=\underset{i=1}{\overset{s}{\sum}}\alpha_i [p_i]$. We call $\alpha_i$ the \emph{charge} of the point $p_i$. 
The \emph{degree} of a divisor is defined as the sum of its charges
\beqq \deg D=\underset{i=1}{\overset{s}{\sum}} \alpha_i .\eeqq
The set of points $p_i$'s for which $\alpha_i\neq 0$ is called the \emph{support} of the divisor.
A divisor is said \emph{generic} if all points are pairwise distinct: $p_i\neq p_j$, for all $i\neq j$.
\hfill $\star$
\ed

\bd
Let $D$ be a divisor of $s$ points on the universal cover $(p_1,\dots,p_s)\in \widetilde{\Sigma_\mathcal{P}}^s$, of degree $0$.
If $\omega$ is a meromorphic 1-form on $\Sigma$ with possible poles at $\mathcal{P}$, the integration on the divisor $D$ is defined as
\beqq
\int_D \omega(z) \coloneqq \sum_{i=1}^{s} \alpha_i \int_o^{p_i} \omega(z),
\eeqq
where $o$ is the basepoint used for the universal cover. Since the degree of the divisor is zero, the integral does not depend on the choice of the base point $o$. 
\hfill $\star$
\ed

\subsubsection{Perturbative wave functions}

\bd[Perturbative wave functions]\label{DefPerturbativeWaveFunctions} Let $\left((\Sigma,x),(\mathcal{A}_i,\mathcal{B}_i)_{i=1}^g\right)$ be some admissible initial data and $D=\underset{i=1}{\overset{s}{\sum}}\alpha_i [p_i]$ a generic divisor on $\widetilde {\Sigma_{\mathcal P}}$. We define the \emph{perturbative wave functions} associated to $D$ by
\beq
\psi(D,\hbar) \coloneqq \exp{\bigg( {\displaystyle  \sum_{h \geq 0} \sum_{n\geq 0} }\frac{\hbar^{2h-2+n}}{n!}\int_D\cdots\int_D \Big(\omega_{h,n}(z_1,\ldots,z_n) - \delta_{h,0} \delta_{n,2} \frac{dx(z_1) dx(z_2)}{(x(z_1)-x(z_2))^2} \Big)\bigg)} .
\eeq
It is defined as a formal power series times exponential terms of the form $e^{-\hbar^{-k} w_k}$. Namely,
\beq
e^{-\hbar^{-2}\omega_{0,0}}e^{-\hbar^{-1}\int_D \omega_{0,1}}\psi(D,\hbar) \in \mathbb C[[\hbar]].
\eeq

More generally, for any $l\geq 1$,
\bea  \label{psilidef}\forall\, i\in \llbracket 1,s\rrbracket\,:\, \psi_{0,i}(D,\hbar)&\coloneqq& \psi(D,\hbar) ,\cr
\forall\, i\in\llbracket 1,s\rrbracket\,,\, l\geq 1\,:\, 
\psi_{l,i}(D,\hbar) &\coloneqq& \bigg[\sum_{h\geq 0} \sum_{n \geq 0} \frac{\hbar^{2h+n}}{n!}  \overbrace{\int_D\cdots\int_D}^n \frac{\hat{Q}_{h,n+1}^{(l)}(p_i;\cdot)}{(dx(p_i))^l} \bigg] \psi(D,\hbar) .
\eea
These are also formal power series in the same space of formal series
\beq
e^{-\hbar^{-2}\omega_{0,0}}e^{-\hbar^{-1}\int_D \omega_{0,1}}\psi_{l,i}(D,\hbar) \in \mathbb C[[\hbar]].
\eeq
\hfill $\star$
\ed

To make notations lighter, we will often omit the divisor $D$ and the formal parameter $\hbar$ as arguments of the perturbative wave functions in expressions where the dependence does not play any significant role.

The first few orders are
\beq
\psi(D,\hbar) = e^{\hbar^{-2}\omega_{0,0}+\omega_{1,0}} \ e^{\hbar^{-1}\int_D \omega_{0,1}} \ \prod_{i<j}\left(\frac{E(p_i,p_j)}{E_0(x(p_i),x(p_j))}\right)^{a_i a_j }
\left(1+O(\hbar)\right),
\eeq
where $E$ is the prime form on $\Sigma$ and $E_0$ is the prime form on the base $\mathbb P^1$,
\beq
E_0(x,x') = \frac{x-x'}{\sqrt{dx dx'}}.
\eeq

\br
For book-keeping, let us explain where the exponent of $\hbar$ comes from. First, the TR differentials are homogeneous under the rescaling $\omega_{0,1}\to \hbar\omega_{0,1}$ (i.e.~all spectral times $t_{p,k}\to \hbar t_{p,k}$ and filling fractions $\epsilon_i\to\hbar \epsilon_i$); they change to
\beq
\omega_{h,n} \to \hbar^{2h-2+n}\omega_{h,n}.
\eeq
This explains the power $\hbar^{2h-2+n}$ in the definition of $\psi(D,\hbar)$.

Our goal is then to find an ODE satisfied by $\psi(D,\hbar)$, and the purpose of $\psi_{l,i}(D,\hbar)$ is to play a role similar to $(\hbar d_{p_i})^l\psi$, i.e.~it is a differential of $\psi$ multiplied by the power $\hbar^l$.
$\hat Q_{h,n+1}^{(l)}$ is made from the product of $l$ TR differentials, whose total homogeneity weight is $\underset{i=1}{\overset{l}{\sum}}(2h_i-2+|\mu_i|+|J_i|) = 2h+n-l$.
Multiplying by $\hbar^l$ gives an order $\hbar^{2h+n}$ in \eqref{psilidef}.
\er

It will also be useful to introduce an element which is fundamental in the theory of topological recursion: the partition function, which is independent of any divisor.

\bd[Perturbative partition function] \label{def-pert-part-function}
 Let $\left((\Sigma,x),(\mathcal{A}_i,\mathcal{B}_i)_{i=1}^g\right)$ be some admissible initial data. We define the \emph{perturbative partition function} as the perturbative wave function of the empty divisor $ Z(\hbar)=\psi(D=\emptyset,\hbar)$, namely
 \beq
 Z(\hbar) \coloneqq \exp{\Bigg( {\displaystyle  \sum_{h \geq 0}}  \hbar^{2h-2} \om_{h,0} \Bigg)}.
 \eeq
Its behavior in $\hbar$ is such that
\beq
e^{-\hbar^{-2} \omega_{0,0}} Z(\hbar) \in \mathbb C[[\hbar]].
\eeq
 \hfill $\star$
\ed
 
While the wave functions are meant to be solutions to a differential equation, the partition function is expected to play the role of an associated tau function from the point of view of isomonodromic or integrable systems.

\subsection{KZ equations for generic divisors}\label{5.2}

The infinite set of loop equations (Theorem \ref{loopeq}) involving all values of  $(h,n,l)\in \mathbb{N}^3$ can be combined into a generating series to form a system of $d\times s$ differential equations satisfied by the wave functions. When interpreting the topological recursion from a conformal field theory perspective \cite{GIS,CFT_TR}, these equations play the role of Knizhnik--Zamolodchikov equations in this context. Hence, we refer to them as \emph{KZ equations} in the rest or the paper.

Using Definition~\ref{DefPerturbativeWaveFunctions} we deduce the following lemma. For completeness, and since this is a key step to get the KZ equations, we include below the idea of the proof, although it already appeared in \cite{BouchardEynard_QC}.

\begin{lemma} \label{lemma-Q-omega} For any $z\in \Sigma \setminus \left(\mathcal{R} \bigcup x^{-1}(\mathcal{P})\right)$ and $\mathbf{z} \in \left[\Sigma \setminus \left(\mathcal{R} \bigcup x^{-1}(x(z))\right)\right]^n$, we have, for all $(h,n,l)\in \mathbb{N}^3$,
\beq
Q_{h;n+1}^{(l)}(x(z);\mathbf{z})
= \hat{Q}_{h;n+1}^{(l)}(z;\mathbf{z})
+ \hat{Q}_{h-1;n+2}^{(l-1)}(z;z,\mathbf{z})
+\!\!\! \sum_{A \sqcup B= \mathbf{z}} \sum_{h_1+h_2=h} \hat{Q}_{h_1,|A|+1}^{(l-1)}(z;A) \om_{h_2,|B|+1}(z,B).
\eeq
\end{lemma}

\begin{proof}
The LHS is defined as a sum over $\beta \subset x^{-1}(x(z))$ in Definition~\ref{def-Q}. This sum  can be split into three subsets giving rise to the three terms in the RHS as follows.
\begin{itemize}
\item The terms with $z \notin \beta$ give rise to the first term $ \hat{Q}_{h;n+1}^{(k)}(z;\mathbf{z})$ in the RHS.
\item The terms with $z \in \beta$ such that the block $\mu_i$ of the set partition $\mu\in\mathcal{S}(\beta)$ with $z\in\mu_i$  does not have cardinality equal to 1, $|\mu_i| \neq 1$, contribute to the second term $ \hat{Q}_{h-1;n+2}^{(k-1)}(z;z,\mathbf{z})$ of the RHS.
\item The terms  with $z \in \beta$ such that the block $\mu_i $ containing $z$ has cardinality $1$
give rise to the last terms.
\end{itemize}
\end{proof}

This lemma allows to prove the main property satisfied by the perturbative wave functions.

\begin{theorem}[General KZ equations]\label{systemPDE} Let $\left((\Sigma,x),(\mathcal{A}_i,\mathcal{B}_i)_{i=1}^g\right)$ be some admissible initial data and $D=\underset{i=1}{\overset{s}{\sum}}\alpha_i [p_i]$ a generic divisor. 
For $i\in\llbracket 1,s\rrbracket$ and $l\in \llbracket 0,d-1\rrbracket$, we have
\beaa
\frac{\hbar}{\alpha_i} \frac{d \psi_{l,i}(D,\hbar)}{dx(p_i)} &=& - \psi_{l+1,i}(D,\hbar) - \hbar \sum_{j \in \llbracket 1,s \rrbracket \setminus \{i\}} \alpha_j \,\frac{\psi_{l,i}(D,\hbar)-\psi_{l,j}(D,\hbar)}{x(p_i)-x(p_j)} \cr
&& + {\displaystyle \sum_{h\geq 0} \sum_{n \geq 0}}\frac{ \hbar^{2h+n}}{n!} \int_{z_1 \in D} \dots \int_{z_n \in D} \widetilde{Q}_{h,n+1}^{(l+1)}(x(p_i);\mathbf{z}) \; \psi(D,\hbar)  \cr
&&+ \left(\frac{1}{\alpha_i} - \alpha_i\right) \Bigg[\sum_{(h,n)\in \mathbb{N}^2} \frac{\hbar^{2h+n+1}}{n!}  \overbrace{\int_D\cdots\int_D}^{n} \frac{d}{dx(p_i)} \bigg(\frac{ \hat{Q}_{h,n+1}^{(l)}(p_i;\cdot)}{\left(dx(p_i)\right)^l}\bigg) \Bigg]  \; \psi(D,\hbar). \cr
\eeaa
Note that if $\alpha_i = \pm 1$, this reduces to
\bea\label{simple_systemPDE}
\frac{\hbar}{\alpha_i} \frac{d \psi_{l,i}(D,\hbar)}{dx(p_i)} &=& - \psi_{l+1,i}(D,\hbar) -  \hbar \sum_{j \in \llbracket 1,s \rrbracket \setminus \{i\}} \alpha_j \,\frac{\psi_{l,i}(D,\hbar)-\psi_{l,j}(D,\hbar)}{x(p_i)-x(p_j)} \cr
&& + {\displaystyle \sum_{h\geq 0} \sum_{n \geq 0}}\frac{ \hbar^{2h+n}}{n!} \int_{z_1 \in D} \dots \int_{z_n \in D} \widetilde{Q}_{h,n+1}^{(l+1)}(x(p_i);\mathbf{z}) \; \psi(D,\hbar)  .\cr
&&
\eea
\end{theorem}

\begin{proof} 
See Appendix~\ref{proof_systemPDE}.
\end{proof}

\subsection{KZ equations for $D=[z]-[\infty^{(\alpha)}]$}\label{5.3}

For some special choices of divisor $D$, one can recombine these equations in order to obtain interesting PDE's satisfied by the perturbative wave function $\psi(D,\hbar)$. For a generic divisor $D$, these PDE's include differential operators $\frac{\partial^k}{\partial \lambda^k}$ with $k > d$ (generically up to $d^2$). These equations are interesting on their own and discussed in \cite{BouchardEynard_QC} for a genus 0 curve and in \cite{GIS} for arbitrary genus. However, our aim here is to build a quantum curve. For this purpose, one is looking for divisors such that the PDE obtained involves only $\frac{\partial^k}{\partial \lambda^k}$ with $k \leq d$. There exist different ways to obtain such a divisor $D$\footnote{See for example \cite{MO19_hyper} for a choice of divisor with two points in the fiber above a common $\lambda$ in the case $d=2$.}. In the present paper, we shall consider a divisor of the type $D = [z]-[\infty^{(\alpha)}]$, with $\alpha\in \llbracket 1,d\rrbracket$. 
However, a divisor $D = [z]-[\infty^{(\alpha)}]$ with $\alpha\in \llbracket 1,d\rrbracket$ is not generic because $\infty^{(\alpha)}\in x^{-1}(\mathcal{P})$, and the purpose of this section is to prove how we may extend the previous generic results to this specific case by considering the limit of generic divisors $D=[z]-[p_2]$, when $p_2\to \infty^{(\alpha)}$. 

\bigskip

Let $z \in \widetilde{\Sigma_\mathcal{P}}$ be a generic point. Let $\alpha\in \llbracket 1,\ell_\infty\rrbracket$. When $D = [z]-[p_2]$, the definition of $\psi(D=[z]-[p_2],\hbar)$ involves an essential singularity as $p_2 \to \infty^{(\alpha)}$. One thus needs to regularize it by the procedure introduced in \cite{EO07}. We remind the reader of the functions introduced in equation \eqref{def-Vp} with $p=\infty^{(\alpha)}$,
\beq
V_{\infty^{(\alpha)}}(z') =-  \sum_{k=1}^{r_{\infty^{(\alpha)}}-1} \frac{t_{\infty^{(\alpha)},k}}{k} \zeta_{\infty^{(\alpha)}}^{-k}(z') + t_{\infty^{(\alpha)},0} \log (\zeta_{\infty^{(\alpha)}}(z')) .
\eeq
The definition is such that the singular part of $ydx$ at $\infty^{(\alpha)}$ is precisely given by $dV_{\infty^{(\alpha)}}$. In order to remove the essential singularity of $\psi(D=[z]-[p_2],\hbar)$ when $p_2\to \infty^{(\alpha)}$, we multiply $\psi(D=[z]-[p_2],\hbar)$ by $\exp\big(\frac{1}{\hbar}\big(V_{\infty^{(\alpha)}}(z)- \int_{p_2}^z dV_{\infty^{(\alpha)}}\big)\big)$ before considering the limit $p_2\to\infty^{(\alpha)}$. This motivates the following definition.

\bd[Regularization of the perturbative wave functions \cite{EO07}]\label{defpsireg}
For any $\alpha \in\llbracket 1,\ell_\infty\rrbracket$, let us define
\beaa
\psi^{\mathrm{reg}}(D=[z]-[\infty^{(\alpha)}],\hbar) 
&\coloneqq & \lim_{p_2\to \infty^{(\alpha)}}
\psi(D=[z]-[p_2],\hbar) \frac{e^{\hbar^{-1}V_{\infty^{(\alpha)}}(p_2)}}{x(p_2)-x(z)}  \sqrt{\frac{dx(p_2)}{d\zeta_{\infty^{(\alpha)}}(p_2)}},
\eeaa
which is worth
\beaa
\psi^{\mathrm{reg}}(D=[z]-[\infty^{(\alpha)}],\hbar) 
&=&  \exp \Bigg(\hbar^{-1}\bigg( V_{\infty^{(\alpha)}}(z)+ \int_{\infty^{(\alpha)}}^z (ydx-dV_{\infty^{(\alpha)}})\bigg)\Bigg) \cr
&& \frac{1}{E(z,\infty^{(\alpha)})\sqrt{dx(z)d\zeta_{\infty^{(\alpha)}}(\infty^{(\alpha)}) }} \cr
&& 
\exp\Bigg(  {\displaystyle  \sum_{h \geq 0} \sum_{n\geq 3\delta_{h,0}} }\frac{\hbar^{2h-2+n}}{n!}\int_{\infty^{(\alpha)}}^z\cdots\int_{\infty^{(\alpha)}}^z \omega_{h,n}  \Bigg) ,
\eeaa
where $E$ is the prime form on $\Sigma$.
Similarly, for $l\geq 0$, we define
\beaa
\psi_l^{\mathrm{reg}}(D=[z]-[\infty^{(\alpha)}],\hbar) 
&\coloneqq & \lim_{p_2\to\infty} \psi_{l,1}^{\mathrm{reg}}(D=[z]-[p_2],\hbar)  \frac{e^{\hbar^{-1}V_{\infty^{(\alpha)}}(p_2)}}{x(p_2)-x(z)}  \sqrt{\frac{dx(p_2)}{d\zeta_{\infty^{(\alpha)}}(p_2)}} \cr
&=& \Bigg( \sum_{h\geq0} \sum_{n\geq 3\delta_{h,0}} \frac{\hbar^{2h+n}}{n!} \int_{\infty^{(\alpha)}}^z\dots \int_{\infty^{(\alpha)}}^z \frac{\hat Q^{(l)}_{h,n+1}(z;z_1,\dots,z_n)}{dx(z)^l}\Bigg) \cr 
&&\psi^{\mathrm{reg}}(D=[z]-[\infty^{(\alpha)}],\hbar) ,
\eeaa
where the integrals of $\hat Q$ are convergent because the $(z_i)_{i\geq 1}$'s that are integrated from $p_2\to\infty^{(\alpha)}$ never belong to a factor $\omega_{0,1}$.
\hfill $\star$
\ed

\medskip

These objects are defined as functions of $z$ on the universal cover of $ \Sigma\setminus x^{-1}(\mathcal P) $, and are WKB formal series of the formal parameter $\hbar$ but we will often omit them in the notations in order to make the reading easier.

These regularized functions are solutions to a simpler version of the system of KZ equations given by equation \eqref{simple_systemPDE}.
\bt [KZ equations for regularized wave functions]\label{reg-KZ-eq} For any $\alpha \in\llbracket 1,\ell_\infty\rrbracket$, the regularized wave functions satisfy, for all $l \in \llbracket 0,d-1\rrbracket$ ,
\bea\label{reg-KZ-eqeq}
&&\!\!\!\!{\hbar} \frac{d}{dx(z)}  \psi_{l}^{\mathrm{reg}}(D=[z]-[\infty^{(\alpha)}],\hbar)  + \psi^{\mathrm{reg}}_{l+1}(D=[z]-[\infty^{(\alpha)}],\hbar) \cr
&&\!\!\!\!= \bigg[ {\displaystyle \sum_{h\geq 0} \sum_{n \geq 0}}\frac{ \hbar^{2h+n} }{n!} \sum_{P \in \mathcal{P}} \sum_{k\in S_P^{(l+1)}} \cr
&&\!\!\!\!\xi_P(x(z))^{-k} \Res_{\lambda \to P}    \xi_P(\lambda)^{k-1} \, d  \xi_P(\lambda)   \int_{z_1 =\infty^{(\alpha)}}^{z_1=z} \dots \int_{z_n  = \infty^{(\alpha)}}^{z_n=z} \frac{{Q}_{h,n+1}^{(l+1)}(\lambda;\mathbf{z})}{(d\lambda)^{l+1}} \bigg]  \psi^{\mathrm{reg}}(D=[z]-[\infty^{(\alpha)}],\hbar),\cr
&& 
\eea
where  the sets $\left(S_P^{(l)}\right)_{P\in \mathcal{P}, l\in\llbracket 1, l\rrbracket} $ are defined in equation \eqref{eq-def-S}.
 \et

\begin{proof}
See Appendix~\ref{proof_reg-KZ-eq}.
\end{proof}

\subsection{Rewriting KZ equations with linear operators}\label{KZ-linear-op}

It was noticed from the beginning of TR in \cite{CEO06, EO07} that TR differentials have the properties that their derivatives with respect to the moduli of the classical spectral curve are integrals similar to \eqref{eq-var-epsilon}. In other words we can trade derivatives with respect to moduli, for integrals on $\Sigma$. 
Here, the right hand side of KZ equations \eqref{reg-KZ-eqeq} contains an integral (a residue), and very often this can be reformulated as a derivative with respect to moduli of the classical spectral curve.
This route would allow to turn the KZ equation into a PDE involving derivatives with respect to the moduli which are familiar in the world of integrable systems. 
Such a PDE could be interesting by itself, however, our goal is to get a quantum curve, i.e.~an ODE. Hence, our aim is to get rid of the derivatives with respect to moduli or, in our case, to be able to remove the integrals of the RHS of the KZ equations.

Let us develop here integral operators that mimic the role of derivatives with respect to moduli.
This is based on the theory of generalized cycles introduced in \cite{GIS} but here we present a standalone shortened version which is only of notational nature. In the present context, we could avoid the introduction of these generalized cycles but this would lead to very long  expressions which we believe would be harder to read.

\subsubsection{Generalized integrals}
Let us now study the RHS of the KZ equations \eqref{reg-KZ-eqeq} for regularized perturbative wave functions.  
For this purpose, observe that the residue in \eqref{reg-KZ-eqeq} extracts the Laurent coefficient of  $\xi_P(\lambda)^{-k}$ in the expansion of the function $\frac{{Q}_{h,n+1}^{(l+1)}(\lambda;\mathbf{z})}{(d\lambda)^{l+1}}$ around $\lambda \to P$,
and since ${Q}_{h,n+1}^{(l+1)}(\lambda;\mathbf{z})$ is an algebraic expression of the $\omega_{h_i,n_i}$, the Laurent expansion coefficients of ${Q}_{h,n+1}^{(l+1)}(\lambda;\mathbf{z})$ are algebraic combinations of the Taylor expansion coefficients of the $\omega_{h_i,n_i}$.
The Taylor expansion coefficients of the $\omega_{h_i,n_i}$, can themselves be obtained as residues.
Let us introduce a few notations in order to define a set of useful residues following \cite{GIS}.

\bd [Generalized cycles]\label{evaluation_idea}
Let $o$ be a reference point in the universal cover of $\Sigma$, we define the set of generalized cycles as
\beq
\mathcal{E}\coloneqq\left\{\mathcal C_{p,k}\right\}_{p \in \Sigma, k \in \mathbb{Z}}\cup  \left\{\mathcal{C}_{o}^p \right\}_{p \in \Sigma} \cup \{\mathcal{A}_i ,\mathcal{B}_i\}_{i=1}^g,
\eeq
where the integration of a meromorphic form $\om$ along such cycles is defined as follows.
\begin{itemize}

\item  $\forall \ p\in \Sigma$, and $\forall  \ k\in \mathbb Z$,
one defines
\beq
\int_{\mathcal C_{p,k}}  \quad : \quad \omega \mapsto \Res_{p} \zeta_p^{-k} \ \omega\,.
\eeq

\item Let $\gamma$ be a Jordan arc from a point $o\in\Sigma$ to a  point $p\in\Sigma$. The integration
\beq
\int_{\mathcal C_{o}^p}  \quad : \quad \omega \mapsto \int_\gamma  \omega
\eeq
 is ill defined on forms $\omega$ that have poles in $\gamma$. One extends the definition to such forms as follows.
\begin{itemize}
\item If $\omega$ has no pole along $\gamma$, we define $\int_{\mathcal{C}_{o,p}}\omega=\int_\gamma \omega$.

\item If $\omega$ has poles along $\gamma$, since a meromorphic 1-form can have only discrete poles, we decompose $\gamma$ as a sum of arcs such that the poles of $\omega$ are at the extremities of the subarcs, in other words we only need to define the action when $p$ is a pole of $\omega$.
For $k\in \llbracket 0,1+\deg_p\omega\rrbracket$, we define $t_{p,k}=\underset{p}{\Res} \zeta_p^{k-1}\omega$ , and $V_p = -\underset{k=1}{\overset{\deg_p\omega}{\sum}} \frac{t_{p,k}}{k} \zeta_p^{-k}$, and we define
\beq
\int_{\mathcal{C}_o^p} \omega = \int_o^p \left(\omega-dV_p-t_{p,0}\frac{d\zeta_p}{\zeta_p}\right) - V_p(o) - t_{p,0} \ln\zeta_p(o),
\eeq
where we choose the branch of the logarithm such that the cut is in the direction opposite to the Jordan arc at $o$. 

\item The result is independent of how $\gamma$ is decomposed into subarcs.

\item This definition is additive under concatenation of Jordan arcs.
\end{itemize}
\item If $\gamma$ is a closed Jordan arc on $\Sigma$, we define the integral of $\omega$ on $\gamma$ by writing $\gamma$ as a sum of open arcs. The result is independent of how we decompose it into open arcs.
\end{itemize}
\hfill $\star$
\ed

\br
Following \cite{GIS}, one can consider the generalized cycles as elements of the algebraic dual of the infinite dimensional space of meromorphic forms on $\Sigma$. From this perspective, we sometimes use the following notation
\beq
\forall \, C \in \mathcal{E} \, , \; \langle C,\omega\rangle \stackrel{{\text{notation}}}{=} \int_C \omega.
\eeq
\er

Our purpose is to use generalized cycles acting only on a small subspace of meromorphic forms, those generated by the $(\omega_{h,n})_{h,n\geq 0}$. In the following we shall need to consider only generalized cycles of type $\mathcal C_{p,k}$, where $p\in x^{-1}(\mathcal P)$, open Jordan arcs ending at $p\in x^{-1}(\mathcal P)$, and closed non-contractible Jordan arcs (thus homologically equivalent to combinations of $\mathcal A_i$ and $\mathcal B_i$ cycles).
We remark the following: let $\omega=\omega_{h,n+1}(z,z_1,\dots,z_n)$ be a TR differential.
Considering its dependence on the first variable $\omega=\omega_{h,n+1}(\cdot,z_1,\dots,z_n)$, it has poles only at $z\in x^{-1}(\mathcal P)$ if $(h,n)=(0,0)$ or at $z_1$ if $(h,n)=(0,1)$ with vanishing residue, or at $\mathcal R$ if $2h+n>1$, with vanishing residues at $\mathcal R$.

Due to vanishing residues, integrals over Jordan arcs $\gamma$
\beq
\langle\gamma,\omega\rangle= \int_\gamma \omega
\eeq
actually depend only on the homotopy class of $\gamma$ on $\Sigma\setminus x^{-1}(\mathcal P)$.
On this subspace we may thus consider that cycles $\mathcal A_i$ or $\mathcal B_i$ in $H_1(\Sigma\setminus x^{-1}(\mathcal P),\mathbb Z)$  define generalized cycles.

In practice, for $p\in x^{-1}(\mathcal{P})$ and $k \in \mathbb{Z}$, $\int_{\mathcal{C}_{p,k}} \om$  extracts the coefficient of $(\zeta_{p})^{k-1} $ in the Laurent expansion of $\om$ at $p$. Another important remark is that only $\om_{0,1}$ has poles in $x^{-1}(\mathcal{P})$ and behaves as
\beq
\begin{cases}
\frac{\om_{0,1}}{d\zeta_{Z_i^{(\alpha)}}} = {\displaystyle \sum_{k=0}^{r_{Z_i^{(\alpha)}}-1} } t_{Z_i^{(\alpha)},k} \,\,\zeta_{Z_i^{(\alpha)}}^{-k-1} + O(1), & \;\; \hbox{around }  Z_i^{(\alpha)}, \cr
\frac{\om_{0,1}}{d\zeta_{\infty^{(\alpha)}}} =   {\displaystyle  \sum_{k=0}^{r_{\infty^{(\alpha)}}-1}} t_{\infty^{(\alpha)},k}\,\, \zeta_{\infty^{(\alpha)}}^{-k-1} + O(1), &  \;\;\hbox{around }  \infty^{(\alpha)}, \cr
\end{cases}
\eeq
so that, for all $(h,n)\in \mathbb{N}^2$,
\bea\label{eq-t-cycles-1}
\forall\, i\in\llbracket1,N\rrbracket \, , \, \forall\, \alpha\in \llbracket 1 , d\rrbracket  \, , \, \forall\, k \leq -r_{Z_i^{(\alpha)}} \, : \; \int_{\mathcal{C}_{Z_i^{(\alpha)},k}} \om_{h,n+1}(\cdot,\mathbf{z}) &=& 0\,,\cr
\forall\, i\in\llbracket1,N\rrbracket \, , \, \forall\, \alpha\in \llbracket 1 , d\rrbracket  \, , \, \forall\, k\in \llbracket -r_{Z_i^{(\alpha)}}+1,0\rrbracket \, : \; \int_{\mathcal{C}_{Z_i^{(\alpha)},k}} \om_{h,n+1}(\cdot,\mathbf{z}) &=& \delta_{h,0} \delta_{n,0} t_{Z_i^{(\alpha)},-k}\,,\cr
 \forall\, \alpha\in \llbracket 1 , d\rrbracket \, , \, \forall\,  k \leq -r_{\infty^{(\alpha)}} \, : \; \int_{\mathcal{C}_{\infty^{(\alpha)},k}} \om_{h,n+1}(\cdot,\mathbf{z}) &=& 0\,,\cr
 \forall\, \alpha\in \llbracket 1 , d\rrbracket \, , \, \forall\, k\in\llbracket -r_{\infty^{(\alpha)}}+1,0\rrbracket \, : \; \int_{\mathcal{C}_{\infty^{(\alpha)},k}} \om_{h,n+1}(\cdot,\mathbf{z}) &=& \delta_{h,0} \delta_{n,0} t_{\infty^{(\alpha)},-k}\,.\cr
&&
\eea

\subsubsection{Definition of the linear operators acting on TR differentials}

The purpose of this section is to rewrite the operators acting on the RHS of the KZ equations \ref{reg-KZ-eq} using the generalized integrals defined in the previous section. In order to be consistent, we need to define the action of the operators only on a subset of differentials forms that includes the TR differentials but also linear combinations, products, exponentials, inverses, etc.~of such quantities. To do so, we first introduce a formal algebra of symbols generated by the TR differentials and explain how the operators act in this algebra. Then we extend the definitions of the action similarly to what we would expect for the action of a derivation operator with respect to spectral times for quantities like exponentials, inverses, etc. that we shall require in the rest of the article. The main advantage of the symbols formalism is to avoid the definition of the operators on the whole set of meromorphic differentials that would be challenging.

\bd[Algebra of generalized integrals]
We define the graded commutative algebra $\check{\mathcal W}$  over $\mathbb C$, as the algebra freely generated 
 by a set of symbols consisting of a pair $(h,n)$ and a symbol $\int_{C_1}\dots\int_{C_n}$, labeled by generalized cycles $C_i\in \mathcal E$, which we denote
\beq
\int_{C_1}\dots\int_{C_n}\omega_{h,n}\,,
\eeq
and quotiented by the linearity relations for cycles $\int_{a_1 C_1+a_2 C_2}\omega_{h,n}=a_1\int_{C_1}\omega_{h,n}+a_2\int_{C_2}\omega_{h,n}$ and distributed in the product
\beq
\check{\mathcal W} = \mathbb C\left[\left\{\int_{C_1}\dots\int_{C_n}\omega_{h,n}\right\}_{h,n\geq 0} \right]_{} \quad / \ (\text{cycles linearity relations}).
\eeq

We define a linear map, called \emph{evaluation map}, that associates to a symbol its value in $\mathbb{C}$ following Definition~\ref{evaluation_idea},
$$
\begin{array}{rrcl}
\text{ev}: & \check{\mathcal W}  & \to & \mathbb{C} \\
&\int_{C_1}\dots\int_{C_n}\omega_{h,n} & \mapsto & \int_{z_1 \in C_1}\dots\int_{z_n \in C_n}\omega_{h,n}(z_1,\dots,z_n ),
\end{array}
$$
where the order of integration is from the rightmost to leftmost.
\hfill $\star$
\ed

Let us make some remarks.
\begin{itemize}
\item Notice that for all $h\geq 0$: $\omega_{h,0}$ is a symbol but is not an element of $\mathbb C$.
\item The algebra is commutative with respect to addition and multiplication, but the integrals themselves are not assumed commutative, i.e. $\int_{C_1}\int_{C_2} \neq \int_{C_2}\int_{C_1}$.
\item Since a generalized cycle $C$ might depend on a point $z$ on the universal cover of the classical spectral curve, the result of the evaluation map can be considered as a, possibly multi-valued, function on $\Sigma$.
\end{itemize}

\bd[Formal algebra]
We define the graded ring of formal Laurent power series, denoted $\check{\mathcal W}((\hbar))$, whose elements are sequences of elements of $\check{\mathcal W}$:
\beq
\check{\mathcal W}((\hbar))=\left\{ w=\sum_{k={k_{\text{min}}}}^{\infty} \hbar^k w_k\,, \ \qquad w_k\in \check{\mathcal W} \ , \ k_{\text{min}}\in \mathbb{Z} \right\}.
\eeq
The degree of $w\in\check{\mathcal W}((\hbar))$ is defined as the smallest non-zero power of $\hbar$,
\beq
 w=\sum_{k={k_{\text{min}}}}^{\infty} \hbar^k w_k \text{ with } k_{\text{min}}\neq 0\,\, \Rightarrow\,\, \deg w= k_{\text{min}}.
\eeq
We define the evaluation by acting term-wise, which results in formal power series of multilinear differential forms.
\hfill $\star$
\ed

\bd[Exponential symbols]
We enlarge our algebra with exponential symbols
\beq
e^{w} \, , \ \text{with} \ 
w\in \check{\mathcal W}[\hbar^{-1}],
\eeq
(i.e.~only negative or zero powers of $\hbar$), and we quotient by relations $e^{w_1+w_2}=e^{w_1}e^{w_2}$,
\beq
\tilde{\mathcal W}= \check{\mathcal W}((\hbar))\Big[\{e^w\}_{w\in \check{\mathcal W}[\hbar^{-1}]}\Big]/(e^{w_1+w_2}=e^{w_1}e^{w_2}).
\eeq
The evaluation map is defined by acting term-wise. 
\hfill $\star$
\ed

\br
Observe that in particular $0\in \check{\mathcal W}[\hbar^{-1}]$, so $e^0\in\tilde{\mathcal W}$, which we can call the identify element $1_{\tilde{\mathcal W}}=e^0$ of the product in $\tilde{\mathcal W}$. For every $w\in \tilde{\mathcal W}$, we identify $w^0=1_{\tilde{\mathcal W}}$.
\er

\br Note that if $w={\displaystyle \sum_{k={k_{\text{min}}}}^{\infty}} \hbar^k w_k$, with strictly positive $k_{\text{min}}>0$, then the exponential 
\beq
e^{w} = \sum_{n=0}^{\infty} \frac{1}{n!} w^n \in \check{\mathcal W}[\hbar] \subset \check{\mathcal W}((\hbar))\Big[\{e^w\}_{w\in \check{\mathcal W}[\hbar^{-1}]}\Big]
\eeq
is already defined in $\check{\mathcal W}((\hbar))\big[\{e^w\}_{w\in \check{\mathcal W}[\hbar^{-1}]}\big]$ and satisfies the quotient relation. That is why we defined new symbols only for negative powers of~$\hbar$.
\er

While $\tilde{\mathcal W}$ is large enough to define a symbolic version of the perturbative wave functions, it does not allow for the definition of the inverse of an element. The existence of the latter is useful for having a shorter presentation of the derivation of the quantum curve and thus we extend $\tilde{\mathcal W}$ by defining inverses of its elements.

\bd[Inverse symbols]
For any symbol $w \in \check{\mathcal W}\setminus\{0\}$, we define a symbol $\frac{1}{w}$ (also denoted $w^{-1}$) modding out by the relation
\beq
\forall \, w \in \check{\mathcal W}\setminus\{0\},\;\; w \, \frac{1}{w} = 1_{\tilde{\mathcal W}}.
\eeq
Finally, we denote 
\beq
\mathcal{W}\coloneqq\frac{\mathbb{C}\left[\left\{\int_{C_1}\dots\int_{C_n}\omega_{h,n}\right\}_{h,n\geq 0} , \left\{\frac{1}{w}\right\}_{w \in \check{\mathcal{W}}\setminus \{0\}}
\right]
((\hbar))\big[\{e^w\}_{w\in \check{\mathcal W}[\hbar^{-1}]}\big]}{ (\text{cycle linearity relations}\, , \,w\, \frac{1}{w} = 1_{\tilde{\mathcal W}} \, , \, e^{w_1+w_2}=e^{w_1}e^{w_2} )}.
\eeq
\hfill $\star$
\ed

Let us emphasize that the inverse of an element in $\mathcal{W}$ is understood from the perspective of WKB series in $\hbar$. This means that, for $w \in \check{\mathcal W}[\hbar^{-1}]$ and $w_k \in  \check{\mathcal W}$, for $k \geq k_{\text{min}}$, one has
\beq
\frac{1}{e^{w} \underset{k\geq k_{\text{min}}}{\sum} \hbar^k w_k} = e^{-w} \hbar^{-k_{\text{min}}} \sum_{l \geq 0} (-1)^l \left(\frac{1}{w_{k_{\text{min}}}}\right)^{l+1} \left(\sum_{k > k_{\text{min}}} w_k \hbar^{k-k_{\text{min}}}\right)^l.
\eeq

\bd[Operators $\left(\mathcal{I}_{C}\right)_{C\in \mathcal{E}}$ acting in $\mathcal W$]\label{DefActionW} 
For a generalized cycle $C\in\mathcal E$, we define the linear operator $\mathcal I_C:\mathcal W\to\mathcal W$ by its action on the symbols,
\bea
\forall\, (h,n)\in \mathbb{N}^2\,:\, \mathcal I_C \left[ \int_{C_1}\dots\int_{C_n}\omega_{h,n}\right] \coloneqq \int_{C_1}\dots\int_{C_n} \int_{C}\omega_{h,n+1}.
\eea
Notice that $\int_C$ is the rightmost.
We extend the action of the operator to any element of $\check{\mathcal{W}}$  by linearity and by the Leibniz rules for products of symbols,
\bea
\mathcal I_C\left[ \text{symbol 1}\times \text{symbol 2}\right] = (\mathcal I_C\left[ \text{symbol 1}\right])\times \text{symbol 2}+\text{symbol 1} \times (\mathcal I_C\left[\text{symbol 2}\right]).
\eea
For exponentials, we define for any  $v\in \check{\mathcal W}[\hbar^{-1}]$ and any operator $\mathcal{I}_{C}$, with $C\in \mathcal{E}$,
\beq \mathcal{I}_C \left[\exp(v)\right]\coloneqq \left(\mathcal{I}_C[v]\right) \,\, \exp(v) .
\eeq
It is easy to check that these operators are compatible with quotients and that they satisfy, for any $w \in \mathcal{W}$,
\beq
\mathcal I_C \left[\frac{1}{w}\right] = - \mathcal I_C \left[w\right] \frac{1}{w^2}.
\eeq
\hfill $\star$
\ed

\medskip

It is worth noticing the action of these operators on some specific elements of $\mathcal W$. Some simple computations lead to the following results.

\bl[Action on spectral times]
Let  the symbolic spectral times be defined by 
\beq
\forall\, p \in x^{-1}(\mathcal{P}) \, , \, \forall k \in \mathbb{N} \, , \; t_{p,k}^{\mathrm{symbol}} \coloneqq \int_{\mathcal C_{p,-k}} \omega_{0,1},
\eeq
where we use the same notation as the spectral time, thanks to
\beq
\forall\, p \in x^{-1}(\mathcal{P}) \, , \, \forall k \in \mathbb{N} \, , \; \mathrm{ev}. t_{p,k}^{\mathrm{symbol}} = t_{p,k} .
\eeq
We have
\beq
\forall\, (p,p') \in (x^{-1}(\mathcal{P}))^2 \, , \, \forall (k,k') \in \mathbb{N}^2 \, , \; \mathrm{ev}.\left(\mathcal I_{\mathcal C_{p',k'}} t^{\mathrm{symbol}}_{p,k} \right)
= \int_{\mathcal C_{p,-k}} \int_{\mathcal C_{p',k'}} \omega_{0,2}
=k \delta_{p,p'} \delta_{k,k'}.
\eeq
From now on we shall drop the superscript ${}^{\mathrm{symbol}}$ since the evaluation indeed gives the spectral times.
\el

\br
This lemma could give the impression that the operator $\mathcal I_{\mathcal C_{p,k}}$ acts as the differential operator $k\frac{\partial}{\partial t_{p,k}}$, but as we already mentioned, the differential operator $\frac{\partial}{\partial t_{p,k}}$ is ill-defined because the spectral times are not independent. 
The operator $\mathcal I_{\mathcal C_{p,k}} $ acting in $\mathcal W$ however is well defined, and is all that we shall need. 
But it may be helpful for the reader to keep in mind that the operators $\mathcal I_C$ mimic some differential operators with respect to moduli.
\er

For $k\geq 1$, we can express our canonical local coordinates as evaluations of symbols.
\beq
x(z)^{k} = -\sum_{\alpha=1}^d  \int_{z''=o^{(\alpha)}}^{z} \ \Res_{z'\to \infty^{(\alpha)}} \zeta_{\infty^{(\alpha)}}^{-k}(z') \omega_{0,2}(z',z'')
= - \mathrm{ev}.\left( 
\sum_{\alpha=1}^d \int_{\mathcal{C}_{o^{(\alpha)}}^z} \ \int_{\mathcal C_{\infty^{(\alpha)},k}}\omega_{0,2}\right)
\eeq
and 
\beq
\ln x(z)  
= \mathrm{ev}.\Bigg( -
\sum_{\alpha=1}^d \int_{\mathcal{C}_{o^{(\alpha)}}^z} \int_{\mathcal{C}_{o^{(\alpha)}}^{\infty^{(\alpha)}}} \omega_{0,2}\Bigg),
\eeq
where $\{ o^{(1)},\dots,o^{(d)}\}=x^{-1}(0)$.
Similarly, for any $i \in \llbracket1,N\rrbracket$,
\beq
\ln (x(z)-\Lambda_i)^{-1} =  \mathrm{ev}.\Bigg( -
\sum_{\alpha=1}^d \int_{\mathcal{C}_{\infty^{(\alpha)}}^z} \int_{\mathcal{C}_{Z_i^{(\alpha)}}^{\infty^{(\alpha)}}} \omega_{0,2}\Bigg),
\eeq
\beq
\forall k \geq 1 \, , \; (x(z)-\Lambda_i)^{-k} =  \mathrm{ev}.\Bigg( 
\sum_{\alpha=1}^d \int_{\mathcal{C}_{\infty^{(\alpha)}}^z} \ \int_{\mathcal C_{Z_i^{(\alpha)},k}}\omega_{0,2}\Bigg).
\eeq
This allows to define these as symbols, and define a symbolic potential.

Keeping in mind that the $\mathcal I_C$ should mimic derivatives, we have the important lemma that canonical coordinates are ``constant''.
\bl
One has, for $k \geq 1$,
\beq
 \mathrm{ev}.\ \mathcal I_C \left(x(z)^{k}\right)^{\mathrm{symbol}} =  \mathrm{ev}.\ \mathcal I_C \ln x(z)
= 0,
\eeq
and
\beq
\forall i \in \llbracket 1,N\rrbracket \, , \;  \mathrm{ev}.\ \mathcal I_C \left((x(z)-\Lambda_i)^{-k}\right)^{\mathrm{symbol}} =  \mathrm{ev}.\ \mathcal I_C \ln  \left(x(z)-\Lambda_i\right)
= 0.
\eeq
\el
\begin{proof}
Let us consider $k>0$, and write
\beq
\mathcal I_C \left(x(z)^{k}\right)^{\mathrm{symbol}} =  
-
\sum_{\alpha=1}^d \int_{\mathcal{C}_{o^{(\alpha)}}^z} \int_{\mathcal C_{\infty^{(\alpha)},k}}  \int_{C}\omega_{0,3},
\eeq
whose evaluation vanishes thanks to the linear loop equation from Theorem~\ref{loopeq}.
The proof is similar for the other cases.
\end{proof}

Further straightforward computations lead to the following properties.
\bl[Operator acting on potentials]
Define the symbolic local potentials for $\alpha \in \llbracket 1,\ell_\infty \rrbracket$
\beq
V_{\infty^{(\alpha)}}^{ \mathrm{symbol}}(z) = \sum_{\beta=1}^d \left(
 \int_{\mathcal C_{\infty^{(\alpha)},0}}\omega_{0,1} \ \int_{\mathcal{C}_{o^{(\beta)}}^z} \int_{\mathcal C_o^{\infty^{(\beta)}}} \omega_{0,2}
-\sum_{k\geq 1} \int_{\mathcal C_{\infty^{(\alpha)},-k}}\omega_{0,1} \ \int_{\mathcal{C}_{o^{(\alpha)}}^z} \int_{\mathcal C_{\infty^{(\beta)},k}} \omega_{0,2} 
 \right)
\eeq
and for $i \in \llbracket 1,N\rrbracket$ and $\alpha \in \llbracket 1,\ell_{\Lambda_i}\rrbracket$,
\beq
V_{Z_i^{(\alpha)}}^{ \mathrm{symbol}}(z) = \sum_{\beta=1}^d \left(
 \int_{\mathcal C_{Z_i^{(\alpha)},0}}\omega_{0,1} \  \int_{\mathcal{C}_{\infty^{(\beta)}}^z} \int_{\mathcal C_o^{Z_i^{(\beta)}}} \omega_{0,2}
-\sum_{k\geq 1} \int_{\mathcal C_{Z_i^{(\alpha)},-k}}\omega_{0,1} \  \int_{\mathcal{C}_{\infty^{(\beta)}}^z} \int_{\mathcal C_{Z_i^{(\beta)},k}} \omega_{0,2} 
 \right).
\eeq
We have, for any $p\in x^{-1}(\mathcal{P})$,
\beq
\mathrm{ev}. V_p^{ \mathrm{symbol}}(z) = V_p(z),
\eeq
 for $k\geq 1$,
\beq
\mathrm{ev}.\left(\mathcal I_{\mathcal C_{p',k}} V^{\mathrm{symbol}}_p(z) \right)=\delta_{p,p'} \zeta_{p}(z)^{-k}
\eeq
and, for $k=0$,
\beq
\mathrm{ev}.\left(\mathcal I_{\mathcal C^{p'}_o} dV^{\mathrm{symbol}}_p \right)= \delta_{p,p'} \zeta_{p}(z)^{-1}.
\eeq
\el

\begin{proposition}[Commutation of the operators] \label{prop-commutator}
For any $n \geq 2$, for any permutation $\sigma \in \mathfrak{S}_n$, and any element $(C_1,\dots,C_n) \in \left(\mathcal{E}\right)^n$, one has for any $ (h,m) \in \mathbb{N}^2$, $(\td{C}_1,\dots,\td{C}_m) \in \left(\mathcal{E}\right)^m$,
\beq
\mathrm{ev}.\mathcal{I}_{C_1}  \mathcal{I}_{C_2}  \dots  \mathcal{I}_{C_n}\int_{\td{C}_1}\cdots\int_{\td{C}_m} {\om}_{h,m} = \mathrm{ev}.\mathcal{I}_{C_{\sigma(1)}}   \mathcal{I}_{C_{\sigma(2)}}  \dots  \mathcal{I}_{C_{\sigma(n)}}\int_{\td{C}_1}\cdots\int_{\td{C}_m} \om_{h,m},
\eeq
if $2h-2+n+m\neq 0$. When $2h-2+n+m=0$, this implies $h=0$, $n=2$, $m=0$, and we have 
the following commutators
\bea
\forall\, p \in x^{-1}(\mathcal{P}) \, , \, \forall\, k\geq 1 \, : \; \mathrm{ev}. \left[\mathcal{I}_{\mathcal{C}_{p,k}} , \mathcal{I}_{\mathcal{C}_{p,-k}}\right] \om_{0,0} &=& k,\cr
\forall\, p \in x^{-1}(\mathcal{P}) \, , \,  \, : \; \mathrm{ev}. \left[\mathcal{I}_{\mathcal{C}_o^{p}},\mathcal{I}_{\mathcal{C}_{p,0}} \right] \om_{0,0} &=& 1,\cr
\forall\, 1=1,\dots, \genus \, , \,  \, : \; \mathrm{ev}. \left[\mathcal{I}_{\mathcal{B}_i} , \mathcal{I}_{\mathcal{A}_j}\right] \om_{0,0} &=& 2\pi\ii \ \delta_{i,j}.
\eea
\end{proposition}

Let us now define the symbolic partner to the prime form.

\bd[Symbolic prime form]
Let $D = {\displaystyle \sum_{i=1}^s} a_i[p_i]$ be a divisor of degree 0 on the universal cover. We define in $\mathcal W$
\beq
E^{\mathrm{symbol}}(D) \coloneqq \exp{\Big(-\frac12 \int_D\int_D \omega_{0,2}\Big)}
\eeq
and the base prime form $E_0(D)$ as
\beq
E_0^{\mathrm{symbol}}(D) \coloneqq \exp{\Big(-\sum_{i=1}^{s-1} \sum_{j=i+1}^s a_i a_j  \sum_{\alpha=1}^d  \int_D\int_{D+[p_i^{(\alpha)}]-[p_i]} \omega_{0,2}\Big)},
\eeq
where we choose $d$ Jordan arcs $[p_i^{(\alpha)}]-[p_i]$ starting at $p_i$ and ending at all preimages $x^{-1}(x(p_i))=\{p_i^{(1)},\dots,p_i^{(d)}\}$.
Their evaluations are linked to the usual prime forms as follows.
\beq
E(D)\coloneqq\mathrm{ev}.E^{\mathrm{symbol}}(D) = \prod_{i=1}^{s-1} \prod_{j=i+1}^s E(p_i,p_j)^{-a_i a_j}  \prod_{i=1}^s d\zeta_{p_i}(p_i)^{\frac12 a_i^2},
\eeq
with $E$ the prime form on $\Sigma$, and
\begin{align}
E_0(D)&\coloneqq\mathrm{ev}.E^{\mathrm{symbol}}_0(D) =  \prod_{i=1}^{s-1} \prod_{j=i+1}^s (x(p_i)-x(p_j))^{-a_i a_j} \nonumber \\
&=  \prod_{i=1}^{s-1} \prod_{j=i+1}^s E_0(x(p_i),x(p_j))^{-a_i a_j}   \prod_{i=1}^s dx(p_i)^{\frac12 a_i^2},
\end{align}
where $E_0(x,x')=(x-x')/\sqrt{dx dx'}$ is the prime form on $\mathbb P^1$.
\hfill $\star$
\ed

Notice that $\mathrm{ev}.\mathcal I_{C} E^{\mathrm{symbol}}_0(D)=0$, i.e.~the base prime form is ``constant'', while  the prime form on $\Sigma$ is not: $\mathrm{ev}.\mathcal I_{C} E^{\mathrm{symbol}}(D)\neq 0$.
From now on, we shall use the same notation for the base prime form and its evaluation, and treat it as a scalar in $\mathcal W$.

\subsubsection{Action on regularized wave functions and rewriting of the KZ equations}
In order to act with the operators $\mathcal{I}_{{C}}$ on the wave functions, one needs to define a symbolic version of the latter.

\bd[Symbolic perturbative wave function]
If $D = {\displaystyle \sum_{i=1}^s} a_i[p_i]$ is a divisor on the universal cover, we define
\beq
\psi^{\mathrm{symbol}}(D,\hbar)\coloneqq  \exp{\left(\sum_{h\geq 0}\sum_{n\geq 0} \frac{\hbar^{2h-2+n}}{n!} \int_D\dots\int_D \omega_{h,n}\right)} \ \ \prod_{i=1}^{s-1} \prod_{j=i+1}^s (x(p_i)-x(p_j))^{-a_i a_j}\quad \in \mathcal W.
\eeq
Equivalently we have,
\beq \label{def-psi-symb}
\psi^{\mathrm{symbol}}(D,\hbar)
= \exp{\left(\sum_{(h,n)\neq(0,2)} \frac{\hbar^{2h-2+n}}{n!} \int_D\dots\int_D \omega_{h,n}\right)} \ \ E^{\mathrm{symbol}}_0(D) \frac{1}{E^{\mathrm{symbol}}(D)}
\quad \in \mathcal W.
\eeq

In the case $D=[z]-[p_2]$, and if we want to set $p_2=\infty^{(\alpha)}$, we define the regularized wave function by just removing the $E_0$ symbol (the only one that is ill-defined at $x(p_2)=\infty$)
\bea
\psi^{\mathrm{reg \; symbol}}([z]-[\infty^{(\alpha)}],\hbar)
= \left[ \psi^{\mathrm{symbol}}([z]-[p_2],\hbar) \, \frac{1}{E_0^{\mathrm{symbol}}([z]-[p_2])} \right]_{p_2=\infty^{(\alpha)}}.
\eea
Its evaluation is
\bea
\mathrm{ev}.\psi^{\mathrm{reg\; symbol}}([z]-[\infty^{(\alpha)}],\hbar)
&=& \exp{\left(\sum_{(h,n)\neq (0,1),(0,2)} \frac{\hbar^{2h-2+n}}{n!} \int_{\infty^{(\alpha)}}^z\dots\int_{\infty^{(\alpha)}}^z \omega_{h,n}\right)}  \cr
&& \exp\left(\hbar^{-1}\left(V_{\infty^{(\alpha)}}(z) + \int_{\infty^{(\alpha)}}^z (\omega_{0,1}-dV_{\infty^{(\alpha)}}) \right)\right) \cr
&& \frac{1}{E(z,\infty^{(\alpha)}) \sqrt{dx(z) d\zeta_{\infty^{(\alpha)}}(\infty^{(\alpha)})}} \cr
&=& \psi^{\mathrm{reg}}([z]-[\infty^{(\alpha)}],\hbar).
\eea
Let us finally define the symbols, for all $l \geq 0$,
\beq
\psi_l^{\mathrm{reg\;  symbol}}([z]-[\infty^{(\alpha)}],\hbar)\coloneqq 
\sum_{\boldsymbol{\beta} \underset{l}{\subseteq} \left(x^{-1}(x(z))\setminus\{z\}\right)} \frac{1 }{l!} \left( \prod_{j=1}^l \mathcal{I}_{\mathcal{C}_{\beta_j,1}} \right) \psi^{\mathrm{reg \; symbol}}([z]-[\infty^{(\alpha)}],\hbar).
\eeq
\hfill $\star$
\ed

Notice that $\mathcal I_C$ acts on the coordinates, and thus on $E_0$ as if they were scalar.
Hence acting with $\mathcal I_C$, $C\in \mathcal E$, we get
\bea\label{ActionOnPsireg} 
\mathcal{I}_C \left[\psi^{\mathrm{reg\; symbol}}(D=[z]-[\infty^{(\alpha)}],\hbar)\right]
&=&\Bigg({\displaystyle  \sum_{h \geq 0} \sum_{n\geq 0} }\frac{\hbar^{2h-2+n}}{n!}\mathcal{I}_C \left[\int_{\infty^{(\alpha)}}^z\cdots\int_{\infty^{(\alpha)}}^z {\omega}_{h,n} \right]\Bigg) \cr
&&  \psi^{\mathrm{reg\; symbol}}(D=[z]-[\infty^{(\alpha)}],\hbar).
\eea  

\br\label{WKBStable} Note that in terms of formal expansions in $\hbar$, we get that 
\beqq\mathrm{ev}.\left(\hbar^2 \mathcal{I}_C \left[\psi^{\mathrm{reg\; symbol}}(D=[z]-[\infty^{(\alpha)}],\hbar)\right]\right)\eeqq 
remains of the same form as the initial perturbative wave functions, i.e.~a WKB series expansion in $\hbar$ of the form $\exp\left(\hbar^{-2}f_{-2}+\hbar^{-1} f_{-1}\right) \underset{k=0}{\overset{\infty}{\sum}} f_k \hbar^k$.   
\er

It is also important to notice that the symbolic wave functions $\psi_l^{\mathrm{reg \; symbol}}([z]-[\infty^{(\alpha)}],\hbar)$ depend on a point $z$ on the universal cover of $\Sigma$ because it involves symbols of the form $\int_{\infty^{(\alpha)}}^z \om$. The latter can be interpreted locally around any point $z_0$ through their series expansion
\beq
\int_{\infty^{(\alpha)}}^z \om = \sum_{k=-\infty}^\infty \zeta_{z_0}(z)^k \, \om_k + \ln(\zeta_{z_0}(z)) \, \om_{\ln},
\eeq
where the coefficients $(\om_k)_{k\in \mathbb{Z}}$ and $\om_{\ln}$ are symbols of the form $\int_{\mathcal{C}_{z_0,l}} \om$. In particular, this allows to consider the action of $\frac{d}{dx(z)}$ on such symbols acting only on the coefficients $\zeta_{z_0}(z)^k$ of this expansion. Since these local coordinates are constants with respect to the action of the operators $\mathcal{I}_{C}$, the action of $\mathcal{I}_{C}$ on such a Laurent series expansion is interpreted as the action on the coefficients $\om_k$.

Using the previous definitions, we may use the operators $\{\mathcal{I}_{C}\}_{C\in \mathcal{E}_\alpha}$ to write the regularized KZ equations \eqref{reg-KZ-eq} in a simpler form. 

\bt \label{final-KZ}For any $\alpha\in\llbracket 1,\ell_\infty \rrbracket$, the regularized wave functions satisfy the KZ equations for all $l\in \llbracket 0,d-1\rrbracket$,
\beq \label{eq-final-KZ}
{\hbar} \frac{d}{dx(z)}  \psi_{l}^{\mathrm{reg \; symbol}}([z]-[\infty^{(\alpha)}])   +\psi_{l+1}^{\mathrm{reg \; symbol}}([z]-[\infty^{(\alpha)}])   
=  \widetilde{\mathcal{L}}_l(x(z)) \left[\psi^{\mathrm{reg \; symbol}}([z]-[\infty^{(\alpha)}]) \right], 
\eeq
where 
\beq
\widetilde{\mathcal{L}}_l(x(z)) = \sum_{P \in \mathcal{P}} \sum_{k \in S_P^{(l+1)}}  \xi_P(x(z))^{-k} \widetilde{\mathcal{L}}_{P,k,l} 
\eeq
 are the multi-linear operators with $\td{\mathcal L}_{P,k,l}$ defined by 
\bea
\widetilde{\mathcal L}_{P,k,l}
&\coloneqq &   \epsilon_P^{l+1}
\Bigg[
 \xi_P(x(z))^{-(l+1)\epsilon_P}   
\sum_{\ell'=0}^{l+1}  \ \sum_{\nu'\subset_{\ell'} \llbracket 1,d \rrbracket} \ \ 
\prod_{j\in \nu'}
\bigg(
\sum_{m=0}^{r_{P^{(j)}}-1}  
\frac{t_{P^{(j)},m}}{d_{P^{(j)}}} \ \xi_P^{-\frac{m}{d_{P^{(j)}}}}   \bigg)
\quad \cr
&& 
\sum_{0\leq \ell'' \leq \frac{l+1-\ell'}{2}}\;\sum_{\substack{\nu''\in\mathcal{S}^{(2)}( \llbracket 1,d \rrbracket\setminus \nu') \\ l(\nu'')=\ell''}}
\ \ 
\prod_{i=1}^{\ell''}  \frac{\hbar^2R(P)_{\nu''_i}}{d_{P^{(\nu_{i,+}'')}} d_{P^{(\nu_{i,-}'')}} } \cr
&&
\sum_{\nu \underset{l+1-\ell'-2\ell''}{\subseteq} \llbracket 1,d \rrbracket\setminus (\nu'\cup\nu'')}
\;\;\prod_{j\in \nu} 
\bigg(\hbar^2\sum_{m=1}^\infty \frac{ \xi_P^{\frac{m}{d_{P^{(j)}}}}}{d_{P^{(j)}}}   \mathcal I_{\mathcal C_{P^{(j)},m}} \bigg)  \Bigg]_{-k}  \cr
&&
-\hbar \delta_{P,\infty} \frac{{\epsilon_\infty^{l+1}}}{d_{\infty^{(\alpha)}}}
\Bigg[
 \xi_\infty(x(z))^{-(l+1)\epsilon_\infty}    
\!\sum_{\ell'=0}^{l} \, \sum_{\nu'\subset_{\ell'} \llbracket 1,d \rrbracket \setminus \{\alpha\}} \,
\prod_{j\in \nu'}
\bigg(
\sum_{m=0}^{r_{\infty^{(j)}}-1}  
\frac{t_{\infty^{(j)},m}}{d_{\infty^{(j)}}} \ \xi_\infty^{-\frac{m}{d_{\infty^{(j)}}}}   \bigg)
\quad \cr
&& 
\sum_{0\leq \ell'' \leq \frac{l+1-\ell'}{2}}\;\sum_{\substack{\nu''\in\mathcal{S}^{(2)}( \llbracket 1,d \rrbracket\setminus (\nu' \cup \{\alpha\}))\\ l(\nu'')=\ell''}}
\ \ 
\prod_{i=1}^{\ell''}  \frac{\hbar^2R(\infty)_{\nu''_i}}{d_{\infty^{(\nu_{i,+}'')}} d_{\infty^{(\nu_{i,-}'')}} } \cr
&&
\sum_{\nu \underset{l-\ell'-2\ell''}{\subseteq} \llbracket 1,d \rrbracket\setminus (\nu'\cup\nu''\cup \{\alpha\})}\;\;
\prod_{j\in \nu} 
\bigg(\hbar^2\sum_{m=1}^\infty \frac{ \xi_\infty^{\frac{m}{d_{\infty^{(j)}}}}}{d_{\infty^{(j)}}}   \mathcal I_{\mathcal C_{\infty^{(j)},m}} \bigg)  \Bigg]_{-k}, 
\eea
where we have denoted $\mathcal{S}^{(2)}(A)\subset \mathcal{S}(A)$ the set of partitions of a set $A$ where all blocks have size~$2$ and the sets $\left(S_P^{(l)}\right)_{P\in \mathcal{P}, l\in\llbracket 1, l\rrbracket} $ are defined in equation \eqref{eq-def-S}.
For every $\nu''_i=\{\pi_{+},\pi_{-}\}$ a pair of distinct indices  $\pi_+\neq\pi_- \in \llbracket 1,d \rrbracket$, we defined
\beq
R(P)_{(\pi_+,\pi_-)} \coloneqq \delta_{P^{(\pi_+)},P^{(\pi_-)}}  \ \frac{\rho_{P^{(\pi_+)}}^{\pi+-\pi_-}}{\left(1-\rho_{P^{(\pi_+)}}^{\pi_+-\pi_-}\right)^2}
\eeq
and $\rho_p$ is a root of unity
\beq
\rho_p = e^{\frac{2\pi\ii}{d_p}}.
\eeq
The subscript $[\cdot]_{-k}$ means keeping only the powers of $\xi_P(x(z))^{-k}$ in the sum, which implies that the sums $\underset{m=1}{\overset{\infty}{\sum}}$  are in fact reduced to finite sums, since
\beq
\frac{m}{d_{P^{(j)}}}\leq -k+(l+1)\epsilon_P + \sum_{p\in x^{-1}(P)} (r_p-1).
\eeq
After evaluation, these KZ equations read
\beq \label{eq-final-KZ-eval}
{\hbar} \frac{d}{dx(z)}  \psi_{l}^{\mathrm{reg}}([z]-[\infty^{(\alpha)}])   +\psi_{l+1}^{\mathrm{reg}}([z]-[\infty^{(\alpha)}])   
= \mathrm{ev} . \ \widetilde{\mathcal{L}}_l(x(z)) \left[\psi^{\mathrm{reg \; symbol}}([z]-[\infty^{(\alpha)}]) \right].
\eeq
\et

\begin{proof}In Appendix~\ref{appB3}. The proof is based on the idea that the residue in Theorem~\ref{reg-KZ-eq} extracts Taylor series coefficients, which can themselves be written as generalized cycle integrals, i.e.~actions of operators $\mathcal I_C$. The Taylor coefficients of $\omega_{0,1}$ and $\omega_{0,2}$ require special treatments and give the spectral times and the $R(P)$ factors.
\end{proof}

\subsection{Monodromies of the perturbative solutions}\label{5.5}

Some important properties of the regularized perturbative wave functions $\psi^{\mathrm{reg}}_{l}([z]-[\infty^{(\alpha)}],\hbar)$ are given by their behavior when $z$ goes around a cycle $\gamma \in \pi_1(\Sigma\setminus x^{-1}(\mathcal P))$. Remember that $\psi(D,\hbar)$ is defined for a divisor on the universal cover, i.e.~$D$ is a Jordan arc, starting at $\infty^{(\alpha)}$ and ending at $z\in\Sigma\setminus\mathcal P$. If $\gamma$ is a closed Jordan arc, starting at $z$ and ending at $z$, $D+\gamma$ denotes another Jordan arc, obtained by concatenation, also starting at $\infty^{(\alpha)}$ and ending at $z$.

We remind that admissible data consists of the classical spectral curve together with a choice of a basis of $H_1(\Sigma\setminus x^{-1}(\mathcal P),\mathbb Z)$, made of small circles $\mathcal C_{p}$ around the poles $p\in x^{-1}(\mathcal P)$ (they are equivalent to the generalized cycles  $2\pi\ii\mathcal C_{p,0}$), and completed by $2g$ Jordan loops
 $(\mathcal{A}_i,\mathcal{B}_i)_{i=1}^g$, such that on $\Sigma$ these $2g$ Jordan loops are representative of the homology basis cycles  $(\mathcal{A}_i,\mathcal{B}_i)_{i=1}^g$. We have the following monodromies of the perturbative wave functions along the elements of this basis.

\bl \label{lemma-mono-pert}
Let $\alpha\in \llbracket 1,d\rrbracket$. One has, for any $l\in \llbracket 0, d-1\rrbracket$,
\beq\label{eq-monodromyCp}
 \forall\, p\in x^{-1}(\mathcal P)\,:\, \psi_{l}([z+\mathcal{C}_p]-[\infty^{(\alpha)}],\hbar) = (-1)^{\delta_{p,\infty^{(\alpha)}}} e^{\frac{2 \pi i t_{p,0}}{\hbar}}\psi_{l}([z]-[\infty^{(\alpha)}],\hbar),
\eeq
\beq\label{eq-monodromyA}
 \forall\, j\in\llbracket 1,g\rrbracket\,:\, \psi_{l}([z+\mathcal{A}_j]-[\infty^{(\alpha)}],\hbar) = e^{\frac{2 \pi i \epsilon_j}{\hbar}}\psi_{l}([z]-[\infty^{(\alpha)}],\hbar)
\eeq
and
\beq\label{eq-monodromyB}
\psi(D+\mathcal B_i,\hbar) 
= \exp\Bigg( {\displaystyle \sum_{(h,n,m) \in \mathbb{N}^3}} \frac{\hbar^{2h-2+n+m}}{n!m!} \overbrace{\int_D\dots\int_D}^n \overbrace{\int_{\mathcal B_i}\dots\int_{\mathcal B_i}}^m \omega_{h,n+m}\Bigg) .
\eeq
\el

\begin{proof}
Let $z\in \Sigma\setminus x^{-1}(\mathcal P)$, $D = [z]-[\infty^{(\alpha)}]$ and let $\{z,z^{(2)},\dots,z^{(d)}\}=x^{-1}(x(z))$.
Using that $[\mathcal I_{\mathcal C{z^{(i_j)},1}},\mathcal I_D]=0$, we have
\beq\label{psilfromI1psi}
\psi_{l,1}(z,D,\hbar) 
= \sum_{2\leq i_1<\dots < i_l\leq d } \prod_{j=1}^l d\zeta_{z^{(i_j)}}(z^{(i_j)}) \ \left( \prod_{j=1}^l  \mathcal I_{\mathcal C_{z^{(i_j)},1}} \right) \  \psi(D,\hbar).
\eeq
Notice also that $[\mathcal I_{\mathcal C{z^{(i_j)},1}},\mathcal I_\gamma]=0$, for $\gamma=\mathcal A_i$, $\gamma=\mathcal B_i$ or $\gamma=\mathcal C_p$.
In other words we need to prove the lemma only for $l=0$.

For $\gamma=\mathcal A_i$, we  have $D\cap \mathcal A_i=0$, so we get
\beq
\psi(D+\mathcal A_i,\hbar) 
= \exp\Bigg({\displaystyle \sum_{(h,n,m) \in \mathbb{N}^3}} \frac{\hbar^{2h-2+n+m}}{n!m!} \overbrace{\int_D\dots\int_D}^n \overbrace{\int_{\mathcal A_i}\dots\int_{\mathcal A_i}}^m \omega_{h,n+m}\Bigg).
\eeq
Notice that if $\gamma=\mathcal A_i$, we have $\int_\gamma \omega_{0,2}(z,\cdot)=0$, which implies (general property of TR invariants in  \cite{EO07}) that $\int_\gamma \omega_{h,n}(z_1,\dots,z_{n-1},\cdot)=0$, for all $(h,n)\neq (0,1)$. For $(h,n)=(0,1)$, we have by definition
\beq
\int_{\mathcal A_i} \omega_{0,1} = 2\pi\ii \epsilon_i.
\eeq
This implies that the sum over $m$ has only 2 non vanishing terms: $m=0$ and $m=1$, which gives
\beq
\psi(D+\mathcal A_i,\hbar) 
= e^{\hbar^{-1} \int_{\mathcal A_i} \omega_{0,1}} \psi(D,\hbar)
= e^{2\pi\ii \hbar^{-1} \epsilon_i} \psi(D,\hbar).
\eeq
The proof for $\gamma=\mathcal C_p$ is similar, the only difference being that
\beq
\mathcal C_p \cap D = \frac{1}{2\pi\ii} \delta_{p,\infty^{(\alpha)}},
\eeq
so that
\beq
\psi(D+\mathcal C_p,\hbar) 
= e^{\pi\ii \delta_{p,\infty^{(\alpha)}}} \  e^{2\pi\ii \hbar^{-1} t_{p,0}} \psi(D,\hbar),
\eeq
leading to the result.

For $\gamma=\mathcal B_i$, we  have $D\cap \mathcal B_i=0$, so we get the result.
\end{proof}

Let us now recall that the $\mathcal B_i$ period of $\om_{h,n+1}$ is equal to the variation of $\om_{h,n}$ with respect to $\epsilon_i$ thanks to \eqref{eq-var-epsilon}. Hence, one can formally interpret \eqref{eq-monodromyB} as
\beq
\psi(D+\mathcal B_i,\hbar) 
= \exp\Bigg({\displaystyle \sum_{(h,n) \in \mathbb{N}^2}} \frac{\hbar^{2h-2+n}}{n!} \overbrace{\int_D\dots\int_D}^n {\displaystyle \sum_{m \geq 0}} \frac{1}{m!} \left(\hbar \frac{\partial}{\partial \epsilon_i} \right)^m \omega_{h,n}\Bigg), 
\eeq
where we recognize the action of the shift operator $e^{ \hbar \frac{\partial}{\partial \epsilon_i}}$ leading to the notation
\beq
\psi_{l}([z+\mathcal{B}_j]-[\infty^{(\alpha)}],\hbar) =   e^{ \hbar \frac{\partial}{\partial \epsilon_j}} \psi_{l}([z]-[\infty^{(\alpha)}],\hbar) =  \psi_{l}([z]-[\infty^{(\alpha)}],\hbar,\epsilon_j\to\epsilon_j+\hbar), 
\eeq
for all $j\in\llbracket 1,g\rrbracket$, where the last equation is to be understood from the point of view of WKB series.

\br
Remark that the KZ equation itself does not depend on $z$ as a point in the universal cover of $\Sigma$ but only on its image $x(z)$. This simple observation implies that, for any 
$ \gamma \in H_1(\Sigma_{\mathcal P},\mathbb Z)$, 
\beq
 \psi_{l}([z+\gamma]-[\infty^{(\alpha)}],\hbar) 
\eeq
is a solution of the same KZ equation as $\psi^{\mathrm{reg}}_{l}([z]-[\infty^{(\alpha)}],\hbar)$, for every $l\in \llbracket0,d\rrbracket$ .

More generally, any linear combination of perturbative wave functions shifted by closed cycles obey the same KZ equation. For any finite family of $c_\gamma$, the following  finite sum satisfies the KZ equations
\beq
\psi_l([z]-[\infty^{(\alpha)}],\hbar,\{c_\gamma\}) \coloneqq \sum_{\gamma\in \pi_1(\Sigma\setminus x^{-1}(\mathcal P))}
c_\gamma \ \psi_l([z]+\gamma-[{\infty^{(\alpha)}}],\hbar).
\eeq
\er

\section{Non-Perturbative KZ equations and Lax system}\label{S6}

We have proved that the topological recursion allows to reconstruct solutions of the associated KZ type equations. We now wish to build solutions to the same equations but having better monodromy properties along the $(\mathcal{B}_i)_{i=1}^g$ cycles. For this purpose, let us consider the non-perturbative partners, introduced in \cite{Eynard_2009, EMhol}, to the above perturbative wave functions.

\subsection{Definitions and properties}\label{6.1}

\subsubsection{Heuristic motivation}

So far, the perturbative wave functions are defined on the universal cover of $\Sigma\setminus x^{-1}(\mathcal P)$, and have non-trivial monodromies. In order to find the quantum curve, we need to find a wave function  with ``trivial'' monodromies given by a simple phase along any element of $\pi_1(\Sigma\setminus x^{-1}(\mathcal P))$. Naively, this could be achieved by considering discrete Fourier transform of the perturbative wave functions.

We have noticed that any linear combination 
\beq
\psi_l([z]-[\infty^{(\alpha)}],\hbar,\{c_\gamma\}) \coloneqq \sum_{\gamma\in \pi_1(\Sigma\setminus x^{-1}(\mathcal P))}
c_\gamma \ \psi_l([z]+\gamma-[{\infty^{(\alpha)}}],\hbar)
\eeq
satisfies the KZ equations.
After a monodromy along $\gamma'$, it is worth
\beq
\psi_l([z]+\gamma'-[\infty^{(\alpha)}],\hbar,\{c_\gamma\}) 
=\psi_l([z]-[\infty^{(\alpha)}],\hbar,\{c_{\gamma-\gamma'}\}) .
\eeq
It has a trivial monodromy along $\gamma'$ if $c_{\gamma-\gamma'}$ is proportional to $c_\gamma$, for example
if there exists $\mathbf\rho$ in the dual of $H_1(\Sigma\setminus x^{-1}(\mathcal P),\mathbb Z)$, such that 
\beq
c_\gamma = e^{\frac{2\pi\ii}{\hbar} \langle\gamma,\mathbf\rho\rangle}.
\eeq
We thus would like to define
\beq
\psi_l^{\infty^{(\alpha)}} (z,\hbar,\mathbf\rho) \coloneqq \sum_{\gamma\in H_1(\Sigma\setminus x^{-1}(\mathcal P),\mathbb Z)}
e^{\frac{2\pi\ii}{ \hbar} \langle\gamma,\mathbf \rho\rangle} \ \psi_l([z]+\gamma-[{\infty^{(\alpha)}}],\hbar).
\eeq
Shifts by cycles $\mathcal A_i$ or $\mathcal C_p$ are already trivial in the sense that they consist in the multiplication by a simple phase factor. Thus we only need to take the sum over $\gamma =\underset{i=1}{\overset{g}{\sum}}n_i \mathcal B_i$, which amounts to $\epsilon_i\to \epsilon_i+\hbar$. 
We see that formally, this wave function would be a sort of Fourier transform of the perturbative wave function, i.e.~if the perturbative wave function were a function of $\boldsymbol{\epsilon} = (\epsilon_1,\dots,\epsilon_g)$, we would write
\beq
\psi_l^{\infty^{(\alpha)}} (z,\hbar;\boldsymbol{\epsilon},\boldsymbol{\rho}) \coloneqq \sum_{\mathbf n\in \mathbb Z^g}
e^{\frac{2\pi\ii}{ \hbar} \underset{j=1}{\overset{g}{\sum}} \rho_j n_j} \ \psi_l([z]-[{\infty^{(\alpha)}}],\hbar,\mathbf \epsilon+\hbar \mathbf n).
\eeq

\begin{remark}[Limitations of the Heuristic definition]
We recall that
\begin{itemize}
\item The filling fraction $\boldsymbol{\epsilon}=(\epsilon_1,\dots,\epsilon_g)$ is not a global coordinate on the space of classical spectral curves with fixed spectral times, it is only a local coordinate.
Hence the set of classical spectral curves with filling fractions  $\mathbf \epsilon+\hbar\mathbb Z$ is not well defined for finite $\hbar$. It makes sense only from a formal WKB series point of view.
\item The sum over $\gamma\in H_1(\Sigma\setminus x^{-1}(\mathcal P),\mathbb Z) $, is not a finite sum, it is not necessarily defined in $\mathcal W$. 
\item All this implies that we shall have to enlarge $\mathcal W$ to contain such sums. We need to consider trans-series in $\hbar$ of the form
\beq \label{trans-series1}
\sum_{\mathbf{n} \in \mathbb{Z}^g} \sum_{r \geq 0} F_{\mathbf{n},r} \hbar^r e^{\frac{1}{\hbar}  {\displaystyle \sum_{j=1}^g} n_j v_j},
\eeq
using the usual ordering of trans-monomials which can be seen from the fact that the sum over $r$ is performed first.
\end{itemize}

Unfortunately, working with trans-series of the form \eqref{trans-series1} would not allow to derive the  quantum curve directly and we shall consider another ordering of the trans-monomials by first summing over $\mathbf{n} \in \mathbb{Z}^g$ and considering series of the form 
\beq \label{trans-series2}
 \sum_{r \geq 0} \sum_{\mathbf{n} \in \mathbb{Z}^g} F_{\mathbf{n},r} \hbar^r e^{\frac{1}{\hbar}  \underset{j=1}{\overset{g}{\sum}} n_j v_j}.
\eeq
In our case, the partial sums $\underset{\mathbf{n} \in \mathbb{Z}^g}{\sum} F_{\mathbf{n},r}  e^{\frac{1}{\hbar}   \underset{j=1}{\overset{g}{\sum}} n_j v_j}$ will give rise to Theta functions through convergent series in the spirit of the trans-asymptotics of \cite{CostinCostin01}.
\end{remark}

\subsubsection{Shifts of the divisors and action on $\om_{1,0}$ and $\om_{2,0}$}

Before going further, let us study how the shifts of the divisor act on the WKB series defining the perturbative wave functions. As a formal series in powers of $\hbar$, with $D=[z]-[\infty^{(\alpha)}]$ we have

\bea
\psi^{\mathrm{reg}}([z]+\gamma-[{\infty^{(\alpha)}}],\hbar)
&=& e^{\underset{h\geq 0}{\sum}\underset{n\geq 0}{\sum} \frac{\hbar^{2h-2+n}}{n!}\int_{\gamma+D} \dots\int_{\gamma+D} \omega_{h,n} } \cr
&=& e^{\hbar^{-2}\omega_{0,0}+\omega_{1,0}} 
e^{\hbar^{-1} \int_D \omega_{0,1}} 
e^{\frac{1}{2}\int_D\int_D \omega_{0,2}} \cr
&& e^{\hbar^{-1} \int_\gamma \omega_{0,1}}
e^{\frac{1}{2} \int_\gamma\int_D \omega_{0,2}}
e^{\frac{1}{2} \int_D\int_\gamma \omega_{0,2}}
e^{\frac{1}{2} \int_\gamma\int_\gamma \omega_{0,2}} \left( 1+\hbar \mathbb C[[\hbar]] \right) .
\eea
The admissible data fixes a basis of $H_1(\Sigma\setminus x^{-1}(\mathcal P),\mathbb Z)$ in which one can decompose $\gamma$ as
\beq
\gamma= \sum_{p\in x^{-1}(\mathcal P)} m_p \mathcal C_p + \sum_{i=1}^g m_i \mathcal A_i + \sum_{i=1}^g n_i \mathcal B_i.
\eeq

We have
\beq
\int_\gamma \omega_{0,1} = 2\pi\ii \left(\sum_{p\in x^{-1}(\mathcal P)} m_p t_{p,0} + \sum_{i=1}^g m_i \epsilon_i \right) + \sum_{i=1}^g n_i \int_{\mathcal B_i}\omega_{0,1},
\eeq
\bea
\frac{1}{2} \int_\gamma\int_\gamma \omega_{0,2} 
&=& \pi\ii \left( \sum_{j=1}^g n_j m_j  + \sum_{i=1}^g \sum_{j=1}^g n_i \tau_{i,j} n_j \right),
\eea
where $\tau$ is the Riemann matrix of periods of $\Sigma$ and 
\beq
\frac{1}{2} \int_D\int_\gamma \omega_{0,2} 
= \pi\ii \sum_{j=1}^g n_j u_j(D)  = \frac12 \int_\gamma\int_D \omega_{0,2} + \pi\ii m_{\infty^{(\alpha)}},
\eeq
where $u_j(D) = \int_D w_j$ with $(w_j)_{j=1}^g$ forming a basis of holomorphic differentials on $\Sigma$ normalized by $\frac{1}{2 \pi \ii} \int_{\mathcal{A}_i} w_j = \delta_{i,j}$.

\subsubsection{Definitions}

As explained above, only the $(\mathcal{B}_i)_{i=1}^g$ cycles give interesting contributions. We shall thus only consider shifted wave functions of the form
\bea
 \psi^{\mathrm{reg\; symbol}}(D+ \mathcal{B}_{i_1}+\dots + \mathcal{B}_{i_k},\hbar)
&=& e^{\hbar^{-2}\omega_{0,0}+\omega_{1,0}}
e^{\hbar^{-1}(\int_D \omega_{0,1}+\underset{j=1}{\overset{k}{\sum}} \oint_{\mathcal B_{i_j}}\omega_{0,1})}\cr
&& e^{2\pi\ii \underset{j=1}{\overset{k}{\sum}} u_{i_j}(D)} e^{\pi\ii  \underset{j,j' =1}{\overset{k}{\sum}} \tau_{i_j,i_{j'}}}  \cr
&& \frac{1}{E^{\mathrm{symbol}}(D)}
 \left( 1 + \sum_{r=1}^\infty \hbar^r G^{(r)}_{(i_1,\dots,i_k)}(D) \right)
\; \in \mathcal W,\cr
&&
\eea
where the coefficients of the $\hbar$ expansion can be explicitly computed. As symbols, a simple computation shows that they read as follow.

\bl
Let $r\geq 0$ and $k\geq 0$.
For any tuple $(i_1,\dots,i_k) \in \llbracket 1, g \rrbracket$ (empty if $k=0$) we define the following symbols in $\mathcal W$.

\begin{itemize}

\item For $r=0$, we define
\beq
G^{(0)}_{(i_1,\dots,i_k)}(D)=\delta_{k,0}.
\eeq
\item For $r\geq 1$, we define
\bea
G^{(r)}_{(i_1,i_2,\dots,i_k)}(D)
&=& \sum_{\ell=1}^{r} \frac{1}{\ell!} \sum_{(h_1,n_1),\dots,(h_\ell,n_\ell)}
\delta\bigg(r=k+\sum_{j=1}^\ell 2h_j-2+n_j\bigg) \cr
&& 
\Bigg( \mathcal I_{\mathcal B_{i_1}}\dots \mathcal I_{\mathcal B_{i_k}} 
\prod_{j=1}^{\ell} \Bigg(\frac{1}{n_j !} \overbrace{\int_D \dots \int_D}^{n_j}   \omega_{h_j,n_j}\Bigg)
\Bigg)_{\text{stable}}  ,
\eea
where the subscript ``stable'' means that after acting with the operators $\mathcal I_{\mathcal B_j}$ with Leibniz rules, we exclude from the final sum all symbols that contain a factor $\overbrace{\int_D \dots \int_D}^{n_j} \overbrace{\int_{\mathcal{B_{\alpha_1}}} \dots \int_{\mathcal{B}_{\alpha_{m_j}}}}^{m_j} \omega_{h_j,n_j+m_j}$ such that $2h_j-2+n_j+m_j\leq 0$.
\end{itemize} 

One has, for any $(i_1,\dots,i_k) \in \llbracket 1,g \rrbracket^k$,
\bea
 \psi^{\mathrm{reg\; symbol}}(D+ \mathcal{B}_{i_1}+\dots + \mathcal{B}_{i_k},\hbar)
&=& e^{\hbar^{-2}\omega_{0,0}+\omega_{1,0}}
e^{\hbar^{-1}(\int_D \omega_{0,1}+\underset{j=1}{\overset{k}{\sum}} \oint_{\mathcal B_{i_j}}\omega_{0,1})}\cr
&& e^{2\pi\ii \underset{j=1}{\overset{k}{\sum}}  u_{i_j}(D)} e^{\pi\ii\underset{j,j'=1}{\overset{k}{\sum}}  \tau_{i_j,i_{j'}}}  \cr
&&  \frac{1}{E^{\mathrm{symbol}}(D)}
 \left( 1 + \sum_{r=1}^\infty \hbar^r G^{(r)}_{(i_1,\dots,i_k)}(D) \right)
\; \in \mathcal W .\cr
&&
\eea
\el

\br
Notice that $G^{(r)}_{(i_1,\dots,i_k)}(D)$ is a finite sum and $G^{(r)}_{(i_1,\dots,i_k)}(D)=0$, if $k>3r$.
Indeed  if we want to satisfy stability we must have $r=\underset{j=1}{\overset{\ell }{\sum}} (2h_j-2+n_j+m_j)$ with $m_1+\dots+m_\ell=k$ and each term is $\geq 1$, therefore $r\geq \ell$. In  $r-k=\underset{j=1}{\overset{\ell }{\sum}} 2h_j-2+n_j$, each term is $\geq -2$, therefore $r-k\geq -2\ell$, which implies $k\leq r+2\ell\leq 3r$.
\er 

These coefficients $G^{(r)}_{(i_1,i_2,\dots,i_k)}(D)$ are easily computed. Let us present the first few ones.
\beq
G^{(1)}_{\emptyset}(D) = \int_{D} \omega_{1,1} + \frac16   \int_D \int_D \int_D  \omega_{0,3}
\, , \qquad 
G^{(1)}_{(i_1)}(D) = \int_{\mathcal B_{i_1}} \omega_{1,1} + \frac12   \int_D \int_D \int_{\mathcal B_{i_1}}  \omega_{0,3},
\eeq
\beq
G^{(1)}_{(i_1,i_2)}(D) = \int_D \int_{\mathcal B_{i_2}}  \int_{\mathcal B_{i_1}}  \omega_{0,3}
\, ,
\qquad
G^{(1)}_{(i_1,i_2,i_3)}(D) = \int_{\mathcal B_{i_3}} \int_{\mathcal B_{i_2}}  \int_{\mathcal B_{i_1}}  \omega_{0,3}
\eeq
and
\bea
G^{(2)}_{(i_1,i_2)}(D) &=& \frac12 \int_{\mathcal B_{i_1}} \omega_{1,1} \int_{\mathcal B_{i_2}} \omega_{1,1} \cr
&& +  \frac12 \int_{\mathcal B_{i_1}} \omega_{1,1}\int_D \int_D \int_{\mathcal B_{i_2}}  \omega_{0,3}+  \frac12 \int_{\mathcal B_{i_2}} \omega_{1,1}\int_D \int_D \int_{\mathcal B_{i_1}}  \omega_{0,3}\cr
&& +  \frac18 \int_D \int_D \int_{\mathcal B_{i_1}}  \omega_{0,3}\int_D \int_D \int_{\mathcal B_{i_2}}  \omega_{0,3} \cr
&& + \frac12 \int_D \int_D \int_{\mathcal B_{i_2}}  \int_{\mathcal B_{i_1}}\omega_{0,4}\,.
\eea

In the following, we use the same notation $G^{(r)}_{(i_1,\dots,i_k)}(D)$ for these coefficients both as symbols and after evaluation, when the context makes clear which one is considered.

Then the method of \cite{Eynard_2009, EMhol} consists in exchanging the formal summation of powers of $\hbar$ with the summation over cycles. 
Remarking that the exponential part with negative powers of $\hbar$ is at most quadratic in the cycles, and the positive powers of $\hbar$ can be expanded as polynomials outside of the exponentials, our heuristic motivation would involve the definition  of theta functions. Let us now introduce them as symbols.

\bd[Theta symbols]
We enlarge the symbol algebra $\mathcal W$ to an algebra $\hat{\mathcal W}$  by introducing the following new symbols.
Let  $k\geq 0$ and $\tau$ the $g\times g$ matrix of symbols with entries defined by
\beq
\forall \, (i,j)\in \llbracket 1,g\rrbracket^2\,:\, \tau_{i,j} = \frac{1}{2\pi\ii} \int_{\mathcal B_{i}}\int_{\mathcal B_j} \omega_{0,2}.
\eeq
For any tuple $(i_1,\dots,i_k) \in \llbracket 1 , g \rrbracket^k$ (empty if $k=0$), and any vector $\mathbf v=(v_1,\dots ,v_g)\in \check{\mathcal W}_0[[\hbar]]^g$,
we introduce the following symbol 
\beq\label{eq:defTheta}
\Theta^{(i_1,\dots,i_k)}(\mathbf v,\tau)
= \sum_{(n_1,\dots,n_g)\in \mathbb Z^g} 
e^{2 \pi \ii \underset{i=1}{\overset{g }{\sum}} n_i v_i}
e^{\pi\ii \underset{i,j=1}{\overset{g }{\sum}} n_i \tau_{i,j} n_j} 
\prod_{r=1}^k n_{i_r} .
\eeq
If $k=0$, we just write it $\Theta^{\emptyset}(v,\tau)=\Theta(v,\tau)$.
We define the degree of the symbol $\Theta$ to be zero.
We shall quotient by the following relations. For $j\in \llbracket 1,g\rrbracket$, 
\beq\label{eqThetaperiod1}
\forall \ j\in \llbracket 1,g\rrbracket, \quad \Theta^{(i_1,\dots,i_k)}(\mathbf v+\mathbf e_j,\tau) \equiv \Theta^{(i_1,\dots,i_k)}(\mathbf v,\tau) ,
\eeq
where
\beq
\mathbf e_j=\left(\int_{\mathcal A_j} \int_{\mathcal B_1} \omega_{0,2} , \dots,\int_{\mathcal A_j} \int_{\mathcal B_g} \omega_{0,2}\right) \in \check{\mathcal W}_0[[\hbar]]^g .
\eeq
For $j\in \llbracket 1,g\rrbracket$, 
\bea\label{eqThetaperiodtau}
\Theta^{(\overbrace{1,\dots,1}^{k_1},\overbrace{2,\dots,2}^{k_2},\dots,\overbrace{g,\dots,g}^{k_g})}(\mathbf v+\td{\mathbf e}_j,\tau)
&\equiv & e^{-v_j} \ e^{-\pi\ii \tau_{j,j}}  \sum_{m=0}^{k_j} (-1)^m \frac{k_j!}{m! (k_j-m)!}  \ \cr
&& \Theta^{(\overbrace{1,\dots,1}^{k_1},\overbrace{2,\dots,2}^{k_2},\dots,\overbrace{j,\dots,j}^{k_j-m},\dots,\overbrace{g,\dots,g}^{k_g})}(\mathbf v,\tau),
\cr&&
\eea
where
\beq
\td{\mathbf e}_j=\left(\int_{\mathcal B_j} \int_{\mathcal B_1} \omega_{0,2} , \dots,\int_{\mathcal B_j} \int_{\mathcal B_g} \omega_{0,2}\right) \in \check{\mathcal W}_0[[\hbar]]^g .
\eeq

We denote this graded ring
\beq
\hat{\mathcal W} = \mathcal W[\Theta]\ /\equiv,
\eeq
where the grading is the power of $\hbar$ (we recall that $\deg \Theta^{(\mathbf k)}=0$).

We define the evaluation of $\Theta$ symbols below in Section~\ref{sec-eval-theta}.
\hfill $\star$
\ed

The operators $\mathcal I_C$ act on elements of $\hat{\mathcal W}$ term by term through their action on elements of $\mathcal{W}$ by using linearity and Leibniz rules. In particular, they act on the symbol $\Theta$ as follows:
\beq\label{defThetaIC}
\mathcal I_C \left[\Theta^{(i_1,\dots,i_k)}(\mathbf v,\tau)\right]
\coloneqq 2 \pi \ii  \sum_{i=1}^g (\mathcal I_C.v_i) \ \Theta^{(i,i_1,\dots,i_k)}(\mathbf v,\tau)
-\pi^2 \sum_{i,j=1}^{g} (\mathcal I_C.\tau_{i,j}) \ \Theta^{(i,j,i_1,\dots,i_k)}(\mathbf v,\tau).
\eeq

\medskip

With these definitions, one can define the non-perturbative partner to the wave functions.

\bd[Non perturbative wave function] \label{def-non-pert-wave-fct}
Let $\boldsymbol{\rho}=(\rho_1,\dots,\rho_g)$ be a constant formal element of $\check{\mathcal{W}}^g$ and $D=[z]-[\infty^{(\alpha)}]$.
We define the symbolic non-perturbative wave function
\beq
\psi_{\mathrm{NP}}^{\mathrm{symbol}}(D;\hbar,\boldsymbol{\rho}) \coloneqq e^{\hbar^{-2}\omega_{0,0}+\omega_{1,0}}  e^{\hbar^{-1}\int_D \omega_{0,1}}\frac{1}{E^{\mathrm{symbol}}(D)} \quad \sum_{r=0}^\infty \hbar^r G^{(r)}(D;\boldsymbol{\rho}) \quad \in \hat{\mathcal W},
\eeq
where
\beq
G^{(r)}(D;\boldsymbol{\rho}) \coloneqq \sum_{k=0}^{3r} \quad
\sum_{i_1,\dots,i_k\in \llbracket 1,g\rrbracket^k}
\Theta^{(i_1,\dots,i_k)}(\mathbf v,\tau) G^{(r)}_{(i_1,\dots,i_k)} (D)
\eeq
and where $\mathbf v=(v_1,\dots,v_g)$ is the vector in $\mathcal W^g$,
\beq
v_j \coloneqq \frac{\rho_j+\phi_j}{\hbar} + \mu_j^{(\alpha)}(z),
\eeq
where
\beq
\phi_j \coloneqq \frac{1}{2 \pi \ii} \oint_{\mathcal{B}_j} \om_{0,1} 
\qquad 
\text{and}
\qquad
\mu_j^{(\alpha)}(z) \coloneqq  \frac{1}{2 \pi \ii} \int_D \oint_{\mathcal{B}_j} \om_{0,2}.
\eeq

For $l\geq 0$, we define 
\beq\label{psilfromI1psiNPalpha}
\psi_{l,\mathrm{NP}}^{\infty^{(\alpha)},\,\mathrm{symbol}}(z,\hbar,\boldsymbol{\rho}) 
\coloneqq \sum_{\mathbf{\beta} \underset{l}{\subseteq} \left(x^{-1}(x(z)) \setminus \{z\}\right)} \frac{1}{l!} \left( \prod_{j=1}^l  \mathcal I_{\mathcal C_{\beta_j,1}} \right) \  \psi_{\mathrm{NP}}^{\mathrm{symbol}}(D;\hbar,\boldsymbol{\rho}).
\eeq
We use them to define a $d \times d$ matrix
\beq
\widehat{\Psi}_{\mathrm{NP}}^{\mathrm{symbol}}(\lambda,\hbar,\boldsymbol{\rho})\coloneqq \left[\psi_{l-1,\mathrm{NP}}^{\infty^{(\alpha)},\,\mathrm{symbol}}(z^{(\beta)}(\lambda),\hbar,\boldsymbol{\rho})\right]_{1\leq l,\beta\leq d},
\eeq
where $z^{(\beta)}(\lambda)$ denotes the $\beta^{\text{th}}$ preimage by $x$ of $\lambda$.

The non-perturbative wave function
\bea
\psi_{\mathrm{NP}}(D;\hbar,\boldsymbol{\rho}) 
&\coloneqq & \mathrm{ev}. \psi_{\mathrm{NP}}^{\mathrm{symbol}}(D;\hbar,\boldsymbol{\rho}) \cr
 & = &
e^{\hbar^{-2}\omega_{0,0}+\omega_{1,0}}  e^{\hbar^{-1}\int_D \omega_{0,1}}
\frac{1}{E(D)}  \sum_{r=0}^\infty \hbar^r G^{(r)}(D; \boldsymbol{\rho}) \quad \in \hat{\mathcal W}
\eea
denotes the result of the evaluation map on the symbolic non-perturbative wave function. Similarly, one defines
\beq
\psi_{l,\mathrm{NP}}^{\infty^{(\alpha)}}(z,\hbar,\boldsymbol{\rho}) \coloneqq \mathrm{ev}.
\psi_{l,\mathrm{NP}}^{\infty^{(\alpha)},\,\mathrm{symbol}}(z,\hbar,\boldsymbol{\rho}) 
\,\,
\text{ and }
\,\,
\widehat{\Psi}_{\mathrm{NP}}(\lambda,\hbar,\boldsymbol{\rho})\coloneqq \left[\psi_{l-1,\mathrm{NP}}^{\infty^{(\alpha)}}(z^{(\beta)}(\lambda),\hbar,\boldsymbol{\rho})\right]_{1\leq l,\beta\leq d}.
\eeq

We define the \emph{non-perturbative partition function} by
\beq
Z_{\mathrm{NP}}^{\mathrm{symbol}}(\hbar,\boldsymbol{\rho})
\coloneqq e^{\hbar^{-2}\omega_{0,0}+\omega_{1,0}} \sum_{r=0}^\infty \hbar^r G^{(r)}(\emptyset;\boldsymbol{\rho}),
\eeq
as well as its evaluated version
\beq
Z_{\mathrm{NP}}(\hbar,\boldsymbol{\rho}) \coloneqq \mathrm{ev}. Z_{\mathrm{NP}}^{\mathrm{symbol}}(\hbar,\boldsymbol{\rho}).
\eeq
\hfill $\star$
\ed

\br
Heuristically, the non-perturbative partition function is the Fourier transform of the perturbative one
\beq
Z_{\mathrm{NP}}(\hbar,\boldsymbol{\epsilon},\boldsymbol{\rho}) \stackrel{\text{formal}}{\coloneqq}
\sum_{\mathbf{n} \in \mathbb{Z}^g} \exp \left( \frac{2 \pi i}{\hbar}  \underset{j=1}{\overset{g}{\sum}}  n_j \rho_j \right) Z(\hbar,\boldsymbol{\epsilon}+ \hbar \mathbf{n}).
\eeq
We shall argue below in the examples that $Z_{\mathrm{NP}}(\hbar,\boldsymbol{\rho})$ is formally a Tau-function of some integrable system. 
The inverse Fourier transform of $Z_{\mathrm{NP}}(\hbar,\boldsymbol{\rho})$, which is the perturbative partition function, can be seen as the ``Witham averaging'' of $Z_{\mathrm{NP}}(\hbar,\boldsymbol{\rho})$, formally replacing all the $\Theta$ by $1$, i.e. erasing  the oscillations.
\er

As we shall see in the examples of Section~\ref{sec-examples}, this non-perturbative partition function is expected to play the role of isomonodromic tau function of the system built by quantization of the classical spectral curve. In order to understand this point of view, let us remark that $Z_{\mathrm{NP}}(\hbar,\boldsymbol{\rho}) $ follows from the expansion of 
$\psi_{0,\mathrm{NP}}^{\infty^{(\alpha)}}(z,\hbar,\boldsymbol{\rho})$ around $z \to \infty^{(\alpha)}$.
Indeed, one has
\beq\label{eq-expansion-psi-infty}
\forall \, \alpha \in\llbracket 1,d \rrbracket\, : \; \psi_{0,\mathrm{NP}}^{\infty^{(\alpha)}}(z,\hbar,\boldsymbol{\rho}) = e^{\hbar^{-1} V_{\infty^{(\alpha)}}(z)} Z_{\mathrm{NP}}(\hbar,\boldsymbol{\rho}) \left(1 + O(x(z)^{-1}) \right),
\eeq
when $z \to \infty^{(\alpha)}$, since the divisor $D=[z]-[\infty^{(\alpha)}]$ gives vanishing integration constants when $z \to \infty^{(\alpha)}$.

\subsection{Trans-series and properties}\label{6.2}

In this section, we collect a few properties of the elements of $\hat{\mathcal{W}}$ which will be useful for the derivation of the quantum curve.

\subsubsection{Evaluation of Theta symbol} \label{sec-eval-theta}

Remark that the evaluation $\mathrm{ev}.  \tau$ is the Riemann matrix of periods of $\Sigma$ with the Torelli marking, and it is a well known theorem that it is a Siegel matrix, i.e.~symmetric $\tau_{i,j}=\tau_{j,i}$ and  $\text{Im } \tau$ is positive definite,
\beq
\tau=\tau^t\, , \qquad
\text{Im } \tau>0.
\eeq
For any Siegel matrix $\tau$, the following sum (called Riemann Theta function)
\beq
\Theta_{\text{Riemann}}(\mathbf v,\tau)
= \sum_{\mathbf n\in \mathbb Z^g} e^{2\pi\ii (\mathbf v,\mathbf n)}  \ e^{\pi\ii \ (\mathbf n,\tau \mathbf n)}
\eeq
is absolutely convergent  (uniformly in $\mathbf v$ in any ball) and is an analytic function of $\mathbf v\in \mathbb C^g $.

It has the quasi-periodicity properties:
\beq
\forall\, \mathbf n\in \mathbb Z^g \; , \,\Theta_{\text{Riemann}}(\mathbf v+\mathbf n,\tau)
=\Theta_{\text{Riemann}}(\mathbf v,\tau)
\eeq
and
\beq
\forall \,\mathbf n\in \mathbb Z^g \; , \,
\Theta_{\text{Riemann}}(\mathbf v+\tau\mathbf n,\tau)
=e^{-2\pi\ii((\mathbf n,\mathbf v) + \frac{1}{2} (\mathbf n,\tau \mathbf n))} \Theta_{\text{Riemann}}(\mathbf v,\tau).
\eeq

The evaluation map applied to the Theta symbol recovers the Riemann Theta function
\beq
\mathrm{ev}.\Theta(\mathbf v,\tau)  = \Theta_{\text{Riemann}}\left(\frac{1}{2\pi\ii}\mathbf v,\tau\right)
\qquad 
\text{and}
\qquad
\mathrm{ev}.\Theta^{(i_1,\dots,i_k)}(\mathbf v,\tau)
= \frac{1}{(2 \pi \ii)^k} \prod_{j=1}^k \frac{\partial}{\partial v_{i_j}} \ \mathrm{ev}.\Theta(\mathbf v,\tau).
\eeq
In other words
\beq\label{eq:defThetadef}
\text{ev}.\Theta^{(i_1,\dots,i_k)}(\mathbf v,\tau)
= \sum_{(n_1,\dots,n_g)\in \mathbb Z^g} 
e^{2 \pi \ii {\displaystyle \sum_{i=1}^g} n_i v_i}
e^{\;\pi\ii\!\! {\displaystyle  \sum_{(i,j) \in \llbracket 1 , g \rrbracket^2}} n_i \tau_{i,j} n_j} 
\prod_{j=1}^k n_{i_j}.
\eeq

\subsubsection{Equalities for trans-series}\label{sectiontransseries}

The non-perturbative wave functions and partition functions are elements of $\hat{\mathcal{W}}$ where the argument $\mathbf{v}$ involves $\hbar$ by taking the form $\mathbf{v} = \hbar^{-1} \mathbf{\phi} + \mathbf{\mu}$. With this particular evaluation of the argument of the function Theta, elements of $\hat{\mathcal{W}}$ take the form of trans-series in $\hbar$
\beq
\sum_{r=0}^\infty \sum_{\mathbf{n} \in \mathbb{Z}^g} \hbar^r e^{\frac{1}{\hbar} {\displaystyle \sum_{j=1}^g} n_j \phi_j} F_{r,\mathbf{n}},
\eeq
where $F_{r,\mathbf{n}} \in \check{\mathcal{W}}$. The equalities that we shall write down from now on must be understood as equalities between trans-series, namely the trans-series above is vanishing if and only if the coefficient $F_{r,\mathbf{n}}$ of each trans-monomial  is vanishing  for any $(r,\mathbf{n})$.

An important remark is that the order of trans-monomial does not affect such an equality, i.e.
\beq\label{eq-order-trans}
\sum_{r=0}^\infty \sum_{\mathbf{n} \in \mathbb{Z}^g} \hbar^r e^{\frac{1}{\hbar}  {\displaystyle \sum_{j=1}^g} n_j \phi_j} F_{r,\mathbf{n}} = 0 
\Leftrightarrow
\sum_{\mathbf{n} \in \mathbb{Z}^g} \sum_{r=0}^\infty  \hbar^r e^{\frac{1}{\hbar} {\displaystyle \sum_{j=1}^g} n_j \phi_j} F_{r,\mathbf{n}} = 0, 
\eeq
but it does not mean that ${\displaystyle \sum_{r=0}^\infty \sum_{\mathbf{n} \in \mathbb{Z}^g}} \hbar^r e^{\frac{1}{\hbar}  {\displaystyle \sum_{j=1}^g} n_j \phi_j} F_{r,\mathbf{n}} $ and ${\displaystyle \sum_{\mathbf{n} \in \mathbb{Z}^g} \sum_{r=0}^\infty } \hbar^r e^{\frac{1}{\hbar}  {\displaystyle \sum_{j=1}^g} n_j \phi_j} F_{r,\mathbf{n}}$ are equal as they are two very different types of formal trans-series.

\subsubsection{Action of the operators $\mathcal{I}_C$ and $\hbar$ grading}

As functions of $\hbar$, these non-perturbative objects are of very different nature compared to their perturbative counterparts whose logarithms admit a formal series expansion in $\hbar$.  Let us emphasize some important properties of these trans-series.

\medskip

We have  observed that the symbolic non-perturbative wave function takes the form
\beq\label{eq-form-NP}
\psi_{0,\mathrm{NP}}^{\infty^{(\alpha)},\mathrm{\,symbol}}(z,\hbar,\boldsymbol{\rho})  = e^{\hbar^{-2} \om_{0,0}  + \hbar^{-1} S_{-1}(z) + S_0(z)} \bigg(1 + {\displaystyle \sum_{m=1}^\infty} \hbar^m \, \Xi_m^{(\infty_\alpha)}(z,\hbar, \boldsymbol{\rho}) \bigg),
\eeq
where $\Xi_m^{(\infty_\alpha)}(z,\hbar, \boldsymbol{\rho})$ are combinations of derivatives of theta functions of the form
$
\left. \frac{\partial^n \Theta(\mathbf{v}, \boldsymbol{\tau})}{\partial v_{i_1} \dots \partial v_{i_n}} \right|_{\mathbf{v} = \frac{\boldsymbol{\phi} + \boldsymbol{\rho}}{\hbar} + \boldsymbol{\mu}^{(\alpha)}(z)}, 
$
whose coefficients are functions of $z$ independent of $\hbar$. Note that the dependence of such functions on $\hbar$ only appears from the evaluation of the argument of the theta function at $\mathbf{v} = \frac{\boldsymbol{\phi} + \rho}{\hbar} + \boldsymbol{\mu}^{(\alpha)}(z)$.

\medskip

In order to prepare the proof of the main theorem, let us see how differential operators act on this type of trans-series. One has two types of very different behaviors.
\begin{itemize}

\item The operator $\hbar \partial_{x(z)}$ acts naively by  increasing the order in $\hbar$ by one
\beq
\hbar \frac{\partial}{\partial x(z)}  \bigg[\sum_{m=1}^{\infty} \hbar^m \, \Xi_m^{(\infty_\alpha)}(z,\hbar, \boldsymbol{\rho}) \bigg]= 
\sum_{m=1}^{\infty} \hbar^{m+1} \, \frac{\partial \Xi_m^{(\infty_\alpha)}(z,\hbar, \boldsymbol{\rho})}{\partial x(z)} .
\eeq

\item On the other hand, the operators $\mathcal{I}_{\mathcal{C}}$ act both through the explicit dependence of the theta functions but also through its action on the elements $(\phi_j)_{j=1}^g$. One has
\bea\label{eq-action-I-trans}
\hbar  \mathcal{I}_{\mathcal{C}}  \bigg[\sum_{m=1}^\infty \hbar^m \, \Xi_m^{(\infty_\alpha)}(z,\hbar, \boldsymbol{\rho}) \bigg] & =& 
\sum_{m=1}^\infty \hbar^{m+1} \, \left. \mathcal{I}_{\mathcal{C}}  \cdot \Xi_m^{(\infty_\alpha)}(z,\hbar, \boldsymbol{\rho}) \right|_{\phi \; \text{fixed}} \cr
&&+ \sum_{j=1}^g \mathcal{I}_{\mathcal{C}} \left[\phi_j \right]
\sum_{m=1}^\infty \hbar^{m} \, \left. \frac{\partial \Xi_m^{(\infty_\alpha)}(z,\hbar, \boldsymbol{\rho})}{\partial \phi_j} \right|_{\phi_j = \frac{1}{2 \pi i } \oint_{\mathcal{B}_j} \om_{0,1}}\!\!\!\!\!\!\!\!\!\!\!\!\!\!\!\!\!\!\!\!\!\!\!\!\!\!\!\!\!\!.\qquad\qquad\qquad\quad
\eea
\end{itemize}

\subsubsection{KZ equations}

One of the two important properties of the non-perturbative wave functions is that they satisfy the same KZ equations as their perturbative partners.

\bt
The symbolic non-perturbative wave functions satisfy the KZ equations
\begin{align}\label{eq-NP-KZ}
{\hbar} \frac{d  \psi_{l,\mathrm{NP}}^{\infty^{(\alpha)},\mathrm{\,symbol}}(z,\hbar,\boldsymbol{\rho}) }{dx(z)} +\psi_{l+1,\mathrm{NP}}^{\infty^{(\alpha)},\mathrm{\,symbol}}(z,\hbar,\boldsymbol{\rho}) =  \cr
 \sum_{P \in \mathcal{P}} \sum_{k \in S_P^{(l+1)}}  \xi_P^{-k}(x(z))  \left[  \widetilde{\mathcal{L}}_{P,k,l}\, \psi_{0,\mathrm{NP}}^{\infty^{(\alpha)},\mathrm{\,symbol}}(z,\hbar,\boldsymbol{\rho}) \right] .
\end{align}
After evaluation, this implies
\beq\label{eq-NP-KZ-ev}
{\hbar} \frac{d  \psi_{l,\mathrm{NP}}^{\infty^{(\alpha)}}(z,\hbar,\boldsymbol{\rho}) }{dx(z)} +\psi_{l+1,\mathrm{NP}}^{\infty^{(\alpha)}}(z,\hbar,\boldsymbol{\rho}) =  
 \sum_{P \in \mathcal{P}} \sum_{k \in S_P^{(l+1)}}  \xi_P^{-k}(x(z))  \mathrm{ev} . \left[  \widetilde{\mathcal{L}}_{P,k,l}\, \psi_{0,\mathrm{NP}}^{\infty^{(\alpha)},\mathrm{\,symbol}}(z,\hbar,\boldsymbol{\rho}) \right] .
\eeq
\et

\begin{proof}
The proof consists in first showing that another trans-series is solution to these KZ equations.

From the previous sections, for any $\mathbf{n} \in \mathbb{Z}^g$, 
\begin{align}
&{\hbar} \frac{d}{dx(z)}  \psi_{l}^{\mathrm{reg \; symbol}}([z+ \sum_{j=1}^g n_j \mathcal{B}_j]-[\infty^{(\alpha)}])   +\psi_{l+1}^{\mathrm{reg \; symbol}}([z+ \sum_{j=1}^g n_j \mathcal{B}_j]-[\infty^{(\alpha)}])   \cr
&=   \left[ \widetilde{\mathcal{L}}_l(x(z)) \left\{\psi^{\mathrm{reg\; symbol}}([z+ \sum_{j=1}^g n_j \mathcal{B}_j]-[\infty^{(\alpha)}]) \right\}\right].
\end{align}
This implies that the Fourier transforms
\beq
\sum_{\mathbf{n} \in \mathbb{Z}^g} e^{\frac{2 \pi \ii}{\hbar} \underset{j=1}{\overset{g}{\sum}} n_j \rho_j}   \psi_{l}^{\mathrm{reg \;  symbol}}([z+ \sum_{j=1}^g n_j \mathcal{B}_j]-[\infty^{(\alpha)}]) 
\eeq
are solutions to the KZ equations as trans-series. 

Let us remark that from the definition of the non-perturbative wave functions theses Fourier transforms are the trans-series obtained by exchanging the sum over $\mathbf{n} \in \mathbb{Z}^g$ and the WKB series in the non-perturbative wave functions. Thanks to the equivalence \eqref{eq-order-trans}, this implies that the non-perturbative wave functions are solutions to the same KZ equations. 
\end{proof}

\subsubsection{Monodromies}

The second important property of the non-perturbative wave functions is their simple monodromy properties.

\bl
For $j\in \llbracket 1,g\rrbracket$, we have 
\beq\label{monodromieNPA}
\psi_{l,\mathrm{NP}}^{\infty^{(\alpha)}}(z+\mathcal{A}_j,\hbar,\boldsymbol{\rho}) = e^{\frac{2 \pi i \epsilon_j}{\hbar}} \psi_{l,\mathrm{NP}}^{\infty^{(\alpha)}}(z,\hbar,\boldsymbol{\rho}),
\eeq

\beq\label{monodromieNPB}
\psi_{l,\mathrm{NP}}^{\infty^{(\alpha)}}(z+\mathcal{B}_j,\hbar,\boldsymbol{\rho})  = e^{-\frac{2 \pi i  \rho_j}{\hbar}} \psi_{l,\mathrm{NP}}^{\infty^{(\alpha)}}(z,\hbar,\boldsymbol{\rho})
\eeq
and $\forall \ p\in x^{-1}(\mathcal P)$
\beq\label{monodromieNPCp}
\psi_{l,\mathrm{NP}}^{\infty^{(\alpha)}}(z+\mathcal{C}_p,\hbar,\boldsymbol{\rho})  = (-1)^{\delta_{p,\infty^{(\alpha)}}} \ e^{\frac{2 \pi i  t_{p,0}}{\hbar}} \psi_{l,\mathrm{NP}}^{\infty^{(\alpha)}}(z,\hbar,\boldsymbol{\rho}).
\eeq
\el

\begin{proof}
The proof follows from the same strategy as the one for proving the KZ equations. We first prove these monodromy properties for the Fourier transforms and then use \eqref{eq-order-trans} to obtain the same result for the non-perturbative wave functions. 

The properties \eqref{monodromieNPA} and \eqref{monodromieNPCp} directly follow from the same properties for the perturbative wave functions which are the coefficients of the Fourier transform. The last property \eqref{monodromieNPB} for the Fourier transform directly comes from the definition.
\end{proof}

This last property, which states that the non-perturbative wave functions have good monodromies both on the $\mathcal{A}$ and $\mathcal{B}$ cycles is the reason why we had to introduce the formal Fourier transform and had to pay the price of working with trans-series. This turns out to be fundamental for deriving a quantum curve  as we shall see in the next sections.

\br
In the following, we will consider not only series of the form \eqref{eq-form-NP} but products, derivatives, inverse and logarithms of such series. We should always keep in mind that these operations have to be understood from the point of view of series in $\hbar$ whose coefficients are oscillating terms. For example, the inverse of such a series is itself a series in $\hbar$ whose coefficients are products of oscillatory terms obtained by acting as if one were working with formal series in $\hbar$.
\er

\subsection{From linear operators to rational functions}\label{6.3}

For any $l\in \llbracket 0,d-1\rrbracket$, we wish to replace the operators $(\widetilde{\mathcal{L}}_{P,k,l})_{P \in \mathcal{P}, k \in S_P^{(l+1)}}$ in the RHS of the KZ equations by derivatives with respect to $x(z)$ in order to derive a quantum curve. First of all, let us remark that the action of   $\underset{P \in \mathcal{P}}{\sum} \underset{k \in S_P^{(l+1)}}{\sum}  \xi_P^{-k}(x(z)) \widetilde{\mathcal{L}}_{P,k,l}$ on $(\om_{h,n})_{h,n\geq 0}$ in $\log \psi_{\mathrm{NP}}^{\mathrm{symbol}}(z,\hbar,\boldsymbol{\rho})  $ gives terms of order $O(\hbar)$ except for the result of the action on $\om_{0,0}$, which gives the coefficient $P_{l+1}(x(z))$ of the classical spectral curve \eqref{def-eq-sp-curve}.

For this purpose, for each $P \in \mathcal{P}$, $l\in \llbracket 0,d-1\rrbracket$ and $k \in S_P^{(l+1)}$, one defines
\beq
{\mathcal{L}_{P,k,l}}\coloneqq \widetilde{\mathcal{L}}_{P,k,l} - P_{P,k}^{(l+1)},
\eeq
where the second term is understood as a multiplication operator by a symbol.

Let us define
\begin{align}
&K_{P,k,l}^{\mathrm{symbol}}(\lambda,\hbar,\boldsymbol{\rho})  \coloneqq \cr 
&\qquad\left[
 \begin{array}{ccc}
\psi_{\mathrm{NP}}^{\mathrm{symbol}}(z^{(1)}(\lambda),\hbar,\boldsymbol{\boldsymbol{\rho}}) & \dots &   \psi_{\mathrm{NP}}^{\mathrm{symbol}}(z^{(d)}(\lambda),\hbar,\boldsymbol{\rho})  \cr
\left[{\mathcal{L}_{P,k,l} } \psi_{\mathrm{NP}}^{\mathrm{symbol}}(z^{(1)}(\lambda),\hbar,\boldsymbol{\rho}) \right] & \dots & \left[{\mathcal{L}_{P,k,l} }  \psi_{\mathrm{NP}}^{\mathrm{symbol}}(z^{(d)}(\lambda),\hbar,\boldsymbol{\rho}) \right]\cr
  \vdots & & \vdots \cr
\left[({\mathcal{L}_{P,k,l} })^{d-1}  \psi_{\mathrm{NP}}^{\mathrm{symbol}}(z^{(1)}(\lambda),\hbar,\boldsymbol{\rho})\right]  & \dots &\left[({\mathcal{L}_{P,k,l} })^{d-1}   \psi_{\mathrm{NP}}^{\mathrm{symbol}}(z^{(d)}(\lambda),\hbar,\boldsymbol{\rho}) \right]\cr  
\end{array}
\right]
\end{align}
and
\beq
C_{P,k,l}^{\mathrm{symbol}}(\lambda,\hbar,\boldsymbol{\rho})\coloneqq  \hbar^{-1} \, K_{P,k,l}^{\mathrm{symbol}}(\lambda,\hbar,\boldsymbol{\boldsymbol{\rho}})  \cdot 
\left(\widehat{\Psi}_{\mathrm{NP}}^{\mathrm{symbol}}(\lambda,\hbar,\boldsymbol{\rho})\right)^{-1},
\eeq
where $\left(\widehat{\Psi}_{\mathrm{NP}}^{\mathrm{symbol}}(\lambda,\hbar,\boldsymbol{\rho})\right)^{-1}$ is defined as a symbol from the adjoint matrix to $\widehat{\Psi}_{\mathrm{NP}}^{\mathrm{symbol}}(\lambda,\hbar,\boldsymbol{\rho})$ and the inverse of $\det (\widehat{\Psi}_{\mathrm{NP}}^{\mathrm{symbol}}(\lambda,\hbar,\boldsymbol{\rho}))$ as a symbol,
\beq
\left(\widehat{\Psi}_{\mathrm{NP}}^{\mathrm{symbol}}(\lambda,\hbar,\boldsymbol{\rho})\right)^{-1} = \frac{1}{\det (\widehat{\Psi}_{\mathrm{NP}}^{\mathrm{symbol}}(\lambda,\hbar,\boldsymbol{\rho}))} \, \mathrm{Adj}\left[\widehat{\Psi}_{\mathrm{NP}}^{\mathrm{symbol}}(\lambda,\hbar,\boldsymbol{\rho}) \right].
\eeq

For all $\lambda$ such that $\widehat{\Psi}_{\mathrm{NP}}(\lambda,\hbar,\boldsymbol{\rho})$ is invertible, i.e.~for $\lambda$ such that the Wronskian $\det (\widehat{\Psi}_{\mathrm{NP}}(\lambda,\hbar,\boldsymbol{\rho}))$ is non vanishing, we can define the evaluated version of these objects by 
\begin{align}
&K_{P,k,l}(\lambda,\hbar,\boldsymbol{\rho})  \coloneqq \cr 
&\qquad\left[
 \begin{array}{ccc}
\psi_{\mathrm{NP}}(z^{(1)}(\lambda),\hbar,\boldsymbol{\boldsymbol{\rho}}) & \dots &   \psi_{\mathrm{NP}}(z^{(d)}(\lambda),\hbar,\boldsymbol{\rho})  \cr
\text{ev}.\left[{\mathcal{L}_{P,k,l} } \psi_{\mathrm{NP}}^{\mathrm{symbol}}(z^{(1)}(\lambda),\hbar,\boldsymbol{\rho}) \right] & \dots & \text{ev}.\left[{\mathcal{L}_{P,k,l} }  \psi_{\mathrm{NP}}^{\mathrm{symbol}}(z^{(d)}(\lambda),\hbar,\boldsymbol{\rho}) \right]\cr
  \vdots & & \vdots \cr
\text{ev}.\left[({\mathcal{L}_{P,k,l} })^{d-1}  \psi_{\mathrm{NP}}^{\mathrm{symbol}}(z^{(1)}(\lambda),\hbar,\boldsymbol{\rho})\right]  & \dots &\text{ev}.\left[({\mathcal{L}_{P,k,l} })^{d-1}   \psi_{\mathrm{NP}}^{\mathrm{symbol}}(z^{(d)}(\lambda),\hbar,\boldsymbol{\rho}) \right]\cr  
\end{array}
\right]\cr
\end{align}
and
\beq
C_{P,k,l}(\lambda,\hbar,\boldsymbol{\rho})\coloneqq  \hbar^{-1} \, K_{P,k,l}(\lambda,\hbar,\boldsymbol{\boldsymbol{\rho}})  \cdot 
\left(\widehat{\Psi}_{\mathrm{NP}}(\lambda,\hbar,\boldsymbol{\rho})\right)^{-1} .
\eeq

In the following, to shorten the notations, we shall most of the time omit the dependence of the functions considered in $\boldsymbol{\epsilon}$ and $\boldsymbol{\rho}$. However, we should keep in mind that  the elements of the matrices denoted by $C_{P,k,l}(\lambda,\hbar)$ are functions of $\hbar$. More precisely, as the non-perturbative wave functions, they are trans-series in $\hbar$ 
\beq\label{eq-trans-C}
C_{P,k,l}(\lambda,\hbar) = \sum_{n \geq 0} \hbar^n C_{P,k,l}^{(n)}(\lambda,\hbar),
\eeq
whose coefficients $C_{P,k,l}^{(n)}(\lambda,\hbar)$ are combinations of products of derivatives of theta functions. Remark that, contrary to the wave function which had an $\hbar$-dependent pre-factor $e^{\hbar^{-2} \om_{0,0}(\boldsymbol{\epsilon})  + \hbar^{-1} S_{-1}(z,\boldsymbol{\epsilon}) + S_0(z,\boldsymbol{\epsilon})}$, the elements of $C_{P,k,l}(\lambda,\hbar)$ do not involve any $\hbar$-dependent pre-factor.

One fundamental point is to understand the properties of these matrices $\left(C_{P,k,l}(\lambda,\hbar)\right)_{l\in \llbracket 0,d-1\rrbracket, P\in \mathcal{P}, k \in S_P^{(l+1)}}$ as functions of $\lambda$. This is summed up in the following theorem.

\bt \label{th-C}
For any $l\in \llbracket 0,d-1\rrbracket$, $P \in \mathcal{P}$ and any $k \in S_P^{(l+1)}$, $C_{P,k,l}(\lambda,\hbar)$ is a rational function of $\lambda$ with possible poles at $\lambda \in \mathcal{P}$, $\lambda \in x(\mathcal{R})$ and $\lambda$ in the set of zeros of the Wronskian $\det \widehat{\Psi}_{\mathrm{NP}}(\lambda,\hbar)$.
This means that, as  a function of $\hbar$, $C_{P,k,l}(\lambda,\hbar)$ is a formal trans-series of the form 
\beq
C_{P,k,l}(\lambda,\hbar) = \sum_{n \geq 0} \hbar^n C_{P,k,l}^{(n)}(\lambda,\hbar),
\eeq
where $C_{P,k,l}^{(n)}(\lambda,\hbar)$ are rational functions of $\lambda$.
\et

\begin{proof}
Let us first show that $C_{P,k,l}(\lambda,\hbar)$ is a mono-valued function of $\lambda \in \mathbb{P}^1$. One important point is to understand how the non-perturbative wave functions behave as functions on the base curve $\mathbb{P}^1$. For this purpose, let us describe how $\pi_1(\mathbb{P}^1 \setminus x(\mathcal{R}))$ acts on these wave functions. When $\lambda$ goes around $\gamma \in \pi_1(\mathbb{P}^1 \setminus x(\mathcal{R}))$, a point $z \in x^{-1}(\lambda)$ follows the composition of a non-contractible path $\Gamma \in H_1(\Sigma,\mathbb{Z})$ and a path going from $z$ to a, possibly distinct, point $z'$ in the fiber above $\lambda$. This means that, for any $j \in \llbracket 1,d\rrbracket$,
\begin{align}
&\exists\, (\mathbf{n}[\gamma,j],\mathbf{m}[\gamma,j]) \in \mathbb{Z}^{2g} \, , \; \exists\, \beta \in \llbracket 1,d  \rrbracket \; \text{ such that } \cr
& \psi_{l,\mathrm{NP}}^{\infty^{(\alpha)}}(z^{(j)}(\lambda+\gamma),\hbar,\boldsymbol{\rho})  = e^{ \frac{2 \pi i}{\hbar}   \underset{a=1}{\overset{g}{\sum}} (n[\gamma,j]_a \epsilon_a + m[\gamma,j]_a \rho_a)}
 \psi_{0,\mathrm{NP}}^{\infty^{(\alpha)}}(z^{(\beta)}(\lambda),\hbar,\boldsymbol{\rho})  .
\end{align}
This implies that there exists a monodromy matrix $M_\gamma(\hbar)$ which has the shape of  a permutation matrix with entries of the form $e^{ \frac{2 \pi i}{\hbar} \underset{a=1}{\overset{g}{\sum}} (n[\gamma,j]_a \epsilon_a + m[\gamma,j]_a \rho_a)}$ instead of $1$'s such that 
\beq
\widehat{\Psi}_{\mathrm{NP}}(\lambda+\gamma,\hbar) = \widehat{\Psi}_{\mathrm{NP}}(\lambda,\hbar) M_\gamma(\hbar)
\qquad 
\text{and}
\qquad
K_{P,k,l}(\lambda+\gamma,\hbar) = K_{P,k,l}(\lambda,\hbar) M_\gamma(\hbar).
\eeq
Since $M_\gamma(\hbar)$ is invertible, this implies that 
\beq
C_{P,k,l}(\lambda+\gamma,\hbar) = C_{P,k,l}(\lambda,\hbar) .
\eeq

It is easy to see that the essential singularities of $\widehat{\Psi}_{\mathrm{NP}}(\lambda,\hbar)$ compensate with the ones on the other side so that $C_{P,k,l}(\lambda,\hbar)$ is a rational function of $\lambda$.

By definition, it is easy to see that it might have poles at $\lambda \in \mathcal{P} \cup x(\mathcal{R})$ from the singularities of  $\widehat{\Psi}_{\mathrm{NP}}(\lambda,\hbar)$ but it can also have singularities at points $\lambda$ where $\widehat{\Psi}_{\mathrm{NP}}(\lambda,\hbar)$ is not invertible, i.e.~when the Wronskian vanishes.

The behavior in $\hbar$  follows from the definition of the operators $\mathcal{L}_{P,k,l}$, where one always has a product of $ \hbar^2 \mathcal{I}_{\mathcal{C}}$ acting on terms of the form $\hbar^{2h-2+n} \om_{h,n}$, with $2h-2+n \geq -1$. We collected the order $O(1)$ terms in $P_l(\lambda)$ so that the resulting terms are of order $O(\hbar)$. Since  $C_{P,k,l}(\lambda,\hbar)$ is of order $\hbar^{-1}$ compared to $\mathcal{L}_{P,k,l}$, it is of order $O(1)$. 
\end{proof}

\br
A priori, the coefficients $C_{P,k,l}^{(n)}(\lambda,\hbar)$ involve theta functions whose argument depends of the different pre-images of $\lambda$ in the classical spectral curve. We have seen that these are actually rational functions of $\lambda$. However, these rational functions of $\lambda$ can still depend on $\hbar$ through derivatives of theta functions of the form 
$
\left. \frac{\partial^n \Theta_{\text{Riemann}}(\mathbf{v}, \boldsymbol{\tau})}{\partial v_{i_1} \dots \partial v_{i_n}} \right|_{\mathbf{v} = \frac{\boldsymbol{\phi} + \boldsymbol{\rho}}{\hbar} } .
$
\er

One can use the elements of the matrices $\left(C_{P,k,l}(\lambda,\hbar)\right)_{l\in \llbracket 0,d-1\rrbracket, P\in \mathcal{P}, k \in S_P^{(l+1)}}$ in order to write the KZ equations simply as differential equations in $\lambda$. For any $l\in \llbracket0,d-1\rrbracket$ and $j\in  \llbracket 1 , d \rrbracket$:
\begin{align}
&{\hbar} \frac{d  \psi_{l,\mathrm{NP}}^{\infty^{(\alpha)}}(z^{(j)}(\lambda),\hbar) }{d\lambda} +\psi_{l+1,\mathrm{NP}}^{\infty^{(\alpha)}}(z^{(j)}(\lambda),\hbar)   = P_{l+1}(\lambda)\, \psi_{0,\mathrm{NP}}^{\infty^{(\alpha)}}(z^{(j)}(\lambda),\hbar) \nonumber \\
&+  \hbar {\displaystyle \sum_{P \in \mathcal{P}} \sum_{k \in S_P^{(l+1)}}}  \xi_P^{-k}(\lambda) {\displaystyle \sum_{m=0}^{d-1}} \left[C_{P,k,l}(\lambda,\hbar)\right]_{2,m+1}  \psi_{m,\mathrm{NP}}^{\infty^{(\alpha)}}(z^{(j)}(\lambda),\hbar).
\end{align}

\subsection{Lax system}\label{6.4}

Under this form, the KZ equations can be translated into a linear differential equation in $\lambda$ compatible with the actions of the linear operators $\left(\mathcal{L}_{P,k,l}\right)_{l\in \llbracket 0,d-1\rrbracket, P\in \mathcal{P}, k \in S_P^{(l+1)}}$.

\bt[ODE and Lax system]
One has
\beq\label{eq-diff-hat}
\hbar  \frac{d \widehat{ \Psi}_{\mathrm{NP}}(\lambda,\hbar) }{d\lambda} =  \bigg[-\widehat{P}(\lambda) + \hbar \sum_{P \in \mathcal{P}}  \sum_{k \in \mathbb{N}}  \xi_P^{-k}(\lambda) 
\widehat{\Delta}_{P,k}(\lambda,\hbar) \bigg] \widehat{\Psi}_{\mathrm{NP}}(\lambda,\hbar),
\eeq
where 
\beq
\widehat{P}(\lambda) \coloneqq 
\left[
 \begin{array}{ccccc}
 - P_1(\lambda) & 1&0  & \dots &  0\cr
 -P_2(\lambda) & 0&1  & \dots &  0\cr
  \vdots & \vdots &\vdots&\ddots & \vdots \cr
  -P_{d-1}(\lambda) & 0&0  & \dots &  1\cr
-P_d(\lambda) & 0 &0 & \dots &  0 \cr  
\end{array}
\right]
\eeq
and for any $P\in \mathcal{P}$ and $k\in \mathbb{N}$
\beq
\widehat{\Delta}_{P,k}(\lambda,\hbar)  \coloneqq 
 \left[
 \begin{array}{ccc}
\left[C_{P,k,0}(\lambda,\hbar)\right]_{2,1} & \dots &  \left[C_{P,k,0}(\lambda,\hbar)\right]_{2,d}  \cr
\left[C_{P,k,1}(\lambda,\hbar)\right]_{2,1}  & \dots & \left[C_{P,k,1}(\lambda,\hbar)\right]_{2,d}  \cr
  \vdots & & \vdots \cr
\left[C_{P,k,d-1}(\lambda,\hbar)\right]_{2,1} & \dots &\left[C_{P,k,d-1}(\lambda,\hbar)\right]_{2,d}  \cr  
\end{array}
\right] 
\eeq
where $C_{P,k,l}(\lambda,\hbar) = 0$ for $k \notin S_P^{(l+1)}$.

For any $P\in \mathcal{P}$, $k\in \mathbb{N}$, $l\in \llbracket 0, d-1\rrbracket$, one has the auxiliary systems
\beq
\hbar^{-1} \mathrm{ev} . \mathcal{L}_{P,k,l} \widehat{\Psi}_{\mathrm{NP}}^{\mathrm{symbol}}(\lambda,\hbar) = \widehat{A}_{P,k,l}(\lambda,\hbar) \widehat{\Psi}_{\mathrm{NP}}(\lambda,\hbar),  
\eeq
with the constraints
\beq
\left[C_{P,k,l}(\lambda,\hbar)\right]_{2,j}= \left[\widehat{A}_{P,k,l}(\lambda,\hbar)\right]_{1,j}, \; \forall\, j\in \llbracket 1, d\rrbracket , 
\eeq
so that the entries of $\widehat{\Delta}_{P,k}(\lambda,\hbar)$ are composed of elements of $\widehat{A}_{P,k,l}(\lambda,\hbar)$.  The elements of $\widehat{A}_{P,k,l}(\lambda,\hbar)$ are $\hbar$-trans-series functions that are rational functions of $\lambda$ similarly to the elements of $C_{P,k,l}(\lambda,\hbar)$. 

One has similar equations at the symbolic level before evaluation but we shall not write them down as they are obtained simply by adding an exponent symbol whenever needed by defining
\beq
\widehat{A}_{P,k,l}^{\mathrm{symbol}}(\lambda,\hbar)  \coloneqq \left(\mathcal{L}_{P,k,l} \left[ \widehat{\Psi}_{\mathrm{NP}}^{\mathrm{symbol}}(\lambda,\hbar) \right]\right) \, \left(  \widehat{\Psi}_{\mathrm{NP}}^{\mathrm{symbol}}(\lambda,\hbar) \right)^{-1} .
\eeq
\et

\subsection{A first gauge transformation to recover the classical spectral curve}

Let us now prove that the matrices $\widehat{\Delta}_{P,k}(\cdot,\hbar)$ do not have any pole on the ramification locus $x(\mathcal{R})$ so that \eqref{eq-diff-hat} produces a quantum curve. For this purpose, it is easier to work in a different basis in which the leading order in $\hbar$ of the Lax matrix is companion-like (and thus recovers the classical spectral curve) and study the associated linear systems. Let us define the matrix  
\beq
G(\lambda)\coloneqq  \left[
 \begin{array}{cccccc}
 1 & 0&0 & \dots & 0 &  0\cr
P_1(\lambda) & -1&0 &  \dots & 0 &  0\cr
P_2(\lambda) & -P_1(\lambda)&1 &  \dots & 0 &  0\cr
\vdots & \vdots &\vdots&\ddots & \vdots & \vdots  \cr
P_{d-2}(\lambda) & -P_{d-3}(\lambda)&P_{d-4}(\lambda) & \dots & (-1)^{d-2} &  0\cr
P_{d-1}(\lambda) &-P_{d-2}(\lambda)&P_{d-3}(\lambda) & \dots & (-1)^{d-2}P_1(\lambda) &  (-1)^{d-1} \cr  
\end{array}
\right] 
\eeq
whose non-zero entries are $G_{i,j}=(-1)^{j-1} P_{i-j}(\lambda)$ for $1\leq j<i\leq d$ and $G_{i,i}=(-1)^{i-1}$ for $i\in \llbracket 1,d\rrbracket$. It can easily be checked that
$\widehat{P}(\lambda) \, G(\lambda)  = -G(\lambda) \widetilde{P}(\lambda)$, where
\beq 
 \widetilde{P}(\lambda) \coloneqq 
 \left[
 \begin{array}{ccccc}
 0 & 1&0  & \dots &  0\cr
0 & 0&1 & \dots &  0\cr
  \vdots & \vdots &\vdots& \ddots & \vdots \cr
  0 & 0&0 &  \dots &  1\cr
  (-1)^{d-1}P_d(\lambda) & (-1)^{d-2}P_{d-1}(\lambda)& (-1)^{d-3}P_{d-2}(\lambda) & \dots &  P_{1}(\lambda) \cr  
\end{array}
\right] 
\eeq
so that 
\beq
 \widetilde{P}(\lambda) = \left(G(\lambda)\right)^{-1} \,  \left(-\widehat{P}(\lambda)\right) \, G(\lambda) .
 \eeq
Note that the characteristic polynomial of $\td{P}(\lambda)$ is precisely the classical spectral curve \eqref{def-eq-sp-curve}: $\det(y I_d-\td{P}(\lambda))=P(\lambda,y)$. We use this matrix $G$ to define a gauge transformation by 
\beq\label{GaugeG}
\widetilde{\Psi}(\lambda,\hbar)\coloneqq \left(G(\lambda)\right)^{-1} \widehat{\Psi}_{\mathrm{NP}}(\lambda,\hbar),
\eeq
which is solution to the differential equation
\beq\label{eq-diff-tilde}
\hbar  \frac{d \widetilde{ \Psi}(\lambda,\hbar) }{d\lambda} =  \widetilde{L}(\lambda,\hbar)  \widetilde{\Psi}(\lambda,\hbar),
\eeq
with
\beq\label{def-tilde-L}
\widetilde{L}(\lambda) \coloneqq \bigg[ \widetilde{P}(\lambda) + \hbar \sum_{P \in \mathcal{P}}  \sum_{k \in \mathbb{N}} \xi_P^{-k}(\lambda) 
\widetilde{\Delta}_{P,k}(\lambda,\hbar) \bigg],
\eeq
where 
\beq \label{eq-diff-tildee}
\widetilde{\Delta}_{P,k}(\lambda,\hbar)  \coloneqq 
 - \hbar \left(G(\lambda)\right)^{-1} \frac{\partial G(\lambda)}{\partial \lambda} + \left(G(\lambda)\right)^{-1} \widehat{\Delta}_{P,k}(\lambda,\hbar) G(\lambda).
\eeq

\br
Let us note that one can define a symbol $G^{\mathrm{symbol}}(\lambda)$ by considering the rational functions $P_j(\lambda)$  as symbols. This allows to define a symbol
\beq
\widetilde{\Psi}^{\mathrm{symbol}}(\lambda,\hbar)\coloneqq (G^{\mathrm{symbol}}(\lambda))^{-1} \widehat{\Psi}_{\mathrm{NP}}^{\mathrm{symbol}}(\lambda,\hbar),
\eeq
where $\left(G^{\mathrm{symbol}}(\lambda)\right)^{-1}$ is defined using its adjoint and the symbolic inverse of its determinant.
We can as well  define a symbolic Lax matrix 
\beq
\widetilde{L}^{\mathrm{symbol}}(\lambda,\hbar) \coloneqq \hbar \frac{\partial \widetilde{\Psi}^{\mathrm{symbol}}(\lambda,\hbar)}{\partial \lambda} \left(\widetilde{\Psi}^{\mathrm{symbol}}(\lambda,\hbar)\right)^{-1},
\eeq
where 
\beq
\left(\widetilde{\Psi}^{\mathrm{symbol}}(\lambda,\hbar)\right)^{-1} \coloneqq \mathrm{Adj} \left[\widetilde{\Psi}^{\mathrm{symbol}}(\lambda,\hbar)\right] \, \frac{1}{\det \widetilde{\Psi}^{\mathrm{symbol}}(\lambda,\hbar)}
\eeq
using the symbolic inverse.
\er

For any $P\in \mathcal{P}$, $l\in \llbracket 0, d-1\rrbracket$, $k\in S_P^{(l+1)}$, one also has the auxiliary systems
\beq
\hbar^{-1}\mathrm{ev} . \mathcal{L}_{P,k,l} \widetilde{\Psi}^{\mathrm{symbol}}(\lambda,\hbar) = \widetilde{A}_{P,k,l}(\lambda,\hbar) \widetilde{\Psi}(\lambda,\hbar),
\eeq
where
\begin{align}
\widetilde{A}_{P,k,l}(\lambda,\hbar)&\coloneqq \left(G(\lambda)\right)^{-1} \bigg( \mathrm{ev} .\Big(\left[ G^{\mathrm{symbol}}(\lambda), \hbar^{-1} \mathcal{L}_{P,k,l} \right] \widetilde{\Psi}^{\mathrm{symbol}}(\lambda,\hbar) \Big)\bigg) \widetilde{\Psi}(\lambda,\hbar)^{-1}\nonumber \\
&+ \left(G(\lambda)\right)^{-1}  \widehat{A}_{P,k,l}(\lambda,\hbar)  G(\lambda)
\end{align}
is a rational function of $\lambda$ with the same poles as $\widehat{A}_{P,k,l}(\lambda,\hbar)$.

\subsection{Behavior of the auxiliary matrices at the ramification points}\label{6.5}
Let us now state one of the main properties of the auxiliary matrices $\left(\widetilde{A}_{P,k,l}(\lambda,\hbar)\right)_{P\in \mathcal{P}, l\in \llbracket 0, d-1\rrbracket, k\in S_P^{(l+1)}}$.

\bt \label{th-no-bp}
For any $P\in \mathcal{P}$, $l\in \llbracket 0, d-1\rrbracket$, $k\in S_P^{(l+1)}$, the matrix $\widetilde{A}_{P,k,l}(\lambda,\hbar)$ does not have any pole at the critical values $u \in x\left(\mathcal{R}\right)$.  This implies that for any $P\in \mathcal{P}$ and $k\in \mathbb{N}$, the matrices $\widehat{\Delta}_{P,k}(\lambda,\hbar)$ and $\widetilde{\Delta}_{P,k}(\lambda,\hbar)$ are regular at the ramification points.
\et

\begin{proof}
This is proved by induction in Appendix~\ref{App-no-bp}. 
\end{proof}

This allows us to obtain the following important property.

\bc
The wave functions matrix $\widehat{\Psi}_{\mathrm{NP}}(\lambda,\hbar)$ is solution to a linear differential equation
\beq\label{eq-diff-hat2}
\hbar  \frac{d \widehat{ \Psi}_{\mathrm{NP}}(\lambda,\hbar) }{d\lambda} = \widehat{L}(\lambda,\hbar) \,  \widehat{\Psi}_{\mathrm{NP}}(\lambda,\hbar),
\eeq
where the rational function
\beq
\widehat{L}(\lambda,\hbar) =  \bigg[-\widehat{P}(\lambda) + \hbar \sum_{P \in \mathcal{P}}  \sum_{k \in \mathbb{N}} \xi_P^{-k}(\lambda) 
\widehat{\Delta}_{P,k}(\lambda,\hbar) \bigg]
\eeq
has poles at $\lambda \in \mathcal{P}$ and at zeros of the Wronskian $\det \widehat{\Psi}_{\mathrm{NP}}(\lambda,\hbar)$, the latter being \emph{apparent singularities} of the system.
\ec

Let us emphasize again that this differential equation has to be understood as an equality order by order in $\hbar$ in the trans-series as explained in Section~\ref{sectiontransseries}.

\medskip

We have thus built a linear differential system of size $d\times d$ whose formal fundamental solution can be computed by topological recursion. This system has poles at the poles of the leading WKB term, namely at poles of the function $y$ but it may also have poles at apparent singularities. We now wish to continue our study of this system and understand better these apparent singularities, find a gauge equivalent system without apparent singularities and write down the associated quantum curve. This is what we are doing in the next section, based on the observation that the position of the apparent singularities can be computed thanks to the KZ equations.

\section{Gauge transformations and quantum curve}\label{S7}

In this section, we deduce quantum curves from the Lax system presented in the previous section, which may feature apparent singularities. We also apply a gauge transformation in order to obtain another linear Lax system without any apparent singularities and whose only poles are at the singularities of the initial classical spectral curve. We finally study the characteristic polynomial of this system and present it as a deformation of the classical spectral curve, allowing to interpret our result as an $\hbar$-family of connections on the base curve $\mathbb{P}^1$.

\subsection{Companion-like system associated to the quantum curve}\label{SectionCompanionlike}

We have shown that the topological recursion allowed us to build a fundamental solution to the size $d$ linear differential  system (Equation \eqref{eq-diff-hat}). This implies that each element of the matrix $\widehat{\Psi}_{\mathrm{NP}}(\lambda,\hbar)$  is itself a solution to a degree $d$ differential equation. In particular, considering the elements of the first line, this implies that these elements are solution to a degree $d$ ODE which we can consider as a quantization of the classical spectral curve.

\bd[Quantum Curve]\label{DefQuantumCurve}For all $j\in \llbracket 1,d\rrbracket$, let us define 
\beq \psi^{(j)}(\lambda,\hbar) \coloneqq \psi_{0,\mathrm{NP}}^{\infty^{(\alpha)}}(z^{(j)}(\lambda),\hbar) .\eeq
Then, for all $ j\in \llbracket 1,d\rrbracket$, $\psi^{(j)}(\lambda,\hbar) $ is solution to a degree $d$ ODE of the form
\beq \label{eq-quant-curve-1}
\forall\, j\in \llbracket 1 , d\rrbracket \,  : \; \sum_{k=0}^d b_{d-k}(\lambda,\hbar) \left(\hbar \frac{\partial}{\partial \lambda}\right)^k \psi^{(j)}(\lambda,\hbar) = 0,
\eeq
for some coefficients $\left(b_l(\lambda,\hbar)\right)_{l\in \llbracket 0,d\rrbracket}$ with $b_0(\lambda,\hbar)=1$. This ODE is known as the ``quantum curve'' associated to the admissible classical spectral curve \eqref{SpectralCurves}.
\hfill $\star$
\ed

\br
Observe that the quantum curve depends on the choice of $\alpha$, but we decided not to indicate this dependence to make notations lighter.
\er

In order to understand better how this procedure is a good quantization of the classical spectral curve, we need to study the properties and singularities of the coefficients $\left(b_{l}(\lambda,\hbar)\right)_{l\in \llbracket 0,d\rrbracket}$ of the quantum curve.

In order to simplify the exposition in this section, we rewrite the quantum curve in a companion-like matrix form.

\bd[Matrix form of the quantum curve]\label{DefPsi}We define the $d\times d$ matrix
\beq
\Psi(\lambda,\hbar) \coloneqq \bigg[ \Big(\hbar \frac{\partial}{\partial \lambda} \Big)^{i-1} \psi^{(j)}(\lambda,\hbar)\bigg]_{1\leq i,j \leq d}
\eeq
and the $d\times d$ companion-like matrix
\beq\label{ExpressionL}
L(\lambda,\hbar) \coloneqq 
 \left[
 \begin{array}{ccccc}
 0 & 1&0  & \dots &  0\cr
0 & 0&1 &  &  \vdots\cr
  \vdots & \ddots &\ddots&\ddots & 0\cr
  0 & 0&0 &  0 &  1\cr
  -b_d(\lambda,\hbar) & -b_{d-1}(\lambda,\hbar)& -b_{d-2}(\lambda,\hbar) & \dots &  -b_{1}(\lambda,\hbar) \cr  
\end{array}
\right] 
\eeq
so that \eqref{eq-quant-curve-1} is equivalent to
\beq
\hbar \frac{\partial}{\partial \lambda} \Psi(\lambda,\hbar) = L(\lambda,\hbar) \Psi(\lambda,\hbar).
\eeq
\hfill $\star$
\ed

This matrix $L(\lambda,\hbar)$ can be obtained by a rational gauge transformation from $\widehat{L}(\lambda,\hbar)$ so that the coefficients $(b_j(\lambda,\hbar))_{j=0}^d$ are rational functions of $\lambda$. From Theorem \ref{th-no-bp}, $\Psi(\lambda,\hbar)$ can have singularities only at $\lambda \in \mathcal{P}$ so that the functions $(b_j(\lambda,\hbar))_{j=0}^d$ can have poles at $\lambda \in \mathcal{P}$. They may as well have simple poles at apparent singularities of the system. It is a classical result, which can be adapted in the present trans-series setup, that these apparent singularities are exactly the values of $\lambda$ which annihilate the Wronskian of this new system $\det \Psi(\lambda,\hbar)$.
Let us also emphasize that, even if it is implicit in the notations,  they are formal series in $\hbar$ of the same type as the elements of the matrices $\left(C_{P,k,l}(\lambda,\hbar)\right)_{P\in \mathcal{P}, l\in \llbracket 0,d-1\rrbracket,k \in S_P^{(l+1)}}$, namely they are series in $\hbar$ whose coefficients are rational functions of $\lambda$ with coefficients given by combinations of derivatives of theta functions.

In the next section, we will study more carefully the properties of these rational functions $\left(b_l(\lambda,\hbar)\right)_{l\in \llbracket 0,d\rrbracket}$ as well as the apparent singularities of this system and how to remove them.

\subsubsection{Wronskian and apparent singularities}\label{SectionWronskian}

Let us first study the apparent singularities of our system.

For this purpose, we wish to identify the structure of the Wronskian $W(\lambda,\hbar) \coloneqq \det \Psi(\lambda,\hbar)$. In order to achieve this aim, let us first study its singularity structure. Thanks to Theorem~\ref{th-no-bp}, 
it might have singularities only when $\lambda \to P \in \mathcal{P}$.
 Let us first remark that one has
\beq
W(\lambda,\hbar) \coloneqq \det \Psi(\lambda,\hbar) = \det \left[ \mathbf{S}(\lambda,\hbar) \right] \, \prod_{j =1}^d \psi^{(j)}(\lambda,\hbar),
\eeq
where $\mathbf{S}(\lambda,\hbar)$ is the $d\times d$ matrix defined by
\beq
\forall \,(i, j)\in \llbracket 1, d\rrbracket^2 \, : \; \left[\mathbf{S}(\lambda,\hbar)\right]_{ij} \coloneqq \left(\psi^{(j)}(\lambda,\hbar) \right)^{-1} \Big(\hbar \frac{\partial }{\partial \lambda} \Big)^{i-1} \psi^{(j)}(\lambda,\hbar) .
\eeq
Let us now recall that the non-perturbative wave function reads
\beq
\psi_{0,NP}^{\infty^{(\alpha)}}(z,\hbar)  = e^{\hbar^{-2} \om_{0,0}(\boldsymbol{\epsilon})  + \hbar^{-1} S_{-1}(z,\boldsymbol{\epsilon}) + S_0(z,\boldsymbol{\epsilon})} \bigg(1 + {\displaystyle \sum_{m=1}^\infty} \hbar^m \, \Xi_m^{(\infty_\alpha)}(z,\hbar) \bigg),
\eeq
where, for any $m\geq 0$, $\Xi_m^{(\infty^{(\alpha)})}(z)$ is holomorphic at $z \in x^{-1}(\mathcal{P})$ and
\beqq
S_{-1}(z,\boldsymbol{\epsilon}) = \int_{\infty^{(\alpha)}}^z ydx + O(1)
\eeqq
and
\beqq
S_0(z,\boldsymbol{\epsilon}) =  \int_{z_1 =\infty^{(\alpha)}}^z \int_{z_2 =\infty^{(\alpha)}}^z  \left[ \om_{0,2}(z_1,z_2) - \frac{dx(z_1) \, dx(z_2)}{(x(z_1)-x(z_2))^2} \right] + O(1)
\eeqq
for $z \to p \in x^{-1}(\mathcal{P})$.

Let us recall as well that 
\beq
\om_{0,2}(z_1,z_2) - \frac{dx(z_1) \, dx(z_2)}{(x(z_1)-x(z_2))^2} = - \sum_{z \in x^{-1}(z_2)\setminus \{z_2\}} \om_{0,2}(z_1,z)
\eeq
to emphasize the fact that this does not have any pole on the diagonal. Hence,
\beq
\hbar^{-2} \om_{0,0}(\boldsymbol{\epsilon})  + \hbar^{-1} S_{-1}(z,\boldsymbol{\epsilon}) + S_0(z,\boldsymbol{\epsilon}) = \hbar^{-1}  \int_{\infty^{(\alpha)}}^z ydx  + \; \text{holomorphic},
\eeq
as $z \to p \in x^{-1}(\mathcal{P})$.

On the other hand, one has
\beq
\det  \left[ \mathbf{S}(\lambda,\hbar) \right] =\det \left[ \left(y(z^{(j)}(\lambda)) \right)^{i-1} \right] + \; \text{sub-leading terms} .
\eeq

This finally gives the asymptotics as $\lambda \to P \in \mathcal{P}$,
\beq
\det \Psi(\lambda,\hbar) = \kappa_P \exp\left[ \hbar^{-1} \int_{\infty}^\lambda P_1(\lambda)  d\lambda \right] \, \xi_P(\lambda)^{-G_P} \, \left[1 + O\left(\xi_P(\lambda)\right)\right],
\eeq
where
\beq
\forall\,  i \in \llbracket 1, N\rrbracket \, : \; G_{\Lambda_i} \leq  {\displaystyle \max_{\boldsymbol{\beta} \underset{\frac{d(d-1)}{2}}{\subseteq}x^{-1}(\Lambda_i)}}\bigg(\sum_{p\in\boldsymbol{\beta}}  r_p\bigg),
\eeq
\beq
G_\infty \leq {\displaystyle \max_{\boldsymbol{\beta} \underset{\frac{d(d-1)}{2}}{\subseteq}x^{-1}(\Lambda_i)}}\bigg(\sum_{p\in\boldsymbol{\beta}} r_p-2 \bigg),
\eeq
and $\kappa_P$ is independent of $\lambda$ and $\hbar$.

This allows to obtain an expression for the Wronskian.

\bc[Expression of the Wronskian]\label{Corollaryqi}
The Wronskian takes the form
\beq
\det \Psi(\lambda,\hbar) = \kappa \frac{  \underset{i=1}{\overset{G}{\prod}}  (\lambda-q_i(\hbar))}{ \underset{i=1}{\overset{N}{\prod}} (\lambda- \Lambda_i)^{G_{\Lambda_i}}} \, \exp\left( \hbar^{-1} \int_{0}^\lambda P_1(\lambda) d\lambda\right),
\eeq
where $\kappa$ is independent of $\lambda$ and $G = \underset{P \in \mathcal{P}}{\sum} G_P$.
\ec

\begin{proof}
The proof is similar to the proof of Theorem \ref{th-C}.

Let us consider the ratio $\frac{\det \Psi(\lambda,\hbar)}{\exp\left( \hbar^{-1} \int_{0}^\lambda P_1(\lambda) d\lambda\right)}$. By studying its possible singularities and monodromies, one sees that it is a rational function of $\lambda$ with poles in $\mathcal{P}$.

More precisely, it implies that 
\beq
\frac{\det \Psi(\lambda,\hbar)}{\exp\left( \hbar^{-1} \int_{0}^\lambda P_1(\lambda) d\lambda\right)} = \sum_{m=0}^\infty \hbar^{m} W_m(\lambda,\hbar),
\eeq
where
\beq
W_m(\lambda,\hbar)  = \frac{\text{Pol}_m(\lambda,\hbar)}{ \underset{i=1}{\overset{N}{\prod}}(\lambda- \Lambda_i)^{G_{\Lambda_i}}},
\eeq
with $\text{Pol}_m(\lambda,\hbar)$ a polynomial in $\lambda$ of degree at most $G$ whose coefficients can depend on $\hbar$ through theta functions and their derivatives so that
\beq
\frac{\det \Psi(\lambda,\hbar)}{\exp\left( \hbar^{-1} \int_{0}^\lambda P_1(\lambda) d\lambda\right)} = \frac{ \underset{m=0}{\overset{\infty}{\sum}}  \hbar^{m} \text{Pol}_m(\lambda,\hbar)}{\underset{i=1}{\overset{N}{\prod}}(\lambda- \Lambda_i)^{G_{\Lambda_i}}},
\eeq
leading to the result with
\beq
 \kappa \,{\displaystyle \prod_{i=1}^G } (\lambda-q_i(\hbar)) = \sum_{m=0}^\infty \hbar^{m} \text{Pol}_m(\lambda,\hbar).
 \eeq
The zeros $q_i(\hbar)$ are obtained by solving the equation with $\underset{m=0}{\overset{\infty}{\sum}}\hbar^{m} \text{Pol}_m(q_i(\hbar),\hbar)$ order by order in $\hbar$ by considering solutions of the form
\beq
q_i(\hbar) = \sum_{m=0}^\infty \hbar^m q_i^{(m)}(\hbar),
\eeq
where the coefficients $q_i^{(m)}(\hbar)$ depend on $\hbar$ only through oscillatory integrals of the form $
\left. \frac{\partial^n \Theta_{\text{Riemann}}(\mathbf{v}, \boldsymbol{\tau})}{\partial v_{i_1} \dots \partial v_{i_n}} \right|_{\mathbf{v} = \frac{\boldsymbol{\phi} + \boldsymbol{\rho}}{\hbar} } .
$
This equation is solved order by order in $\hbar$ so that the number of solutions is fixed by the number of solutions of the leading order equation in $\hbar$ which has $G$ distinct solutions from the admissibility assumption (Definition~\ref{DefAdmissibleSpectralCurve}) of the classical spectral curve. Hence, $q_i(\hbar) \neq q_j(\hbar)$, for $i \neq j$.
\end{proof}

Let us remark that for generic values of the spectral times $(t_{p,k})_{p\in x^{-1}(\mathcal{P}),k\in\llbracket 0,r_p-1\rrbracket}$, the number of zeros of the Wronskian is expected to be equal to the genus $g$ of the classical spectral curve (see Examples \ref{sec-examples}). The zeros $(q_i(\hbar))_{i=1}^G$ are trans-series in $\hbar$ that can be computed through the asymptotics of $\Psi$ at its poles. The latter may be computed by expanding the KZ equations around the poles $P \in \mathcal{P}$. We shall see this procedure in the examples of Section~\ref{sec-examples}.

\subsubsection{Behavior of the quantum curve at poles}\label{Sectionbehaviorquantumpoles}

We now know that the coefficients $\left(b_j(\lambda,\hbar)\right)_{j\in\llbracket 1, d \rrbracket}$ have simple poles at $\lambda \in\{q_i\}_{i=1}^G$ and possibly higher order poles at $\lambda \in \mathcal{P}$. From a quantization perspective, one expects the leading order of $b_j(\lambda,\hbar)$ to behave as its classical limit $P_j(\lambda)$ at any pole in $\mathcal{P}$. Let us check this and be more precise in our description of $b_j(\lambda,\hbar)$.

 The coefficients $\left(b_j(\lambda,\hbar)\right)_{j=1}^d$ of the Lax matrix $L(\lambda,\hbar)$ are solutions to a system of $d$ equations
\beq
\sum_{k=0}^{d-1}  \frac{b_{d-k}(\lambda,\hbar)}{\psi^{(j)}(\lambda,\hbar)} \frac{\hbar^k\partial^k \psi^{(j)}(\lambda,\hbar)}{\partial \lambda^k} = -  \frac{1}{\psi^{(j)}(\lambda,\hbar)} \frac{\hbar^d\partial^d \psi^{(j)}(\lambda,\hbar)}{\partial \lambda^d},
\eeq
for $j\in \llbracket 1,d\rrbracket$.
Away from the zeros $\lambda = q_i(\hbar)$ of the Wronskian, this equation has a unique solution in $\left(b_j(\lambda,\hbar)\right)_{j=1}^d$. Let us study this solution as $\lambda \to P \in \mathcal{P}$. We shall first evaluate the behavior of the coefficients of this equation.

\bl
One has
\beq
\forall j \in \llbracket 1 , d \rrbracket \, , \; \forall\, k \geq 1 \, , \; \frac{1}{\psi^{(j)}(\lambda)} \frac{\hbar^k\partial^k \psi^{(j)}\left(\lambda\right)}{\partial \lambda^k} = y(z^{(j)}(\lambda))^k + O\left(\frac{\partial y(z^{(j)}(\lambda))^{k-1}}{\partial \lambda} \right),
\eeq
when $\lambda \to P \in \mathcal{P}$.
\el

\begin{proof}
We shall proceed by induction. For $k=1$, this is obviously valid. Let us assume that the proposition is valid for $k-1$, let us prove it for $k$.
A simple computation shows that, for all $k \geq 2$,
\begin{align}
& \qquad\qquad\qquad\qquad\qquad\qquad\quad \frac{1}{\psi^{(j)}(x(z))} \frac{\hbar^k\partial^k \psi^{(j)}(x(z))}{\partial x(z)^k} =  \\
& \frac{ \hbar\partial }{\partial x(z)} \left[\frac{1}{\psi^{(j)}(x(z))}\frac{\hbar^{k-1}\partial^{k-1} \psi^{(j)}(x(z))}{\partial x(z)^{k-1}}\right] + \frac{1}{\psi^{(j)}(x(z))} \frac{\hbar\psi^{(j)}(x(z))}{\partial x(z)}  \frac{1}{\psi^{(j)}(x(z))} \frac{\hbar^{k-1}\partial^{k-1} \psi^{(j)}(x(z))}{\partial x(z)^{k-1}}, \nonumber
\end{align} 
which implies the lemma.
\end{proof}

Using the behavior of the function $y$ at its poles, this implies that, around a pole $P \in \mathcal{P}$, one has, for all $j \in \llbracket 1 , d \rrbracket$  
\beq
 \forall\, k \geq 1 \, , \; \frac{1}{\psi^{(j)}(\lambda)} \frac{\hbar^k\partial^k \psi^{(j)}\left(\lambda\right)}{\partial \lambda^k} = y(z^{(j)}(\lambda))^k + O\left(\zeta_{p^{(j)}}(\lambda)^{(k-1)(-r_{p^{(j)}}-\epsilon_P d_{p^{(j)}}+1)-\epsilon_P d_{p^{(j)}}} \right)
\eeq
where $p^{(j)}$ denotes the pre-image of $P$ in the $j^{\text{th}}$ sheet of the cover.

On the other hand, one knows that the rational functions $\left(P_l(\lambda)\right)_{l=0}^d$ satisfy the classical spectral curve equation, i.e.
\beq \label{remarkrevision}
\forall\,  j\in \llbracket 1,d\rrbracket \, : \; \sum_{k=0}^{d-1} (-1)^{d-k}P_{d-k}(\lambda) \left[y(z^{(j)}(\lambda))\right]^k 
= -  \left[y(z^{(j)}(\lambda))\right]^d.
\eeq

{Hence, the leading order terms of $b_k(\lambda,\hbar)$ as $\lambda \to P\in \mathcal{P}$ match those of $(-1)^{d-k}P_{d-k}(\lambda)$ and  thus  are $\hbar$-independent. The exact number of terms independent of $\hbar$ can be computed on a case by case basis.\footnote{These terms generically correspond to the Casimirs obtained by studying the associated isospectral system.}

\subsection{Pole structure of the quantum curve}


Since coefficients $\left(b_i(\lambda,\hbar)\right)_{i\in \llbracket 1, d\rrbracket}$ of the quantum curve might only have simple poles at apparent singularities, we may regroup the observation above together with results of Sections \ref{SectionWronskian} and \ref{Sectionbehaviorquantumpoles} to obtain the following expression.

\bt\label{deformed-curve}
The coefficients of the quantum curve read
\beq
\forall\, l \in \llbracket 1,d\rrbracket \, : \; b_{l}(\lambda,\hbar) = \sum_{P \in \mathcal{P}} \sum_{k \in \mathcal{S}_{P}^{(l)}} B_{P,k}^{(l)}(\hbar) \xi_{P}(\lambda)^{-k} + \sum_{i=1}^G \frac{p_i^{(l)}(\hbar)}{\lambda-q_i(\hbar)},
\eeq
for some trans-series $p_i^{(l)}(\hbar)$ and $B_{P,k}^{(l)}(\hbar)$ of the same nature as the singularities $q_i(\hbar)$. 

Moreover,
for any $l\in \llbracket 1,d\rrbracket$, some of the leading coefficients $B_{P,\mathcal{S}_{P}^{(l)}}^{(l)}(\hbar)$ of the singular part of $b_{l}(\lambda,\hbar) $ around any $P \in \mathcal{P}$ do not depend on $\hbar$ since they are equal to the leading coefficient of $(-1)^l P_l(\lambda)$.

\et


In addition to the behavior of the coefficients of the quantum curve around its poles, it is possible to study how the later behaves as $\hbar \to 0$.
For this purpose, let us recall that the non-perturbative wave functions are obtained as evaluations of theta-symbols. From this perspective, any expression of the form $\frac{1}{\psi^{(j)}(\lambda)} \frac{\hbar^k\partial^k \psi^{(j)}\left(\lambda\right)}{\partial \lambda^k}$, for $k \in \mathbb{N}$, admits a well defined $\hbar \to 0$ limit as
\beq \label{Limithbar0specific}
\lim_{\hbar \to 0} \frac{1}{\psi^{(j)}(\lambda)} \frac{\hbar^k\partial^k \psi^{(j)}\left(\lambda\right)}{\partial \lambda^k} = y(z^{(j)})^k.
\eeq

Let us remark that the coefficients $\left(b_i(\lambda,\hbar)\right)_{i\in \llbracket 1, d\rrbracket}$ can be expressed through
$$
L(\lambda,\hbar) = \hbar \frac{\partial \Psi(\lambda,\hbar)}{\partial \lambda} \cdot \Psi(\lambda,\hbar)^{-1}.
$$
From this perspective, any $b_i(\lambda,\hbar)$ can be expressed in terms of combinations of expressions of the form $\frac{1}{\psi^{(j)}(\lambda)} \frac{\hbar^k\partial^k \psi^{(j)}\left(\lambda\right)}{\partial \lambda^k}$. For example, the coefficient $b_1(\lambda,\hbar)$ is given by the logarithmic derivative of the Wronskian
$$
b_1(\lambda,\hbar) = - \frac{1}{W(\lambda,\hbar)} \frac{\partial W(\lambda,\hbar)}{\partial \lambda}.$$
This implies that $b_i(\lambda,\hbar)$ admits a well defined limit as $\hbar \to 0$ as a trans-series built out of theta-symbols.

One can thus consider the $\hbar \to 0$ limit of the quantum curve which is nothing but the classical spectral curve thanks to \eqref{Limithbar0specific}. One thus obtains our second main result motivating the name ``quantum curve'' in Definition~\ref{DefQuantumCurve}.

\bt\label{limit-curve}
One has 
\beq
 \forall\, l\in \llbracket 1,d\rrbracket \, : \;   b_l(\lambda,\hbar) \overset{\hbar \to 0}{\to} (-1)^l P_l(\lambda).
\eeq
\et

Let us emphasize that the position of the apparent singularities as $q_i(\hbar)$, the corresponding residues $p_i^{(l)}(\hbar)$ as well as the coefficients $B_{P,k}^{(l)}(\hbar)$ can be computed explicitly in terms of the coefficients of the expansion of the non-perturbative wave functions around its singularities $p \in \mathcal{P}$. They can thus be expressed in terms of quantities computed by the topological recursion procedure. In examples presented in Section~\ref{sec-examples}, we alternatively write these coefficients in terms of the spectral times $\left(t_{p,k}\right)_{p\in x^{-1}(\mathcal{P}),k\in \llbracket 0, r_p-1\rrbracket}$ and $g$ Hamiltonians $\left(H_i\right)_{i=1}^g$ whose expressions are explicit in terms of the Darboux coordinates $\left(q_i,p_i\right)_{i=1}^g$ and spectral times. These results are similar to the ones obtained for hyper-elliptic curves in \cite{MO19_hyper} and we believe that they should hold in the general framework. We leave this conjecture for future works. 

\subsection{Gauge transformation to remove apparent singularities}\label{SectionGaugeTransfo}

We have built a differential equation (Definition~\eqref{DefQuantumCurve}) with the expected singularities as $\lambda \to \mathcal{P}$. However, equation \eqref{eq-quant-curve-1} involves singularities away from $\mathcal{P}$ as well. 
Since the solution $\Psi(\lambda,\hbar)$ is regular at these points $\left(q_i(\hbar)\right)_{i=1}^G$, these are apparent singularities, meaning that they are at most simple poles of the coefficients $(b_l(\lambda,\hbar))_{l=1}^d$. In addition, one knows that there exists a gauge transformation transforming this system into a system without apparent singularities. Let us exhibit this gauge transformation. 

\subsubsection{Expression of the gauge transformation removing apparent singularities}

\bd[Gauge transformation]\label{GaugeTransfo}
We define the following gauge transformation matrix,
\beqq
J(\lambda,\hbar) \coloneqq
 \left[
 \begin{array}{ccccc}
 1 & 0  & \dots &0&  0\cr
0 & 1 & \dots & 0& 0\cr
  \ddots & \ddots &\ddots & \ddots \cr
  0 & 0 &  \dots &  1&0 \cr
  \frac{Q_d(\lambda,\hbar)}{ \underset{i=1}{\overset{G}{\prod}} (\lambda-q_i(\hbar)) }
&  \frac{Q_{d-1}(\lambda,\hbar)}{ \underset{i=1}{\overset{G}{\prod}} (\lambda-q_i(\hbar)) }
&  \dots &  \frac{Q_2(\lambda,\hbar)}{ \underset{i=1}{\overset{G}{\prod}} (\lambda-q_i(\hbar))}
 &  \frac{Q_1(\lambda,\hbar)}{ \underset{i=1}{\overset{G}{\prod}} (\lambda-q_i(\hbar))}
    \cr  
\end{array}
\right]
\eeqq 
where the coefficients $(Q_i(\lambda,\hbar))_{i=1}^d$ are defined by interpolation by the constraints
\begin{itemize}
\item for $j\in \llbracket 1, d\rrbracket$, $Q_j(\lambda,\hbar)$ is a polynomial of degree at most  $G-1$;
\item for $j\in \llbracket 1, d\rrbracket$ we have
\beq\label{pi2}
\forall\,  i\in \llbracket 1,G\rrbracket \, : \; \frac{Q_j(q_i(\hbar),\hbar)}{\underset{k=1}{\overset{N}{\prod}} (q_i(\hbar)-\Lambda_k)^{G_{\Lambda_k}}}= - \hbar^{-1} \Res_{\lambda \to q_i(\hbar)} b_j(\lambda,\hbar) \, d\lambda = - \hbar^{-1} \,  p_i^{(j)}(\hbar).
\eeq
\end{itemize}
\hfill $\star$
\ed

In particular, since $b_1(\lambda,\hbar) = - \hbar \frac{\partial \log \det \Psi(\lambda,\hbar)}{\partial \lambda}$, one has 
\beq
\forall\, i\in \llbracket 1,G\rrbracket \, : \; \frac{Q_1(q_i(\hbar),\hbar)}{ \underset{k=1}{\overset{N}{\prod}} (q_i(\hbar)-\Lambda_k)^{G_{\Lambda_k}}} = - \hbar^{-1} \Res_{\lambda \to q_i(\hbar)} b_1(\lambda,\hbar) \, d\lambda = 1,
\eeq
so that
\beq
Q_1(\lambda,\hbar) = \underset{k=1}{\overset{N}{\prod}} (\lambda-\Lambda_k)^{G_{\Lambda_k}}
\eeq
does not depend on $\hbar$.

This implies that
\beq
\det J(\lambda,\hbar) = \frac{ \underset{k=1}{\overset{N}{\prod}} (\lambda-\Lambda_k)^{G_{\Lambda_k}}}{\underset{i=1}{\overset{G}{\prod}} (\lambda-q_i(\hbar))}.
\eeq 

This gauge transformation finally allows us to obtain our final system of ODEs.

\bt\label{mainThm}
The matrix 
\beq
\check{\Psi}(\lambda,\hbar) \coloneqq J(\lambda,\hbar) \Psi(\lambda,\hbar)
\eeq
is solution to the linear differential equation
\beq\label{DefCheckL}
\hbar \frac{\partial \check{\Psi}(\lambda,\hbar)}{\partial \lambda} = \check{L}(\lambda,\hbar) \check{\Psi}(\lambda,\hbar),
\eeq
where the Lax matrix
\beq
\check{L}(\lambda,\hbar) \coloneqq J(\lambda,\hbar) L(\lambda,\hbar) J(\lambda,\hbar)^{-1} + \hbar \frac{\partial J(\lambda,\hbar)}{\partial \lambda} J(\lambda,\hbar)^{-1}
\eeq
is a rational function of $\lambda$ with poles only at $\lambda \in \mathcal{P}$ whose expression is given in the following Proposition~\ref{ExpressionLnew}. 
\et

\begin{proof}
Let us first study the possible singularities of $\check{\Psi}(\lambda,\hbar)$. It inherits the singularities at $\lambda \in \mathcal{P}$ from $\Psi(\lambda,\hbar)$ but could now have simple poles at $\lambda = q_i(\hbar)$ for any $i\in \llbracket 1,G\rrbracket$, from the poles of $J(\lambda,\hbar)$. Let us rule out this possibility. For this purpose, we can compute the expansion of $J(\lambda,\hbar)$ around $\lambda = q_i$. It reads
\bea
J(\lambda,\hbar) &=& 
 \left[
 \begin{array}{ccccc}
 0 & 0  & \dots &0&  0\cr
0 & 0 & \dots & 0& 0\cr
  \vdots & \ddots &\ddots & \ddots&\vdots \cr
  0 & 0 &  \dots &  0&0 \cr
{\displaystyle  \Res_{\lambda \to q_i}} b_d(\lambda,\hbar) \, d\lambda& {\displaystyle  \Res_{\lambda \to q_i}} b_{d-1}(\lambda,\hbar) \, d\lambda &  \dots &  {\displaystyle  \Res_{\lambda \to q_i}}b_2(\lambda,\hbar) \, d\lambda  & {\displaystyle  \Res_{\lambda \to q_i}} b_1(\lambda,\hbar) \, d\lambda \cr  
\end{array}
\right] \cr
&&\times
 \frac{ \underset{k=1}{\overset{N}{\prod}} (q_i(\hbar)-\Lambda_k)^{G_{\Lambda_k}} }{\underset{j\in \llbracket 1,d\rrbracket \setminus \{i\}}{\prod} (q_i(\hbar)-q_j(\hbar))} \, \frac{1}{(\lambda-q_i(\hbar))} + O(1) .
\eea
The residue of $\check{\Psi}(\lambda,\hbar)$ at the simple pole $\lambda = q_i(\hbar)$ is thus proportional to 
\begin{align}
&\Res_{\lambda \to q_i(\hbar)} \check{\Psi}(\lambda,\hbar) d\lambda  = \nonumber\\
& \frac{ \underset{k=1}{\overset{N}{\prod}} (q_i(\hbar)-\Lambda_k)^{G_{\Lambda_k}} }{\underset{j\in \llbracket 1,d\rrbracket \setminus \{i\}}{\prod} (q_i(\hbar)-q_j(\hbar))} \,  \Res_{\lambda \to q_i(\hbar)}  \left[
 \begin{array}{ccccc}
 0 & 0  & \dots &0&  0\cr
0 & 0 & \dots & 0& 0\cr
  \ddots & \ddots &\ddots & \ddots \cr
  0 & 0 &  \dots &  0&0 \cr
 b_d(\lambda,\hbar) &  b_{d-1}(\lambda,\hbar)&  \dots &b_2(\lambda,\hbar) & b_1(\lambda,\hbar)  \cr  
\end{array}
\right] \Psi(\lambda,\hbar) \,d\lambda.
\end{align}
Hence
\beq
\forall\, l\in \llbracket1,d-1\rrbracket \, , \; \forall\, m\in \llbracket 1 , d\rrbracket \, : \; \Res_{\lambda \to q_i} \left[\check{\Psi}(\lambda,\hbar)\right]_{l,m} d\lambda = 0,
\eeq
while, for $j\in \llbracket 1,d\rrbracket$, we have
\beq
\Res_{\lambda \to q_i(\hbar)} \left[\check{\Psi}(\lambda,\hbar)\right]_{d,j} d\lambda  = \frac{\underset{k=1}{\overset{N}{\prod}} (q_i(\hbar)-\Lambda_k)^{G_{\Lambda_k}} }{\underset{j\in \llbracket 1,d\rrbracket \setminus \{i\}}{\prod} (q_i(\hbar)-q_j(\hbar))}  \Res_{\lambda \to q_i(\hbar)} \sum_{k=0}^{d-1} b_{d-k}(\lambda,\hbar) \left(\hbar \frac{\partial}{\partial \lambda}\right)^k \psi^{(j)}(\lambda,\hbar) d\lambda .
\eeq
From \eqref{eq-quant-curve-1}, this is equal to 
\beq
\Res_{\lambda \to q_i(\hbar)} \left[\check{\Psi}(\lambda,\hbar)\right]_{d,j} d\lambda  =  -  \frac{ \underset{k=1}{\overset{N}{\prod}} (q_i(\hbar)-\Lambda_k)^{G_{\Lambda_k}} }{\underset{j\in \llbracket 1,d\rrbracket \setminus \{i\}} {\prod} (q_i(\hbar)-q_j(\hbar))}   \Res_{\lambda \to q_i(\hbar)} \left(\hbar \frac{\partial}{\partial \lambda}\right)^d \psi^{(j)}(\lambda,\hbar) d\lambda,
\eeq
which is vanishing since $\psi^{(j)}(\lambda,\hbar)$ does not have any singularity at $\lambda = q_i(\hbar)$. Hence, $\check{\Psi}(\lambda,\hbar)$ does not have any pole at $\lambda = q_i(\hbar)$ and only has singularities at $\lambda \in \mathcal{P}$.

This means that, apart from $\lambda \in \mathcal{P}$, $\check{L}(\lambda,\hbar)$ might have only apparent singularities. Let us recall that these apparent singularities can only appear as zeros of the Wronskian of this system  $\det \check{\Psi} (\lambda,\hbar)$.

One can easily compute
\beq
\det \check{\Psi} (\lambda,\hbar) = \det(J(\lambda,\hbar) ) \, \det(\Psi(\lambda,\hbar)) = \kappa  \, \exp\left( \hbar^{-1} \int_{0}^\lambda P_1(\lambda) d\lambda\right),
\eeq
which is not vanishing away from $\lambda \in \mathcal{P}$. Thus $\check{L}(\lambda,\hbar)$ does not have any apparent singularity.
\end{proof}

Let us now compute explicitly $\check{L}(\lambda,\hbar)$. A straightforward computation gives the following proposition.

\begin{proposition}[Expression of $\check{L}$]\label{ExpressionLnew} The entries of the matrix $\check{L}(\lambda,\hbar)$ are 
\beqq
\forall\, l\in \llbracket 1,d-2\rrbracket \, , \; \forall\, m\in \llbracket 1,d\rrbracket \, :\; \left[\check{L}(\lambda,\hbar) \right]_{l,m} = \delta_{l,m-1},
\eeqq
\beaa
\forall\, m\in \llbracket 1,d-1\rrbracket \, : \left[\check{L}(\lambda,\hbar) \right]_{d-1,m} &=& - \frac{Q_{d-m+1}(\lambda,\hbar)}{\underset{k=1}{\overset{N}{\prod}} (\lambda-\Lambda_k)^{G_{\Lambda_k}}},
\cr
 \left[\check{L}(\lambda,\hbar) \right]_{d-1,d} &=& \frac{\underset{j=1}{\overset{G}{\prod}} (\lambda - q_j(\hbar)) }{\underset{k=1}{\overset{N}{\prod}} (\lambda-\Lambda_k)^{G_{\Lambda_k}}},
\eeaa
\beaa
\left[\check{L}(\lambda,\hbar) \right]_{d,1} & =&  \hbar \frac{ \partial \bigg(\frac{Q_d(\lambda)(\hbar)}{\underset{j=1}{\overset{G}{\prod}} (\lambda-q_j(\hbar))}\bigg)}{\partial \lambda} -b_d(\lambda,\hbar) \frac{\underset{k=1}{\overset{N}{\prod}} (\lambda-\Lambda_k)^{G_k}}{\underset{j=1}{\overset{G}{\prod}} (\lambda - q_j(\hbar))}\cr 
&& - P_1(\lambda) \frac{Q_d(\lambda,\hbar)}{\underset{j=1}{\overset{G}{\prod}} (\lambda-q_j(\hbar))} - \frac{Q_2(\lambda,\hbar) \, Q_d(\lambda,\hbar)}{\underset{j=1}{\overset{G}{\prod}}(\lambda-q_j(\hbar)) \, {\underset{k=1}{\overset{N}{\prod}} (\lambda-\Lambda_k)^{G_{\Lambda_k}}}},
\eeaa
and for $m\in \llbracket 1,d-1\rrbracket$:
\beaa
\left[\check{L}(\lambda,\hbar) \right]_{d,m} &=&
\hbar \frac{ \partial \bigg(\frac{Q_{d-m+1}(\lambda,\hbar)}{\underset{j=1}{\overset{G}{\prod}} (\lambda-q_j(\hbar))}\bigg)}{\partial \lambda}  - b_{d-m+1}(\lambda,\hbar) \frac{\underset{k=1}{\overset{N}{\prod}} (\lambda-\Lambda_k)^{G_k}}{\underset{j=1}{\overset{G}{\prod}} (\lambda - q_j(\hbar))}  - P_1(\lambda) \frac{Q_{d-m+1}(\lambda,\hbar)}{\underset{j=1}{\overset{G}{\prod}} (\lambda-q_j(\hbar))} \cr
&&- \frac{Q_2(\lambda,\hbar) \, Q_{d-m+1}(\lambda,\hbar)}{\underset{j=1}{\overset{G}{\prod}}(\lambda-q_j(\hbar)) \, {\underset{k=1}{\overset{N}{\prod}} (\lambda-\Lambda_k)^{G_{\Lambda_k}}}} + \frac{Q_{d-m+2}(\lambda,\hbar)}{\underset{j=1}{\overset{G}{\prod}} (\lambda-q_j(\hbar))} 
\eeaa
and finally
\beqq
\left[\check{L}(\lambda,\hbar) \right]_{d,d} = P_1(\lambda) + \frac{Q_2(\lambda,\hbar)}{\underset{k=1}{\overset{N}{\prod}} (\lambda-\Lambda_k)^{G_{\Lambda_k}}} .
\eeqq

Let us remind the reader that the polynomials $\left(Q_j(\lambda,\hbar)\right)_{j=1}^d$ are defined by interpolation in Definition \ref{GaugeTransfo}.
\end{proposition}

\subsubsection{$\hbar$-Deformed spectral curve}

The $\hbar$-deformed  spectral curve associated to the differential system \eqref{DefCheckL} is defined by $\det(y- \check{L}(\lambda,\hbar)) $. It reads
\beq
\det( y- \check{L}(\lambda,\hbar))   =  y^d + \sum_{k=0}^{d-1} y^k \check{b}_{d-k}(\lambda,\hbar),
\eeq
where for all $j\in \llbracket1, d\rrbracket$, we have
\beq
\check{b}_{j}(\lambda,\hbar) =  b_j(\lambda,\hbar) + \hbar \frac{Q_j(\lambda,\hbar)}{\underset{k=1}{\overset{N}{\prod}}(\lambda-\Lambda_k)^{G_{\Lambda_k}}} \sum_{i=1}^G \frac{1}{\lambda-q_i(\hbar)} - \frac{\hbar}{\underset{k=1}{\overset{N}{\prod}} (\lambda-\Lambda_k)^{G_{\Lambda_k}}} \frac{\partial Q_j(\lambda,\hbar)}{\partial \lambda} .
\eeq

Remark in particular that
\beq
\check{b}_1(\lambda,\hbar) = -P_1(\lambda)
\eeq
is independent of $\hbar$, while the other coefficients read, for all $j\in \llbracket 2, d\rrbracket$,
\begin{align}
\check{b}_{j}(\lambda,\hbar)& = \sum_{P \in \mathcal{P}} \sum_{k \in \mathcal{S}_{P}^{(j)}} B_{P,k}^{(j)}\, \xi_{P}(\lambda)^{-k}  + \sum_{i=1}^G \frac{1}{\lambda-q_i(\hbar)} \Bigg[\hbar \frac{Q_j(\lambda,\hbar)}{\underset{k=1}{\overset{N}{\prod}} (\lambda-\Lambda_k)^{G_{\Lambda_k}}} + p_i^{(j)}(\hbar) \Bigg]\nonumber\\
& \qquad - \frac{\hbar}{\underset{k=1}{\overset{N}{\prod}} (\lambda-\Lambda_k)^{G_{\Lambda_k}}} \frac{\partial Q_j(\lambda,\hbar)}{\partial \lambda},
\end{align}
whose poles at $(q_i(\hbar))_{i=1}^d$ cancel due to the conditions on the polynomials $(Q_i)_{i=1}^d$ following Definition~\ref{GaugeTransfo}.

\medskip

In \cite{MOsl2}, following \cite{Darboux_coord93}, it was observed that the zeros of the Wronskian could be used to define Darboux coordinates describing our Lax matrix as a point in a coadjoint orbit. In order to prepare future work explaining our result from the point of view of isomonodromic systems, let us observe that we have obtained a very similar structure in this more general context. This is summarized by the following result.

\bt
For any $i\in \llbracket 1,G\rrbracket$, the pair $\left(q_i, P_1(q_i) - \hbar^{-1} p_i^{(2)}\right)$ defines a point on the $\hbar$-deformed classical spectral curve, i.e.
\beq
\det(P_1(q_i) - \hbar^{-1} p_i^{(2)}(\hbar)- \check{L}(q_i,\hbar))  = 0.
\eeq
\et

\begin{proof}
The theorem follows from the observation that the last column of the matrix $y-\check{L}(q_i,\hbar)$ has only one non-vanishing element given by 
\beq
\left[y-\check{L}(q_i,\hbar)\right]_{d,d} = y - P_1(\lambda) - \frac{Q_2(\lambda,\hbar)}{\underset{k=1}{\overset{N}{\prod}} (\lambda-\Lambda_k)^{G_{\Lambda_k}}}.
\eeq
Hence, its determinant is proportional to this element that vanishes for $y = P_1(q_i) - \hbar^{-1} p_i^{(2)}(\hbar)$.
\end{proof}

\begin{remark}We note that the quantum curve of Definition \ref{DefQuantumCurve} and its parametrization using $\left(q_i, p_i:=P_1(q_i) - \hbar^{-1} p_i^{(2)}\right)_{1\leq i\leq G}$ recover results of Dubrovin and Mazzocco \cite{DubrovinMazzocco} in the case of Fuchsian singularities. Indeed, the parametrization is the same in both setups since the $\left(q_i\right)_{1\leq i\leq G}$ are the apparent singularities (i.e. zeros of the Wronskian as observed in Corollary \ref{Corollaryqi}) while from \eqref{pi2} we have
\beq \forall \,i\in \llbracket 1,G\rrbracket \,:\, p_i^{(2)}=\underset{\lambda \to q_i}\Res\, b_2(\lambda)\eeq
as used by Dubrovin and Mazzocco (See page $5$ of \cite{DubrovinMazzocco}). The main difference is that we obtain a formal family of quantum curves parametrized by the parameter $\hbar$ that could be inserted in \cite{DubrovinMazzocco} by a simple rescaling of the spectral parameter, the position of the poles and their monodromies and of the wave matrix. Up to this trivial rescaling, the quantum curves (eq. $1.8$ of \cite{DubrovinMazzocco}) are the same because they have the same pole structure in both setups.
\end{remark}

\subsection{Summary of the different interesting gauges}

From Section~\ref{6.4} up to Section~\ref{SectionGaugeTransfo}, we used $4$ different gauges. Each of them has advantages and limitations that we detail here. Note that none of the gauge transformations modify the first line of the wave functions matrix so that it remains the same in the $4$ different gauges (and is used to define the quantum curve in Definition~\ref{DefQuantumCurve}).
\begin{itemize}\item \underline{Gauge $\widehat{\Psi}$}: This is the natural gauge coming from KZ equations and it provides compatible auxiliary systems $\left(\mathcal{L}_{P,k,l}\right)_{P\in \mathcal{P}, l\in \llbracket 0, d-1\rrbracket, k\in S_P^{(l+1)}}$. The main limitations are that the corresponding Lax matrix $\widehat{L}$ is not companion-like and that even its order $\hbar^0$ is not companion-like. This means that both the classical and quantum spectral curves are not directly readable from the Lax matrix. This gauge may also contain apparent singularities.
\item \underline{Gauge $\widetilde{\Psi}$}: This gauge shares many properties of the previous one and it is obtained from it through an $\hbar^0$ gauge transformation \eqref{GaugeG}. The main difference is that the leading order in $\hbar$ of $\widetilde{L}$ is companion-like so that the classical spectral curve is directly recovered from its last line. However, this does not hold for any order in $\hbar$ so that the quantum curve is not directly readable in this gauge. This gauge may also contains apparent singularities.
\item \underline{Gauge $\Psi$}: This is a gauge in which the corresponding Lax matrix $L$ is companion-like at all orders in $\hbar$. This means that both the quantum and classical curves are directly read from the last line of $L$ and its $\hbar\to 0$ limit. This gauge may contain apparent singularities but it provides a natural framework for Darboux coordinates and isomonodromic deformations.
\item \underline{Gauge $\check{\Psi}$}: This gauge is obtained from the previous one with an explicit matrix given in Definition~\ref{GaugeTransfo}. In this gauge, the Lax matrix $\check{L}$ has no apparent singularities. This allows to interpret $\check{L}(\lambda,\hbar) d\lambda$ as an $\hbar$-familly of Higgs fields giving rise to a flow in the corresponding Hitchin system. However in this gauge, $\check{L}$ is no longer companion-like (last two lines are non-trivial) so it is less adapted to directly read the classical and quantum curves. This gauge should also provide Lax systems generalizing Jimbo--Miwa Lax pairs for the Painlev\'{e} cases \cite{JimboMiwa}.   
\end{itemize}

\subsection{Practical computations}

In this section, we have built a quantum curve  \eqref{eq-quant-curve-1} and some associated linear systems $L(\lambda,\hbar)$ (given by equation~\eqref{ExpressionL} and Theorem~\ref{deformed-curve}) and $\check{L}(\lambda,\hbar)$ (given by Proposition~\ref{ExpressionLnew}). Both systems are expressed in terms of the apparent singularities $(q_i(\hbar))_{i=1}^d$. As we shall see in the examples, the position of these apparent singularities can be expressed in terms of the action of the linear operators $\mathcal{I}_C$ on the non-perturbative wave functions $\psi_{\mathrm{NP}}^{(\infty^{(\alpha)})}$. In practice, here is the procedure one should follow to quantize a classical spectral curve.

\begin{enumerate}
\item Write down the KZ equations \eqref{eq-NP-KZ} satisfied by the non-perturbative wave function.

\item Expand these KZ equations around each pole $\lambda \to P \in \mathcal{P}$  in order to obtain an expression of the coefficients of the asymptotic expansion of $\psi_{0,\mathrm{NP}}^{(\infty^{(\alpha)})}$ in terms of the action of the operators $\mathcal{I}_C$.

\item Use the latter expressions to compute the Wronskian of the system thanks to its expansion around its poles. This allows to compute the position of the apparent singularities $(q_i(\hbar))_{i=1}^d$.

\item Finally write down the linear system obtained, as well as the associated quantum curve, and use the compatibility of the system to recover its properties.
\end{enumerate}

\section{Examples}\label{sec-examples}

In this section we present the details of the quantization procedure applied to two particular cases, of degree $2$ and $3$ respectively.

\subsection{Painlev\'e 2 example}
In this section, we consider the simplest example of degree $2$ where the classical spectral curve has genus $1$ and a unique pole located at infinity. This type of curve was already quantized in \cite{MOsl2} in a slightly different context by imposing an additional symmetry corresponding to $P_1(\lambda) = 0$. Moreover, in \cite{MOsl2}, the authors considered a divisor of the form $D=[z]-[\sigma_a(z)]$ and not $D=[z]-[\infty^{(\alpha)}]$ as in the present paper. However, in both cases, one manages to find a quantum curve and recovers a 2-parameter solution of the Painlev\'e $2$ equation from the associated isomonodromic system.

We consider the case when $d=2$, $N=0$, $r_\infty^{(1)} = 2$ and $r_\infty^{(2)} = 4$. We consider $l_\infty = 2$ so that there are two points above infinity denoted by $\infty^{(1)}$ and $\infty^{(2)}$ respectively with $d_{\infty^{(1)}} = d_{\infty^{(2)}} = 1$.

\subsubsection{Classical spectral curve}\label{P2cla}

Let us consider a two-sheeted cover of the sphere defined by the equation
\beq\label{P2ClassicalSpectralCurve}
y^2 - P_1(\lambda) y + P_2(\lambda)= 0,
\eeq
where
\beq\label{def-P1}
P_1(\lambda) = P_{\infty,2}^{(1)} \lambda^2 + P_{\infty,1}^{(1)} \lambda + P_{\infty,0}^{(1)}
\eeq
and
\beq\label{def-P2}
P_2(\lambda) = P_{\infty,4}^{(2)} \lambda^4 + P_{\infty,3}^{(2)} \lambda^3+ P_{\infty,2}^{(2)} \lambda^2 + P_{\infty,1}^{(2)} \lambda + P_{\infty,0}^{(2)}.
\eeq
The function $y$ admits the following expansions around the two poles $\left(\infty^{(i)}\right)_{i=1}^2$,
\beq
y(z) = - t_{i,3} x(z)^2  - t_{i,2} x(z) -  t_{i,1} -  t_{i,0} x(z)^{-1} - \frac{\partial \om_{0,0}}{\partial t_{i,1}}x(z)^{-2} +  O\left(x(z)^{-3}\right),  \text{ as } z \to \infty^{(i)}, 
\eeq
where we used the general theory of topological recursion to express the order $x(z)^{-1}$ term.

\medskip

In terms of these spectral times $\left(t_{i,j}\right)_{1\leq i\leq 2, 1\leq j\leq 3}$, one has
\beq \label{P2Spectimes}
P_{\infty,2}^{(1)} = -t_{1,3} - t_{2,3}, \qquad
P_{\infty,1}^{(1)}  = - t_{1,2} - t_{2,2}, \qquad
P_{\infty,0}^{(1)} = - t_{1,1} -t_{2,1}, \qquad
0 =  - t_{1,0} -t_{2,0}, 
\eeq
as well as
\beq\label{P2Spectimes2}
P_{\infty,4}^{(2)} = t_{1,3} t_{2,3}
 , \quad
P_{\infty,3}^{(2)} = t_{1,2} t_{2,3} +  t_{1,3} t_{2,2}
 , \quad
P_{\infty,2}^{(2)} = t_{1,2} t_{2,2} + t_{1,3} t_{2,1} + t_{1,1} t_{2,3}
\eeq
and
\beq\label{P2Spectimes3}
P_{\infty,1}^{(2)} = t_{1,3} t_{2,0} + t_{1,0} t_{2,3} + t_{1,2} t_{2,1} + t_{1,1} t_{2,2}
 , \quad
P_{\infty,0}^{(2)} = t_{1,2} t_{2,0} + t_{1,0} t_{2,2} + t_{1,1} t_{2,1} + t_{1,3} \frac{\partial \om_{0,0}}{\partial t_{2,1}}  + t_{2,3} \frac{\partial \om_{0,0}}{\partial t_{1,1}},
\eeq
where $\left(\omega_{h,n}\right)_{h,n\geq 0}$ corresponds to the differential forms computed by the topological recursion applied on the classical spectral curve \eqref{P2ClassicalSpectralCurve}.

\subsubsection{KZ equations}
For this example, we choose to consider the divisor $D = [z]-[\infty^{(1)}]$. The choice of $\infty^{(2)}$ would lead to similar results. 

The non-perturbative wave functions satisfy the KZ equations
\beq
\left\{
\begin{array}{l}
\hbar \frac{\partial \psi_{0,NP}^{\infty^{(1)}}(z,\hbar)}{\partial x(z)} +  \psi_{1,NP}^{\infty^{(1)}}(z,\hbar) = P_1(x(z))  \psi_{0,NP}^{\infty^{(1)}}(z,\hbar), \cr
\hbar \frac{\partial \psi_{1,NP}^{\infty^{(1)}}(z,\hbar)}{\partial x(z)}  = P_2(x(z))  \psi_{0,NP}^{\infty^{(1)}}(z,\hbar) + \hbar \text{ev} . \mathcal{L}_{KZ}(x(z))  \left[\psi_{0,NP}^{\infty^{(1)} , \mathrm{symbol}}(z,\hbar)\right], \cr
\end{array} \right.
\eeq
where
\beq
\mathcal{L}_{KZ}(\lambda) \coloneqq \hbar t_{1,3} \mathcal{I}_{\mathcal{C}_{\infty^{(2)},1}} + \hbar t_{2,3} \mathcal{I}_{\mathcal{C}_{\infty^{(1)},1}} -  t_{2,3} \lambda -  t_{2,2}.
\eeq

\subsubsection{Asymptotics}

The asymptotics of the non-perturbative wave functions when $\lambda \to \infty$ being given by the ones of the perturbative ones, let us study the asymptotic behavior of $\psi^{\mathrm{reg}}([z^{(i)}(\lambda]-[\infty^{(1)}],\hbar)$ for $i\in\{1,2\}$ when $\lambda \to \infty$. From the general theory, one has
\beq
\psi^{\mathrm{reg}}([z^{(1)}(\lambda)]-[\infty^{(1)}],\hbar) = \exp \left[[\hbar^{-1} (V_{\infty^{(1)}}(z^{(1)}(\lambda)) + O(1))\right] (C_1+O(\lambda^{-1}))
\eeq
and
\beq
\psi^{\mathrm{reg}}([z^{(2)}(\lambda)]-[\infty^{(1)}],\hbar) = \exp \left[[\hbar^{-1} (V_{\infty^{(2)}}(z^{(2)}(\lambda)) + O(1))\right] \lambda^{-1} (C_2+O(\lambda^{-1})),
\eeq
for non-vanishing constants $(C_i)_{i=1}^2$.

This leads to a (non-perturbative) Wronskian of the form
\bea \label{WronskianP2}
W(\lambda,\hbar)&\coloneqq& \hbar \frac{\partial \psi_{NP}^{\mathrm{reg}}([z^{(2)}(\lambda)]-[\infty^{(1)}])}{\partial \lambda} \psi_{NP}^{\mathrm{reg}}([z^{(1)}(\lambda)]-[\infty^{(1)}]) \cr
&&- \hbar \frac{\partial \psi_{NP}^{\mathrm{reg}}([z^{(1)}(\lambda)]-[\infty^{(1)}])}{\partial \lambda} \psi_{NP}^{\mathrm{reg}}([z^{(2)}(\lambda)]-[\infty^{(1)}]) \cr
&=& \kappa  \exp \left[\hbar^{-1} \int_{0}^\lambda P_1(\lambda) d\lambda\right] (\lambda - q) .
\eea
Let us now compute the asymptotics of the action of the operator $\mathcal{L}$ on the perturbative wave function. From the regularized KZ equations, one has
\beq
\frac{\text{ev}. \mathcal{L}_{KZ} \psi^{\mathrm{reg \; symbol}}(z^{(1)}(\lambda))}{\psi^{reg}(z^{(1)}(\lambda))}= - 2 \hbar t_{2,3} \lambda + O(1)
\eeq
and
\beq
 \frac{\text{ev}.\mathcal{L}_{KZ} \psi^{\mathrm{reg \; symbol}}(z^{(2)}(\lambda))}{\psi^{reg}(z^{(2)}(\lambda))}=  -  \hbar (t_{2,3} + t_{1,3}) \lambda + O(1),
\eeq
as $\lambda \to \infty$.

This implies that 
\bea\label{eq-L11}
&&\frac{\text{ev}. \mathcal{L}_{KZ}[\psi_{NP}^{\mathrm{reg\; symbol}}(z^{(2)}(\lambda))] \hbar \partial_\lambda [\psi_{NP}^{\mathrm{reg}}(z^{(1)}(\lambda))] - \text{ev}.\mathcal{L}_{KZ}[\psi_{NP}^{\mathrm{reg \; symbol}}(z^{(1)}(\lambda))] \hbar \partial_\lambda [\psi_{NP}^{\mathrm{reg}}(z^{(2)}(\lambda))] }{W(\lambda)} \cr
&&= \hbar (t_{1,3}+2 t_{2,3})\lambda + H(\hbar) - \frac{p}{\lambda-q}, 
\eea
where $H$ and $p$ are independent of $\lambda$.

This imposes the form of the Lax system presented in the next section. As we shall see, the coefficient $ \hbar (t_{1,3}+2 t_{2,3})$ can be recovered from the compatibility of the Lax system. We derived it here directly from the action of the operator $\mathcal{L}_{KZ}$ only as a consistency check of our formulation.

\subsubsection{Lax pair formulation}

The general theory developed in this article implies that the wave functions satisfy a quantum curve that is equivalent to a differential system. The non-perturbative wave function
$$
\Psi(\lambda,\hbar)=\begin{pmatrix}\psi_{0,\mathrm{NP}}^{\infty^{(1)}}(z^{(1)}(\lambda),\hbar)&\psi_{0,\mathrm{NP}}^{\infty^{(1)}}(z^{(2)}(\lambda),\hbar)\\
\hbar \partial_\lambda \psi_{0,\mathrm{NP}}^{\infty^{(1)}}(z^{(1)}(\lambda),\hbar)& \hbar\partial_\lambda \psi_{0,\mathrm{NP}}^{\infty^{(1)}}(z^{(2)}(\lambda),\hbar)  \end{pmatrix}
$$
is solution to 
\beq \label{LaxPairP2}\hbar \partial_\lambda \Psi(\lambda,\hbar)=\begin{pmatrix}0&1\\ -P_2(\lambda) +\hbar P_1'(\lambda)+H-\frac{p}{\lambda-q}+\hbar \alpha \lambda & P_1(\lambda)+ \frac{\hbar}{\lambda-q} \end{pmatrix} \Psi(\lambda,\hbar)=L(\lambda,\hbar)\Psi(\lambda,\hbar),\eeq
where $\alpha=t_{1,3}+2t_{2,3}$ thanks to \eqref{eq-L11}. We shall see later that the value of $\alpha$ could alternatively be recovered by the asymptotics of the wave functions.
It is complemented by 
\beq \label{LaxPairP22-1}
\text{ev}.\mathcal{L}_{KZ}(\lambda)[\Psi^{\mathrm{symbol}}(\lambda,\hbar)]
=A_{KZ}(\lambda,\hbar)\Psi(\lambda,\hbar),\eeq
where the first line of $A_{KZ}(\lambda,\hbar)$ is constrained by the KZ equations to be
\beq
A_{KZ}(\lambda,\hbar)
=\begin{pmatrix} -\alpha\lambda-\frac{H}{\hbar}+\frac{p}{\hbar(\lambda-q)}& -\frac{1}{\lambda-q} \\ \left[A_{KZ}\right]_{2,1}(\lambda,\hbar)&\left[A_{KZ}\right]_{2,2}(\lambda,\hbar)\end{pmatrix}.
\eeq
Because $\mathcal{L}_{KZ}(\lambda)$ depends on $\lambda$, $\partial_\lambda \mathcal{L}_{KZ}(\lambda)$ is non-vanishing. In order to get rid of this complication, let us define a new operator
\beq
 \mathcal{L}\coloneqq \mathcal{L}_{KZ}(\lambda) +t_{2,3} \lambda + t_{2,2} = \hbar t_{1,3} \mathcal{I}_{\mathcal{C}_{\infty^{(2)},1}} + \hbar t_{2,3} \mathcal{I}_{\mathcal{C}_{\infty^{(1)},1}}.
 \eeq
 
 The compatible system now reads
 \bea \label{LaxPairP22-2}
\text{ev}.\mathcal{L}[\Psi^{\mathrm{symbol}}(\lambda,\hbar)]&=&\begin{pmatrix} (-\alpha+t_{2,3})\lambda + t_{2,2}  -\frac{H}{\hbar}+\frac{p}{\hbar(\lambda-q)}& -\frac{1}{\lambda-q} \\ A_{2,1}(\lambda,\hbar)&A_{2,2}(\lambda,\hbar)\end{pmatrix} \Psi(\lambda,\hbar) \cr
&=&\begin{pmatrix} P_{\infty,2}^{(1)} \lambda + t_{2,2}  -\frac{H}{\hbar}+\frac{p}{\hbar(\lambda-q)}& -\frac{1}{\lambda-q} \\ A_{2,1}(\lambda,\hbar)&A_{2,2}(\lambda,\hbar)\end{pmatrix} \Psi(\lambda,\hbar) 
 \coloneqq A(\lambda,\hbar)\Psi(\lambda,\hbar).\cr
&&
\eea

In order to ease the notations, we shall replace the operator $\text{ev}.\mathcal{L}$ acting on symbols by an operator $\mathcal{L}$ acting on $\Psi(\lambda,\hbar)$ so that \eqref{LaxPairP22-2} is replaced by 
 \beq \label{LaxPairP22}
\mathcal{L}[\Psi(\lambda,\hbar)]
 = A(\lambda,\hbar)\Psi(\lambda,\hbar).
\eeq
The operator $\mathcal{L}$ acts on any function $w$ built out of generalized integrals of differentials $\om_{h,n}$ through the map $w \to w^{\mathrm{symbol}}$ which replaces a generalized cycle by its symbolic partner. It can be easily checked that any function on which $\mathcal{L}$ acts later on has a unique symbolic partner through this map on which $\text{ev}.\mathcal{L}$ has a well defined action.

The compatibility of the system reads
\beq\label{CompatibilityP2} \mathcal{L}[L(\lambda,\hbar)]=\hbar \partial_\lambda A(\lambda,\hbar)+\left[A(\lambda,\hbar),L(\lambda,\hbar)\right]. \eeq
The first line of the previous compatibility equations (whose l.h.s.~is null) determines the last line of $A(\lambda,\hbar)$. Straightforward computations show that
\bea A_{2,1}(\lambda,\hbar)&=&A_{1,2}(\lambda,\hbar)L_{2,1}(\lambda,\hbar)+\hbar \partial_\lambda A_{1,1}(\lambda,\hbar),\cr
A_{2,2}(\lambda,\hbar)&=&A_{1,1}(\lambda,\hbar)+A_{1,2}(\lambda,\hbar)L_{2,2}(\lambda,\hbar)+\hbar \partial_\lambda A_{1,2}(\lambda,\hbar).
\eea
Using the complete knowledge of the matrix $A(\lambda,\hbar)$, one may now look at the last line of the compatibility equation \eqref{CompatibilityP2}. We get, for the entry $(2,2)$,
\beq \mathcal{L}[L_{2,2}(\lambda,\hbar)]=\hbar \partial_\lambda \Tr A(\lambda,\hbar).\eeq
The l.h.s.~of the previous equation reads
\beq \mathcal{L}[L_{2,2}(\lambda,\hbar)]=\mathcal{L}[P_1](\lambda,\hbar)+\frac{\hbar \mathcal{L}[q]}{(\lambda-q)^2},\eeq
while the r.h.s.~reads
\beq \hbar \partial_\lambda \Tr A(\lambda,\hbar)=\frac{\hbar P_1(q)-2p}{(\lambda-q)^2}-\hbar P_{\infty,2}^{(1)}+2\hbar P_{\infty,2}^{(1)}.\eeq
Thus, we end up with
\begin{align}\label{P2LP1}\mathcal{L}[P_{\infty,2}^{(1)}]=\mathcal{L}[P_{\infty,1}^{(1)}] =0\,,\, \mathcal{L}[P_{\infty,0}^{(1)}]=-\hbar P_{\infty,2}^{(1)}+2\hbar P_{\infty,2}^{(1)}=-\hbar(t_{1,3}+t_{2,3}) \cr \Leftrightarrow  \mathcal{L}[P_1](\lambda)=-\hbar(t_{1,3}+t_{2,3}) \end{align}
and
\beq \label{P2Hamilton1} \mathcal{L}[q]= P_1(q)-2\frac{p}{\hbar}.\eeq
The same computation may be applied to the entry $(2,1)$ of the compatibility equation \eqref{CompatibilityP2}, and it is equivalent to 
\bea  \label{P2Hamilton2}\mathcal{L}[P_{\infty,4}^{(2)}]&=&0\cr
\mathcal{L}[P_{\infty,3}^{(2)}]&=&0\cr
\mathcal{L}[P_{\infty,2}^{(2)}]&=&-2\hbar P_{\infty,4}^{(2)}+ \hbar \left[P_{\infty,2}^{(1)}\right]^2 \cr
\mathcal{L}[P_{\infty,1}^{(2)}]&=&-\hbar P_{\infty,3}^{(2)}+\hbar P_{\infty,1}^{(1)} P_{\infty,2}^{(1)}\cr 
\mathcal{L}[P_{\infty,0}^{(2)}]-\mathcal{L}[H]&=&2\hbar P_{\infty,4}^{(2)}q^2+\hbar P_{\infty,3}^{(2)}q- P_{\infty,2}^{(1)} p+\hbar P_{\infty,0}^{(1)} P_{\infty,2}^{(1)}\cr
\mathcal{L}[p]&=&-P_1'(q) p+\hbar P_2'(q)+\hbar^2 t_{2,3}\cr
H&=&\frac{p^2}{\hbar^2}- P_1(q)\frac{p}{\hbar}+  P_2(q)-\hbar P_1'(q)+\hbar (P_{\infty,2}^{(1)}-t_{2,3}) q\,.
\eea
In particular, combining \eqref{P2Hamilton1} and \eqref{P2Hamilton2} we get that the evolution of $(q,p)$ satisfies the Hamiltonian system
\beq \mathcal{L}[q]=-\hbar\frac{\partial H_0}{\partial p} \,,\,\;\; \mathcal{L}[p]=\hbar\frac{\partial H_0}{\partial q},\eeq
with the Hamiltonian
\begin{align}\label{HamiltonP2}H_0(p,q,\hbar)&= \frac{p^2}{\hbar^2}- P_1(q)\frac{p}{\hbar}+  P_2(q)-\hbar P_1'(q)+\hbar q(2P_{\infty,2}^{(1)}-t_{2,3})\cr
&=\frac{p^2}{\hbar^2}- \left( P_{\infty,2}^{(1)}q^2+P_{\infty,1}^{(1)}q+P_{\infty,0}^{(1)} \right)\frac{p}{\hbar}\cr &+\left(P_{\infty,4}^{(2)}q^4+P_{\infty,3}^{(2)}q^3+P_{\infty,2}^{(2)}q^2+P_{\infty,1}^{(2)}q+P_{\infty,0}^{(2)}+2\hbar P_{\infty,2}^{(1)} q \right)\,,
&&
\end{align}
so that 
\beq H=H_0(p,q,\hbar)+ \hbar (t_{1,3}+t_{2,3}) q .\eeq
We may use the connection between coefficients $\left(P_{\infty,i}^{(j)}\right)_{1\leq i\leq 4,\,1\leq j\leq 2}$ and the spectral times $(t_{i,j})$'s given by \eqref{P2Spectimes},\eqref{P2Spectimes2},\eqref{P2Spectimes3} to obtain
\beq \mathcal{L}[t_{1,3}]=\mathcal{L}[t_{2,3}]=\mathcal{L}[t_{1,2}]=\mathcal{L}[t_{1,0}]=\mathcal{L}[t_{2,0}]=0\,,\,\mathcal{L}[t_{1,1}]=\hbar t_{2,3}\,,\, \mathcal{L}[t_{2,1}]=\hbar t_{1,3},\eeq
indicating that the operator $\mathcal{L}$ corresponds to the following differential operator
\beq \label{LP2} \mathcal{L} = \hbar t_{2,3}\partial_{t_{1,1}}+\hbar t_{1,3} \partial_{t_{2,1}}.\eeq
In particular, since $\alpha=t_{1,3}+2t_{2,3}$, we obtain that $\mathcal{L}[\alpha]=0$ .We may also rewrite the action of $\mathcal{L}$  in terms of the spectral times
\bea \mathcal{L}[P_{\infty,4}^{(2)}]&=&\mathcal{L}[P_{\infty,3}^{(2)}]=0 \cr
\mathcal{L}[P_{\infty,2}^{(2)}]&=&\hbar(t_{1,3}^2+t_{2,3}^2)\,\,,\,\, \mathcal{L}[P_{\infty,1}^{(2)}]=\hbar(t_{2,3}t_{2,2}+t_{1,3}t_{1,2})\cr
\mathcal{L}[P_{\infty,2}^{(1)}]&=&\mathcal{L}[P_{\infty,1}^{(1)}] =0\,,\, \mathcal{L}[P_{\infty,0}^{(1)}]=-\hbar(t_{1,3}+t_{2,3})\cr 
\mathcal{L}[P_{\infty,0}^{(2)}]-\mathcal{L}[H]&=&2\hbar t_{1,3}t_{2,3}q^2+\hbar (t_{1,2}t_{2,3}+t_{2,2}t_{1,3})q + (t_{1,3}+t_{2,3})  p+\hbar (t_{1,1}+t_{2,1}) (t_{1,3}+t_{2,3})\cr
\mathcal{L}[q]&=&-2\frac{p}{\hbar}-(t_{1,3}+t_{2,3})q^2-(t_{1,2}+t_{2,2})q-t_{1,1}-t_{2,1}\cr
\mathcal{L}[p]&=&(2(t_{1,3}+t_{2,3})q+t_{1,2}+t_{2,2})p+4\hbar t_{1,3}t_{2,3}q^3+3\hbar(t_{1,3}t_{2,2}+t_{2,3}t_{1,2})q^2\cr
&&+2\hbar(t_{1,3}t_{2,1}+t_{2,3}t_{1,1}+t_{1,2}t_{2,2})q\cr
&&+\hbar(t_{1,2}t_{2,1}+t_{1,1}t_{2,2}-(t_{1,3}-t_{2,3})t_{1,0})-\hbar^2 t_{2,3}\,.
\eea
Combining the last two equations, we may obtain the following equation for $q$,
\bea \label{EvolutionP2}\mathcal{L}^2[q]
&=&2(t_{1,3}-t_{2,3})^2q^3+3(t_{1,3}-t_{2,3})(t_{1,2}-t_{2,2})q^2\cr
&&+\left( (t_{1,2}-t_{2,2})^2+2(t_{1,3}-t_{2,3})(t_{1,1}-t_{2,1})\right)q\cr
&&+(t_{1,2}-t_{2,2})(t_{1,1}-t_{2,1})+(2t_{1,0}-\hbar)(t_{1,3}-t_{2,3})\,.
\eea

\begin{remark}[Normalization of the non-perturbative wave functions and coefficient $\alpha$]\label{RemarkIso} The value of the coefficient $\alpha=t_{1,3}+2t_{2,3}$ comes from the choice of normalization of the wave functions. Indeed, the Lax pair \eqref{LaxPairP2} and \eqref{LaxPairP22} is consistent for any value of $\alpha$. However, this value is imposed by the choice of normalization of the wave functions at infinity,
\bea \label{P2FormalExpansion}\ln \Psi_{1,1}&=&-\frac{t_{1,3}}{3\hbar}\lambda^3-\frac{t_{1,2}}{2\hbar}\lambda^2- \frac{t_{1,1}}{\hbar}\lambda-\frac{t_{1,0}}{\hbar}\ln \lambda+S_{1,0}+\sum_{k=2}^{\infty} \frac{S_{1,k}}{(k-1)\lambda^{k-1}},\cr
 \ln \Psi_{1,2}&=&-\frac{t_{2,3}}{3\hbar}\lambda^3-\frac{t_{2,2}}{2\hbar}\lambda^2- \frac{t_{2,1}}{\hbar}\lambda-\frac{t_{2,0}}{\hbar}\ln \lambda-\ln \lambda+S_{2,0}+\sum_{k=2}^{\infty} \frac{S_{2,k}}{(k-1)\lambda^{k-1}}.
\eea
Indeed, this formal expansion and the fact that $t_{2,0}=-t_{1,0}$ and $t_{1,3}-t_{2,3}\neq0$ imply that $L_{2,1}(\lambda,\hbar)=\frac{\hbar^2}{W(\lambda,\hbar)}\left( \frac{\partial^2 \Psi_{1,1} }{\partial \lambda^2} \frac{\partial \Psi_{1,2} }{\partial \lambda}-\frac{\partial \Psi_{1,1} }{\partial \lambda} \frac{\partial^2 \Psi_{1,2} }{\partial \lambda^2}\right)$ has a formal expansion at infinity,
\bea L_{2,1}(\lambda,\hbar)&=&-t_{1,3}t_{2,3}\lambda^4-(t_{1,3}t_{2,2}+t_{2,3}t_{1,2})\lambda^3- (t_{2,1}t_{1,3}+t_{1,1}t_{2,3}+t_{2,2}t_{1,2})\lambda^2\cr
&& -(t_{1,3}t_{2,0}+t_{1,0}t_{2,3}+t_{1,2}t_{2,1}+t_{1,1}t_{2,2}+\hbar t_{1,3})\lambda+ O(1),
\eea
where the Wronskian $W(\lambda,\hbar)$ is given by \eqref{WronskianP2} with $\kappa=(t_{1,3}-t_{2,3})e^{S_{1,0}+S_{2,0}}$.
Thus $L_{2,1}(\lambda,\hbar)$ only matches the formula given by \eqref{LaxPairP2},
\beq L_{2,1}(\lambda,\hbar)=-P_{\infty,4}^{(2)}\lambda^4-P_{\infty,3}^{(2)}\lambda^3-P_{\infty,2}^{(2)}\lambda^2- (P_{\infty,1}^{(2)}-2\hbar P_{\infty,2}^{(1)}-\hbar \alpha)\lambda +O(1), \eeq
for $t_{1,3}=-2P_{\infty,2}^{(1)}- \alpha $, i.e.~$\alpha=t_{1,3}+2t_{2,3}$.

Note also that $S_{1,0}=\ln Z_{\mathrm{NP}}(\hbar,\rho)$ because of \eqref{eq-expansion-psi-infty}.
\end{remark} 

\begin{remark} In agreement with the general theory developed in this article, the formal expansion of the wave functions \eqref{P2FormalExpansion} as well as the condition on the Wronskian \eqref{WronskianP2} completely determine the Lax pair matrices using $L=\hbar (\partial_\lambda \Psi )\Psi^{-1}$ and $A=\mathcal{L}[\Psi] \,\Psi^{-1}$. 
\end{remark}

Using Remark \ref{RemarkIso}, we may look at the terms $O\left(\frac{1}{\lambda^2}\right)$ in entry $A_{1,2}(\lambda,\hbar)$ as $\lambda\to\infty$. The coefficient is given by $-q$ in the Lax pair formulation while it corresponds to $\frac{\mathcal{L}[S_{2,0}]-\mathcal{L}[S_{1,0}]+t_{1,2}-t_{2,2}}{t_{1,3}-t_{2,3}}$ using \eqref{P2FormalExpansion}. Similarly, we may look at coefficient $O(1)$ when $\lambda\to \infty$ of the entry $A_{1,1}(\lambda,\hbar)$. We get $\frac{t_{1,3}\mathcal{L}[S_{2,0}] -t_{2,3}\mathcal{L}[S_{1,0}]+t_{2,3}t_{1,2}-t_{1,3}t_{2,2}}{t_{1,3}-t_{2,3}}=-\frac{H}{\hbar}+t_{2,2}$. The last two equations are equivalent to
\beq \mathcal{L}[S_{1,0}]=t_{1,3}\,q-\frac{H}{\hbar}+t_{1,2}+t_{2,2} \,\,,\,\,  \mathcal{L}[S_{2,0}] = t_{2,3}\,q-\frac{H}{\hbar}+ 2t_{2,2}.\eeq
In particular, since $S_{1,0}=\ln Z_{\mathrm{NP}}(\hbar,\rho)$, we end up with
\bea\label{ZNPP2} \mathcal{L}[\ln Z_{\mathrm{NP}}]&=& t_{1,3}\,q-\frac{H}{\hbar}+t_{1,2}+t_{2,2}=-\frac{H_0}{\hbar}-t_{2,3}q+t_{1,2}+t_{2,2}\cr
&=&-\frac{1}{\hbar}\left[\frac{p^2}{\hbar^2}- P_1(q)\frac{p}{\hbar}+  P_2(q)-\hbar P_1'(q)\right] -2P_{\infty,2}^{(1)}q +t_{1,2}+t_{2,2}\,.
\eea
This relation connects the non-perturbative partition function with the Hamiltonian of the system. This relation shall be useful to relate the non-perturbative partition function with the isomonodromic tau-function associated to the Lax system. We let this perspective for future works where we hope to prove that such relations hold for all cases presented in this article.

\subsubsection{Gauge transformation to remove the apparent singularity}
We may also perform a gauge transformation to remove the apparent singularity of the Lax system \eqref{LaxPairP2} and \eqref{LaxPairP22}. The general theory indicates that we should define
\beq \check{\Psi}(\lambda,\hbar)= J(\lambda,\hbar) \Psi(\lambda,\hbar) \text{ with } J(\lambda,\hbar)=\begin{pmatrix}1&0\\ -\frac{p}{\hbar(\lambda-q)}& \frac{1}{\lambda-q} \end{pmatrix},\eeq
so that we obtain a new Lax system for $\check{\Psi}(\lambda,\hbar)$ given by
\beq \hbar \partial_\lambda \check{\Psi}(\lambda,\hbar)=\check{L}(\lambda,\hbar)\check{\Psi}(\lambda,\hbar) \,\,,\,\, \mathcal{L}[\check{\Psi}(\lambda,\hbar)]=\check{A}(\lambda,\hbar)\check{\Psi}(\lambda,\hbar),\eeq
with
\bea \check{L}(\lambda,\hbar)&=&\begin{pmatrix}\frac{p}{\hbar}& \lambda-q \\ -( (\lambda+q)(t_{1,3}+t_{2,3})+t_{2,2}+t_{1,2})\frac{p}{\hbar}+ Q_3(\lambda,\hbar)&
-\frac{p}{\hbar}+P_1(\lambda)
\end{pmatrix}\cr 
\check{A}(\lambda,\hbar)&=&\begin{pmatrix} -(t_{1,3}+t_{2,3})\lambda-\frac{H}{\hbar}+t_{2,2}& -1 \\
(t_{1,3}+t_{2,3})\frac{p}{\hbar}+Q_2(\lambda,\hbar)& (t_{1,3}+t_{2,3})q+t_{1,2}+2t_{2,2}-\frac{H}{\hbar}
\end{pmatrix},\quad
\eea
where
\bea Q_3(\lambda,\hbar)&=&-P_{\infty,4}^{(2)}\lambda^3- (P_{\infty,4}^{(2)}q+P_{\infty,3}^{(2)})\lambda^2-(P_{\infty,4}^{(2)}q^2+P_{\infty,3}^{(2)}q+P_{\infty,2}^{(2)})\lambda \cr
&&+(P_{\infty,4}^{(2)}q^3+P_{\infty,3}^{(2)}q^2+P_{\infty,2}^{(2)}q+P_{\infty,1}^{(2)}+\hbar t_{1,3})\,,\cr
Q_2(\lambda,\hbar)&=&P_{\infty,4}^{(2)}\lambda^2+2P_{\infty,4}^{(2)}q \lambda+P_{\infty,3}^{(2)}\lambda+(3P_{\infty,4}^{(2)}q^2+2P_{\infty,3}^{(2)}q+P_{\infty,2}^{(2)})\,.\cr
&&
\eea
The deformed classical spectral curve is given by $\det(yI_2-\check{L}(\lambda,\hbar))=0$. It reads 
\beq \label{P2DeformedSpecCurve} P_{\text{defo}}(\lambda,y)=y^2- P_1(\lambda) y+ P_2(\lambda)-\hbar P_1'(\lambda)-\hbar (t_{1,3}+2t_{2,3})\lambda-H=0\,.\eeq
In particular, we observe that
\beq P_{\text{defo}}\left(q,\frac{p}{\hbar}\right)=0\,.\eeq

\subsubsection{Recovering the Painlev\'{e} $2$ equation}
From the previous section, we know that the operator $\mathcal{L}$ corresponds to $\mathcal{L} = \hbar(  t_{2,3}\partial_{t_{1,1}}+t_{1,3} \partial_{t_{2,1}})$. This suggests that there exists a natural change of variables from $(t_{1,1},t_{2,1})$ to $(\tau,\td{\tau})$ such that $\mathcal{L}$ corresponds to $\hbar \partial_\tau$ while $\td{\tau}$ should satisfy $\partial_\tau \td{\tau}=0$. We define
\bea \tau&=&\frac{1}{t_{1,3}-t_{2,3}}\left(t_{2,1}-t_{1,1}\right)\, \,,\, \,\, \, \td{\tau}=\frac{1}{t_{1,3}-t_{2,3}}\left( t_{1,3}t_{1,1}-t_{2,3}t_{2,1}\right)  \cr
\Leftrightarrow \,\, t_{1,1}&=&t_{2,3}\tau+\td{\tau}\,, \, \,\, \, t_{2,1}=t_{1,3}\,\tau+\td{\tau}\,.
\eea
This gives
\beq \partial_{t_{1,1}}=-\frac{1}{t_{1,3}-t_{2,3}} \partial_\tau +\frac{t_{1,3}}{t_{1,3}-t_{2,3}}\partial_{\td{\tau}}\, \text{ and }\, \partial_{t_{2,1}}=\frac{1}{t_{1,3}-t_{2,3}} \partial_\tau -\frac{t_{2,3}}{t_{1,3}-t_{2,3}} \partial_{\td{\tau}}\,.\eeq
In particular $\mathcal{L}=\hbar \partial_\tau$, $\hbar \partial_\tau \td{\tau}=0$, $\partial_{\td{\tau}}\tau=0$ and $\mathcal{L}[\td{\tau}]=0$. Observe also that $t_{1,1}-t_{2,1}=-(t_{1,3}-t_{2,3})\tau$.
Equation \eqref{EvolutionP2} reads
\bea\hbar^2\partial^2_{\tau^2}q&=&2(t_{1,3}-t_{2,3})^2q^3+3(t_{1,3}-t_{2,3})(t_{1,2}-t_{2,2})q^2\cr
&&+\left( (t_{1,2}-t_{2,2})^2+2(t_{1,3}-t_{2,3})(t_{1,1}-t_{2,1})\right)q\cr
&&+(t_{1,2}-t_{2,2})(t_{1,1}-t_{2,1})+(2t_{1,0}-\hbar)(t_{1,3}-t_{2,3})\cr
&=&2(t_{1,3}-t_{2,3})^2q^3+3(t_{1,3}-t_{2,3})(t_{1,2}-t_{2,2})q^2\cr 
&&+\left( (t_{1,2}-t_{2,2})^2-2(t_{1,3}-t_{2,3})^2\tau\right)q\cr
&&-(t_{1,2}-t_{2,2})(t_{1,3}-t_{2,3})\tau+(2t_{1,0}-\hbar)(t_{1,3}-t_{2,3}).
\eea
In particular it does not depend on $\td{\tau}$. The last differential equation is equivalent to the Painlev\'{e} $2$ equation after a proper rescaling $(q,\tau)\leftrightarrow (\td{q},t)$. Indeed, let us define
\begin{align}\label{FinalChangeVariableP2}&\tau = \left(\frac{-1}{2(t_{1,3}-t_{2,3})^2}\right)^{\frac{1}{3}} t-\frac{(t_{1,2}-t_{2,2})^2}{4(t_{1,3}-t_{2,3})^2} \;\text{ and }\, q=\left(\frac{-2}{t_{1,3}-t_{2,3}}\right)^{\frac{1}{3}}\td{q}-\frac{t_{1,2}-t_{2,2}}{2(t_{1,3}-t_{2,3})}\Leftrightarrow \cr
&t = \left(-2(t_{1,3}-t_{2,3})^2\right)^{\frac{1}{3}} \left(\tau+\frac{(t_{1,2}-t_{2,2})^2}{4(t_{1,3}-t_{2,3})^2}\right) \text{ and }\, \td{q}=\left(\frac{-(t_{1,3}-t_{2,3})}{2}\right)^{\frac{1}{3}}\left(q+\frac{t_{1,2}-t_{2,2}}{2(t_{1,3}-t_{2,3})}\right).  \cr
\end{align}
Then we obtain that $\td{q}$ satisfies the following Painlev\'{e} $2$ equation
\beq\label{P2} \hbar^2 \partial^2_{t^2}\td{q}=2\td{q}^3+t\td{q}-\left(t_{1,0}-\frac{\hbar}{2}\right).\eeq
Note that the final rescaling \eqref{FinalChangeVariableP2} greatly simplifies when $t_{1,2}=t_{2,2}$ and $t_{1,3}-t_{2,3}=-1$.

\subsection{A $\text{Gl}_3$ example}\label{SectionGL3Example}

In this section, we consider a degree $3$ example where the classical spectral curve has genus $1$ and a single singularity at infinity. In particular, we show how to derive the quantum curve in this context using the general results obtained above. For this purpose, let us consider a classical spectral curve defined by an algebraic equation of the form
\beqq
y^3 - (P_{\infty,1}^{(1)} \lambda + P_{\infty,0}^{(1)} ) y^2 + ( P_{\infty,2}^{(2)} \lambda^2 + P_{\infty,1}^{(2)} \lambda + P_{\infty,0}^{(2)}) y - P_{\infty,3}^{(3)} \lambda^3 - P_{\infty,2}^{(3)} \lambda^2 - P_{\infty,1}^{(3)} \lambda - P_{\infty,0}^{(3)} = 0,
\eeqq
where the coefficients $\left(P_{\infty,i}^{(j)}\right)_{0\leq i\leq 2, 1\leq j\leq 3}$ are generic in such a way that the curve has genus $1$ and that there are three distinct points $(\infty^{(1)},\infty^{(2)},\infty^{(3)})$ in the fiber $x^{-1}(\infty)$ above infinity. This corresponds to the case $N=0$ and $r_{{\infty}^{(j)}} = 3$ for $j\in \llbracket 1,3\rrbracket $. We define
\bea P_1(\lambda)&=&P_{\infty,1}^{(1)} \lambda + P_{\infty,0}^{(1)}\,,\cr
P_2(\lambda)&=&P_{\infty,2}^{(2)} \lambda^2 + P_{\infty,1}^{(2)} \lambda + P_{\infty,0}^{(2)}\,,\cr
P_3(\lambda)&=&P_{\infty,3}^{(3)} \lambda^3 + P_{\infty,2}^{(3)} \lambda^2 + P_{\infty,1}^{(3)} \lambda + P_{\infty,0}^{(3)}\,.
\eea
The coefficients $\left(P_{\infty,i}^{(j)}\right)_{0\leq i\leq 3, 1\leq j\leq 3}$ are connected to the spectral times $(t_{i,j})_{1\leq i \leq 3 , 0\leq j\leq 2}$ by the relations
\beq 
P_{\infty,1}^{(1)} =  -\sum_{i=1}^3  t_{i,2} \, , \quad P_{\infty,0}^{(1)} =  - \sum_{i=1}^3  t_{i,1} \, , \quad 0 = \sum_{i=1}^3 t_{i,0}\,,
\eeq
as well as 
\bea
P_{\infty,2}^{(2)} &=&  \sum_{1\leq i<j\leq 3} t_{i,2} t_{j,2}=t_{1,2}t_{2,2}+t_{1,2}t_{3,2}+t_{2,2}t_{3,2}\,, \cr
P_{\infty,1}^{(2)} &=&  \sum_{i=1}^3\sum_{j\neq i} t_{i,1} t_{j,2} =t_{1,1}t_{2,2}+t_{1,1}t_{3,2}+t_{2,1}t_{1,2}+t_{2,1}t_{3,2}+t_{3,1}t_{1,2}+t_{3,1}t_{2,2}\,,\cr
P_{\infty,0}^{(2)} &=&   \sum_{i=1}^3 \sum_{j\neq i} t_{i,0} t_{j,2} + \sum_{1\leq i<j\leq 3} t_{i,1} t_{j,1} \cr
&=&t_{1,0}t_{2,2}+t_{1,0}t_{3,2}+t_{2,0}t_{1,2}+t_{2,0}t_{3,2}+t_{3,0}t_{1,2}+t_{3,0}t_{2,2} +t_{1,1}t_{2,1}+t_{1,1}t_{3,1}+t_{2,1}t_{3,1}\cr
&&
\eea
and
\bea
P_{\infty,3}^{(3)} &=& - t_{1,2} t_{2,2} t_{3,2}\,,\cr
P_{\infty,2}^{(3)} &=& - t_{1,1} t_{2,2} t_{3,2} - t_{1,2} t_{2,1} t_{3,2} - t_{1,2} t_{2,2} t_{3,1}\,,\cr
P_{\infty,1}^{(3)} &= & - \sum_{k_1+k_2+k_3 = 4} t_{1,k_1} t_{2,k_2} t_{3,k_3}\cr
&=&-t_{1,0}t_{2,2}t_{3,2}-t_{1,1}t_{2,1}t_{3,2}-t_{1,1}t_{2,2}t_{3,1} - t_{1,2}t_{2,2}t_{3,0}-t_{1,2}t_{2,1}t_{3,1}-t_{1,2}t_{2,0}t_{3,2}\,,\cr
P_{\infty,0}^{(3)} &=&  - \sum_{k_1+k_2+k_3 = 3} t_{1,k_1} t_{2,k_2} t_{3,k_3}\cr
&=&-t_{1,0}t_{2,1}t_{3,2}-t_{1,0}t_{2,2}t_{3,1}-t_{1,1}t_{2,1}t_{3,1}-t_{1,1}t_{2,2}t_{3,0}\cr
&& \quad -t_{1,1}t_{2,0}t_{3,2}-t_{1,2}t_{2,1}t_{3,0}-t_{1,2}t_{2,0}t_{3,1}\,.\cr
&&
\eea
The computations being rather long, we provide the detailed example in Appendix \ref{app-example-gl3} and we shall only present the important result there. 

\medskip

The quantum curve corresponds to the Lax pair $L(\lambda,\hbar) =$
\beq 
\begin{pmatrix} 0&1&0\\0&0&1\\
P_3(\lambda)-\hbar P_2'(\lambda)+  \hbar (P_{\infty,2}^{(2)}+t_{2,2}t_{3,2}) \lambda - H + \frac{p_1}{\lambda-q} &
-P_2(\lambda)-\hbar t_{1,2}  + \frac{p_2}{\lambda- q}&
P_1(\lambda) +\frac{\hbar}{\lambda-q} 
\end{pmatrix}\,,
\eeq
with the associated auxiliary matrix $\mathcal{L}[\Psi(\lambda,\hbar)]=A(\lambda,\hbar) \Psi(\lambda,\hbar)$ of the form
\beq
A(\lambda,\hbar) = 
\begin{pmatrix}  P_{\infty,2}^{(2)} \lambda - \frac{H}{\hbar} + \frac{p_1}{\hbar(\lambda-q)} &  \frac{p_2}{\hbar(\lambda- q)}+t_{1,2} &\frac{1}{\lambda-q}\\A_{2,1}&A_{2,2}&A_{2,3}\\
A_{3,1}&A_{3,2}&A_{3,3}
\end{pmatrix}\,.
\eeq

\medskip

Compatibility of this system provides the following equations
\bea \mathcal{L}[t_{1,0}]&=&\mathcal{L}[t_{2,0}]=\mathcal{L}[t_{3,0}]=\mathcal{L}[t_{3,2}]=\mathcal{L}[t_{2,2}]=\mathcal{L}[t_{1,2}]=0\cr
\mathcal{L}[t_{1,1}]&=&-\hbar(t_{1,2}t_{2,2}+t_{1,2}t_{3,2}+t_{2,2}t_{3,2})\cr
\mathcal{L}[t_{2,1}]&=&-\hbar(t_{1,2}t_{3,2}+t_{2,2}t_{3,2}+t_{2,2}^2)\cr
\mathcal{L}[t_{3,1}]&=&-\hbar(t_{1,2}t_{2,2}+t_{2,2}t_{3,2}+t_{3,2}^2)\cr
p_1&=&\frac{p_2^2}{\hbar}+P_1(q_2)p_2+\hbar P_2(q)+\hbar^2 t_{1,2}\cr
H&=&-H_0(q,p_2,\hbar)+(t_{1,2}+P_{\infty,1}^{(1)})p_2-\hbar P_{\infty,2}^{(2)}q-\hbar P_{\infty,1}^{(2)}-\hbar P_{\infty,0}^{(1)}t_{1,2}\cr 
\mathcal{L}[q]&=& -3\frac{p_2^2}{\hbar^2}-4P_1(q)\frac{p_2}{\hbar}-P_2(q)-P_1(q)^2-\hbar(P_{\infty,1}^{(1)}+2t_{1,2}) \cr
\mathcal{L}[p_2]&=&2P_{\infty,1}^{(1)}\frac{p_2^2}{\hbar}+(P_2'(q)+2P'_1(q)P_1(q))p_2-\hbar P_3'(q)+\hbar P_2(q)P_1'(q)+\hbar P_2'(q)P_1(q)\cr
&&+\hbar^2(P_{\infty,1}^{(1)}t_{1,2}-t_{2,2}t_{3,2})\,.
\eea
In particular, the system is Hamiltonian with $H_0(q,p_2,\hbar)$ given by
\bea H_0(q,p_2,\hbar)&=&\frac{p_2^3}{\hbar^3}+2P_1(q)\frac{p_2^2}{\hbar^2}+(P_2(q)+P_1(q)^2+\hbar(P_{\infty,1}^{(1)}+2t_{1,2}))\frac{p_2}{\hbar}-P_3(q)+P_1(q)P_2(q)\cr
&&+\hbar(P_{\infty,1}^{(1)}t_{1,2}-t_{2,2}t_{3,2})q\,,
\eea
satisfying
\beq \hbar \partial_{p_2} H_0(q,p_2,\hbar)=-\mathcal{L}[q] \text{ and } \hbar\partial_q H_0(q,p_2,\hbar)= \mathcal{L}[p_2]\,. \eeq
Combining these relations, we show that $q$ satisfies a PDE of the form
\beq \alpha_0(q;\hbar)\mathcal{L}^3[q]+ \alpha_1(q;\hbar)\mathcal{L}^2[q]\mathcal{L}[q]+\alpha_2(q;\hbar) (\mathcal{L}[q])^2+\alpha_3(q;\hbar) \mathcal{L}^2[q]+ \alpha_4(q;\hbar) \mathcal{L}[q]+\alpha_5(q;\hbar)=0\,,\eeq
with polynomial coefficients $(\alpha_i(q;\hbar))_{1\leq i\leq 5}$ in $q$ and $\hbar$ whose expressions are given in \eqref{CoeffGl33} (quantities involved in \eqref{CoeffGl33} are defined in \eqref{CoeffGl31}).

Finally, we get that
\begin{align} \mathcal{L}&=-\hbar(t_{1,2}t_{2,2}+t_{1,2}t_{3,2}+t_{2,2}t_{3,2})\partial_{t_{1,1}}-\hbar(t_{1,2}t_{3,2}+t_{2,2}t_{3,2}+t_{2,2}^2)\partial_{t_{2,1}}-\cr& \hbar(t_{1,2}t_{2,2}+t_{2,2}t_{3,2}+t_{3,2}^2)\partial_{t_{3,1}},\end{align}
and prove that the change of variables
\beq \begin{pmatrix}\tau_1\\ \tau_2\\ \tau_3\end{pmatrix} =B \begin{pmatrix} t_{1,1}\\ t_{2,1}\\ t_{3,1}\end{pmatrix} \,\Leftrightarrow \,   \begin{pmatrix} t_{1,1}\\ t_{2,1}\\ t_{3,1}\end{pmatrix} =B^{-1}\begin{pmatrix}\tau_1\\ \tau_2\\ \tau_3\end{pmatrix}= \begin{pmatrix} -(t_{1,2}t_{2,2}+t_{1,2}t_{3,2}+t_{2,2}t_{3,2})& 1&t_{1,2}\\
-(t_{1,2}t_{3,2}+t_{2,2}t_{3,2}+t_{2,2}^2)& 1&t_{2,2}\\
-(t_{1,2}t_{2,2}+t_{2,2}t_{3,2}+t_{3,2}^2)& 1&t_{3,2}
\end{pmatrix}\begin{pmatrix}\tau_1\\ \tau_2\\ \tau_3\end{pmatrix}
\eeq 
is such that $\mathcal{L}=\hbar \partial_{\tau_1}$.

\section{Conclusion and outlooks}\label{S9}

In this paper, we have been able to quantize any admissible spectral curve by building a formal solution to the associated quantum curve using the topological recursion. A first natural question which should be addressed is the generalization of this procedure to other classical spectral curves. First of all, we expect that it is possible to drop the admissibility conditions we introduced without changing the structure of the proof. It should also be possible to consider any classical spectral curve given by a Higgs field on a higher genus base curve and allow for more complicated singularities such as logarithmic ones. This step would in particular give access to curves over $\mathbb{C}^*$ appearing in the study of Gromov--Witten invariants of Toric Calabi--Yau three-folds by mirror symmetry. Considering the structure of the proofs presented in this article we expect that such a large generalization is possible.

\medskip

We have observed in the first example of Section~\ref{sec-examples} that the non-perturbative partition function plays the role of isomonodromic tau function of the system. We expect this observation to be a general phenomenon. Namely, we conjecture that the non-perturbative partition function is an isomonodromic tau function for the Lax system built by topological recursion. In order to prove it, we plan to express the coefficients of the deformed spectral curve in terms of variations of the logarithm of the partition function with respect to the spectral times. At the same time, this would allow to describe the $\hbar$ evolution of the connection built in this paper as a flow in the associated moduli space of connections, making the link with Hitchin systems.

\medskip

In this paper we have considered formal $\hbar$ trans-series solutions to the quantum curve.  However, these trans-series are in no way convergent when $\hbar$ is a small parameter. One important question thus remains: how is it related to asymptotics of actual convergent solutions of the quantum curve equation when $\hbar \to 0$? In order to answer this question, we would like to be able to trade off the formal Fourier transform considered in this paper against a proper Fourier transform. However, it is not clear in which context this can be safely performed since the two orderings of the infinite sums  involved lead to different ordering of trans-monomials considered for defining the formal trans-series. This issue already appeared in the context of trans-asymptotics developed by Costin \cite{Costin-P1,CostinCostin01,CostinCostinHuang15}, who proved that, in a particular setting, formal trans-series of the form considered in this paper give good asymptotics to proper solutions of a differential equation. We hope to be able to adapt this procedure to the present set-up and make the link with the result on resurgence analysis encountered in a similar context \cite{ASV12,SchiappaVaz}. Very recently, N. Nikolaev \cite{nikolaev_2022,nikolaev_2022_2} obtained tremendous results promoting formal WKB expansion in $\hbar$ to some analytic objects defined in some Stokes sectors in the $\mathfrak{gl}_2$ case. It would be very interesting to see how these results help in the understanding of the formal setup developed in this article. Upgrading formal quantities to analytic objects would in particular allow to study the asymptotics of bi-orthogonal polynomials arising in hermitian random matrix models and make rigorous some saddle-point approximations in that context. Since our partition functions and wave functions have $2g$ parameters $(\boldsymbol{\epsilon},\boldsymbol{\rho})$, we expect to get a general isomonodromic tau function in this way and should be able to study the Riemann--Hilbert problems they solve. We wish to map these free parameters to coordinates in the associated moduli space and study their map to Stokes data.

\medskip

Finally, our results can be interpreted from a geometric quantization perspective. Indeed, the initial data considered is an algebraic curve on the one hand and a choice of Torelli marking on the other hand. This second ingredient is similar to a choice of real polarization for the quantization of our algebraic curve. It would be very interesting to translate our procedure into the language of geometric quantization and study the effect of a choice of polarization. In particular, it would be important to derive Hitchin's equations governing a change in the choice of polarization.

\section*{Acknowledgements}

This paper is partly a result of the ERC-SyG project, Recursive and Exact New Quantum Theory (ReNewQuantum) which received funding from the European Research Council (ERC) under the European Union's Horizon 2020 research and innovation program under grant agreement No 810573. The second author was supported by the public grant ``Jacques Hadamard'' as part of the Investissement d'avenir project, reference ANR-11-LABX-0056-LMH, LabEx LMH and currently receives funding from  the  European  Research Council  (ERC)  under  the  European  Union's Horizon  2020  research and  innovation  program (grant  agreement  No.~ERC-2016-STG 716083  ``CombiTop''). The work of N.~O.~is supported in part by the NCCR SwissMAP of the Swiss National Science Foundation.

\appendix

\renewcommand{\theequation}{\thesection-\arabic{equation}}

\section{Proof of the loop equations and ramification points at poles}
\label{appendix-proof-loopeq}

In this section, we prove Theorem~\ref{loopeq} and Lemma~\ref{lemRamptsPoles}.

\subsection{Case $(h,n) = (0,0)$}
This follows directly from the expression \eqref{def-Q01}.

\subsection{Case $(h,n) = (0,1)$}

From the definition, an elementary computation gives
\beq
\frac{Q_{0,2}^{(l)}(\lambda,z_1)}{(d\lambda)^l} = (-1)^l \sum_{z \in x^{-1}(\lambda)} \frac{ \om_{0,2}(z,z_1)}{d \lambda} \left[ \sum_{k=0}^{l-1}P_{l-1-k}(\lambda) y(z)^k\right].
\eeq
It is a meromorphic function of $\lambda$ on the base, i.e.~a rational fraction of $\lambda$.
Since there is the denominator $d\lambda$ it seems that it could have a pole at $a$, a  zero of $d\lambda$, i.e.~a ramification point. 
Let $\zeta_a=\sqrt{\lambda-x(a)}$ be a local coordinate near $a$, then $d\lambda$ vanishes to the first order (from our assumption of simple ramification points), and all the other terms are finite at the ramification points.
Therefore 
$$\frac{Q_{0,2}^{(l)}(\lambda,z_1)}{(d\lambda)^l} = O\bigg(\frac{1}{\sqrt{\lambda-x(a)}}\bigg).
$$ 
Since a rational fraction can have poles only with integer powers, this implies that 
$$\frac{Q_{0,2}^{(l)}(\lambda,z_1)}{(d\lambda)^l} = O(1),
$$ 
i.e.~there is no pole at $a$.

\subsection{Stable cases}
In this section, we deal with the stable cases corresponding to $(h,n) \notin \{(0,0),(0,1)\}$.

\medskip 

{\bf{Step 1: Rewriting the topological recursion}}

The proof is more easily obtained by using the general formalism of \emph{global topological recursion} developed in \cite{BouchardEynard} where authors prove that recursion \eqref{def-tr} defining the forms $\om_{h,n}$ can be traded for a more global formulation which reads
\beq\label{eq-TRhigherorder}
\om_{h,n+1}(z_0,\mathbf{z}) = \sum_{a \in \mathcal{R}} \Res_{z \to a} \left[ \sum_{k=1}^{d-1} \sum_{\mathbf{\beta}  \underset{k}{\subseteq} x^{-1}(x(z)) \setminus \{z\}} (-1)^{k+1} \frac{\int_{\alpha}^z \om_{0,2}(z_0, \cdot)}{\underset{i=1}{\overset{k}{\prod}} \left[\om_{0,1}(z)- \om_{0,1}(\beta_i)\right]} {U}_{h,n+1}^{(k+1)}(z,\mathbf{\beta};\mathbf{z})
\right],
\eeq
with
\beq
{U}_{h,n+1}^{(k)}(\mathbf{\beta};\mathbf{z})\coloneqq  \sum_{\mu \in \mathcal{S}(\beta)} \,  \underset{ \underset{i=1}{\overset{l(\mu)}{\bigsqcup}} J_i = \mathbf{z}}{\sum}\,\, \,  \overset{\hbox{no } (0,1)}{\sum_{{ \underset{i=1}{\overset{l(\mu)}{\sum}}} g_i= h+l(\mu)-k}} 
\bigg[
\prod_{i=1}^{l(\mu)} \om_{g_i,|\mu_i|+|J_i|}(\mu_i,J_i) \bigg], 
\eeq
where the symbol $ \overset{\hbox{no } (0,1)}{\sum}$ means that one considers only terms with $(g_i,|\mu_i|+|J_i|) \notin \{(0,1)\}$ and $\alpha \in \Sigma$ is an arbitrary reference point.

In \cite{BouchardEynard_QC}, Cauchy formula on $\Sigma$ allowed to prove that this version of the topological recursion can be equivalently written as  
\beq \label{eq-new-TR-2}
0= \sum_{a \in \mathcal{R}} \Res_{z \to a} \left[\int_{\alpha}^z \om_{0,2}(z_0,\cdot)\right] \frac{dx(z)}{\frac{\partial P(x(z),y(z))}{\partial y}} Q_{h,n+1}(x(z),y(z);\mathbf{z}),
\eeq
where we abusively use the notation
\beq
\frac{\partial P(x(z),y(z'))}{\partial y} \coloneqq \left. \frac{\partial P(\lambda,y)}{\partial y}\right|_{\substack{ \lambda = x(z) \cr y = y(z') } }
 \eeq
and $Q_{h,n+1}(x(z),y(z);\mathbf{z})$ is given by Definition \ref{def-Q-x-y}.

\medskip

{\bf{Step 2: Proving that $ \frac{dx(z)}{\frac{\partial P}{\partial y}(x(z),y(z))} Q_{h,n+1}(x(z),y(z);\mathbf{z}) $ is holomorphic at the ramification points}}

To prove this second step, let us assume that 
 $ \frac{dx(z)}{\frac{\partial P}{\partial y}(x(z),y(z))} Q_{h,n+1}(x(z),y(z);\mathbf{z}) $ has a pole at $z=a \in \mathcal{R}$. Hence, in any local coordinate, it reads
\beq
  \frac{ dx(z)}{\frac{\partial P}{\partial y}(x(z),y(z))} Q_{h,n+1}(x(z),y(z);\mathbf{z})  = \frac{S_{h,n+1}(\mathbf{z}) }{(z-a)^{m+1}} \left(1+ O(z-a)\right),
\eeq
for some $m\geq 1$ and $S_{h,n+1}(\mathbf{z}) \neq 0$. On the other hand, the expansion of $\om_{0,2}$ in the same local coordinate takes the form
\beq
\int_{\alpha}^z\om_{0,2}(z_0,\cdot) =\sum_{k=0}^\infty (z-a)^k f_{a,k}(z_0),
\eeq
where $f_{a,k}(z_0)$ is a 1-form behaving as
\beq
f_{a,k}(z_0) = \frac{dz_0}{(z_0-a)^{k+1}} \left(1 +O(z_0-a)\right),
\eeq
as $z_0 \to a$. The evaluation above shows that the leading order of \eqref{eq-new-TR-2} in $(z_0-a)$ implies that $S_{h,n+1}(\mathbf{z}) = 0$, leading to a contradiction.

{\bf{Step 3: Proving that $\frac{Q_{h,n+1}^{(k)}(\lambda;\mathbf{z})}{dx(z)^k}$ is holomorphic at the ramification points.}}

Let us recall that
\beq
Q_{h,n+1}(\lambda,y;\mathbf{z})= \sum_{k=1}^d (-1)^k y^{r-k} \frac{Q_{h,n+1}^{(k)}(\lambda;\mathbf{z})}{d\lambda^k},
\eeq
so that
\beq\label{eq-rel-Q}
Q_{h,n+1}(x(z),y(z);\mathbf{z}) = \sum_{k=1}^d (-1)^k y(z)^{r-k} \frac{Q_{h,n+1}^{(k)}(x(z);\mathbf{z})}{dx(z)^k}.
\eeq
We shall now use an interpolation formula to extract the coefficients $\frac{Q_{h,n+1}^{(k)}(x(z);\mathbf{z})}{dx(z)^k}$ out of $Q_{h,n+1}(x(z),y(z);\mathbf{z})$.

For this purpose, let us write 
\beq
\frac{P(x(z),y)}{y-y(z)} = 
 \sum_{l=0}^{d-1} (-1)^l y^{d-1-l} U^{(l)}(z),
\eeq
where 
\beq
U^{(l)}(z)\coloneqq \sum_{\mathbf{\beta} \underset{l}{\subseteq} x^{-1}(x(z)) \setminus \{z\}} \prod_{i=1}^l y(\beta_i).
\eeq
The evaluation at $y = y(z)$  gives
\beq
\frac{\partial P}{\partial y}(x(z),y(z)) = 
 \sum_{l=0}^{d-1} (-1)^l  y(z)^{d-1-l} U^{(l)}(z),
\eeq
while the evaluation at $y = y(z')$ with $x(z) = x(z')$ but $z \neq z'$ implies that
\beq
\forall\, z' \in x^{-1}(x(z)) \setminus \{z\} \, , \;\;\; 0 = 
 \sum_{l=0}^{d-1} (-1)^l  y(z')^{d-1-l} U^{(l)}(z) .
\eeq
One can use these two relations to compute
\begin{align}
 Q_{h,n+1}(x(z),y(z);\mathbf{z}) &= 
\sum_{z' \in x^{-1}(x(z))} \frac{Q_{h,n+1}(x(z),y(z');\mathbf{z})}{ \frac{\partial P}{\partial y}(x(z),y(z')) } 
 \frac{\partial P}{\partial y}(x(z),y(z))\, \delta_{z,z'} \cr
 &= 
 \sum_{z' \in x^{-1}(x(z))} \frac{Q_{h,n+1}(x(z),y(z');\mathbf{z})}{ \frac{\partial P}{\partial y}(x(z),y(z')) }  
 \sum_{l=0}^{d-1} (-1)^l  y(z)^{d-1-l} U^{(l)}(z') . \cr
 \end{align}
 Exchanging the summations and shifting the index $l$ by 1, one gets
\beq
 Q_{h,n+1}(x(z),y(z);\mathbf{z})  =  -  \sum_{l=1}^{d} (-1)^l  y(z)^{d-l}  \sum_{z' \in x^{-1}(x(z))} \frac{Q_{h,n+1}(x(z),y(z');\mathbf{z})}{ \frac{\partial P}{\partial y}(x(z),y(z')) } 
  U^{(l-1)}(z') .
\eeq
Comparing with \eqref{eq-rel-Q}, one obtains
\beq
 \frac{Q_{h,n+1}^{(k)}(x(z);\mathbf{z})}{dx(z)^k}
  = -  \sum_{z' \in x^{-1}(x(z))} \frac{Q_{h,n+1}(x(z),y(z');\mathbf{z})}{ \frac{\partial P}{\partial y}(x(z),y(z')) }  
  U^{(l-1)}(z') .
  \eeq

Since we have proved that $ \frac{Q_{h,n+1}(x(z),y(z');\mathbf{z})}{ \frac{\partial P}{\partial y}(x(z),y(z')) } $ is holomorphic at the ramification points and $U^{(l-1)}(z')$ obviously is, one can conclude that $ \frac{Q_{h,n+1}^{(k)}(x(z);\mathbf{z})}{dx(z)^k}$ does not have any pole at $z = a \in \mathcal{R}$.

\subsection{Ramification points at poles}\label{Applemmaramificationpointspoles}

The points of $x^{-1}(\mathcal P)$ play an important role in this article.
A large part of our derivations rely on the fact that $\omega_{h,n}(z_1,\dots,z_n)$ for $(h,n)\neq(0,1)$ have no poles at any $z_i \in x^{-1}(\mathcal P)$.
Let us consider some subclass of spectral curves for which this holds.

Let us write the ramification profile over $P\in \mathcal P$ as
\beq
x^{-1}(P) = \{ P^{(1)},\dots,P^{(\ell_P)}\},
\eeq
\beq
d_{P^{(\alpha)}} = \deg_{P^{(\alpha)}} x.
\eeq
Let us denote $\zeta_\alpha$ the canonical local coordinate near $P^{(\alpha)}$, and let us denote
\beq
\rho_{P^{(\alpha)}} = e^{\frac{2\pi\ii }{d_{P^{(\alpha)}}}}
\eeq
the root of unity.

The meromorphic function $y$ has a Taylor--Laurent expansion given by the spectral times
\beq
ydx = \sum_{k=0}^{r_{P^{(\alpha)}}-1} t_{ P^{(\alpha)},k}\, \zeta_\alpha^{-k-1} d\zeta_\alpha + \text{analytic at } P^{(\alpha)}.
\eeq

\bl
If for all $p\in x^{-1}(\mathcal P)$ we have $r_p\geq 3$ and $t_{p,r_p-2}\neq 0$, then
$\omega_{h,n}$ for $(h,n)\neq(0,1)$ are analytic at $x^{-1}(\mathcal P)$.

\el

\begin{proof}
The proof proceeds by recursion on $2h-2+n$.
It is clearly true for $2h-2+n=0$, i.e. for $\omega_{0,2}$.

Let us assume that it holds up to $2h-2+n$ and prove it for $\omega_{h,n+1}$ using eq \eqref{eq-TRhigherorder}. Let $p\in x^{-1}(\mathcal P)$.

By recursion hypothesis, the factor $U^{(k+1)}_{h,n+1}(z,\beta;z)$ has no pole at $z=p$ nor $\beta_i=p$, but it may have poles at $z=\beta_i$ or $\beta_i=\beta_j$ if they appear in $\omega_{0,2}$.
Remark that
\beq
\omega_{0,2}(\beta_i,\beta_j) = \frac{\rho_p^{\beta_j}}{(\rho_p^{\beta_i}-\rho_p^{\beta_j})^2} \zeta_p^{-2} d\zeta_p^2,
\eeq
i.e.~it has a double pole (notice that $\rho_p^{\beta_i}-\rho_p^{\beta_j}\neq 0$ because $\beta_i$ and $\beta_j$ are in $\llbracket 0,d_p\rrbracket$).

The maximum number of factors $\omega_{0,2}$ is $\frac{k+1}{2}$, which means that the degree of the pole is at most $k+1$.

Let us now study the behavior of the denominator with $ydx$.
We have
\beq
ydx(z)-ydx(\beta_j) 
= t_{p,r_p-1} (1-\rho_p^{r_p \beta_j}) \zeta_p^{-r_p} d\zeta_p +   t_{p,r_p-2} (1-\rho_p^{(r_p-1) \beta_j}) \zeta_p^{1-r_p} d\zeta_p + O\left(\zeta_P^{2-r_p}\right)d\zeta_p.
\eeq
If $r_p\beta_j$ is not multiple of $d_p$, this quantity is of order $O\left(\zeta_p^{-r_p}\right)$, and if $r_p\beta_j$ is multiple of $d_p$, then $(r_p-1)\beta_j$ can't be multiple of $d_p$, and thus in both cases this quantity has a pole of degree at least $(r_p-1)$.

In \eqref{eq-TRhigherorder}, there are $k$ factors $\frac{1}{ydx(z)-ydx(\beta_j)}$, therefore the integrand is of order at least $O\left(\zeta_p^{-(k+1)+k(r_p-1)}\right)$.

If $r_p\geq 3$, the integrand behaves as $O(\zeta_p^{k-1})$, which has no pole because $k\geq 1$. Therefore, the residue vanishes.

This proves that residues coming from $a=p$ vanish, and hence $\omega_{h,n}$ is the same as if we had taken residues only at $a\in \mathcal R$, which have no pole at $z\in x^{-1}(\mathcal P)$.
\end{proof}

\br
This condition is sufficient, but it might not be necessary.
\er

\section{Proof of the KZ equations}

In this section we prove the KZ equations for a generic divisor $D$, then we regularize them for a particular divisor $D= [z]-[\infty^{(\alpha)}]$ and finally, we re-write them making use of cycles.

\subsection{Proof for a generic divisor}\label{proof_systemPDE}

In this section we give the proof of Theorem~\ref{systemPDE}. Let us compute the differential $\hbar d_{p_i} \psi_{l,i}(D,\hbar)$. It reads
\bea\label{eq-dPsi-1}
\hbar d_{p_i} \psi_{l,i}(D,\hbar) &=&  \hbar d_{p_i} \Bigg[\sum_{(h,n)\in \mathbb{N}^2} \frac{\hbar^{2h+n}}{n!}  \overbrace{\int_D\cdots\int_D}^n \frac{\hat{Q}_{h,n+1}^{(l)}(p_i;\cdot)}{\left(dx(p_i)\right)^l} \Bigg]  \; \psi(D,\hbar)  \cr
&& \qquad 
+  \Bigg[\sum_{(h,n)\in \mathbb{N}^2} \frac{\hbar^{2h+n}}{n!}  \overbrace{\int_D\cdots\int_D}^n \frac{\hat{Q}_{h,n+1}^{(l)}(p_i;\cdot)}{\left(dx(p_i)\right)^l} \Bigg]  \; \hbar d_{p_i}  \left(\psi(D,\hbar)  \right) \cr
&=&  \Bigg[\sum_{(h,n)\in \mathbb{N}^2} \frac{\hbar^{2h+n+1}}{n!}  \overbrace{\int_D\cdots\int_D}^{n} d_{p_i} \bigg(\frac{ \hat{Q}_{h,n+1}^{(l)}(p_i;\cdot)}{\left(dx(p_i)\right)^l}\bigg) \Bigg]  \; \psi(D,\hbar) \cr
&& + \alpha_i \Bigg[\sum_{(h,n)\in \mathbb{N}^2} \frac{\hbar^{2h+n+1}}{(n-1)!}  \overbrace{\int_D\cdots\int_D}^{n-1} \frac{\hat{Q}_{h,n+1}^{(l)}(p_i;p_i,\cdot)}{\left(dx(p_i)\right)^l} \Bigg]  \; \psi(D,\hbar) \cr
&& + \Bigg[\sum_{(h,n)\in \mathbb{N}^2} \frac{\hbar^{2h+n}}{n!}  \overbrace{\int_D\cdots\int_D}^n \frac{\hat{Q}_{h,n+1}^{(l)}(p_i;\cdot) }{\left(dx(p_i)\right)^l}\Bigg]  \; \hbar d_{p_i}  \left(\psi(D,\hbar)  \right).
\eea

The first of the last three terms is obtained by action of the differential operator inside the integrals while the second one comes from the action on any of the $n$ integrals themselves.

\medskip

Let us now compute the different terms.  One has
\beq
\hbar d_{p_i}  \left(\psi(D,\hbar)  \right)  = \alpha_i \Bigg[\sum_{h \geq 0} \sum_{n\geq 1} \frac{\hbar^{2h-1+n}}{(n-1)!} \overbrace{\int_D\cdots\int_D}^{n-1} \bigg[\omega_{h,n}(p_i, \cdot) - \delta_{h,0} \delta_{n,2} \frac{dx(p_i) dx(\cdot)}{(x(p_i)-x(\cdot))^2} \bigg] \Bigg] \psi(D,\hbar),
\eeq
so that the sum of the second and third terms reads
\beq
\begin{array}{l}
\alpha_i \sum_{h\geq 0} \sum_{n \geq 1}   \frac{\hbar^{2h+n+1}}{(n-1)!}  \overbrace{\int_D\cdots\int_D}^{n-1} \frac{\hat{Q}_{h,n+1}^{(l)}(p_i;p_i,\cdot)}{\left(dx(p_i)\right)^l} \; \psi(D,\hbar) \cr
+  \alpha_i {\displaystyle \sum_{h_1,h_2 \geq 0} \sum_{n_1,n_2 \geq 0} } \frac{\hbar^{2h_1+2h_2+n_1+n_2}}{n_1! n_2!} \overbrace{\int_D\cdots\int_D}^{n_1} \frac{\hat{Q}_{h_1,n_1+1}^{(l)}(p_i;\cdot)}{\left(dx(p_i)\right)^l}\cr 
\quad \overbrace{\int_D\cdots\int_D}^{n_2} \left[\omega_{h_2,n_2+1}(p_i, \cdot) - \delta_{h_2,0} \delta_{n_2,1} \frac{dx(p_i) dx(\cdot)}{(x(p_i)-x(\cdot))^2} \right] 
 \; \psi(D,\hbar),  \cr
\end{array}
\eeq
which can be re-organized as 
\beq
\begin{array}{l}
\alpha_i {\displaystyle \sum_{h\geq 0} \sum_{n \geq 0}} \hbar^{2h+n} \Bigg[\frac{1}{n!}  \overbrace{\int_D\cdots\int_D}^{n} \frac{\hat{Q}_{h-1,n+2}^{(l)}(p_i;p_i,\cdot)}{\left(dx(p_i)\right)^l} \cr
  +  {\displaystyle \sum_{h_1+h_2 =h} \sum_{n_1+n_2 =n} } \frac{1}{n_1! n_2!} \overbrace{\int_D\cdots\int_D}^{n_1} \frac{\hat{Q}_{h_1,n_1+1}^{(l)}(p_i;\cdot) }{\left(dx(p_i)\right)^l}\cr
  \overbrace{\int_D\cdots\int_D}^{n_2} \left[\omega_{h_2,n_2+1}(p_i, \cdot) - \delta_{h_2,0} \delta_{n_2,1} \frac{dx(p_i) dx(\cdot)}{(x(p_i)-x(\cdot))^2}  \right] \Bigg] \; \psi(D,\hbar) 
 .\cr
\end{array}
\eeq

Using Lemma~\ref{lemma-Q-omega} and Definition~\ref{DefPerturbativeWaveFunctions}, this reads 
\begin{align}
&- \alpha_i \psi_{l+1,i}(D,\hbar) dx(p_i)+ \alpha_i {\displaystyle \sum_{h\geq 0} \sum_{n \geq 0}}\frac{ \hbar^{2h+n}}{n!} \nonumber \\
& \quad \int_{z_1 \in D} \dots \int_{z_n \in D} \bigg[\frac{Q_{h,n+1}^{(l+1)}(x(p_i);\mathbf{z}) }{\left(dx(p_i)\right)^l} - \sum_{j=1}^n \frac{\hat{Q}_{h,n}^{(l)}(p_i;\mathbf{z} \setminus \{z_j\}) }{\left(dx(p_i)\right)^l} \frac{dx(p_i) dx(z_j)}{(x(p_i)-x(z_j))^2}\bigg] \psi(D,\hbar) .
\end{align}
Plugging the definition given by \eqref{def-Q-tilde}, this reads
\begin{align}
& - \alpha_i \psi_{l+1,i}(D,\hbar) dx(p_i)
+ \alpha_i {\displaystyle \sum_{h\geq 0} \sum_{n \geq 0}}\frac{ \hbar^{2h+n} dx(p_i) }{n!} \int_{z_1 \in D} \dots \int_{z_n \in D} \widetilde{Q}_{h,n+1}^{(l+1)}(x(p_i);\mathbf{z}) \psi(D,\hbar) \cr
& + \alpha_i {\displaystyle \sum_{h\geq 0} \sum_{n \geq 0}}\frac{ \hbar^{2h+n}  }{n!}  \int_{z_1 \in D} \dots \int_{z_n \in D} {\displaystyle \sum_{j=1}^n} d_{z_j} \Bigg[\frac{dx(p_i) \Big(\frac{\hat{Q}_{h,n}^{(l)}(z_j;\mathbf{z} \setminus \{z_j\})}{dx(z_j)^l}- \frac{\hat{Q}_{h,n}^{(l)}(p_i;\mathbf{z} \setminus \{z_j\})}{dx(p_i)^l} \Big)}{x(p_i)-x(z_j)} \Bigg] \; \psi(D,\hbar). \cr
\end{align}
One can evaluate the last term by integration along the divisor $D$ and, taking into account that it has degree 0, one gets
\begin{align}
&- \alpha_i \psi_{l+1,i}(D,\hbar) dx(p_i) - \hbar \alpha_i dx(p_i) {\displaystyle \sum_{j \in \llbracket 1 , s \rrbracket \setminus  \{i\}}} \alpha_j \frac{\psi_{l,i}(D,\hbar)-\psi_{l,j}(D,\hbar)}{x(p_i)-x(p_j)} \cr
&- \alpha_i^2 \Bigg[{\displaystyle \sum_{(h,n) \in \mathbb{N}^2} }\frac{\hbar^{2h+n+1}}{n!}  \overbrace{\int_D\cdots\int_D}^{n} d_{p_i} \bigg(\frac{\hat{Q}_{h,n+1}^{(l)}(p_i;\cdot) }{\left(dx(p_i)\right)^l}\bigg) \Bigg]  \; \psi(D,\hbar) \cr
&+ \alpha_i {\displaystyle \sum_{h\geq 0} \sum_{n \geq 0}}\frac{ \hbar^{2h+n} dx(p_i) }{n!} \int_{z_1 \in D} \dots \int_{z_n \in D} \widetilde{Q}_{h,n+1}^{(l+1)}(x(p_i);\mathbf{z}) \psi(D,\hbar) .\cr
\end{align}
Plugging this into \eqref{eq-dPsi-1}, this gives
\begin{align}
\hbar d_{p_i} \psi_{l,i}(D,\hbar) &= - \alpha_i \psi_{l+1,i}(D,\hbar) dx(p_i) - \hbar \alpha_i dx(p_i) \sum_{j \in \llbracket 1 , s \rrbracket \setminus \{i\}} \alpha_j \frac{\psi_{l,i}(D,\hbar)-\psi_{l,j}(D,\hbar)}{x(p_i)-x(p_j)} \cr
&+ (1- \alpha_i^2) \Bigg[\sum_{(h,n) \in \mathbb{N}^2} \frac{\hbar^{2h+n+1}}{n!}  \overbrace{\int_D\cdots\int_D}^{n} d_{p_i} \bigg( \frac{\hat{Q}_{h,n+1}^{(l)}(p_i;\cdot)}{\left(dx(p_i)\right)^l}\bigg) \Bigg]  \; \psi(D,\hbar) \cr
&+ \alpha_i {\displaystyle \sum_{h\geq 0} \sum_{n \geq 0}}\frac{ \hbar^{2h+n} dx(p_i)}{n!} \int_{z_1 \in D} \dots \int_{z_n \in D} \widetilde{Q}_{h,n+1}^{(l+1)}(x(p_i);\mathbf{z}) \psi(D,\hbar) .\cr
\end{align}

\subsection{Proof for the special divisor}\label{proof_reg-KZ-eq}
In this section we prove Theorem~\ref{reg-KZ-eq} corresponding to the special choice of divisor: $D = [z]-[\infty^{(\alpha)}]$, with $z\notin x^{-1}(\mathcal{P}) \cup x^{-1}(x(\mathcal{R}))$ a generic point in a small neighborhood of $\infty^{(\alpha)}$. 

\bigskip

For a two point divisor $D = [z]-[p_2]$, the first KZ equation of Theorem~\ref{systemPDE} reads
\beqq
\begin{array}{rcl}
{\hbar} \frac{d \psi_{l,1}(D = [z]-[p_2],\hbar)}{dx(z)} & = & - \psi_{l+1,1}(D = [z]-[p_2],\hbar) +  \hbar \frac{\psi_{l,1}(D= [z]-[p_2],\hbar)-\psi_{l,2}(D= [z]-[p_2],\hbar)}{x(z)-x(p_2)} \cr
&& \quad  + \bigg[{\displaystyle \sum_{h\geq 0} \sum_{n \geq 0}}\frac{ \hbar^{2h+n} }{n!} \int_{z_1 \in D} \dots \int_{z_n \in D} \widetilde{Q}_{h,n+1}^{(l+1)}(x(z);\mathbf{z}) \bigg]\; \psi(D = [z]-[p_2],\hbar). \cr
\end{array}
\eeqq
Multiplying this equation by $ e^{ \hbar^{-1} V_{\infty^{(\alpha)}}(p_2) } \frac{1}{x(p_2)}\sqrt{\frac{dx(p_2)}{d\zeta_{\infty^{(\alpha)}}(p_2)}}$,
 one gets
\beqq
 {\hbar} \frac{d  \psi_{l,1}^{(p_2)} }{dx(z)}  = - \psi_{l+1,1}^{(p_2)} +  \hbar \frac{\psi_{l,1}^{(p_2)}-\psi_{l,2}^{(p_2)}}{x(z)-x(p_2)} + \bigg[{\displaystyle \sum_{h\geq 0} \sum_{n \geq 0}}\frac{ \hbar^{2h+n} }{n!} \int_{z_1 \in D} \dots \int_{z_n \in D} \widetilde{Q}_{h,n+1}^{(l+1)}(x(z);\mathbf{z}) \bigg]  \psi^{(p_2)}_{0,1}, 
\eeqq
 where we denote for simplicity:
\beqq
\psi_{l,i}^{(p_2)}\coloneqq e^{ \hbar^{-1} V_{\infty^{(\alpha)}}(p_2) }\frac{1}{x(p_2)} \sqrt{\frac{dx(p_2)}{d\zeta_{\infty^{(\alpha)}}(p_2)}}  \psi_{l,i}\left(D=[z]-[p_2],\hbar\right),
\eeqq
for $i\in\{1,2\}$ and  $l\in \llbracket 0, d-1\rrbracket$.
They are defined in such a way that
\beqq
\psi^{\mathrm{reg}}_l(D=[z]-[\infty^{(\alpha)}],\hbar)\coloneqq \lim_{p_2 \to \infty^{(\alpha)}}  \psi_{l,1}^{(p_2)},\; \text{ for all }  l\in \llbracket 0, d-1\rrbracket.
\eeqq
 
By definition, $ \psi_{l,1}^{(p_2)}$ is holomorphic as $p_2 \to \infty^{(\alpha)}$, so that one can separate the singular and regular terms in this limit by writing
  \begin{align*}
& {\hbar} \frac{d  \psi_{l,1}^{(p_2)} }{dx(z)} + \psi_{l+1,1}^{(p_2)}  - \hbar \frac{\psi_{l,1}^{(p_2)}}{x(z)-x(p_2)} \cr
& =   -  \hbar \frac{\psi_{l,2}^{(p_2)}}{x(z)-x(p_2)} + \bigg[{\displaystyle \sum_{h\geq 0} \sum_{n \geq 0}}\frac{ \hbar^{2h+n} }{n!} \int_{z_1 \in D} \dots \int_{z_n \in D} \widetilde{Q}_{h,n+1}^{(l+1)}(x(z);\mathbf{z}) \bigg]  \psi^{(p_2)}_{0,1}  .
 \end{align*}
 
 The LHS admits a limit as $p_2 \to \infty^{(\alpha)}$ and reads
\beqq
 \hbar \frac{d  \psi^{\mathrm{reg}}_l(D=[z]-[\infty^{(\alpha)}],\hbar) }{dx(z)} + \psi^{\mathrm{reg}}_{l+1}(D=[z]-[\infty^{(\alpha)}],\hbar) .
\eeqq
This implies that the RHS admits a limit as $p_2 \to \infty^{(\alpha)}$ as well. In order to study this limit, let us write this RHS as
\beq\label{before_limit}
\Bigg[{\displaystyle \sum_{h\geq 0} \sum_{n \geq 0}}\frac{ \hbar^{2h+n} }{n!} \int_{z_1 \in D} \dots \int_{z_n \in D} \bigg(\widetilde{Q}_{h,n+1}^{(l+1)}(x(z);\mathbf{z}) 
- \hbar \frac{\hat{Q}_{h,n+1}^{(l)}(p_2;\mathbf{z}) }{(x(z)-x(p_2)) \, (dx(p_2))^l} \bigg)
\Bigg]  \psi^{(p_2)}_{0,1},
\eeq
where we have used that $\psi^{(p_2)}_{0,2}=\psi^{(p_2)}_{0,1}$ because  $\psi_{0,2}(D,\hbar)=\psi_{0,1}(D,\hbar)=\psi(D,\hbar)$ for any divisor $D$. The factor $\psi^{(p_2)}_{0,1}$ satisfies
\beqq
\psi^{(p_2)}_{0,1} \to \psi^{\mathrm{reg}}(D=[z]-[\infty^{(\alpha)}],\hbar), \qquad \hbox{as }  p_2 \to \infty^{(\alpha)},
\eeqq
so that the RHS tends to 
\begin{align*}
\psi^{\mathrm{reg}}([z]-[\infty^{(\alpha)}],\hbar) \, \cdot  \lim_{p_2 \to \infty^{(\alpha)}}  \Bigg[{\displaystyle \sum_{h\geq 0} \sum_{n \geq 0}}\frac{ \hbar^{2h+n} }{n!} \int_{z_1 \in D}  \dots \int_{z_n \in D} \bigg(\widetilde{Q}_{h,n+1}^{(l+1)}(x(z);\mathbf{z}) \quad\quad\quad\quad \\
- \hbar \frac{\hat{Q}_{h,n+1}^{(l)}(p_2;\mathbf{z}) }{(x(z)-x(p_2)) \, (dx(p_2))^l} \bigg)
\Bigg]  .
\end{align*}
Let $\lambda\in \mathbb P^1\setminus \mathcal P$, and $\lambda\neq x(z)$, we have
\bea
 \int_{z_1 \in D}  \dots \int_{z_n \in D} \widetilde{Q}_{h,n+1}^{(l+1)}(\lambda;\mathbf{z}) 
&=&
\int_{z_1 \in D}  \dots \int_{z_n \in D} {Q}_{h,n+1}^{(l+1)}(\lambda;\mathbf{z})/d\lambda^{l+1} \cr
&& -\frac{n}{\lambda-x(z)}\int_{z_1 \in D}  \dots \int_{z_{n-1} \in D} {\hat Q}_{h,n}^{(l)}(z;\mathbf{z})/dx(z)^l \cr
&& +\frac{n}{\lambda-x(p_2)}\int_{z_1 \in D}  \dots \int_{z_{n-1} \in D} {\hat Q}_{h,n}^{(l)}(p_2;\mathbf{z})/dx(p_2)^l \ . \cr &&
\eea
The third term exactly cancels when we sum over $n$, and the first 2 terms do have a limit as $p_2\to\infty^{(\alpha)}$.
Indeed if $l\geq 1$, in the definition of $Q_{h,n+1}^{(l+1)}(\lambda;z_1,\dots,z_n)$, there is no $\omega_{0,1}(z_i)$ for any $i\in \llbracket1,n\rrbracket$, and there is no $\omega_{0,2}(z_i,z_j)$. Thus the integrals $z_i\in D$ are convergent in the limit $p_2\to\infty^{(\alpha)}$.

This implies that the RHS tends to 
\begin{align*}
\psi^{\mathrm{reg}}([z]-[\infty^{(\alpha)}],\hbar) \, \cdot  \lim_{\lambda \to x(z)}  \Bigg[{\displaystyle \sum_{h\geq 0} \sum_{n \geq 0}}\frac{ \hbar^{2h+n} }{n!} \int_{z_1 \in [z]-[\infty^{(\alpha)}]}  \dots \int_{z_n \in [z]-[\infty^{(\alpha)}]} \bigg({Q}_{h,n+1}^{(l+1)}(\lambda;\mathbf{z})/d\lambda^{l+1} \quad \\
- \hbar \frac{\hat{Q}_{h,n+1}^{(l)}(z;\mathbf{z}) }{(\lambda-x(z)) \, (dx(z))^l} \bigg)
\Bigg]  .
\end{align*}
Let us write the limit as a residue at $\lambda\to x(z)$ as
\begin{align*}
\Res_{\lambda \to x(z)} \frac{d\lambda}{\lambda-x(z)} \Bigg[{\displaystyle \sum_{h\geq 0} \sum_{n \geq 0}}\frac{ \hbar^{2h+n} }{n!} \int_{z_1 \in [z]-[\infty^{(\alpha)}]}  \dots \int_{z_n \in [z]-[\infty^{(\alpha)}]} \bigg(\frac{{Q}_{h,n+1}^{(l+1)}(\lambda;\mathbf{z})}{d\lambda^{l+1}} \quad\quad\\ - \hbar \frac{\hat{Q}_{h,n+1}^{(l)}(z;\mathbf{z}) }{(\lambda-x(z)) \, (dx(z))^l} \bigg)
\Bigg]  .
\end{align*}
The integrand is a rational fraction of $\lambda$, with singularities at $x(z)$ and at $\lambda\in \mathcal P$, therefore, moving the integration contour we get
\begin{align*}
-\sum_{P\in \mathcal P} \Res_{\lambda \to P} \frac{d\lambda}{\lambda-x(z)} \Bigg[{\displaystyle \sum_{h\geq 0} \sum_{n \geq 0}}\frac{ \hbar^{2h+n} }{n!} \int_{z_1 \in [z]-[\infty^{(\alpha)}]}  \dots \int_{z_n \in [z]-[\infty^{(\alpha)}]} \bigg(\frac{{Q}_{h,n+1}^{(l+1)}(\lambda;\mathbf{z})}{d\lambda^{l+1}}\quad\quad\\- \hbar \frac{\hat{Q}_{h,n+1}^{(l)}(z;\mathbf{z}) }{(\lambda-x(z)) \, (dx(z))^l} \bigg)
\Bigg]  .
\end{align*}
The last term in $\frac{1}{(\lambda-x(z))^2}$  yields no residues at $\mathcal P$.
For the other term, let us Taylor expand the $\frac{1}{\lambda-x(z)}$ at $\lambda\to P$ as
\beq
\frac{d\lambda}{\lambda-x(z)} = 
-\sum_{k=0}^\infty \xi_P(x(z))^{-k} \xi_P(\lambda)^{k-1}d\xi_P(\lambda) .
\eeq
This gives
\bea
&&\sum_{P\in \mathcal P} \sum_{k=0}^\infty \xi_P(x(z))^{-k} \Res_{\lambda \to P} \xi_P(\lambda)^{k-1} d\xi_P(\lambda)\cr 
&&{\displaystyle \sum_{h\geq 0} \sum_{n \geq 0}}\frac{ \hbar^{2h+n} }{n!} \int_{z_1 \in [z]-[\infty^{(\alpha)}]}  \dots \int_{z_n \in [z]-[\infty^{(\alpha)}]} {Q}_{h,n+1}^{(l+1)}(\lambda;\mathbf{z}) /d\lambda^{l+1}  .
\eea
This leads to the KZ equation
\bea\label{KZappB}
 &&{\hbar} \frac{d  \psi^{\mathrm{reg}}_l([z]-[\infty^{(\alpha)}],\hbar) }{dx(z)}  + \psi^{\mathrm{reg}}_{l+1}([z]-[\infty^{(\alpha)}],\hbar) \cr
&&= \bigg[ {\displaystyle \sum_{h\geq 0} \sum_{n \geq 0}}\frac{ \hbar^{2h+n} }{n!} {\displaystyle \sum_{P \in \mathcal{P}} \sum_{k\in S_P^{(l+1)}}} \xi_P(x(z))^{-k}  {\displaystyle \Res_{\lambda \to P}} \xi_P(\lambda)^{k-1} \, d  \xi_P(\lambda) \cr  
&&\quad \int_{z_1 =\infty^{(\alpha)}}^z \dots \int_{z_n  = \infty^{(\alpha)}}^z \frac{{Q}_{h,n+1}^{(l+1)}(\lambda;\mathbf{z})}{(d\lambda)^{l+1}} \bigg]  \psi^{\mathrm{reg}}([z]-[\infty^{(\alpha)}],\hbar).
\eea

\subsection{Rewriting the KZ equations with cycles}\label{appB3}

For every pole $P\in \mathcal P$, let us write its preimages $x^{-1}(P)=\{p^{(1)},\ldots,p^{(\ell_P)}\}$ keeping track of their multiplicities as follows
\beq
P=(P^{(1)},\dots,P^{(d)}) = (\overbrace{p^{(1)},\dots,p^{(1)}}^{d_{p^{(1)}}},\overbrace{p^{(2)},\dots,p^{(2)}}^{d_{p^{(2)}}},\dots,\overbrace{p^{(\ell_P)},\dots,p^{(\ell_P)}}^{d_{p^{(\ell_P)}}}).
\eeq
For $j\in \llbracket 1,d\rrbracket$, we will say that $P^{(j)}$ corresponds to $p^{(k)}$ if $1+\underset{r=1}{\overset{k-1}{\sum}}d_{p^{(r)}}\leq j \leq \underset{r=1}{\overset{k}{\sum}}d_{p^{(r)}}$. We also denote $d_{p^{(k)}}$ by $d_{P^{(j)}}$, for every $P^{(j)}$ that corresponds to $p^{(k)}$. Note that there are $d_{p^{(k)}}$ different $j$'s that give $P^{(j)}$'s corresponding to the same $p^{(k)}$. 

Let us introduce roots of unity
\beq
\rho_{P^{(j)}}=e^{\frac{2\pi\ii}{d_{P^{(j)}}}}.
\eeq
In a neighborhood of $\lambda\to P$, we denote the preimages 
$$
x^{-1}(\lambda) = \{q_1,\dots,q_d\},
$$
labeled in such a way that when $\lambda\to P$, we have $q_j\to P^{(j)}$, and the local coordinates are
\beq
\zeta_{P^{(j)}}(q_j) = (\rho_{P^{(j)}})^j \ \left(\xi_{P}(\lambda)\right)^{\frac{1}{d_{P^{(j)}}}}.
\eeq
In particular, if ${P^{(j)}}={P^{(j')}}$, the local coordinates are proportional by a root of unity
\beq
\zeta_{P^{(j)}}=\left(\rho_{P^{(j)}}\right)^{j-j'}  \ \zeta_{P^{(j')}}.
\eeq
Then let us Taylor expand $Q^{(l+1)}_{h,n+1}(\lambda;\mathbf z)$ in its first variable in the limit $\lambda\to P$.
For $(h,n+l+1)\neq (0,2)$, the Taylor series of $\omega_{h,n+l+1}(q_{i_1},\dots,q_{i_{l+1}},z_1,\dots,z_n)$, has its coefficients obtained by residues, i.e.~by generalized cycle integrals
\bea\label{B3-eq-taylorexpomega}
\omega_{h,n+l+1}(q_{i_1},\dots,q_{i_{l+1}},z_1,\dots,z_n)
&=&\!\!\!\!\!\! \sum_{k_1,\dots,k_{l+1}=-r_P}^\infty 
\prod_{j=1}^{l+1}
\zeta_{P^{(i_j)}}(q_{i_j})^{k_j-1}d\zeta_{P^{(i_{j})}}(q_{i_{j}})
 \cr
&& \int_{\mathcal C_{P^{(i_1)},k_1}}\!\!\!\dots \int_{\mathcal C_{P^{(i_{l+1})},k_{l+1}}} \!\!\!\!\!\!\omega_{h,n+l+1}(\cdot,\dots,\cdot,\mathbf z) .
\eea
Non strictly positive values of $k_i$s can occur only for $\omega_{0,1}$.
We denote the lower bound $r_P=\max_{j=1,\dots,d}\ (r_{P^{(j)}}-1) $.

We can also write the Taylor expansion for the integrals with $z_1,\dots,z_n$ on $D$:
\bea\label{B3-eq-taylorexpintD}
&& \int_{z_1\in D}\dots \int_{z_n\in D}\omega_{h,n+l+1}(q_{i_1},\dots,q_{i_{l+1}},z_1,\dots,z_n) \cr
&=& \sum_{k_1,\dots,k_{l+1}=-r_P}^\infty 
\prod_{j=1}^{l+1}
\zeta_{P^{(i_j)}}(q_{i_j})^{k_j-1}d\zeta_{P^{(i_{j})}}(q_{i_{j}})
 \cr
&& \int_{\mathcal C_{P^{(i_1)},k_1}}\dots \int_{\mathcal C_{P^{(i_{l+1})},k_{l+1}}}
\int_D\dots \int_D \omega_{h,n+l+1}. 
\eea
As in the last expression, for simplicity, we omit the variables that we integrate over when there is no possible confusion.
Notice that the $\int_D$ integrals are rightmost, they are performed before the $\mathcal{C}_{p,k}$ integrals, i.e.~before taking the Taylor expansion coefficients, and \eqref{B3-eq-taylorexpintD} is not the integral of \eqref{B3-eq-taylorexpomega} in general.
However the only case where the order of integration does not commute is for $\omega_{0,2}$, as we saw in Proposition~\ref{prop-commutator}, and in that case
\bea
\int_{z_1\in D}\omega_{0,2}(q_{i_1},z_1) 
&=& \sum_{k_1=0}^\infty 
\zeta_{P^{(i_1)}}(q_{i_1})^{k_1-1}d\zeta_{P^{(i_{1})}}(q_{i_{1}})
 \int_{\mathcal C_{P^{(i_1)},k_1}}
\int_D \omega_{0,2}  \cr
&=& -\delta_{P^{(i_1)},\infty^{(\alpha)}} \zeta_{P^{(i_1)}}(q_{i_1})^{-1}d\zeta_{P^{(i_{1})}}(q_{i_{1}}) \cr
&& + \sum_{k_1=1}^\infty 
\zeta_{P^{(i_1)}}(q_{i_1})^{k_1-1}d\zeta_{P^{(i_{1})}}(q_{i_{1}})
\int_D  \int_{\mathcal C_{P^{(i_1)},k_1}}
\omega_{0,2} .
\eea
From the definition of $Q^{(l+1)}_{h,n+1}(\lambda;\mathbf z)$, the $\omega_{0,2}(q_{i_1},q_{i_2})$ that can appear there will always be evaluated at $q_{i_1}\neq q_{i_2}$. However, we can expand them around $q_{i_j}\rightarrow P^{(i_j)}$, $j\in\{1,2\}$, with $P^{(i_1)}=P^{(i_2)}$.
The Taylor expansion of $\omega_{0,2}$ is
\bea
&& \omega_{0,2}(q_{i_1},q_{i_2}) -
\delta_{P^{(i_1)},P^{(i_2)}} \  \frac{d\zeta_{P^{(i_1)}}(q_{i_1}) d\zeta_{P^{(i_2)}}(q_{i_2})}{(\zeta_{P^{(i_1)}}(q_{i_1})-\zeta_{P^{(i_2)}}(q_{i_2}) )^2} \cr
&= & 
\sum_{k_1,k_2=1}^\infty \zeta_{P^{(i_1)}}(q_{i_1})^{k_1-1}d\zeta_{P^{(i_1)}}(q_{i_1}) \zeta_{P^{(i_2)}}(q_{i_2})^{k_2-1} d\zeta_{P^{(i_2)}}(q_{i_2}) \cr
&& 
\oint_{z'_1\in\mathcal C_{P^{(i_1)},k_1}}\oint_{z'_2\in\mathcal C_{P^{(i_2)},k_2}} \left( \omega_{0,2}(z'_1,z'_2)
-\delta_{P^{(i_1)},P^{(i_2)}} \ \frac{d\zeta_{P^{(i_1)}}(z'_{1}) d\zeta_{P^{(i_2)}}(z'_{2})}{(\zeta_{P^{(i_1)}}(z'_{1})-\zeta_{P^{(i_2)}}(z'_{2}) )^2} \right)  \cr
&= & 
\sum_{k_1,k_2=1}^\infty \zeta_{P^{(i_1)}}(q_{i_1})^{k_1-1}d\zeta_{P^{(i_1)}}(q_{i_1}) \zeta_{P^{(i_2)}}(q_{i_2})^{k_2-1} d\zeta_{P^{(i_2)}}(q_{i_2}) \cr
&& 
\oint_{z'_1\in\mathcal C_{P^{(i_1)},k_1}} 
\Big( - \delta_{P^{(i_1)},P^{(i_2)}}  \ k_2 \zeta_{P^{(i_1)}}(z'_{1})^{-k_2-1} d \zeta_{P^{(i_1)}}(z'_{1})+
\oint_{z'_2\in\mathcal C_{P^{(i_2)},k_2}} \omega_{0,2}(z'_1,z'_2) \Big)  \cr
&= & 
\sum_{k_1,k_2=1}^\infty \zeta_{P^{(i_1)}}(q_{i_1})^{k_1-1}d\zeta_{P^{(i_1)}}(q_{i_1}) \zeta_{P^{(i_2)}}(q_{i_2})^{k_2-1} d\zeta_{P^{(i_2)}}(q_{i_2}) \cr
&& 
\oint_{z'_1\in\mathcal C_{P^{(i_1)},k_1}} 
\oint_{z'_2\in\mathcal C_{P^{(i_2)},k_2}} \omega_{0,2}(z'_1,z'_2)  .
\eea
Then, notice that if $P^{(i_1)}=P^{(i_2)}$, $\frac{\zeta_{P^{(i_1)}}}{\zeta_{P^{(i_2)}}} = \rho_{P^{(i_1)}}^{i_1-i_2}$ is a root of unity.
This implies
\beq
\frac{d\zeta_{P^{(i_1)}}(q_{i_1})) d\zeta_{P^{(i_2)}}(q_{i_2}))}{(\zeta_{P^{(i_1)}}(q_{i_1}))-\zeta_{P^{(i_2)}}(q_{i_2})) )^2}
= \frac{\rho_{P^{(i_1)}}^{i_1-i_2}}{(1-\rho_{P^{(i_1)}}^{i_1-i_2})^2} \ \zeta_{P^{(i_1)}}(q_{i_1})^{-1}d\zeta_{P^{(i_1)}}(q_{i_1})  \zeta_{P^{(i_2)}}(q_{i_2})^{-1} d\zeta_{P^{(i_2)}}(q_{i_2}).
\eeq
In other words
\bea
\omega_{0,2}(q_{i_1},q_{i_2}) 
&=& \delta_{P^{(i_1)},P^{(i_2)}} \ \frac{\rho_{P^{(i_1)}}^{i_1-i_2}}{(1-\rho_{P^{(i_1)}}^{i_1-i_2})^2} \  \zeta_{P^{(i_1)}}(q_{i_1})^{-1}d\zeta_{P^{(i_1)}}(q_{i_1})  \zeta_{P^{(i_2)}}(q_{i_2})^{-1} d\zeta_{P^{(i_2)}}(q_{i_2})
\cr
&& +
\sum_{k_1,k_2=1}^\infty \zeta_{P^{(i_1)}}(q_{i_1}))^{k_1-1}d\zeta_{P^{(i_1)}}(q_{i_1})) \zeta_{P^{(i_2)}}(q_{i_2}))^{k_2-1} d\zeta_{P^{(i_2)}}(q_{i_2})) \cr
&& 
\oint_{\mathcal C_{P^{(i_1)},k_1}} 
\oint_{\mathcal C_{P^{(i_2)},k_2}} \omega_{0,2} .
\eea
In some sense this is as if we extended the sum in the second line to include the cases $k_1=0$ and $k_2=0$ by defining:
\beq
\oint_{\mathcal C_{P^{(i_1)},0}} 
\oint_{\mathcal C_{P^{(i_2)},0}} \omega_{0,2}  
\coloneqq R(P)_{i_1,i_2},  
\eeq
where we defined the $d\times d$ matrix  $R(P)$ to be 
\beq
\forall \, i\in\llbracket 1,d\rrbracket \,\, 
R(P)_{i,i} \coloneqq 0 
\, , \quad 
\forall \ i_1\neq i_2 \ \ \ 
R(P)_{i_1,i_2} \coloneqq \delta_{P^{(i_1)},P^{(i_2)}} \ \frac{\rho_{P^{(i_1)}}^{i_1-i_2}}{(1-\rho_{P^{(i_1)}}^{i_1-i_2})^2}.
\eeq

\medskip

With this definition at hand, we can now write
the RHS of the KZ equation \eqref{KZappB}, after expanding the $Q^{(l+1)}_{h,n+1}(\lambda;\mathbf z)$ around $\lambda\to P$, as
\bea
&& \Bigg[ {\displaystyle \sum_{h, n \geq 0}}\frac{ \hbar^{2h+n} }{n!} {\displaystyle \sum_{P \in \mathcal{P}} \sum_{k\in S_P^{(l+1)}}} \xi_P(x(z))^{-k}  \cr
&& \sum_{\beta \underset{l+1}{\subseteq} x^{-1}(\lambda)} \sum_{\nu \in \mathcal{S}(\beta)} \, {\displaystyle \sum_{J_1\sqcup\dots\sqcup J_{l(\nu)}=\mathbf z}}\,  \sum_{ \underset{i=1}{\overset{l(\nu)}{\sum}} h_i= h+l(\nu)-l-1}
\sum_{k_{i,j}=1 -r_{P^{(\nu_{i,j})}}}^\infty  \cr
&&
\Res_{\lambda=P} \xi_P(\lambda)^{k-1}d\xi_P(\lambda) 
\prod_{i=1}^{l(\nu)} \prod_{j=1}^{|\nu_i|}
\zeta_{P^{(\nu_{i,j})}}^{k_{i,j}-1}d\zeta_{P^{(\nu_{i,j})}}/d\lambda
\cr
&& 
\prod_{i=1}^{l(\nu)} \left( \left( \prod_{j=1}^{|\nu_i|}
 \int_{\mathcal C_{P^{(\nu_{i,j})},k_{i,j}}}  
 \prod_{j=1}^{|J_i|}  \int_{[z]-[\infty^{(\alpha)}]} 
 \right) \omega_{h_{i},|\nu_i|+|J_i|} \right)\Bigg]  \psi^{\mathrm{reg}}([z]-[\infty^{(\alpha)}],\hbar).
\eea
Notice that only the elements of $J_i$ are integrated from $\infty^{(\alpha)}$ to $z$, and since $|\nu_i|>0$, it can never be an integral of $\omega_{0,1}$, i.e.~the integrand never has poles at $\infty^{(\alpha)}$ and thus these integrals are well defined.

Note that only $\omega_{0,1}$ factors can bring some $k_{i,j}<0$ and only $\omega_{0,1}$ and $\omega_{0,2}$ factors can bring some $k_{i,j}=0$.

Let us now compute the residues at $\lambda\to P$.
We have,

\begin{itemize}

\item If $P=\Lambda_i$ is a finite pole, we denote $\epsilon_P=1$, and  its multiplicity $d_{P^{(j)}}=\text{order}_{P^{(j)}} x >0$. The canonical base local coordinate is $\xi_P = \lambda-\Lambda_i$, $d\xi_P=d\lambda$, and writing the preimages of $P$ with multiplicities as $\{P^{(1)},\dots,P^{(d)}\}$, we have
\beq
\zeta_{P^{(j)}} = \xi_{P}^{\frac{1}{d_{P^{(j)}}}}\; \text{ and }\;\; d\zeta_{P^{(j)}} = \frac{1}{d_{P^{(j)}}}\xi_{P}^{\frac{1}{d_{P^{(j)}}}-1} d \xi_{P}.
\eeq
The residue at $\lambda=P$ thus selects
\beq
-k = -l-1+\sum_{i=1}^{l(\nu)} \sum_{j=1}^{|\nu_i|} \frac{k_{i,j}}{d_{P^{(\nu_{i,j})}}}.
\eeq

\item If $P=\infty$ is an infinite pole, we denote $\epsilon_P=-1$, and its multiplicity $d_{P^{(j)}}=-\text{order}_{P^{(j)}} x >0$.
We have $\xi_P = \frac{1}{\lambda}$, $d\lambda=-\xi_P^{-2}d\xi_P$, and writing the preimages of $P$ with multiplicities as $\{P^{(1)},\dots,P^{(d)}\}$, we have
\beq
\zeta_{P^{(j)}} = \xi_{P}^{\frac{1}{d_{P^{(j)}}}}\; \text{ and }\;\; d\zeta_{P^{(j)}} = \frac{1}{d_{P^{(j)}}}\xi_{P}^{\frac{1}{d_{P^{(j)}}}-1} d \xi_{P}.
\eeq
The residue at $\lambda=P$ thus selects
\beq
-k = l+1+\sum_{i=1}^{l(\nu)} \sum_{j=1}^{|\nu_i|} \frac{k_{i,j}}{d_{P^{(\nu_{i,j})}}}.
\eeq
\end{itemize}

In both cases
\beq
 \sum_{i=1}^{l(\nu)} \sum_{j=1}^{|\nu_i|} \frac{k_{i,j}}{d_{P^{(\nu_{i,j})}}} = (l+1)\epsilon_P -k.
\eeq
Some of the $k_{i,j}$s can be negative and with lower bound $k_{i,j}>-r_{P^{(\nu_{i,j})}}$.
This implies that the sum over the $k_{i,j}$ must in fact be a finite sum, with an upper bound.

We have that the RHS of \eqref{KZappB} reads
\bea
&& \Bigg[ {\displaystyle \sum_{h, n \geq 0}}\frac{ \hbar^{2h+n} }{n!} {\displaystyle \sum_{P \in \mathcal{P}} \sum_{k\in S_P^{(l+1)}}} \xi_P(x(z))^{-k}  \cr
&& \sum_{\beta \underset{l+1}{\subseteq} x^{-1}(\lambda)} \sum_{\nu \in \mathcal{S}(\beta)} \, {\displaystyle \sum_{J_1\sqcup\dots\sqcup J_{l(\nu)}=\mathbf z}}\,  \sum_{ \underset{i=1}{\overset{l(\nu)}{\sum}} h_i= h+l(\nu)-l-1}
\sum_{k_{i,j}=1 -r_{P^{(\nu_{i,j})}}}^\infty  \cr
&&
\delta\left( \sum_{i=1}^{l(\nu)} \sum_{j=1}^{|\nu_i|} \frac{k_{i,j}}{d_{P^{(\nu_{i,j})}}} = (l+1)\epsilon_P -k \right)
\prod_{i=1}^{l(\nu)} \prod_{j=1}^{\nu_i} \frac{\epsilon_P}{d_{P^{(\nu_{i,j})}}}
\cr
&& 
\prod_{i=1}^{l(\nu)} \left( \left( \prod_{j=1}^{|\nu_i|}
 \int_{\mathcal C_{P^{(\nu_{i,j})},k_{i,j}}}  
 \prod_{j=1}^{|J_i|}  \int_{[z]-[\infty^{(\alpha)}]} 
 \right) \omega_{h_{i},|\nu_i|+|J_i|} \right)\Bigg]  \psi^{\mathrm{reg}}([z]-[\infty^{(\alpha)}],\hbar),
\eea
where we use the Kr\"onecker $\delta$ to encode the constraint on the $k_{i,j}$s.
In this sum there are terms with $k_{i,j}<0$  that may only come from $\omega_{0,1}$,
there are terms with $k_{i,j}=0$  that may only come from $\omega_{0,1}$
 and $\omega_{0,2}$, and there are the strictly positive $k_{i,j}>0$.

Let us separate the negative and vanishing $k_{i,j}s$ from the strictly positive.

We shall decompose the partition $\nu$ into 3 pieces:
\begin{itemize}
\item A piece that we re-call $\nu$ for the strictly positive $k_{i,j}s$:
\beq
\nu=\bigsqcup_{i\in\llbracket 1, \ell\rrbracket }\{\nu_{i,1},\dots,\nu_{i,|\nu_i|}\},
\eeq
with all indices $(i,j)$, $j\in\llbracket 1, |\nu_i|\rrbracket$, distinct, corresponding to elements in $\llbracket 1,d\rrbracket$, and giving an order to the sets.
To it we associate the sequence of exponents
\beq
((k_{1,1},\dots,k_{1,|\nu_1|}),(k_{2,1},\dots,k_{2,|\nu_2|}),\dots,(k_{\ell,1},\dots,k_{\ell,|\nu_\ell|})),
\eeq
with all $k_{i,j}>0$.

\item A piece that we call $\nu'$ for the negative or vanishing $k_{i,j}s$ corresponding to integrals of $\omega_{0,1}$. All its parts must be of size 1, and instead of encoding it as a set partition, we can just denote it as a set of indices $\nu'=\{\nu'_1,\dots,\nu'_{\ell'}\}$, with $\nu'_j\in \llbracket 1,d\rrbracket\setminus\nu$.
To it we associate the exponents
\beq
k'_{1},\dots,k'_{\ell'},
\qquad \text{with} \quad
k'_i\in \llbracket 0,r_{P^{(\nu'_i)}}-1\rrbracket .
\eeq
The integrals associated to it give the spectral times
\beq
\oint_{\mathcal C_{P^{(\nu'_i)},-k'_i}} \omega_{0,1} =t_{P^{(\nu'_i)},k'_i}.
\eeq

\item A piece that we call $\nu''$ for the zero values corresponding to the pseudo-integrals of $\omega_{0,2}$ which are worth $R(p)_{i_1,i_2}$. 
All its parts must be of size $2$ and its indices are taken $\nu''_{i,\pm}\in\llbracket 1,d\rrbracket\setminus (\nu\cup \nu')$: 
$$
\nu'' = \{\{\nu''_{1,+},\nu''_{1,-}\},\{\nu''_{2,+},\nu''_{2,-}\},\dots,\{\nu''_{\ell'',+},\nu''_{\ell'',-}\}\}\in \mathcal{S}^{(2)}( \llbracket 1,d \rrbracket\setminus (\nu\cup \nu')),
$$ 
where we denoted $\mathcal{S}^{(2)}( \llbracket 1,d \rrbracket\setminus (\nu\cup \nu'))$ the set of set partitions where all the parts have size $2$. 
To it we shall associate zero exponents $k''_{i,\pm}=0$, and the factor
\beq
C_{\nu''_i} \coloneqq \oint_{\mathcal C_{P^{(\nu''_{i,+})},0}}\oint_{\mathcal C_{P^{(\nu''_{i,-})},0}} \omega_{0,2} = R(p)_{\nu''_{i,+},\nu''_{i,-}} = R(p)_{\nu''_i}.
\eeq

\end{itemize}

Finally, one could have a factor of the form 
$$ \int_{\mathcal C_{P^{(\nu_{i,j})},k_{i,j}}}    \int_{[z]-[\infty^{(\alpha)}]} 
 \omega_{0,2} ,
$$
with $P^{(\nu_{i,j})} = \infty^{(\alpha)}$ and $k_{i,j} = 0$ giving a contribution equal to -1. 
In this case, one has the same contributions as above for $\nu \cup \nu' \cup \nu''$ not containing $\alpha$ and $|\nu \cup \nu' \cup \nu''|=l$.

In the end we get the RHS of the KZ equation divided by $\psi^{\mathrm{reg}}([z]-[\infty^{(\alpha)}],\hbar)$ as
\bea
&& {\displaystyle \sum_{h, n \geq 0}}\frac{ \hbar^{2h+n} }{n!} {\displaystyle \sum_{P \in \mathcal{P}} \sum_{k\in S_P^{(l+1)}}} \xi_P(x(z))^{-k}  \sum_{\ell'=0}^{l+1} \; \sum_{\nu'\subset_{\ell'} \llbracket 1,d \rrbracket}
\sum_{k'_{i}=0}^{r_{P^{(\nu'_{i})}}-1}    
\sum_{0\leq \ell'' \leq \frac{l+1-\ell'}{2}}\;\sum_{\substack{\nu''\in\mathcal{S}^{(2)}( \llbracket 1,d \rrbracket\setminus \nu') \\ l(\nu'')=\ell''}}   \cr
&& \sum_{\beta \underset{l+1-\ell'-2\ell''}{\subseteq} \llbracket 1,d \rrbracket\setminus (\nu'\cup\nu'')}\;\; \sum_{\nu \in \mathcal{S}(\beta)}   \;\;\sum_{ \underset{i=1}{\overset{l(\nu)}{\sum}} h_i= h+l(\nu)-|\nu|-\ell''}
\;\;\sum_{\substack{k_{i,j}=1\\ i=1,\ldots,l(\nu) \\j=1,\ldots,|\nu_i|}}^\infty  
\;\sum_{n_1+\dots + n_{l(\nu)}=n} \frac{n!}{n_1! \dots n_{l(\nu)}!} 
\cr
&&
\delta\left( \sum_{i=1}^{l(\nu)} \sum_{j=1}^{|\nu_i|} \frac{k_{i,j}}{d_{P^{(\nu_{i,j})}}} = (l+1)\epsilon_P -k +\sum_{i=1}^{\ell'}\frac{k'_i}{d_{P^{(\nu'_i)}}} \right)   
\prod_{i=1}^{\ell'}  \frac{\epsilon_P}{d_{P^{(\nu'_{i})}}}
\prod_{i=1}^{\ell''} \frac{1}{d_{P^{(\nu''_{i,+})}}d_{P^{(\nu''_{i,-})}}}
\cr
&& 
\prod_{i=1}^{l(\nu)}\prod_{j=1}^{\nu_i} \frac{\epsilon_P}{d_{P^{(\nu_{i,j})}}}\prod_{i=1}^{\ell'}  t_{P^{(\nu'_i)},k'_{i}}
\ \ 
\prod_{i=1}^{\ell''}  R(P)_{\nu''_i}
\ \ 
\prod_{i=1}^{l(\nu)} \left( \left(  \prod_{j=1}^{|\nu_i|}
 \int_{\mathcal C_{P^{(\nu_{i,j})},k_{i,j}}}  \prod_{j=1}^{n_i}  \int_{[z]-[\infty^{(\alpha)}]} \right) \omega_{h_{i},|\nu_i|+n_i} \right)\cr
&&- \sum_{h,n \geq 0} \frac{\hbar^{2h+n}}{n!} 
\sum_{k\in S_{\infty}^{(l+1)}} \xi_\infty(x(z))^{-k}  
\sum_{\ell'=0}^{l} \; \sum_{\nu'\subset_{\ell'} \left( \llbracket 1,d \rrbracket \setminus \{\alpha\}\right)}
\sum_{k'_{i}=0}^{r_{\infty^{(\nu'_{i})}}-1}    
\sum_{0\leq \ell'' \leq \frac{l-\ell'}{2}}\;\sum_{\substack{\nu''\in\mathcal{S}^{(2)}( \llbracket 1,d \rrbracket\setminus \nu'\cup \{\alpha\}) \\ l(\nu'')=\ell''}}  \cr
&& 
\sum_{\beta \underset{l-\ell'-2\ell''}{\subseteq} \llbracket 1,d \rrbracket\setminus (\nu'\cup\nu''\cup \{\alpha\})}\;\; \sum_{\nu \in \mathcal{S}(\beta)}   \;\;\sum_{ \underset{i=1}{\overset{l(\nu)}{\sum}} h_i= h+l(\nu)-|\nu|-\ell''}
\;\;\sum_{\substack{k_{i,j}=1\\ i=1,\ldots,l(\nu) \\j=1,\ldots,|\nu_i|}}^\infty  
 \cr
&&
\sum_{n_1+\dots + n_{l(\nu)}+1=n} \frac{n!}{n_1! \dots n_{l(\nu)}!}\;\;\delta\left( \sum_{i=1}^{l(\nu)} \sum_{j=1}^{|\nu_i|} \frac{k_{i,j}}{d_{\infty^{(\nu_{i,j})}}} = (l+1)\epsilon_\infty -k +\sum_{i=1}^{\ell'}\frac{k'_i}{d_{\infty^{(\nu'_i)}}} \right)  \cr
&& \frac{\epsilon_{\infty}}{d_{\infty^{(\alpha)}}} \prod_{i=1}^{l(\nu)} \prod_{j=1}^{\nu_i} \frac{\epsilon_\infty}{d_{\infty^{(\nu_{i,j})}}}
\prod_{i=1}^{\ell'}  \frac{\epsilon_\infty}{d_{\infty^{(\nu'_{i})}}}
\prod_{i=1}^{\ell''} \frac{1}{d_{\infty^{(\nu''_{i,+})}}d_{\infty^{(\nu''_{i,-})}}}
\cr
&& 
\prod_{i=1}^{\ell'}  t_{P^{(\nu'_i)},k'_{i}}
\ \ 
\prod_{i=1}^{\ell''}  R(\infty)_{\nu''_i}
\ \ 
\prod_{i=1}^{l(\nu)} \left( \left(  \prod_{j=1}^{|\nu_i|}
 \int_{\mathcal C_{\infty^{(\nu_{i,j})},k_{i,j}}}  \prod_{j=1}^{n_i}  \int_{[z]-[\infty^{(\alpha)}]} \right) \omega_{h_{i},|\nu_i|+n_i} \right).
\eea

Now all cycle integrals are with generalized cycles with $k_{i,j}>0$, and these commute with all other integrals present in this expression; in particular with integrals $\int_{[z]-[\infty^{(\alpha)}]}$, so the order of integrations doesn't matter anymore.
Moreover, we have a formal power series of $\hbar$, whose coefficients are rational fractions of $x(z)$ with poles at $\mathcal P$, and whose coefficients are algebraic combinations of generalized cycle integrals of the $\omega_{h_i,n_i+|\nu_i|}$, therefore it can be written as the evaluation of an element of our symbolic algebra $\mathcal W$.

So let us rewrite the last expression in $\mathcal W$ as
\bea
\!\!\!\!\!\!\!\!\!&\!\!\! & \!\!\! \hbar^{l+1}  \sum_{P\in \mathcal P} {\epsilon_P^{l+1}}
\sum_{k\in S_P^{(l+1)}} \xi_P(x(z))^{-k}  
\sum_{\ell'=0}^{l+1}  \; \sum_{\nu'\subset_{\ell'} \llbracket 1,d \rrbracket}
\sum_{k'_{i}=0}^{r_{P^{(\nu'_{i})}}-1} \sum_{0\leq \ell'' \leq \frac{l+1-\ell'}{2}}\;\sum_{\substack{\nu''\in\mathcal{S}^{(2)}( \llbracket 1,d \rrbracket\setminus \nu') \\ l(\nu'')=\ell''}}   \cr
\!\!\!\!\!\!\!\!\!&\!\!\! & \!\!\!\sum_{\beta \underset{l+1-\ell'-2\ell''}{\subseteq} \llbracket 1,d \rrbracket\setminus (\nu'\cup\nu'')}\;\; \sum_{\nu \in \mathcal{S}(\beta)}   \;\;
\;\;  \frac{1}{
 {\displaystyle \prod_{i=1}^{l(\nu)} \prod_{j=1}^{\nu_i} d_{P^{(\nu_{i,j})}}
\prod_{i=1}^{\ell'}  d_{P^{(\nu'_{i})}}
\prod_{i=1}^{\ell''} d_{P^{(\nu''_{i,+})}}d_{P^{(\nu''_{i,-})}}}
}
\cr
\!\!\!\!\!\!\!\!\!&\!\!\! &
\!\!\!\sum_{\substack{k_{i,j}=1\\ i=1,\ldots,l(\nu) \\j=1,\ldots,|\nu_i|}}^\infty  \delta\left( \sum_{i=1}^{l(\nu)} \sum_{j=1}^{|\nu_i|} \frac{k_{i,j}}{d_{P^{(\nu_{i,j})}}} = (l+1)\epsilon_P -k +\sum_{i=1}^{\ell'}\frac{k'_i}{d_{P^{(\nu'_i)}}} \right)  \;
\prod_{i=1}^{\ell'} \hbar^{-1}  t_{P^{(\nu'_i)},k'_{i}}\;
\prod_{i=1}^{\ell''}  R(P)_{\nu''_i} \cr
\!\!\!\!\!\!\!\!\!&\!\!\! &\!\!\!\sum_{h_1,\dots,h_{l(\nu)}} \;
 \sum_{n_1+\dots + n_{l(\nu)}=n} \; 
\prod_{i=1}^{l(\nu)}\left(\frac{\hbar^{2h_i-2+n_i+|\nu_i|}}{n_i!}   \prod_{j=1}^{n_i}  \int_{[z]-[\infty^{(\alpha)}]}  
\prod_{j=1}^{|\nu_i|}  \int_{\mathcal C_{P^{(\nu_{i,j})},k_{i,j}}}   \omega_{h_{i},|\nu_i|+n_i} \right) \cr
\!\!\!\!\!\!\!\!\!&\!\!\! &\!\!\! - \hbar^{l+1} \epsilon_\infty^{l+1}
\sum_{k\in S_{\infty}^{(l+1)}}\xi_\infty(x(z))^{-k}  
\sum_{\ell'=0}^{l}  \; \sum_{\nu'\subset_{\ell'} \llbracket 1,d \rrbracket \setminus \{\alpha\}}
\sum_{k'_{i}=0}^{r_{\infty^{(\nu'_{i})}}-1}\sum_{0\leq \ell'' \leq \frac{l-\ell'}{2}}\;\sum_{\substack{\nu''\in\mathcal{S}^{(2)}( \llbracket 1,d \rrbracket\setminus \nu'\cup \{\alpha\}) \\ l(\nu'')=\ell''}}  \cr  
\!\!\!\!\!\!\!\!\!&\!\!\! & 
\!\!\!\sum_{\beta \underset{l-\ell'-2\ell''}{\subseteq} \llbracket 1,d \rrbracket\setminus (\nu'\cup\nu''\cup \{\alpha\})}\;\; \sum_{\nu \in \mathcal{S}(\beta)}
\;\;\frac{ 1}{ d_{\infty^{(\alpha)}}
 {\displaystyle \prod_{i=1}^{l(\nu)} \prod_{j=1}^{\nu_i} d_{\infty^{(\nu_{i,j})}}
\prod_{i=1}^{\ell'}  d_{\infty^{(\nu'_{i})}}
\prod_{i=1}^{\ell''} d_{\infty^{(\nu''_{i,+})}}d_{\infty^{(\nu''_{i,-})}}}
}   
\cr
\!\!\!\!\!\!\!\!\!&\!\!\! &
\!\!\!\sum_{\substack{k_{i,j}=1\\ i=1,\ldots,l(\nu) \\j=1,\ldots,|\nu_i|}}^\infty \delta\left( \sum_{i=1}^{l(\nu)} \sum_{j=1}^{|\nu_i|} \frac{k_{i,j}}{d_{\infty^{(\nu_{i,j})}}} = (l+1)\epsilon_\infty -k +\sum_{i=1}^{\ell'}\frac{k'_i}{d_{\infty^{(\nu'_i)}}} \right)  \;
\prod_{i=1}^{\ell'}   \hbar^{-1}t_{\infty^{(\nu'_i)},k'_{i}}\;
\prod_{i=1}^{\ell''}  R(\infty)_{\nu''_i} \cr
\!\!\!\!\!\!\!\!\!&\!\!\! &  
\!\!\!\sum_{h_1,\dots,h_{l(\nu)}} \;
 \sum_{n_1+\dots + n_{l(\nu)}+1=n}\;
\prod_{i=1}^{l(\nu)}\left(\frac{\hbar^{2h_i-2+n_i+|\nu_i|}}{n_i!}    \prod_{j=1}^{n_i}  \int_{[z]-[\infty^{(\alpha)}]}  
\prod_{j=1}^{|\nu_i|}  \int_{\mathcal C_{\infty^{(\nu_{i,j})},k_{i,j}}}   \omega_{h_{i},|\nu_i|+n_i} \right) . \quad
\eea

The following lemma will help recognize this expression as an operator acting on $\psi$.
\bl
Let $l>0$ and let $C_1,\dots,C_l \in \mathcal E$ be some generalized cycles of strictly positive type (i.e.~$C_i=\mathcal C_{p_i,k_i}$ with $k_i>0$).
Let $D$ be a divisor of degree 0.
We have

\bea\label{appB3lemmaeq1}
\hbar^{2l} \mathcal I_{C_{1}}\dots \mathcal I_{C_l}
\psi^{\mathrm{reg\; symbol}}(D,\hbar)
\!\!\!\!\!\!&=&\!\!\!\!\!\! \bigg[ \sum_{h, n \geq 0} \frac{\hbar^{2h+n}}{n!}
\sum_{\nu\in\mathcal{S}(\llbracket 1,l\rrbracket)} \;
\sum_{n_1+\dots+n_{l(\nu)}=n}\;\sum_{h_1+\dots+h_{l(\nu)}=h-l+l(\nu)} \cr
\!\!\!\!\!\!&&\!\!\!\!\!\!
\frac{n!}{\prod_{i=1}^{l(\nu)} n_i!}  \prod_{i=1}^{l(\nu)} \Big(\prod_{j=1}^{n_i}\int_D \prod_{j\in \nu_i}\int_{C_j} \omega_{h_i,|\nu_i|+n_i}\Big)
\bigg]  \psi^{\mathrm{reg\; symbol}}(D,\hbar)\cr&&
\eea
and
\bea\label{appB3lemmaeq2}
\text{ev}.\mathcal I_{C_1}\dots \mathcal I_{C_l}
\psi^{\mathrm{reg\; symbol}}(D,\hbar)
\!\!\!\!\!\!&=&\!\!\!\!\!\!  \bigg[ \sum_{h, n \geq 0} \sum_{\beta \underset{l}{\subseteq} x^{-1}(\lambda)} \sum_{\nu \in \mathcal{S}(\beta)} \frac{\hbar^{2h+n-2l}}{n!} \int_{z_1\in D} \dots \int_{z_n\in D} \cr
&& \sum_{J_1\sqcup\dots\sqcup J_{l(\nu)}=\{z_1,\dots,z_n\}}
\;\sum_{h_1+\dots+h_{l(\nu)}=h-l+l(\nu)} \cr
&&
\prod_{i=1}^{l(\nu)} \Big(\prod_{j\in \nu_i}\int_{C_j}  \omega_{h_i,|\nu_i|+|J_i|}(\nu_i,J_i)\Big)
\bigg] \ \psi^{\mathrm{reg}}(D,\hbar). \cr&&
\eea

\el

\begin{proof}
Let us start by proving \eqref{appB3lemmaeq1} by recursion on $l$.
It is true for $l=1$ by our definition of applying $\mathcal I_{C_1}$ to an exponential \eqref{ActionOnPsireg}:
\beq\label{l1}
\hbar^2\mathcal I_{C_{1}}\psi^{\mathrm{reg\; symbol}}(D,\hbar)= \bigg[ \sum_{h, n \geq 0} \frac{\hbar^{2h+n}}{n!}
 \Big(\int_D\cdots\int_D \int_{C_1} \omega_{h,1+n}\Big)
\bigg]  \psi^{\mathrm{symbol}}(D,\hbar).
\eeq
Then, assuming it is true up to $l-1$, let us prove it for $l$. We have by induction hypothesis
\begin{align}
\hbar^{2l}\mathcal I_{C_{1}}\dots \mathcal I_{C_l}
\psi^{\mathrm{reg\; symbol}}(D,\hbar)
&=\hbar^2\mathcal I_{C_l}\Bigg( \bigg[ \sum_{h, n\geq 0} \frac{\hbar^{2h+n}}{n!}
\sum_{\nu\in\mathcal{S}(\llbracket 1,l-1\rrbracket)} \,
\sum_{n_1+\dots+n_{l(\nu)}=n}\,\sum_{\substack{h_1+\dots+h_{l(\nu)}\\=h-l+1+l(\nu)}} \cr
&
\frac{n!}{\prod_{i=1}^{l(\nu)} n_i!}  \prod_{i=1}^{l(\nu)} \Big(\prod_{j=1}^{n_i}\int_D \prod_{j\in \nu_i}\int_{C_j} \omega_{h_i,|\nu_i|+n_i}\Big)
\bigg]  \psi^{\mathrm{reg\; symbol}}(D,\hbar)\Bigg).
\end{align}
Then we apply Leibniz rule
\begin{align}
\hbar^{2l} \mathcal I_{C_{1}}\dots \mathcal I_{C_l}
\psi^{\mathrm{reg\; symbol}}(D,\hbar) 
&=   \bigg[ \sum_{h, n\geq 0} \frac{\hbar^{2h+n}}{n!}
\sum_{\nu\in\mathcal{S}(\llbracket 1,l-1\rrbracket)} \,
\sum_{n_1+\dots+n_{l(\nu)}=n}\,\sum_{h_1+\dots+h_{l(\nu)}=h-l+1+l(\nu)} \cr
&
\frac{n!}{\prod_{i=1}^{l(\nu)} n_i!}  \prod_{i=1}^{l(\nu)} \Big(\prod_{j=1}^{n_i}\int_D \prod_{j\in \nu_i}\int_{C_j} \omega_{h_i,|\nu_i|+n_i}\Big)
\bigg]  \hbar^2\mathcal I_{C_l}  \psi^{\mathrm{reg\; symbol}}(D,\hbar)\cr
& + \Big[ \sum_{h, n\geq 0} \frac{\hbar^{2h+2+n}}{n!}
\sum_{\nu\in\mathcal{S}(\llbracket 1,l-1\rrbracket)} \,
\sum_{n_1+\dots+n_{l(\nu)}=n}\,\sum_{h_1+\dots+h_{l(\nu)}=h-l+l(\nu)} \cr
&
\frac{n!}{\prod_{i=1}^{l(\nu)} n_i!}  \mathcal I_{C_l}\bigg(\prod_{i=1}^{l(\nu)} \Big(\prod_{j=1}^{n_i}\int_D \prod_{j\in \nu_i}\int_{C_j}  \omega_{h_i,|\nu_i|+n_i}\Big)\bigg)
\bigg] \psi^{\mathrm{reg\; symbol}}(D,\hbar).
\end{align}
The first term can be written using \eqref{l1} as sum over partitions of $\llbracket 1, l\rrbracket$ the form $\td\nu=(\nu_1,\dots,\nu_{l(\nu)},1)$, i.e.~a partition obtained from $\nu$ by adding one more block of size $1$, hence $l(\tilde{\nu})=l(\nu)+1$ and $|\tilde{\nu}|=|\nu|+1$. Then we have that the new genus is $\tilde{h}=h+h_{l(\tilde{\nu})}$, giving $h_1+\dots+h_{l(\tilde{\nu})}=\tilde{h}-l+1+l(\nu)=\tilde{h}-l+l(\tilde{\nu})$. We also have $\tilde{n}=n+n_{l(\tilde{\nu})}$.
In the second term, $\mathcal I_{C_l}$ acts by Leibniz rule on the product, by acting on one of the factors, which amounts to adding the index $l$ to a block of $\nu$ in all possible ways. This gives $\hbar^{2h+2+n}$, with a new $\tilde{h}=h+1$ and $\tilde{\nu}\in\mathcal{S}(\llbracket 1, l\rrbracket)$, so $h_1+\dots+h_{l(\tilde{\nu})}=\tilde{h}-l+l(\tilde{\nu})$, with $l(\tilde{\nu})=l(\nu)$ and $\tilde{h}-l=h-l+1$, as before.
Eventually these two terms together give the sum over all partitions of the set $\llbracket 1,l \rrbracket$, i.e.~the formula \eqref{appB3lemmaeq1} at rank $l$.
\smallskip

To prove \eqref{appB3lemmaeq2}, we just recognize that the factor
\beq
\frac{n!}{\prod_{i=1}^{l(\nu)} n_i!}
\eeq
is the number of ways of writing
\beq
\{ z_1,\dots,z_n \} = J_1 \sqcup J_2 \sqcup \dots \sqcup J_{l(\nu)}.
\eeq
This ends the proof of the lemma.
\end{proof}

We define the operator
\bea
\td{\mathcal L}_l(x(z))
&\coloneqq &  \sum_{P\in \mathcal P} {\epsilon_P^{l+1}}
\Bigg[
 \xi_P(x(z))^{-(l+1)\epsilon_P}   
\sum_{\ell'=0}^{l+1}  \ \sum_{\nu'\subset_{\ell'} \llbracket 1,d \rrbracket} \ \ 
\prod_{j\in \nu'}
\bigg(
\sum_{k=0}^{r_{P^{(j)}}-1}  
\frac{t_{P^{(j)},k}}{d_{P^{(j)}}} \ \xi_P^{-k/d_{P^{(j)}}}   \bigg)
\quad \cr
&& 
\sum_{0\leq \ell'' \leq \frac{l+1-\ell'}{2}}\;\sum_{\substack{\nu''\in\mathcal{S}^{(2)}( \llbracket 1,d \rrbracket\setminus \nu') \\ l(\nu'')=\ell''}}
\ \ 
\prod_{i=1}^{\ell''}  \frac{\hbar^2R(P)_{\nu''_i}}{d_{P^{(\nu_{i,+}'')}} d_{P^{(\nu_{i,-}'')}} } \cr
&&
\sum_{\nu \underset{l+1-\ell'-2\ell''}{\subseteq} \llbracket 1,d \rrbracket\setminus (\nu'\cup\nu'')}
\;\;\prod_{j\in \nu} 
\bigg(\hbar^2\sum_{k=1}^\infty \frac{ \xi_P^{{k}/{d_{P^{(j)}}}}}{d_{P^{(j)}}}   \mathcal I_{\mathcal C_{P^{(j)},k}} \bigg)  \Bigg]_{\leq 0}  \cr
&&
- \hbar \frac{{\epsilon_\infty^{l+1}}}{d_{\infty^{(\alpha)}}}
\Bigg[
 \xi_\infty(x(z))^{-(l+1)\epsilon_\infty}    
\!\sum_{\ell'=0}^{l} \, \sum_{\nu'\subset_{\ell'} \llbracket 1,d \rrbracket \setminus \{\alpha\}} \,
\prod_{j\in \nu'}
\bigg(
\sum_{k=0}^{r_{\infty^{(j)}}-1}  
\frac{t_{\infty^{(j)},k}}{d_{\infty^{(j)}}} \ \xi_\infty^{-k/d_{\infty^{(j)}}}   \bigg)
\quad \cr
&& 
\sum_{0\leq \ell'' \leq \frac{l+1-\ell'}{2}}\;\sum_{\substack{\nu''\in\mathcal{S}^{(2)}( \llbracket 1,d \rrbracket\setminus (\nu' \cup \{\alpha\}))\\ l(\nu'')=\ell''}}
\ \ 
\prod_{i=1}^{\ell''}  \frac{\hbar^2R(\infty)_{\nu''_i}}{d_{\infty^{(\nu_{i,+}'')}} d_{\infty^{(\nu_{i,-}'')}} } \cr
&&
\sum_{\nu \underset{l-\ell'-2\ell''}{\subseteq} \llbracket 1,d \rrbracket\setminus (\nu'\cup\nu''\cup \{\alpha\})}\;\;
\prod_{j\in \nu} 
\bigg(\hbar^2\sum_{k=1}^\infty \frac{ \xi_\infty^{{k}/{d_{\infty^{(j)}}}}}{d_{\infty^{(j)}}}   \mathcal I_{\mathcal C_{\infty^{(j)},k}} \bigg)  \Bigg]_{\leq 0}, 
\eea
where the notation $[\cdot]_{\leq 0}$ means we keep only the terms that give $\xi_P(x(z))^{-m}$, with $m\in S_P^{(l)}$.
Then \eqref{appB3lemmaeq1} implies that the RHS of the KZ equation is
\beq 
{\hbar} \frac{d}{dx(z)}  \psi_{l}^{\mathrm{reg}}([z]-[\infty^{(\alpha)}])   +\psi_{l+1}^{\mathrm{reg}}([z]-[\infty^{(\alpha)}])   
= \text{ev}. \widetilde{\mathcal{L}}_l(x(z)) \left[\psi^{\text{reg  symbol}}([z]-[\infty^{(\alpha)}]) \right].
\eeq
This ends the proof of Theorem \ref{final-KZ}.

\section{Proof of Theorem \ref{th-no-bp}: No pole at ramification points}\label{App-no-bp}

\subsection{General idea of the proof}

For any triple $J=(P,k,l)_{P\in \mathcal{P}, k\in \mathbb{N}, l\in \llbracket 0,d-1\rrbracket}$, $ \mathcal{L}_{P,k,l}$ is a linear combination of operators of the form ${\displaystyle \prod_{j=1}^m} \left(\hbar^{2} \mathcal{I}_{\mathcal C_j}\right)$
with generalized cycles of the form
$\mathcal C_j = {\mathcal{C}_{p_j,k_j}} $ with $x(p_j) = P$, $k_j \geq 1$ and $m \leq l$. Let us note that the coefficients of this linear combination are independent of $\lambda$.

For any set of generalized cycles $(\mathcal C_j)_{j=1}^m$ of this form, one can define the symbolic matrix 
\beq
\widetilde{A}_{\mathcal C_1,\dots,\mathcal C_m}^{\mathrm{symbol}} \coloneqq \hbar^{-1} \left(\prod_{j=1}^m \left(\hbar^{2} \mathcal{I}_{\mathcal C_j} \right) \tilde{\Psi}^{\mathrm{symbol}}\right)   \left(\tilde{\Psi}^{\mathrm{symbol}}\right)^{-1}
\eeq
and its evaluation
\beq
\widetilde{A}_{\mathcal C_1,\dots,\mathcal C_m} \coloneqq \text{ev}. \widetilde{A}_{\mathcal C_1,\dots,\mathcal C_m}^{\mathrm{symbol}} .
\eeq

It is defined in such a way that the systems
\beq\label{eq-comp-syst-symbol}
\left\{
\begin{array}{lll}
\hbar  \frac{d \widetilde{ \Psi}^{\mathrm{symbol}}(\lambda) }{d\lambda} & = & \widetilde{L}^{\mathrm{symbol}}(\lambda) \widetilde{ \Psi}^{\mathrm{symbol}}(\lambda), \\[0.3em]
 \hbar^{-1} \left({\displaystyle \prod_{j=1}^m} \left(\hbar^{2} \mathcal{I}_{\mathcal C_j} \right)\right) \widetilde{\Psi}^{\mathrm{symbol}}(\lambda) & = & \widetilde{A}_{\mathcal C_1,\dots,\mathcal C_m}^{\mathrm{symbol}}(\lambda) \widetilde{\Psi}^{\mathrm{symbol}}(\lambda) \cr
\end{array}
\right.
\eeq
and
\beq\label{eq-comp-syst}
\left\{
\begin{array}{lll}
\hbar  \frac{d \widetilde{ \Psi}(\lambda) }{d\lambda} & = & \widetilde{L}(\lambda) \widetilde{ \Psi}(\lambda), \\[0.3em]
 \text{ev}. \left[\hbar^{-1} \left({\displaystyle \prod_{j=1}^m} \left(\hbar^{2} \mathcal{I}_{\mathcal C_j} \right)\right) \widetilde{\Psi}^{\mathrm{symbol}}(\lambda) \right] & = & \widetilde{A}_{\mathcal C_1,\dots,\mathcal C_m}(\lambda) \widetilde{\Psi}(\lambda) \cr
\end{array}
\right.
\eeq
are compatible. All these matrices are trans-series of the form
\beq
\widetilde{L}(\lambda) = \sum_{p\geq 0} \hbar^p L^{(p)}(\lambda,\hbar)
 , \qquad
\widetilde{L}^{\mathrm{symbol}}(\lambda) = \sum_{p\geq 0} \hbar^p L^{(p),\mathrm{\,symbol}}(\lambda,\hbar)
 \eeq
 and
 \beq
\widetilde{A}_{\mathcal C_1,\dots,\mathcal C_m} (\lambda) = \sum_{p \geq 0} \hbar^p \widetilde{A}_{\mathcal C_1,\dots,\mathcal C_m} ^{(p)}(\lambda,\hbar)
, \qquad
\widetilde{A}_{\mathcal C_1,\dots,\mathcal C_m}^{\mathrm{symbol}} (\lambda) = \sum_{p \geq 0} \hbar^p \widetilde{A}_{\mathcal C_1,\dots,\mathcal C_m} ^{(p),\mathrm{\,symbol}}(\lambda,\hbar).
\eeq

\medskip

\emph{The proof is done by induction on $h\geq 0$ regarding the following proposition.
\beq
 \mathcal{P}_h\,:\, \begin{array}{c}
`` \,\forall\, m \geq 1 \, , \; \forall\, P \in \mathcal{P} \, , \; \forall\, (p_1,\dots,p_m) \in (x^{-1}(P))^m \, , \; \forall\, (k_1,\dots,k_m) \in \left(\mathbb{N}^*\right)^m
  \,   : \cr
  \,\widetilde{A}_{\mathcal{C}_{p_1,k_1},\dots,\mathcal{C}_{p_m,k_m}}^{(h)}(\lambda) = O(1) \, , \; \text{ when }\, \lambda \to u \,\text{ for all }\, u \in x\left(\mathcal{R}\right).\text{''}
  \end{array}
  \eeq}

\noindent{For} every $h\geq 0$, we thus assume by induction that for all $a\in \llbracket 0, h-1\rrbracket$, $\mathcal{P}_a$ is satisfied.
Note that this assumption is empty for the initialization case $h=0$. We now want to prove that $\mathcal{P}_h$ holds.

First of all, by the definition \eqref{def-tilde-L} of $\widetilde{L}(\lambda,\hbar)$, the possible poles of $ L^{(a)}(\lambda)$ at $\lambda = u$ come from the poles of some $\widetilde{A}_{P,K,l}^{(m)}$ with $m<a$.
Hence, the previous assumption implies that
\beq
\forall\, a \in \llbracket 0, h\rrbracket \,,\, L^{(a)}(\lambda) = O(1), \, \text{ when } \, \lambda \to u , \,\text{ for all }\, u\in x\left(\mathcal{R}\right) .
\eeq
Hence proving $ \mathcal{P}_h$ for any $h \in \mathbb{N}$ proves the theorem.

For proving $ \mathcal{P}_h$, we shall begin by proving it for $m\geq 2$ by using an induction formula for $\widetilde{A}_{\mathcal{C}_{p_1,k_1},\dots,\mathcal{C}_{p_m,k_m}}^{(h)}(\lambda)$.

We then prove it for $m=1$ by using the compatibility of the system \eqref{eq-comp-syst-symbol}.

\subsection{Proof of $\mathcal{P}_h$ for $m \geq 2$}

One can compute the matrices $\widetilde{A}_{\mathcal C_1,\dots,\mathcal C_m}(\lambda)$ by induction on $m$ through the formula
\beq
\forall\, m \geq 2 \, , \; \widetilde{A}_{\mathcal C_1,\dots,\mathcal C_m}^{\mathrm{symbol}} (\lambda) = \hbar^2 \mathcal{I}_{\mathcal{C}_{m}} \widetilde{A}_{\mathcal C_1,\dots,\mathcal C_{m-1}}^{\mathrm{symbol}} (\lambda) +
\hbar \widetilde{A}_{\mathcal C_1,\dots,\mathcal C_{m-1}}^{\mathrm{symbol}} (\lambda)  \widetilde{A}_{\mathcal C_{m}}^{\mathrm{symbol}} (\lambda) .
\eeq
After evaluation, this reads
\beq
\forall\, m \geq 2 \, , \;  \widetilde{A}_{\mathcal C_1,\dots,\mathcal C_m} (\lambda) = \hbar^2 \text{ev} . \mathcal{I}_{\mathcal{C}_{m}} \widetilde{A}_{\mathcal C_1,\dots,\mathcal C_{m-1}}^{\mathrm{symbol}} (\lambda) +
\hbar \widetilde{A}_{\mathcal C_1,\dots,\mathcal C_{m-1}} (\lambda)  \widetilde{A}_{\mathcal C_{m}} (\lambda) .
\eeq
To order $h$ in the $\hbar$ expansion as trans-series, the RHS involves only terms of the form $\widetilde{A}_{\mathcal{C}_1,\dots,\mathcal{C}_n}^{(a)}$ with $a<h$ which do not have  not  have any pole as $\lambda \to u \in \mathcal{R}$ thanks to the induction property $\mathcal{P}_a$, $a \leq h-1$.  Hence it is regular and one has
\beq
\forall\, m \geq 2 \, , \; \widetilde{A}_{\mathcal C_1,\dots,\mathcal C_m}^{(h)} (\lambda) = O(1)
\eeq
when $\lambda \to u$.

\subsection{Preparation of the proof of $\mathcal{P}_h$ for $m =1$}

In this section, we collect a few properties which will prove to be useful for the proof in the case $m=1$.

\subsubsection{Leading order of the Lax matrix $\widetilde{L}(\lambda)$}

One of the main ingredients of our proof is the form of the leading order $L^{(0)}(\lambda)$ of the Lax matrix $\widetilde{L}(\lambda)$.

For $\lambda \notin x\left(\mathcal{R}\right)\cup \mathcal{P}$, since $L^{(0)}(\lambda)$ is a companion matrix with distinct eigenvalues,
\beq
L^{(0)}(\lambda) = 
 \left[
 \begin{array}{cccccc}
 0 & 1&0 &  \dots &  0\cr
0 & 0&1 &  \dots &  0\cr
  \vdots & \vdots &\vdots&\ddots & \vdots \cr
  0 & 0&0  & \dots &  1\cr
  (-1)^{d-1}P_d(\lambda) & (-1)^{d-2}P_{d-1}(\lambda)& (-1)^{d-3}P_{d-2}(\lambda) & \dots &  P_{1}(\lambda) \cr  
\end{array}
\right].
\eeq
It is diagonalized by a Vandermonde matrix $V(\lambda)$,
\beq
V^{-1}(\lambda) L^{(0)}(\lambda) V(\lambda) =  Y(\lambda),
\eeq
where
\beq
Y(\lambda) = \diag (y_1(\lambda),y_2(\lambda),\dots,y_d(\lambda)),
\eeq
with $y_i(\lambda)$ being the value of $y(z^{(i)}(\lambda))$, with $\{z^{(i)}(\lambda)\}_{i=1}^d = x^{-1}(\lambda)$ chosen such that 
\beq
y_1(u) = y_2(u)
\eeq
and $V(\lambda)$ is the Vandermonde matrix with entries $\left[V(\lambda)\right]_{k,l} = y_l(\lambda)^{k-1}$, for all $(k,l)\in \llbracket 1,d\rrbracket^2$.

\br\label{RemarkRegularBranchpoints}Since the spectral curve is admissible, we have that $\forall\, (k,l) \in\llbracket 3,d\rrbracket^2$ : $y_k(u)\neq y_1(u)$ and $y_k(u)\neq y_l(u)$. Moreover, $y'_2(u)=- y_1'(u)$, with $y'_1(u)\neq y'_2(u)$. In particular, we also have $y'_1(u)\neq 0$ and $y_2'(u)\neq 0$.
\er

\subsubsection{Action of $\mathcal{I}_{C}$ on trans-series}

\label{subsection-action-trans}

Let us remind the reader how our linear operators act on the trans-series considered. From \eqref{eq-action-I-trans}, which we recall here for convenience,
\begin{align}
\hbar  \mathcal{I}_{\mathcal{C}}  \left[\sum_{m=1}^\infty \hbar^m \, \Xi_m^{(\infty_\alpha)}(z,\hbar, \boldsymbol{\epsilon}, \boldsymbol{\rho}) \right] = 
\sum_{m=1}^\infty \hbar^{m+1} \, \left. \mathcal{I}_{\mathcal{C}}  \cdot \Xi_m^{(\infty_\alpha)}(z,\hbar, \boldsymbol{\epsilon}, \boldsymbol{\rho}) \right|_{\phi \; \text{fixed}}\cr
+ \sum_{j=1}^g \mathcal{I}_{\mathcal{C}} \left[\phi_j \right]
\sum_{m=1}^\infty \hbar^{m} \, \left. \frac{\partial\, \Xi_m^{(\infty_\alpha)}(z,\hbar, \boldsymbol{\epsilon}, \boldsymbol{\rho})}{\partial \phi_j} \right|_{\phi_j = \frac{1}{2 \pi i } \oint_{\mathcal{B}_j} \om_{0,1}}\!\!\!\!\!\!\!\!\!\!\!\!\!\!\!\!\!\!,
\end{align}
$\hbar  \mathcal{I}_{\mathcal{C}}$ raises the $\hbar$ order by one except when acting on the $(\phi_i)_{i=1}^g$ dependence of our trans-series. In particular, when acting on the Vandermonde matrix $V(\lambda)$, its inverse or the rational functions $\xi_P(\lambda)$, $\hbar  \mathcal{I}_{\mathcal{C}}$ raises the order of $\hbar$ by one. 

\subsubsection{Relation between the Lax and auxiliary matrices}

Let $k\in \mathbb{N}$ and $P\in \mathcal{P}$. The entries of $\widetilde{\Delta}_{P,k}(\lambda)$ defined from \eqref{eq-diff-tilde} and \eqref{eq-diff-tildee} are expressed in terms of the entries of the auxiliary matrices $\widehat{A}_{P,k,l}(\lambda)$. For later use, let us remind it here. From the definition, one has
\beq\label{EqDelta}
\widetilde{\Delta}_{P,k}(\lambda)= -\hbar (G(\lambda))^{-1} \frac{\partial G(\lambda)}{\partial \lambda} + (G(\lambda))^{-1}  \widehat{\Delta}_{P,k} (\lambda) G(\lambda),
\eeq
\beq
\left[ \widehat{\Delta}_{P,k} (\lambda)\right]_{r,s} = \left[\widehat{A}_{P,k,r-1}\right]_{1,s},
\eeq
and
\bea
\widetilde{A}_{P,k,l}(\lambda,\hbar)&\coloneqq& \left(G(\lambda)\right)^{-1} \mathrm{ev.}\left( \left[ G^{\mathrm{symbol}}(\lambda), \hbar^{-1} \mathcal{L}_{P,k,l} \right] \widetilde{\Psi}^{\mathrm{symbol}}(\lambda,\hbar) \right) \widetilde{\Psi}(\lambda,\hbar)^{-1} \cr
&&+ \left(G(\lambda)\right)^{-1}  \widehat{A}_{P,k,l}(\lambda,\hbar)  G(\lambda) .
\eea
One can combine these two equations in order to get an expression of $\widetilde{\Delta}_{P,k}(\lambda)$ in terms of the $\widetilde{A}_{P',k',l'}$.
At the end of the proof, we shall use only the leading order expansion of this equation around the branch points.

\subsection{Proof of $\mathcal{P}_h$ for $m =1$}

In order to prove that $ \mathcal{P}_h$ is true for $m=1$, let us consider an arbitrary generalized cycle $\mathcal{C}\coloneqq\mathcal{C}_{p,k}$ for $p \in x^{-1}(\mathcal{P})$ and $k \geq 1$. We shall prove that $\widetilde{A}_{\mathcal{C}}^{(h)}(\lambda) = O(1)$ as $\lambda \to u \in \mathcal{R}$.

The compatibility of the system \eqref{eq-comp-syst-symbol} reads
\beq
\hbar \frac{\partial \widetilde{A}_{\mathcal{C}}^{\mathrm{symbol}}(\lambda)}{\partial \lambda} = \hbar \mathcal{I}_{\mathcal{C}} \widetilde{L}^{\mathrm{symbol}}(\lambda) + \left[ \widetilde{L}^{\mathrm{symbol}}(\lambda),\widetilde{A}_{\mathcal{C}}^{\mathrm{symbol}}(\lambda)\right].
\eeq

Let us now consider a given $u\in x\left(\mathcal{R}\right)$ for the rest of the proof.
We shall denote $s_\mathcal{C} \in \mathbb{N}$ the order of the pole of $\widetilde{A}_{\mathcal{C}}^{(h)}(\lambda) $ at $\lambda\to u$. Note that we may have $s_\mathcal{C}=0$ if $\widetilde{A}_{\mathcal{C}}^{(h)}(\lambda)$ is regular at $\lambda\to u$. We then define $s=\underset{(p,k)\in  x^{-1}(\mathcal{P})\times \mathbb{N}^*}{\max}(s_{\mathcal{C}_{p,k}})$. If $s=0$, then $\mathcal{P}_h$ holds. Let us thus assume by reductio ad absurdum that $s>0$ and consider $\mathcal{T}=\{(p,k)\in  x^{-1}(\mathcal{P})\times \mathbb{N}^*\,\text{ such that }\, s_{\mathcal{C}_{p,k}}=s\}$ the subset of indices for which the poles are of maximum order. By assumption $\mathcal{T}\neq \emptyset$.
\medskip

For any $\mathcal{C}_0  = \mathcal{C}_{p,k}$ with $(p,k)\in \mathcal{T}$, the expansion of $\widetilde{A}_{\mathcal C_0}^{(h)}(\lambda)$ around $\lambda=u$ reads
\beq\label{poledef}
\widetilde{A}_{\mathcal C_0}^{(h)}(\lambda) =  \frac{a_{{\mathcal{C}_0}}}{(\lambda-u)^s} + \frac{b_{\mathcal C_0}}{(\lambda-u)^{s-1}} +O\left((\lambda-u)^{-s+2}\right), \, \text{ with }\, a_{\mathcal C_0}\neq 0,
\eeq
for some $d\times d$ matrices $a_{\mathcal C_0}$ and $b_{\mathcal C_0}$.

\medskip

At order $\hbar^{h}$ and $\hbar^{h+1}$, the compatibility condition takes the form, for any $\mathcal{C}_0\in \mathcal{T}$,
\beq\label{compatibility-rec}
\left\{
\begin{array}{l}
\left[\widetilde{A}_{\mathcal C_0}^{(h)}(\lambda),L^{(0)}(\lambda)\right]  = - \underset{0\leq m \leq h-1}{\sum} \left[\widetilde{A}_{\mathcal C_0}^{(m)}(\lambda),L^{(h-m)}(\lambda)\right] + \frac{\partial \widetilde{A}_{\mathcal C_0}^{(h-1)}(\lambda)}{\partial \lambda}  - \left(\hbar \, \text{ev} . \mathcal{I}_{\mathcal C_0} \widetilde{L}^{\mathrm{symbol}}(\lambda) \right)^{(h)},\cr
 \underset{0\leq m\leq h+1}{\sum} \left[\widetilde{A}_{\mathcal C_0}^{(m)}(\lambda),L^{(h-m+1)}(\lambda)\right]  =  \frac{\partial \widetilde{A}_{\mathcal C_0}^{(h)}}{\partial \lambda}  - \left(\hbar \,  \text{ev}.\mathcal{I}_{\mathcal C_0} \widetilde{L}^{\mathrm{symbol}}(\lambda) \right)^{(h+1)},\cr
\end{array}
\right.
\eeq
where we remind the reader that  the notation $(F(\lambda))^{(p)}$ means extracting the coefficient of order $\hbar^p$ in the trans-series $F(\lambda)$.

\subsubsection{Expansion of the Vandermonde matrix and its inverse when $\lambda\to u$}

The expansions of the Vandermonde matrix $V(\lambda)$ and its inverse  $V^{-1}(\lambda)$  at $\lambda\to u$ read
\bea
V(\lambda)&=&V_0+V_1(\lambda-u)^{\frac{1}{2}}+V_2(\lambda-u)+O\left((\lambda-u)^{\frac{3}{2}}\right) ,\cr
V^{-1}(\lambda)&=&\frac{B_0}{(\lambda-u)^{\frac{1}{2}}}+B_1+B_2(\lambda-u)^{\frac{1}{2}} +O\left((\lambda-u)\right) .
\eea
It is straightforward computations from the identity $V^{-1}(\lambda)V(\lambda)=I_d$ around $\lambda=u$ (providing $B_0V_0=0$, $B_1V_0+B_0V_1=I_d$, $B_2V_0+B_1V_1+B_0V_2=0$ and $B_3V_0+B_2V_1+B_1V_2+B_0V_3=0$ and so on) and our knowledge of $V(\lambda)$ to show that the matrices $(B_k)_{k\geq 0}$ are given by
\beq M (B_k)^t=F_k \,\,\Leftrightarrow \,\, B_k=(F_k)^t (M^{-1})^t  \,\,,\,\,\forall \, k\geq 0,\eeq
where the matrix $M$ does not depend on $k$ and is given by
\beq M\coloneqq \left(\begin{array}{cccc}1&y_2(u)&\dots &y_2(u)^{d-1}\\
1&y_3(u)&\dots &y_3(u)^{d-1}\\
\vdots&\vdots&  &\vdots\\
1&y_d(u)&\dots&y_d(u)^{d-1}\\
0 & (y'_1(u)-y_2'(u))&\dots &  (y'_1(u)-y_2'(u))(d-1)y_2(u)^{d-2}
\end{array}\right),
\eeq
while matrices $(F_k)_{k\geq 0}$ are given by
\beq F_0=\left(\begin{array}{rrrrr}0&0 &0&\dots& 0\\
\vdots & \vdots & \vdots & &\vdots\\
0&0 &0&\dots &0\\
1&-1&0&\dots&0
\end{array}\right),
\eeq
\beq \forall\ (i,j)\in\llbracket 1,d\rrbracket^2\, , \, (F_1)_{i,j}= \left\{
    \begin{array}{ll}
       \delta_{i+1,j}-(B_0V_1)_{j,i+1}, & \mbox{if } i\leq d-1,\\
        (B_0V_2)_{j,2}-(B_0V_2)_{j,1}, & \mbox{if } i=d,
    \end{array}
\right.
\eeq
and for all $k\geq 2$,
\beq  \forall\ (i,j)\in\llbracket 1,d\rrbracket^2\, , \, (F_k)_{i,j}= \left\{
    \begin{array}{ll}
       -\underset{m=0}{\overset{k-1}{\sum}}(B_mV_{k-m})_{j,i+1}, & \mbox{if } i\leq d-1,\\
        \left(\underset{m=0}{\overset{k-1}{\sum}}B_mV_{k+1-m}\right)_{j,2}-\left(\underset{m=0}{\overset{k-1}{\sum}}B_mV_{k+1-m}\right)_{j,1}, & \mbox{if } i=d.
    \end{array}
\right.
\eeq
Note that $M$ is invertible since its determinant is 
\beq \label{detM} \det M=(-1)^{d}(y'_1(u)-y_2'(u))\left(\underset{j=3}{\overset{d}{\prod}}(y_2(u)-y_j(u))^2\right)\prod_{3\leq i<j\leq d}(y_i(u)-y_j(u)),
\eeq
which is non-vanishing due to Remark \ref{RemarkRegularBranchpoints}.

Remark also that $B_0$ and $V_0$ have a peculiar form and read
\bea
 B_0&=&\begin{pmatrix}(B_0)_{1,1}&\dots &(B_0)_{1,d}\\
-(B_0)_{1,1}&\dots &-(B_0)_{1,d}\\
0&\dots &0\\
\vdots& &\vdots\\
0&\dots&0
\end{pmatrix} ,\cr
V_0&=&\begin{pmatrix} 1&1& 1&\dots &1\\
y_2(u)& y_2(u)&y_3(u)&\dots & y_d(u)\\
\vdots&\vdots& \vdots& &\vdots\\
y_2(u)^{d-1}&y_2(u)^{d-1}&y_3(u)^{d-1}&\dots&y_d(u)^{d-1}
\end{pmatrix} .
\eea
In particular, we get that, for any $d \times d$ matrix $K$, we always have
\bea \label{properties-B0-V0}
0&=&(B_0KV_0)_{i,j}\,,\, \,\forall\, i\geq 3 \text{ and } j\geq 1,\cr
0&=&(B_0KV_0)_{1,j}+(B_0KV_0)_{2,j}\,,\, \,\forall \,j\geq 1,\cr
0&=&(B_0KV_0)_{i,1}-(B_0KV_0)_{i,2}\,,\,\,\forall \,i\geq 1 .
\eea

\subsubsection{Expansion of the compatibility conditions}

For any $\mathcal C_0 = \mathcal C_{p,k}$ with $(p,k)\in \mathcal{T}$, let us now write down the expansion of the compatibility conditions around $\lambda = u$. Remark that, by the induction hypothesis, the RHS of \eqref{compatibility-rec} behave at worst as $(\lambda-u)^{-s}$  while the LHS behaves as $(\lambda-u)^{-s-1}$.
Hence, after conjugation for diagonalizing the leading order of $\widetilde{L}(\lambda)$, the coefficients of $(\lambda - u)^{-s-\frac{1}{2}}$, $(\lambda - u)^{-s}$ and $(\lambda - u)^{-s+\frac{1}{2}}$ of the first line of \eqref{compatibility-rec} respectively read, for any $(k,l)\in \llbracket1,d\rrbracket^2$,
\beq\label{compatibility-rec-1}
\left\{
\begin{array}{l}
\left[B_0 a_{\mathcal{C}_0} V_0\right]_{k,l} (y_k(u)-y_l(u)) = 0, \cr
\left[B_0 a_{\mathcal{C}_0}  V_0\right]_{k,l} \left[y_k'(u)-y_l'(u)\right] + \left[B_0 a_{\mathcal{C}_0}  V_1 + B_1 a_{\mathcal{C}_0} V_0 \right]_{k,l} (y_k(u)-y_l(u)) = 0, \cr
\left[B_0 a_{\mathcal{C}_0}  V_0\right]_{k,l} \frac{y_k''(u)-y_l''(u)}{2} + \left[B_0 a_{\mathcal{C}_0}  V_1 + B_1 a_{\mathcal{C}_0}  V_0 \right]_{k,l} (y_k'(u)-y_l'(u)) + \cr
\quad + 
 \left[B_0 a_{\mathcal{C}_0}  V_2 + B_1 a_{\mathcal{C}_0}  V_1 + B_2 a_{\mathcal{C}_0} V_0+  B_0 b_{\mathcal{C}_0}  V_0 \right]_{k,l} (y_k(u)-y_l(u)) = \left[B_0 R_{\mathcal{C}_0}  V_0\right]_{k,l} \delta_{s,1}, \cr
\end{array}
\right.
\eeq
where
\beq \label{DefRi0}
R_{\mathcal{C}_0}  \coloneqq \lim_{\lambda \to u} (\lambda - u)^{\frac{1}{2}}  \left[
\underset{0\leq m\leq h-1}{\sum} \left[- \widetilde{A}_{\mathcal{C}_0} ^{(h)}(\lambda),L^{(h-m)}(\lambda)\right] + \frac{\partial \widetilde{A}_{\mathcal{C}_0} ^{(h-1)}(\lambda)}{\partial \lambda}  - \left(\hbar \, \text{ev}. \mathcal{I}_{\mathcal{C}_0}  \widetilde{L}^{\mathrm{symbol}}(\lambda) \right)^{(h)}
\right] .
\eeq
Remark that the RHS of the last line of \eqref{compatibility-rec-1} is vanishing except if we have a simple pole (i.e. if $s=1$) since we have proved that the induction hypothesis implies that $L^{(h)}(\lambda)=O(1)$ as $\lambda \to u$.

\medskip

Let us now write the second line of \eqref{compatibility-rec} as
\bea
 \frac{\partial \widetilde{A}_{\mathcal{C}_0}^{(h)}(\lambda)}{\partial \lambda}  &=& \left(\hbar \, \text{ev} . \mathcal{I}_{\mathcal{C}_0} \widetilde{L}^{\mathrm{symbol}}(\lambda) \right)^{(h+1)} + \underset{0\leq m\leq h-1}{\sum} \left[\widetilde{A}_{\mathcal{C}_0}^{(m)}(\lambda),L^{(h-m+1)}(\lambda)\right]  \cr
&&+ [\widetilde{A}_{\mathcal{C}_0}^{(h)}(\lambda),L^{(1)}(\lambda)] + [\widetilde{A}_{\mathcal{C}_0}^{(h+1)}(\lambda),L^{(0)}(\lambda)] .
\eea

After conjugation with the Vandermonde matrix, $L^{(0)}(\lambda)$ is diagonalized and the commutator involving $\widetilde{A}_{\mathcal{C}_0}^{(h+1)}$ has a vanishing diagonal. The LHS behaves as $(\lambda-u)^{-s-\frac{3}{2}}$, while the RHS behaves at most as $(\lambda-u)^{-s-\frac{1}{2}}$ on the diagonal. Thus, on the diagonal, the expansion around $\lambda =u$ reads, for any $k\in \llbracket 1, d \rrbracket$,
\beq\label{compatibility-rec-2}
\left\{
\begin{array}{l}
\left[B_0 a_{\mathcal{C}_0} V_0\right]_{k,k}  = 0, \cr
\left[B_0 a_{\mathcal{C}_0} V_1 + B_1 a_{\mathcal{C}_0} V_0 \right]_{k,k} = 0, \cr
-s \left[B_0 a_{\mathcal{C}_0} V_2 + B_2 a_{\mathcal{C}_0} V_0 + B_1 a_{\mathcal{C}_0} V_1 \right]_{k,k} -(s-1) \left[B_0 b_{\mathcal{C}_0} V_0 \right]_{k,k} = \left[B_0\, K_{\mathcal{C}_0}\, V_0\right]_{k,k}, \cr
\end{array}
\right.
\eeq
where 
\beq\label{def-Ji}
K_{\mathcal{C}_0}= \lim_{\lambda \to u} (\lambda - u)^{s}  \left( \left(\hbar \, \text{ev} .\mathcal{I}_{\mathcal{C}_0} \widetilde{L}^{\mathrm{symbol}}(\lambda) \right)^{(h+1)} + [\widetilde{A}_{\mathcal{C}_0}^{(h)}(\lambda),L^{(1)}(\lambda)] 
\right) .
\eeq

We shall come back to a more precise expression for this quantity later in the proof.

\subsubsection{Computation of $a_{\mathcal{C}_0}$ in terms of its $(1,d)$ entry}
Let $\mathcal{C}_0 = \mathcal{C}_{p,k}$ with $(p,k) \in \mathcal{T}$. The first line of \eqref{compatibility-rec-1}, together with the first line of \eqref{compatibility-rec-2}, implies that $\left[B_0 a_{\mathcal{C}_0} V_0\right]_{k,l}$ is vanishing for any $(k,l) \notin \{(1,2),(2,1)\}$. The evaluation of the second line of \eqref{compatibility-rec-1} for $(k,l) \in \{(1,2),(2,1)\}$ implies that $\left[B_0 a_{\mathcal{C}_0} V_0\right]_{1,2} =\left[B_0 a_{\mathcal{C}_0} V_0\right]_{2,1} = 0$. One thus has
\beq\label{ai0Eq1}
B_0 a_{\mathcal{C}_0}V_0 = 0.
\eeq

Using the second equation of \eqref{properties-B0-V0}, the same considerations applied to the second and third line of  \eqref{compatibility-rec-1} and the second line of  \eqref{compatibility-rec-2} lead to
\beq\label{ai0Eq2}
B_0 a_{\mathcal{C}_0} V_1 + B_1 a_{\mathcal{C}_0} V_0 = 0.
\eeq

Equations \eqref{ai0Eq1} and \eqref{ai0Eq2} imply the following lemma.

\begin{lemma}\label{Lemmaai0} For any $\mathcal{C}_0 = \mathcal{C}_{p,k}$ with $(p,k) \in \mathcal{T}$, the matrix $a_{\mathcal{C}_0}$ satisfies
\beq B_0a_{\mathcal{C}_0}=0 \,\, \text{ and  }\,\; a_{\mathcal{C}_0}V_0=0 .\eeq
\end{lemma}

\begin{proof}
Let us recall the non-overdetermined linear system determining $B_0$ in terms of $V_0$ and $V_1$. It reads
\beq \label{B0}
\left\{
    \begin{array}{ll}
      (B_0V_0)_{1,l}&=0\, \,,\, \forall \, l\geq 1 ,\\
			 (B_0V_0)_{k,l}&=0\,\,,\,\, \forall \, k\geq 3 \text{ and } l\geq 1 ,\\
			(B_0V_1)_{k,1}-(B_0V_1)_{k,2}&=\delta_{k,1}-\delta_{k,2} \,\,,\,\, \forall \, k\geq 1 .
    \end{array}
\right.
\eeq

Equations \eqref{ai0Eq1} and \eqref{ai0Eq2} show that the matrix $C \coloneqq B_0a_{\mathcal{C}_0}$ satisfies the system $CV_0=0$ and $(CV_1)_{k,1}-(CV_1)_{k,2}=0$, for all $k\geq 1$. This corresponds to the same linear system as the one determining $B_0$, except that the RHS is null. Hence, we get that $C=0 (M^{-1})^t=0$. Therefore, we get $C=0$, i.e. $B_0a_{\mathcal{C}_0}=0$. \\
In the same way, let us define $\td{C} \coloneqq a_{\mathcal{C}_0}V_0$. It satisfies from \eqref{ai0Eq1} and \eqref{ai0Eq2}
\beq B_0\td{C}=0 \text{ and } B_1 \td{C}=0 .\eeq
Let us introduce $\hat{C}=\td{C}^t(M^{-1})$. Since $B_0=F_0^{\,t}(M^{-1})^t$ and $B_1=F_1^{\,t}(M^{-1})^t$, we get that the previous equations are equivalent to
\beq \hat{C} F_0=0\text{ and }\hat{C} F_1=0 .\eeq
Since 
\beq F_0=\left(\begin{array}{rrrrr}0&0 &0&\dots& 0\\
\vdots & \vdots & \vdots & &\vdots\\
0&0 &0&\dots &0\\
1&-1&0&\dots&0
\end{array}\right) \text{ ,  }\;\;
 F_1=\left(\begin{array}{cccccc}
(F_1)_{1,1} &1-(F_1)_{1,1}&0&\dots&\dots& 0\\ 
(F_1)_{2,1}&-(F_1)_{2,1}&1&0&\dots&0\\
\vdots&\vdots&0&\ddots&&\vdots\\
&&&&&0\\
\vdots&\vdots&&\ddots&\ddots&1\\
(F_1)_{d,1}&-(F_1)_{d,1}&0&\dots&\dots&0\\
\end{array}\right) ,
\eeq
we get that $\hat{C}F_1=0$ implies that, for all $(i,j)\in\llbracket 1,d\rrbracket\times \llbracket 1,d-1\rrbracket$, $\hat{C}_{i,j}=0$. Indeed,
\bea \forall \,i\in\llbracket 1,d\rrbracket\,,\, j\in\llbracket 3,d\rrbracket\,:\, 0&=&(\hat{C}F_1)_{i,j}=\sum_{k=1}^d\hat{C}_{i,k}(F_1)_{k,j}=\hat{C}_{i,j-1},\cr
\forall \,i \in \llbracket 1,d\rrbracket\,:\,0&=&(\hat{C}F_1)_{i,1}+(\hat{C}F_1)_{i,2}=\sum_{k=1}^d\hat{C}_{i,k}\left( (F_1)_{k,1}+(F_1)_{k,2}\right)=\hat{C}_{i,1} .\cr
&&
\eea

Eventually, $\hat{C}F_0=0$ leads to 
\beq \forall\, i\in\llbracket 1,d\rrbracket\,:\, 0=(\hat{C}F_0)_{i,1}=\sum_{k=1}^d \hat{C}_{i,k}(F_0)_{k,1}=\hat{C}_{i,d},\eeq
so that $\hat{C}=0$, i.e.~$a_{\mathcal{C}_0} V_0=0$.
\end{proof}

\bigskip

Lemma~\ref{Lemmaai0} gives us the first two equations for $a_{\mathcal{C}_0}$. The coefficient of $(\lambda -u)^{-s+1}$ of the first line of \eqref{compatibility-rec} (which would be the fourth line of \eqref{compatibility-rec-1}, but we did not write it down), together with the last equation of \eqref{properties-B0-V0}, gives the third equation for $a_{\mathcal{C}_0}$. 
\beq
\left\{
    \begin{array}{ll}
      B_0 a_{\mathcal{C}_0}&=0 ,\\
			a_{\mathcal{C}_0}V_0&=0 ,\\
			(B_1a_{\mathcal{C}_0}V_1)_{i,1}-(B_1a_{\mathcal{C}_0}V_1)_{i,2}&=0 ,\,\, \forall\, i\geq 3 .
    \end{array}
\right.
\eeq

We define $\td{a}_{\mathcal{C}_0} \coloneqq (a_{\mathcal{C}_0})^{\, t} M^{-1}$ (i.e.~$a_{\mathcal{C}_0}=M^t (\td{a}_{\mathcal{C}_0})^t$). The last set of equations is equivalent to
\beq
\left\{
    \begin{array}{ll}
      \td{a}_{\mathcal{C}_0} \,F_0&=0 ,\\
			(\td{a}_{\mathcal{C}_0})^t V_0&=0 ,\\
			\big((V_1)^t\,\td{a}_{\mathcal{C}_0}\,F_1\big)_{1,i}-\big((V_1)^t\,\td{a}_{\mathcal{C}_0}\,F_1\big)_{2,i}&=0\,\,\,\,,\,\,  \forall\, i\geq 3 .
    \end{array}
\right.
\eeq

Since $(F_0)_{i,j}=\delta_{(i,j)=(d,1)}- \delta_{(i,j)=(d,2)}$, equation $\td{a}_{\mathcal{C}_0}\,F_0=0$ is equivalent to (only entries $(\td{a}_{\mathcal{C}_0}\,F_0)_{i,1}$, with $i\geq 1$, are sufficient)
\beq \big(\td{a}_{\mathcal{C}_0}\big)_{i,d}=0,\; \forall \,i\in \llbracket 1,d\rrbracket, \eeq
i.e.~the last column of $\td{a}_{\mathcal{C}_0}$ is vanishing. For $j\in\llbracket 1,d\rrbracket$, let us define the vector $\td{\mathbf{a}}_{{\mathcal{C}_0},j}\coloneqq \big((\td{a}_{\mathcal{C}_0})_{1,j},\dots,(\td{a}_{\mathcal{C}_0})_{d,j}\big)^t$ corresponding to the $j^{\text{th}}$ column of $\td{a}_{\mathcal{C}_0}$. For $j, k\geq 2$, equations $\big((\td{a}_{\mathcal{C}_0})^t V_0\big)_{j,k}=0$ and  $\big((V_1)^t\,\td{a}_{\mathcal{C}_0}\,F_1\big)_{1,j+1}-\big((V_1)^t\,\td{a}_{\mathcal{C}_0}\,F_1\big)_{2,j+1}=0$ are equivalent to
\bea \label{SysAj}
\left\{
    \begin{array}{ll}
     0& = \underset{r=1}{\overset{d}{\sum}} y_k(u)^{r-1}\big(\mathbf{\td{a}}_{\mathcal{C}_0,j}\big)_r  \,\,,\,\, \forall \,j, k\in\llbracket 2,d\rrbracket, \\
		 0&=(y'_1(u)-y'_2(u))\underset{r=2}{\overset{d}{\sum}}(r-1)y_2^{r-2}(u)\big(\mathbf{\td{a}}_{\mathcal{C}_0,j}\big)_r \,\,,\,\, \forall \, j\in\llbracket 2,d\rrbracket.
    \end{array}
\right.
\eea
These equations are equivalent to saying that, for all $j\geq 2$, $M\td{\mathbf{a}}_{\mathcal{C}_0,j}=0$, i.e.~$\td{\mathbf{a}}_{\mathcal{C}_0,j}=0$, since $M$ is invertible. Note that in order to obtain the last equation of \eqref{SysAj}, we have used $(F_1)_{s,j+1}=\delta_{s,j}$ for $j\geq 2$ and $s\leq d-1$, and the fact that $(\td{a}_{\mathcal{C}_0})_{r,d}=0$, for all $r\geq 1$, so that $(F_1)_{d,j+1}$ does not contribute.

In terms of the initial matrix $a_{\mathcal{C}_0}=M^t (\td{a}_{\mathcal{C}_0})^t$, this is equivalent to saying that, for all $(k,l)\in\llbracket 1,d\rrbracket^2$,
\beq \left(a_{\mathcal{C}_0}\right)_{k,l}=\sum_{n=1}^d M_{n,k}(\td{a}_{\mathcal{C}_0})_{l,n} =M_{1,k}(\td{a}_{\mathcal{C}_0})_{l,1}=y_2(u)^{k-1}(\td{a}_{\mathcal{C}_0})_{l,1} .
\eeq
In other words, it means that the $l^{\text{th}}$ column of $a_{\mathcal{C}_0}$ is $(\td{a}_{\mathcal{C}_0})_{l,1}\left(1,y_2(u),\dots,y_2(u)^{d-1}\right)^t$.

So far, only the first column of $\td{a}_{\mathcal{C}_0}$ remains. Using $\big((\td{a}_{\mathcal{C}_0})^t V_0\big)_{1,k}=0$ for all $k\geq 2$, we get that 
\beq \label{eqA1} \begin{pmatrix}1&y_2(u)&\dots &y_2(u)^{d-1}\\
\vdots&\vdots&&\vdots\\
1&y_d(u)&\dots&y_d(u)^{d-1}\end{pmatrix}\begin{pmatrix}(\td{a}_{\mathcal{C}_0})_{1,1}\\ \vdots\\ (\td{a}_{\mathcal{C}_0})_{d,1}\end{pmatrix}=0 .
\eeq
We may eventually use \eqref{eqA1} to get the following proposition.

\begin{proposition}[Expression of $a_{\mathcal{C}_0}$ in terms of its $(1,d)$ entry]\label{PropositionMatrixai0} For any $\mathcal{C}_0 \coloneqq \mathcal{C}_{p,k}$ with $(p,k) \in \mathcal{T}$, the matrix $a_{\mathcal{C}_0}$ is equal to
\beqq \left(a_{\mathcal{C}_0}\right)_{1,d}\! \begin{pmatrix}
(-1)^{d-1}E_{d-1}(y_2(u),\dots,y_d(u))&\!\!\!\!\!\!\!\! \cdots\!\!\!\!\!\!\!\! & -E_1(y_2(u),\dots,y_d(u))&\!\!\!\!   1\\
y_2(u)(-1)^{d-1}E_{d-1}(y_2(u),\dots,y_d(u))&\!\!\!\!\!\!\!\! \cdots\!\!\!\!\!\!\!\!  &-y_2(u) E_1(y_2(u),\dots,y_d(u))&\!\!\!\!   y_2(u)\\
\vdots&\!\!  \!\! &\vdots&\!\!\!\! \vdots\\
y_2(u)^{d-1}(-1)^{d-1}E_{d-1}(y_2(u),\dots,y_d(u))& \!\!\!\!\!\!\!\! \cdots\!\!\!\!\!\!\!\!  & -y_2^{d-1}(u) E_1(y_2(u),\dots,y_d(u))& \!\! y_2(u)^{d-1}\! 
\end{pmatrix},
\eeqq
where $\forall\, j\in \llbracket 1,d\rrbracket\,: \,E_j(\lambda_2,\dots,\lambda_d)=\underset{2\leq i_1<\dots<i_j\leq d}{\sum} \lambda_{i_1}\dots \lambda_{i_j}$. (In particular, $E_1(\lambda_2,\dots,\lambda_d)=\lambda_2+\dots+\lambda_d$ and $E_{d-1}(\lambda_2,\dots,\lambda_d)=\lambda_2\dots \lambda_d$).
\end{proposition}

Note that Proposition \ref{PropositionMatrixai0} implies that, for any $\mathcal{C}_0 \coloneqq \mathcal{C}_{p,k}$ with $(p,k) \in \mathcal{T}$,
\beq \left[a_{\mathcal{C}_0},L^{(0)}(\lambda)\right]=0\eeq
and that we only need one remaining independent equation to get the complete expression of $a_{\mathcal{C}_0}$. Moreover, we have, for all $(k,l)\in \llbracket 1,d\rrbracket^2$,
\begin{align}\label{eq-B1AV1} 
[B_1a_{\mathcal{C}_0}V_1]_{k,l}&=((F_1)^t(M^{-1})^ta_{\mathcal{C}_0}V_1)_{k,l}=((F_1)^t\,(\td{a}_{\mathcal{C}_0})^t\,V_1)_{k,l}\nonumber\\
&=\sum_{1\leq m,n\leq d}(F_1)_{m,k}(\td{a}_{\mathcal{C}_0})_{n,m}(V_1)_{n,l}\nonumber\\
&=\sum_{n=1}^d(F_1)_{1,k}(\td{a}_{\mathcal{C}_0})_{n,1}(V_1)_{n,l}\nonumber\\
&=\left(-\frac{y'_2(u)}{y'_1(u)-y'_2(u)}\delta_{k,1}+\frac{y'_1(u)}{y'_1(u)-y'_2(u)}\delta_{k,2}\right)\left(\sum_{n=1}^d (\td{a}_{\mathcal{C}_0})_{n,1}(V_1)_{n,l}\right)\nonumber\\
&=\left(-\frac{y'_2(u)}{y'_1(u)-y'_2(u)}\delta_{k,1}+\frac{y'_1(u)}{y'_1(u)-y'_2(u)}\delta_{k,2}\right)\left(\sum_{n=2}^d(n-1)y'_l(u) \, y_l(u)^{n-2}(\td{a}_{\mathcal{C}_0})_{n,1}\right).
\end{align}

\subsubsection{Proving that $\left(a_{\mathcal{C}_0}\right)_{1,d} = 0$.}

Let us finally prove that the last remaining unknown coefficient of $a_{\mathcal{C}_0}$ is vanishing for any $\mathcal{C}_0\in \mathcal{T}$. For this purpose, we shall consider two different cases depending on the degree $s$ of the pole at $\lambda \to u$ of $\widetilde{A}_{\mathcal{C}_0}^{(h)}(\lambda)$.

\bigskip

\noindent{\bf{Higher order pole: $s \geq 2$}}

\medskip

If $s \geq 2$, the RHS of the last line of \eqref{compatibility-rec-1}, as well as the RHS of the following order in the expansion around $\lambda = u$, vanish. This allows to conclude that,
\beq\label{condition-sum}
\forall\, k \in \llbracket 2,d \rrbracket, \, \forall \, l\in \llbracket 1,d\rrbracket \, : \; \left[B_1 a_{\mathcal{C}_0} V_1 + B_0 b_{\mathcal{C}_0}V_0 \right]_{k,l} = 0.
\eeq
Using the last line of \eqref{properties-B0-V0} and subtracting cases $l=1$ and $l=2$ in the previous equation we get that
\beq \forall\, k \in \llbracket 2,d \rrbracket : \; \left[B_1 a_{\mathcal{C}_0} V_1\right]_{k,1}-\left[B_1 a_{\mathcal{C}_0} V_1\right]_{k,2}=0.\eeq
In particular for $k=2$, we end up with
\beq \left[B_1 a_{\mathcal{C}_0} V_1\right]_{2,1}-\left[B_1 a_{\mathcal{C}_0} V_1\right]_{2,2}=0,\eeq
which, using \eqref{eq-B1AV1}, and the fact that $y_1(u)=y_2(u)$ implies
\beq y_1'(u)\left(\sum_{n=2}^d(n-1)\,y_2(u)^{n-2}(\td{a}_{\mathcal{C}_0})_{n,1}\right)=0.\eeq
Together with \eqref{eqA1}, we get the linear system
\beq  \begin{pmatrix}1&y_2(u)&y_2(u)^2 &\dots&y_2(u)^{d-1}\\
\vdots&\vdots&\vdots&&\vdots\\
1&y_d(u)&y_d(u)^2&\dots&y_d(u)^{d-1}\\
0& y'_1(u)& 2y'_1(u) y_2(u)&\dots & (d-1)y'_1(u)  y_2(u)^{d-2}
\end{pmatrix}\begin{pmatrix}(\td{a}_{\mathcal{C}_0})_{1,1}\\ \vdots\\\vdots\\ (\td{a}_{\mathcal{C}_0})_{d,1}\end{pmatrix}=0,
\eeq
whose determinant is similar to the one of $M$ and reads
\beq
(-1)^{d-1} y_1'(u) \left(\underset{j=3}{\overset{d}{\prod}}(y_2(u)-y_j(u))^2\right)\prod_{3\leq i<j\leq d}(y_i(u)-y_j(u)) .
\eeq
From Remark \ref{RemarkRegularBranchpoints}, this determinant is not vanishing so that the linear system has a unique trivial solution
\beq
 \forall \, k\in \llbracket 1,d\rrbracket\,:\, (\td{a}_{\mathcal{C}_0})_{k,1} =0.
 \eeq
Thus, from Proposition~\ref{PropositionMatrixai0}, we end up with $\td{a}_{\mathcal{C}_0}=0$, i.e.~$a_{\mathcal{C}_0}=0$ for all $\mathcal{C}_0 = \mathcal{C}_{p,k}$ with $(p,k) \in \mathcal{T}$, which is a contradiction with \eqref{poledef}.

\bigskip

\noindent{\bf{Simple pole $s=1$}}

\medskip

The derivation of the same result of $s = 1$ requires to be more precise since the RHS of the last line of \eqref{compatibility-rec-1} does not vanish.

\medskip

The last line of \eqref{compatibility-rec-2} reads
\beq \label{Eqai01d}
 - \left[B_1 a_{\mathcal{C}_0}V_1 \right]_{k,k} 
 = \left[B_0 \,K_{\mathcal{C}_0}\, V_0\right]_{k,k},\; \forall\, k\in \llbracket 1,d\rrbracket.
\eeq
Let us now prove that the RHS is vanishing using our knowledge on $K_{\mathcal{C}_0}$. We remind the reader of its definition from \eqref{def-Ji}:
\beq K_{\mathcal{C}_0}= \lim_{\lambda \to u} (\lambda - u) \left[ \left(\hbar \, \text{ev} .\mathcal{I}_{\mathcal{C}_0} \widetilde{L}^{\mathrm{symbol}}(\lambda) \right)^{(h+1)} + [\widetilde{A}_{\mathcal{C}_0}^{(h)}(\lambda),L^{(1)}(\lambda)] 
\right].\eeq

Since, for any $\mathcal{C}_0$, $B_0 a_{\mathcal{C}_0} = a_{\mathcal{C}_0} V_0 = 0$, the commutator $[\widetilde{A}_{\mathcal{C}_0}^{(h)},L^{(1)}]$ satisfies 
 \beq
V^{-1}(\lambda) \left[\widetilde{A}_{\mathcal{C}_0}^{(h)},L^{(1)}\right] V(\lambda) = \frac{B_0 a_{\mathcal{C}_0} R_1 V_0 - B_0 R_1 a_{\mathcal{C}_0} V_0}{(\lambda - u)^{\frac{3}{2}}} + O\left((\lambda-u)^{-1}\right) =  O\left((\lambda-u)^{-1}\right),
\eeq
where we have denoted $L^{(1)}(\lambda)\overset{\lambda\to u}{=}R_1+ o(1)$. 

Thus, for any $\mathcal{C}_0 \coloneqq \mathcal{C}_{p,k}$ with $(p.k) \in \mathcal{T}$, $\left[\widetilde{A}_{\mathcal{C}_0}^{(h)},L^{(1)}\right]$  does not provide any contribution to $K_{\mathcal{C}_0}$. 

Following the discussion of Section~\ref{subsection-action-trans} , the only contribution from $[\hbar \text{ev} .\mathcal{I}_{\mathcal{C}_0} \widetilde{L}^{\mathrm{symbol}}(\lambda)]_{l,m}$ at order $\hbar^{h+1}$ in the trans-series expansion is
\beq
\sum_{P \in \mathcal{P}} \sum_{K\in \mathbb{N}} \xi_P^{-K}(\lambda) \sum_{j=1}^g \, \text{ev} . \mathcal{I}_{\mathcal{C}_0}\left[\phi_j\right] \, \left. \frac{\partial \left[\widetilde{\Delta}_{P,K}^{(h+1)}(\lambda)\right]_{l,m}}{\partial \phi_j} \right|_{\phi_j = \frac{1}{2 \pi i } \oint_{\mathcal{B}_j} \om_{0,1}}\!\!\!\!\!\!\!\!\!\!\!\!\!\!\!\!\!\!.
\eeq

Note that the sum over $K\in \mathbb{N}$ is finite since $\{\widetilde{\Delta}_{P,K}(\lambda)\}_{K\geq 0}$ contains only a finite number of non-zero elements. 

Using  \eqref{EqDelta}, the leading order of $[\hbar \text{ev} . \mathcal{I}_{\mathcal{C}_0} \widetilde{L}^{\mathrm{symbol}}(\lambda)]_{i,m}^{(h+1)}$ in the expansion around $\lambda = u$ reads
\begin{align}&
[\hbar \text{ev} . \mathcal{I}_{\mathcal{C}_0} \widetilde{L}^{\mathrm{symbol}}(\lambda)]_{i,m}^{(h+1)}=\cr
&\sum_{P \in \mathcal{P}} \sum_{K\in \mathbb{N}} \xi_P^{-K}(u) \sum_{j=1}^g \mathcal{I}_{\mathcal{C}_0}^{(0)}\left[\phi_j\right] \,\sum_{r=1}^d \left[G^{-1}(u)\right]_{i,r} \sum_{l=1}^d \frac{[G]_{1,l}}{(\lambda - u)}\left. \frac{\partial \left[a_{P,K,r-1}\right]_{l,m}}{\partial \phi_j} \right|_{\phi_j = \frac{1}{2 \pi i } \oint_{\mathcal{B}_j} \om_{0,1}} + O\left(1\right),\cr
\end{align}
where $a_{P,k,l}$ is defined as the leading order of $\widetilde{A}_{P,K,l}^{(h)}(\lambda)$ when $\lambda \to u$, i.e. 
\beq
\widetilde{A}_{P,K,l}^{(h)}(\lambda) = \frac{a_{P,K,l}}{\lambda-u} + O(1),
\eeq
as $\lambda \to u$. Since we proved the part of the proposition $\mathcal{P}_h$ corresponding to $m\geq 2$, the matrices $a_{P,K,l}$ are linear combinations of matrices $a_{\mathcal{C}}$ with coefficients independent of $\phi_j$,
\beq
\forall (P,K,l) \in \mathcal{P} \times \mathbb{N} \times \llbracket 1, d \rrbracket \, , \; a_{P,K,l} = \sum_{(p_0,k_0) \in \mathcal{T}} \alpha_{P,K,l}^{p_0,k_0} a_{\mathcal{C}_{p_0.k_0}},
\eeq
where $\alpha_{P,K,l}^{p_0,k_0} \in \mathbb{C}$ is independent of $\phi_j$.

Hence, for any $\mathcal{C}_0 \coloneqq \mathcal{C}_{p_0,k_0}$ with $(p_0,k_0) \in \mathcal{T}$ and $k\in \llbracket 1,d\rrbracket$, the RHS of equation \eqref{Eqai01d} reduces to
\begin{align}\label{B0KV0}
\left[B_0 \,K_{\mathcal{C}_0} \,V_0\right]_{k,k}
& = \sum_{P \in \mathcal{P}} \sum_{K\in \mathbb{N}} \xi_P^{-K}(u) \sum_{j=1}^g \mathcal{I}_{\mathcal{C}_0}^{(0)}\left[\phi_j\right] \,\sum_{r=1}^d \left[B_0 G^{-1}(u)\right]_{k,r}\nonumber \\
& \sum_{l=1}^d \frac{[G]_{1,l}}{(\lambda - u)}\left. \frac{\partial \left[a_{P,K,r-1}\right]_{l,m} \left[V_0\right]_{m,k}}{\partial \phi_j} \right|_{\phi_j = \frac{1}{2 \pi i } \oint_{\mathcal{B}_j} \om_{0,1}}
\nonumber \\
& = \sum_{P \in \mathcal{P}} \sum_{K\in \mathbb{N}} \xi_P^{-K}(u) \sum_{j=1}^g \mathcal{I}_{\mathcal{C}_0}^{(0)}\left[\phi_j\right] \,\sum_{r=1}^d \left[B_0 G^{-1}(u)\right]_{k,r}\nonumber \\
& \sum_{l=1}^d \frac{[G]_{1,l}}{(\lambda - u)}\left. \frac{\partial \left( \left[a_{P,K,r-1}\right]_{l,m} y_k(u)^{m-1}\right)}{\partial \phi_j} \right|_{\phi_j = \frac{1}{2 \pi i } \oint_{\mathcal{B}_j} \om_{0,1}}
\!\!\!\!\!\!  .
\end{align}

Since $a_{P,K,l}$ is a linear combination of matrices $a_{\mathcal{C}_{p_0,k_0}}$ for $(P,K,l) \in \mathcal{P} \times \mathbb{N} \times \llbracket 1, d \rrbracket$ the expression of 
 of $\left[a_{P,K,l}\right]_{1\leq l,m\leq d}$ in terms of $\left[a_{P,K,l}\right]_{1,d}$ is the same as the expression of  
 $\left((a_{\mathcal{C}_0})_{l,m}\right)_{1\leq l,m\leq d}$ in terms of $(a_{\mathcal{C}_0})_{1,d}$ given by Proposition~\ref{PropositionMatrixai0} and one has
\beq
\forall\, (P,K,r) \in \mathcal{P}\times\mathbb{N}\times \llbracket 1,d\rrbracket\,, \forall \,(m, n) \in  \llbracket 1,d\rrbracket^2\, :\, \underset{m=1}{\overset{d}{\sum}} (a_{P,K,r-1})_{l,m}  y_k(u)^{m-1} = 0.
\eeq
Thus, equation \eqref{B0KV0} simplifies into
\beq
\forall \,k\in \llbracket1 , d\rrbracket \, : \;  \left[B_0 \,K_{\mathcal{C}_0}\,V_0\right]_{k,k}  = 0,
 \eeq
implying from \eqref{Eqai01d} the following lemma.

\begin{lemma}\label{condition-s} For all $\mathcal{C}_0 = \mathcal{C}_{p_0,k_0}$ with $(p_0,k_0) \in \mathcal{T}$, the matrix $a_{\mathcal{C}_0}$ satisfies
\beq\label{cond-simple-pole}
\forall\, k\in \llbracket 1,d\rrbracket\,:\, \, \left[B_1 a_{\mathcal{C}_0} V_1 \right]_{k,k} = 0.
\eeq
\end{lemma}

Plugging \eqref{eq-B1AV1} into \eqref{cond-simple-pole} and using the fact (see Remark~\ref{RemarkRegularBranchpoints}) that $y'_1(u) \neq y'_2(u)$ and $y'_1(u) y'_2(u) \neq 0$, one gets, for $k=2$,
\beq
\left(\sum_{n=2}^d(n-1)y'_2(u) \, y_2(u)^{n-2}(\td{a}_{\mathcal{C}_0})_{n,1}\right) = 0.
\eeq
As in the previous case, together with \eqref{eqA1}, this leads to the linear system
\beq  \begin{pmatrix}1&y_2(u)&y_2(u)^2 &\dots&y_2(u)^{d-1}\\
\vdots&\vdots&\vdots&&\vdots\\
1&y_d(u)&y_d(u)^2&\dots&y_d(u)^{d-1}\\
0& y'_2(u)& 2y'_2(u) y_2(u)&\dots & (d-1)y'_2(u)  y_2(u)^{d-2}
\end{pmatrix}\begin{pmatrix}(\td{a}_{\mathcal{C}_0})_{1,1}\\ \vdots\\\vdots\\ (\td{a}_{\mathcal{C}_0})_{d,1}\end{pmatrix}=0,
\eeq
whose determinant reads
\beq
(-1)^{d-1} y_2'(u) \left(\underset{j=3}{\overset{d}{\prod}}(y_2(u)-y_j(u))^2\right)\prod_{3\leq i<j\leq d}(y_i(u)-y_j(u)) .
\eeq
From Remark \ref{RemarkRegularBranchpoints}, this determinant is not vanishing so that the linear system has a unique trivial solution
\beq
 \forall \, k\in \llbracket 1,d\rrbracket\,:\, (\td{a}_{\mathcal{C}_0})_{k,1} =0.
 \eeq
Thus, from Proposition~\ref{PropositionMatrixai0}, we end up with $\td{a}_{\mathcal{C}_0}=0$, i.e.~$a_{\mathcal{C}_0}=0$ for any  $\mathcal{C}_0 = \mathcal{C}_{p_0,k_0}$ with $(p_0,k_0) \in \mathcal{T}$, which is a contradiction with \eqref{poledef}.

\section{Detailed computations for the $\text{Gl}_3$ example}
\label{app-example-gl3}

In this appendix, we present the detailed computation of the example presented in Section~\ref{SectionGL3Example}.
 
\subsection{Classical spectral curve}

Let us consider  a three-sheeted cover of the sphere defined by an equation of the form
\beq
y^3 - (P_{\infty,1}^{(1)} \lambda + P_{\infty,0}^{(1)} ) y^2 + ( P_{\infty,2}^{(2)} \lambda^2 + P_{\infty,1}^{(2)} \lambda + P_{\infty,0}^{(2)}) y - P_{\infty,3}^{(3)} \lambda^3 - P_{\infty,2}^{(3)} \lambda^2 - P_{\infty,1}^{(3)} \lambda - P_{\infty,0}^{(3)} = 0.
\eeq
where the coefficients $P_{\infty,i}^{(j)}$ are generic in such a way that the curve has genus 1 and that there are three distinct points $(\infty^{(1)},\infty^{(2)},\infty^{(3)})$ in the fiber $x^{-1}(\infty)$ above infinity. This corresponds to the case $N=0$ and $r_{\infty^{(j)}} = 3$ for $j=1,2,3$. We recall that the general notations of the article correspond to:
\bea P_1(\lambda)&=&P_{\infty,1}^{(1)} \lambda + P_{\infty,0}^{(1)}\,,\cr
P_2(\lambda)&=&P_{\infty,2}^{(2)} \lambda^2 + P_{\infty,1}^{(2)} \lambda + P_{\infty,0}^{(2)}\,,\cr
P_3(\lambda)&=&P_{\infty,3}^{(3)} \lambda^3 + P_{\infty,2}^{(3)} \lambda^2 + P_{\infty,1}^{(3)} \lambda + P_{\infty,0}^{(3)}\,.
\eea

Around the three points $(\infty^{(1)},\infty^{(2)},\infty^{(3)})$, the function $y(z)$ admits the expansion
\beq
\forall\, i \in \{1,2,3\} \, , \; y(z) = - t_{i,2} x(z) -  t_{i,1} -  t_{i,0} x(z)^{-1} + O\left(x(z)^{-2}\right),  \text{ as }  z \to \infty^{(i)},
\eeq
where we used the notation $t_{i,j}\coloneqq t_{\infty^{(i)},j}$ in order to simplify the reading.
These coefficients satisfy
\beq \label{Pinfty1j}
P_{\infty,1}^{(1)} =  -\sum_{i=1}^3  t_{i,2} \quad , \quad P_{\infty,0}^{(1)} =  - \sum_{i=1}^3  t_{i,1} \quad , \quad 0 = \sum_{i=1}^3 t_{i,0},
\eeq
as well as 
\bea\label{Pinfty2j}
P_{\infty,2}^{(2)} &=&  \sum_{1\leq i<j\leq 3} t_{i,2} t_{j,2}=t_{1,2}t_{2,2}+t_{1,2}t_{3,2}+t_{2,2}t_{3,2} \cr
P_{\infty,1}^{(2)} &=&  \sum_{i=1}^3\sum_{j\neq i} t_{i,1} t_{j,2} =t_{1,1}t_{2,2}+t_{1,1}t_{3,2}+t_{2,1}t_{1,2}+t_{2,1}t_{3,2}+t_{3,1}t_{1,2}+t_{3,1}t_{2,2}\cr
P_{\infty,0}^{(2)} &=&   \sum_{i=1}^3 \sum_{j\neq i} t_{i,0} t_{j,2} + \sum_{1\leq i<j\leq 3} t_{i,1} t_{j,1} \cr
&=&t_{1,0}t_{2,2}+t_{1,0}t_{3,2}+t_{2,0}t_{1,2}+t_{2,0}t_{3,2}+t_{3,0}t_{1,2}+t_{3,0}t_{2,2} +t_{1,1}t_{2,1}+t_{1,1}t_{3,1}+t_{2,1}t_{3,1}\cr
&&
\eea
\bea\label{Pinfty3j}
P_{\infty,3}^{(3)} &=& - t_{1,2} t_{2,2} t_{3,2}\cr
P_{\infty,2}^{(3)} &=& - t_{1,1} t_{2,2} t_{3,2} - t_{1,2} t_{2,1} t_{3,2} - t_{1,2} t_{2,2} t_{3,1}\cr
P_{\infty,1}^{(3)} &= & - \sum_{k_1+k_2+k_3 = 4} t_{1,k_1} t_{2,k_2} t_{3,k_3}\cr
&=&-t_{1,0}t_{2,2}t_{3,2}-t_{1,1}t_{2,1}t_{3,2}-t_{1,1}t_{2,2}t_{3,1} - t_{1,2}t_{2,2}t_{3,0}-t_{1,2}t_{2,1}t_{3,1}-t_{1,2}t_{2,0}t_{3,2}\cr
P_{\infty,0}^{(3)} &=&  - \sum_{k_1+k_2+k_3 = 3} t_{1,k_1} t_{2,k_2} t_{3,k_3}\cr
&=&-t_{1,0}t_{2,1}t_{3,2}-t_{1,0}t_{2,2}t_{3,1}-t_{1,1}t_{2,1}t_{3,1}-t_{1,1}t_{2,2}t_{3,0}\cr
&-&t_{1,1}t_{2,0}t_{3,2}-t_{1,2}t_{2,1}t_{3,0}-t_{1,2}t_{2,0}t_{3,1}.\cr
&&
\eea

\subsection{KZ equations}
Let us now write the KZ equations satisfied by the wave functions. They read
\beq
\left\{
\begin{array}{l}
\hbar \frac{\partial \psi(z)}{\partial x(z)} + \psi_1(z) =  x(z) \mathcal{L}_{\infty,1}^{(0)} \psi(z) + \mathcal{L}_{\infty,0}^{(0)} \psi(z) \cr
\hbar \frac{\partial \psi_1(z)}{\partial x(z)} + \psi_2(z) = x(z)^2 \mathcal{L}_{\infty,2}^{(1)} \psi(z) + x(z) \mathcal{L}_{\infty,1}^{(1)} \psi(z) + \mathcal{L}_{\infty,0}^{(1)} \psi(z)\cr
\hbar \frac{\partial \psi_2(z)}{\partial x(z)}  = x(z)^3 \mathcal{L}_{\infty,3}^{(2)} \psi(z)+  x(z)^2 \mathcal{L}_{\infty,2}^{(2)} \psi(z) + x(z) \mathcal{L}_{\infty,1}^{(2)} \psi(z) + \mathcal{L}_{\infty,0}^{(2)} \psi(z). \cr
\end{array}
\right.
\eeq

Hence, the KZ equations read
\beq
\left\{
\begin{array}{l}
\hbar \frac{\partial \psi(z)}{\partial x(z)} + \psi_1(z) =  P_1(\lambda) \psi(z)\cr
\hbar \frac{\partial \psi_1(z)}{\partial x(z)} + \psi_2(z) = \left[P_2(\lambda) - \hbar t_{2,2} - \hbar t_{3,2} \right] \psi(z)\cr
\hbar \frac{\partial \psi_2(z)}{\partial x(z)}  =  P_3(\lambda) \psi(z) +  \hbar \mathcal{L}_{KZ} \psi(z), \cr
\end{array}
\right. 
\eeq
where
\beq
\mathcal{L}_{KZ}(\lambda) \coloneqq  \hbar \left[t_{1,2} t_{2,2} \mathcal{I}_{\mathcal{C}_{\infty^{(3)},2}} + t_{1,2} t_{3,2} \mathcal{I}_{\mathcal{C}_{\infty^{(2)},2}} + t_{3,2} t_{2,2} \mathcal{I}_{\mathcal{C}_{\infty^{(1)},2}} \right] + \hbar t_{2,2} t_{3,2} \lambda + t_{2,1} t_{3,2} + t_{2,2} t_{3,1}.
\eeq

\subsection{Expansion around $\infty$ and apparent singularities}

In this example, we consider the divisor $D = [z]-[\infty^{(1)}]$. The perturbative wave functions have the following asymptotics as $\lambda \to \infty$.
\beq
\psi^{reg}\left([z^{(i)}(\lambda)]-[\infty^{(1)}],\hbar\right) = 
\left\{
\begin{array}{lcl}
\exp \left[\hbar^{-1}(V_{\infty^{(1)}}(z^{(1)}(\lambda) +O(1))\right] (C_1 + O(\lambda^{-1})), & \text{if} & i=1, \cr
\lambda^{-1} \, \exp \left[\hbar^{-1}(V_{\infty^{(i)}}(z^{(i)}(\lambda) +O(1))\right] (C_i + O(\lambda^{-1})), & \text{if} & i\neq1. \cr
\end{array}
\right. 
\eeq

This allows to compute the Wronskian of the system 
\bea
W(\lambda,\hbar) &\coloneqq& \psi_{\mathrm{NP}}(z^{(1)}(\lambda)) \hbar \frac{\partial  \psi_{\mathrm{NP}}(z^{(2)}(\lambda)) }{\partial \lambda} \hbar^2 \frac{\partial^2  \psi_{\mathrm{NP}}(z^{(3)}(\lambda)) }{\partial \lambda^2} \cr
&& + \psi_{\mathrm{NP}}(z^{(2)}(\lambda)) \hbar \frac{\partial  \psi_{\mathrm{NP}}(z^{(3)}(\lambda)) }{\partial \lambda} \hbar^2 \frac{\partial^2  \psi_{\mathrm{NP}}(z^{(1)}(\lambda)) }{\partial \lambda^2} \cr
&& + \psi_{\mathrm{NP}}(z^{(3)}(\lambda)) \hbar \frac{\partial  \psi_{\mathrm{NP}}(z^{(1)}(\lambda)) }{\partial \lambda} \hbar^2 \frac{\partial^2  \psi_{\mathrm{NP}}(z^{(2)}(\lambda)) }{\partial \lambda^2} \cr
&& - \psi_{\mathrm{NP}}(z^{(1)}(\lambda)) \hbar \frac{\partial  \psi_{\mathrm{NP}}(z^{(3)}(\lambda)) }{\partial \lambda} \hbar^2 \frac{\partial^2  \psi_{\mathrm{NP}}(z^{(2)}(\lambda)) }{\partial \lambda^2} \cr
&& - \psi_{\mathrm{NP}}(z^{(2)}(\lambda)) \hbar \frac{\partial  \psi_{\mathrm{NP}}(z^{(1)}(\lambda)) }{\partial \lambda} \hbar^2 \frac{\partial^2  \psi_{\mathrm{NP}}(z^{(3)}(\lambda)) }{\partial \lambda^2} \cr
&& - \psi_{\mathrm{NP}}(z^{(3)}(\lambda)) \hbar \frac{\partial  \psi_{\mathrm{NP}}(z^{(2)}(\lambda)) }{\partial \lambda} \hbar^2 \frac{\partial^2  \psi_{\mathrm{NP}}(z^{(1)}(\lambda)) }{\partial \lambda^2}, \cr
&&
\eea
which takes the form
\beq\label{WronskianGL3}
W(\lambda,\hbar) = 
\kappa \exp\left( \hbar^{-1} \int_{0}^\lambda P_1(\lambda) d\lambda\right)  \left(\lambda- q(\hbar)\right).
\eeq

\subsection{Quantum curve}

Let us now compute the quantum curve of this system. For this purpose, we should compute the entries $\left[C_{\infty}(\lambda,\hbar)\right]_{2,i}$ with $i\in\{1,2,3\}$, of the matrix
\small{\beq
\begin{array}{l}
C_\infty(\lambda,\hbar) \coloneqq \cr \!\!\frac{1}{\hbar} \! 
\left[
 \begin{array}{ccc}
\!\!\psi_{\mathrm{NP}}(z^{(1)}(\lambda),\hbar) & \!\!\psi_{\mathrm{NP}}(z^{(2)}(\lambda),\hbar) &  \!\! \psi_{\mathrm{NP}}(z^{(3)}(\lambda),\hbar)\!\cr
\!\!\text{ev} . \left[{\mathcal{L}_{KZ} }\psi_{\mathrm{NP}}^{\mathrm{symbol}}(z^{(1)}(\lambda),\hbar) \right] & \!\!\text{ev}.{\mathcal{L}_{KZ} }\left[\psi_{\mathrm{NP}}(z^{(2)}(\lambda),\hbar)\right] & \!\!\text{ev}.\left[{\mathcal{L}_{KZ} }\psi_{\mathrm{NP}}(z^{(3)}(\lambda),\hbar)\right]\! \cr
\!\!\text{ev}. \left[{\mathcal{L}^2 }\psi_{\mathrm{NP}}^{\mathrm{symbol}}(z^{(1)}(\lambda),\hbar)  \right]& \!\!\text{ev}. \left[{\mathcal{L}^2 }\psi_{\mathrm{NP}}^{\mathrm{symbol}}(z^{(2)}(\lambda),\hbar) \right]& \!\!\text{ev}. \left[{\mathcal{L}^2 }\psi_{\mathrm{NP}}^{\mathrm{symbol}}(z^{(3)}(\lambda),\hbar)\right]\! \cr
\end{array}
\right] 
\!\!\cdot 
\!\left({\Psi}_{\mathrm{NP}}(\lambda,\hbar)\right)^{-1} \cr
\end{array}
\eeq}\normalsize{as} a rational function of $\lambda$ with poles at $\infty$ and at the zero $q(\hbar)$ of $W(\lambda,\hbar)$. Since it is a rational function in $\lambda$, one only needs its expansion around its poles and it takes the form of $\frac{\text{Pol}(\lambda)}{(\lambda-q)}$ where $\text{Pol}(\lambda)$ is a polynomial which we can compute through the asymptotic expansion of 
$C_\infty(\lambda,\hbar) $ around infinity. To compute these polynomials, we shall use the expression of $\left({\Psi}_{\mathrm{NP}}(\lambda,\hbar)\right)^{-1}$ in terms of the matrix of ${\Psi}_{\mathrm{NP}}(\lambda,\hbar)$ and its determinant $W(\lambda,\hbar)$. 

One obtains expressions of the form 
\beq
\left[C_\infty(\lambda,\hbar)\right]_{2,1} = \hbar \left(P_{\infty,2}^{(2)}+ t_{2,2} t_{3,2}\right) \lambda - H + \frac{p_1}{\lambda-q} ,
\eeq
\beq
\left[C_\infty(\lambda,\hbar)\right]_{2,2} = \hbar \left(P_{\infty,1}^{(1)}+t_{2,2}+t_{3,2}\right) + \frac{p_2}{\lambda- q}
\eeq
and
\beq
\left[C_\infty(\lambda,\hbar)\right]_{2,3} = \frac{\hbar}{\lambda - q} .
\eeq

This leads to a Lax matrix of the form given in the next section.

\subsection{Lax pair and Hamiltonian system}
From the previous results, the Lax pair corresponding to this example is given by $L(\lambda,\hbar) =$

\beq 
\begin{pmatrix} 0&1&0\\0&0&1\\
P_3(\lambda)-\hbar P_2'(\lambda)+  \hbar (P_{\infty,2}^{(2)}+t_{2,2}t_{3,2}) \lambda - H + \frac{p_1}{\lambda-q} &
-P_2(\lambda)-\hbar t_{1,2}  + \frac{p_2}{\lambda- q}&
P_1(\lambda) +\frac{\hbar}{\lambda-q} 
\end{pmatrix},
\eeq
with the associated auxiliary matrix
\beq
A_{KZ}(\lambda,\hbar) = 
\begin{pmatrix}  (P_{\infty,2}^{(2)}+t_{2,2}t_{3,2})\lambda - \frac{H}{\hbar} + \frac{p_1}{\hbar(\lambda-q)} &  \frac{p_2}{\hbar(\lambda- q)}+t_{1,2}&\frac{1}{\lambda-q}\\A_{2,1}&A_{2,2}&A_{2,3}\\
A_{3,1}&A_{3,2}&A_{3,3}
\end{pmatrix} .
\eeq

The coefficients of the Lax pair correspond to the normalization of the non-perturbative wave functions at $\lambda=\infty$,
\bea \label{PSIGL3}\Psi_{1,1}(\lambda,\hbar)&=&\exp\Bigg(-\frac{t_{1,2}}{2\hbar}\lambda^2-\frac{t_{1,1}}{\hbar}\lambda -\frac{t_{1,0}}{\hbar}\ln  \lambda +S_{1,0}+\sum_{k=2}^{\infty} \frac{S_{1,k}}{(k-1)\lambda^{k-1}}
\Bigg), \cr
\Psi_{1,2}(\lambda,\hbar)&=&\exp\Bigg(-\frac{t_{2,2}}{2\hbar}\lambda^2-\frac{t_{2,1}}{\hbar}\lambda -\frac{t_{2,0}}{\hbar}\ln \lambda-\ln\lambda  +S_{2,0}+\sum_{k=2}^{\infty} \frac{S_{2,k}}{(k-1)\lambda^{k-1}}
\Bigg), \cr
\Psi_{1,3}(\lambda,\hbar)&=&\exp\Bigg(-\frac{t_{3,2}}{2\hbar}\lambda^2-\frac{t_{3,1}}{\hbar}\lambda -\frac{t_{3,0}}{\hbar}\ln \lambda-\ln\lambda  +S_{3,0}+\sum_{k=2}^{\infty} \frac{S_{3,k}}{(k-1)\lambda^{k-1}}\Bigg),\cr
&&
\eea
from which one may recover the general form of $L(\lambda,\hbar)$ and $A(\lambda,\hbar)$ at infinity using $L(\lambda,\hbar)=\hbar (\partial_\lambda \Psi ) \Psi^{-1}$, $A(\lambda,\hbar)=(\mathcal{L}[\Psi]) \Psi^{-1}$ and the form of the Wronskian \eqref{WronskianGL3} with $\kappa= -(t_{3,2}-t_{1,2})(t_{3,2}-t_{2,2})(t_{2,2}-t_{1,2})e^{S_{1,0}+S_{2,0}+S_{3,0}}$. Note in particular that the first entry does not have the $-\ln \lambda$ term in its formal expansion around $\lambda=\infty$  contrary to all other entries.

In order to obtain an operator commuting with $\partial_\lambda$ we shall define $\mathcal{L}$ acting directly on $\Psi(\lambda,\hbar)$ out of the operator $\mathcal{L}_{KZ} -  t_{2,2} t_{3,2}$ as in the degree 2 example so that $\mathcal{L} \Psi(\lambda,\hbar) = \text{ev} . \left[\mathcal{L}_{KZ} \Psi^{\mathrm{symbol}}(\lambda,\hbar)\right]$.
 Hence, the compatible system is 
\beq
\mathcal{L} \Psi(\lambda,\hbar) = A(\lambda,\hbar) \Psi(\lambda,\hbar),
\eeq
with
\beq
A(\lambda,\hbar) = 
\begin{pmatrix}  P_{\infty,2}^{(2)} \lambda - \frac{H}{\hbar} + \frac{p_1}{\hbar(\lambda-q)} &  \frac{p_2}{\hbar(\lambda- q)}+t_{1,2} &\frac{1}{\lambda-q}\\A_{2,1}&A_{2,2}&A_{2,3}\\
A_{3,1}&A_{3,2}&A_{3,3}
\end{pmatrix} .
\eeq

\medskip

We may now use the compatibility relation
\beq \mathcal{L}[L(\lambda,\hbar)]=\hbar \partial_\lambda A(\lambda,\hbar)+ [A(\lambda,\hbar),L(\lambda,\hbar)] .\eeq
The l.h.s.~of the first two lines is null so that we obtain the expression of the last two lines of $A(\lambda,\hbar)$. The expressions being rather long, we omit them here. Eventually, the third line of the compatibility relation provides the following equations by identification of the singular parts at $\lambda=q$ and at $\lambda=\infty$,
\bea \label{GL3Evolution1}
\mathcal{L}[P_{\infty,1}^{(1)}]&=&0\cr
\mathcal{L}[P_{\infty,0}^{(1)}]&=& \hbar\left( (P_{\infty,1}^{(1)})^2+P_{\infty,2}^{(2)}+ t_{1,2}P_{\infty,1}^{(1)}\right)\cr
\mathcal{L}[P_{\infty,2}^{(2)}]&=&0\cr
\mathcal{L}[P_{\infty,1}^{(2)}]&=&\hbar( -3P_{\infty,3}^{(3)}+3P_{\infty,1}^{(1)}P_{\infty,2}^{(2)} +2t_{1,2}P_{\infty,2}^{(2)})\cr
\mathcal{L}[P_{\infty,0}^{(2)}]&=&\hbar(P_{\infty,1}^{(2)}P_{\infty,1}^{(1)}-P_{\infty,2}^{(3)}+P_{\infty,0}^{(1)}P_{\infty,2}^{(2)}+P_{\infty,1}^{(2)}t_{1,2}-\mathcal{L}[t_{1,2}])\cr
\mathcal{L}[P_{\infty,3}^{(3)}]&=&0\cr
\mathcal{L}[P_{\infty,2}^{(3)}]&=&\hbar(P_{\infty,1}^{(1)}P_{\infty,3}^{(3)}+(P_{\infty,2}^{(2)})^2+3t_{1,2}P_{\infty,3}^{(3)})\cr
\mathcal{L}[P_{\infty,1}^{(3)}]&=&\hbar(-P_{\infty,0}^{(1)}P_{\infty,3}^{(3)}+P_{\infty,1}^{(1)}P_{\infty,2}^{(3)}+P_{\infty,2}^{(2)}P_{\infty,1}^{(2)}+2t_{1,2}P_{\infty,2}^{(3)}\cr
&&-t_{3,2}\mathcal{L}[t_{2,2}]-t_{2,2}\mathcal{L}[t_{3,2}])\cr
p_1&=&\frac{p_2^2}{\hbar}+P_1(q_2)p_2+\hbar P_2(q)+\hbar^2 t_{1,2}\cr
\mathcal{L}[q]&=& -3\frac{p_2^2}{\hbar^2}-4P_1(q)\frac{p_2}{\hbar}-P_2(q)-P_1(q)^2-\hbar(P_{\infty,1}^{(1)}+2t_{1,2}) \cr
\mathcal{L}[p_2]&=&2P_{\infty,1}^{(1)}\frac{p_2^2}{\hbar}+(P_2'(q)+2P'_1(q)P_1(q))p_2-\hbar P_3'(q)+\hbar P_2(q)P_1'(q)+\hbar P_2'(q)P_1(q)\cr
&&+\hbar^2(P_{\infty,1}^{(1)}t_{1,2}-t_{2,2}t_{3,2})\cr
H&=&-H_0(q,p_2,\hbar)+(t_{1,2}+P_{\infty,1}^{(1)})p_2-\hbar P_{\infty,2}^{(2)}q-\hbar P_{\infty,1}^{(2)}-\hbar P_{\infty,0}^{(1)}t_{1,2}\,,\cr
&& 
\eea
where the Hamiltonian $H_0(q,p_2,\hbar)$ is given by
\bea H_0(q,p_2,\hbar)&=&\frac{p_2^3}{\hbar^3}+2P_1(q)\frac{p_2^2}{\hbar^2}+(P_2(q)+P_1(q)^2+\hbar(P_{\infty,1}^{(1)}+2t_{1,2}))\frac{p_2}{\hbar}-P_3(q)+P_1(q)P_2(q)\cr
&&+\hbar(P_{\infty,1}^{(1)}t_{1,2}-t_{2,2}t_{3,2})q
\eea
and satisfies
\beq \hbar \partial_{p_2} H_0(q,p_2,\hbar)=-\mathcal{L}[q] \text{ and } \hbar\partial_q H_0(q,p_2,\hbar)= \mathcal{L}[p_2]\,. \eeq

\medskip

One may use the relations between spectral times and coefficients $(P_{\infty,i}^{(j)})_{i\leq 2, j\leq 3}$ given by \eqref{Pinfty1j}, \eqref{Pinfty2j} and \eqref{Pinfty3j} to rewrite \eqref{GL3Evolution1} as
\bea \label{GL3Evolution2}\mathcal{L}[t_{1,0}]&=&\mathcal{L}[t_{2,0}]=\mathcal{L}[t_{3,0}]=\mathcal{L}[t_{3,2}]=\mathcal{L}[t_{2,2}]=\mathcal{L}[t_{1,2}]=0\cr
\mathcal{L}[t_{1,1}]&=&-\hbar(t_{1,2}t_{2,2}+t_{1,2}t_{3,2}+t_{2,2}t_{3,2})\cr
\mathcal{L}[t_{2,1}]&=&-\hbar(t_{1,2}t_{3,2}+t_{2,2}t_{3,2}+t_{2,2}^2)\cr
\mathcal{L}[t_{3,1}]&=&-\hbar(t_{1,2}t_{2,2}+t_{2,2}t_{3,2}+t_{3,2}^2)\,.
\eea
Thus, we get that
\bea \mathcal{L}&=&-\hbar(t_{1,2}t_{2,2}+t_{1,2}t_{3,2}+t_{2,2}t_{3,2})\partial_{t_{1,1}}-\hbar(t_{1,2}t_{3,2}+t_{2,2}t_{3,2}+t_{2,2}^2)\partial_{t_{2,1}}\cr
&&- \hbar(t_{1,2}t_{2,2}+t_{2,2}t_{3,2}+t_{3,2}^2)\partial_{t_{3,1}}\,.\label{GL3L}\eea
Note that we may now obtain $\mathcal{L}[P_{\infty,i}^{(j)}]$ for $(i,j)\in\llbracket 0,2\rrbracket \times \llbracket 1,3\rrbracket$ in terms of of the spectral times $(t_{i,j})_{i\leq 2, j\leq 3}$ using \eqref{Pinfty1j}, \eqref{Pinfty2j}, \eqref{Pinfty3j} and \eqref{GL3Evolution2}. The results being long and not particularly enlightening, we omit them here. 

\medskip

We may now write the evolution equation for $q$. We first observe that
\bea \label{EqL2q}\mathcal{L}^2[q]&=&\left(2P_1'(q)P_1(q)-3P_2'(q)\right)\frac{p_2^2}{\hbar^2}\cr
&&+\big(6P_3'(q)-2P_1'(q)P_2(q)-6P_2'(q)P_1(q)+4P_1'(q)P_1(q)^2-\hbar(4P_{\infty,2}^{(2)}\cr
&&+2P_{\infty,1}^{(1)}t_{1,2}-6t_{2,2}t_{3,2})\big)\frac{p_2}{\hbar}+(P_2(q)+P_1(q)^2)(P_2'(q)+2P_1'(q)P_1(q))\cr
&&+4P_1(q)(P_3'(q)-P_1'(q)P_2(q)-P_1(q)P_2'(q))-\hbar\big((t_{1,2}\cr
&&-t_{3,2})(t_{1,2}-t_{2,2})(2t_{1,2}+t_{3,2}+t_{2,2})q \cr
&&+(t_{1,2}-t_{3,2})(t_{1,2}-t_{2,2})(t_{3,1}+2t_{1,1}+t_{2,1}) \big)\,.\cr
&&
\eea
In particular, we get that
\beq\label{Relation1} \mathcal{L}^2[q]+\frac{1}{3}(2P_1'(q)P_1(q)-3P_2'(q))\mathcal{L}[q]=R_0(q;h)\frac{p_2}{\hbar}+R_3(q)+\hbar R_1(q),
\eeq
where 
\bea  \label{CoeffGl31}
R_0(\lambda;h)&=&6P_3'(\lambda)-2P_2(\lambda)P_1'(\lambda)-2P_2'(\lambda)P_1(\lambda)+\frac{4}{3}P_1(\lambda)^2P_1'(\lambda) \cr
&&+2\hbar(t_{1,2}-t_{2,2})(t_{1,2}-t_{3,2})\,,\cr
R_3(\lambda)&=&(P_2(\lambda)+P_1(\lambda)^2)(P_2'(\lambda)+2P_1'(\lambda)P_1(\lambda))\cr
&&+4P_1(\lambda)(P_3'(\lambda)-P_1'(\lambda)P_2(\lambda)-P_1(\lambda)P_2'(\lambda))\cr
&&-\frac{1}{3}(2P_1'(\lambda)P_1(\lambda)+3P_2'(\lambda))(P_2(\lambda)+P_1(\lambda)^2)\cr
&=&\left(2P_2'(\lambda)+\frac{4}{3}P_1'(\lambda)P_1(\lambda)\right)(P_2(\lambda)+P_1(\lambda)^2)\cr
&&+4P_1(\lambda)(P_3'(\lambda)-P_1'(\lambda)P_2(\lambda)-P_1(\lambda)P_2'(\lambda))\,,\cr
R_1(\lambda)&=&-\big((t_{1,2}-t_{3,2})(t_{1,2}-t_{2,2})(2t_{1,2}+t_{3,2}+t_{2,2})\lambda\cr
&& +(t_{1,2}-t_{3,2})(t_{1,2}-t_{2,2})(t_{3,1}+2t_{1,1}+t_{2,1}) \big)\cr
&&-\frac{1}{3}(2P_1'(\lambda)P_1(\lambda)-3P_2'(\lambda))(P_{\infty,1}^{(1)}+2t_{1,2})\,.
\eea

We also have 
\bea \label{Relation2} 3\frac{\mathcal{L}[p_2]}{\hbar}+2P_1'(q) \mathcal{L}[q]&=&-(2P_1'(q)P_1(q)-3P_2'(q))\frac{p_2}{\hbar}-3P_3'(q)+P_2(q)P_1'(q)+3P_2'(q)P_1(q)\cr
&&-2P_1'(q)P_1(q)^2-\hbar(2(P_{\infty,1}^{(1)})^2+P_{\infty,1}^{(1)}t_{1,2}+3t_{2,2}t_{3,2})\,.
\eea

Applying $\mathcal{L}$ to \eqref{Relation1} leads to
\bea &&\mathcal{L}^3[q]+\frac{1}{3}(2P_1'(q)P_1(q)-3P_2'(q))\mathcal{L}^2[q] \cr
&&+ \frac{1}{3}(2P_1'(q)^2-3P_2''(q))(\mathcal{L}[q])^2+\frac{1}{3}(\mathcal{L}[2P_1P_1'-3P_2'](q) )\mathcal{L}[q] \cr
&&=(R_0'(q;\hbar)\mathcal{L}[q]+\mathcal{L}[R_0](q;\hbar))\frac{p_2}{\hbar}+R_0(q;\hbar)\frac{\mathcal{L}[p_2]}{\hbar}\cr
&&+(R_3'(q)+\hbar R_1'(q))\mathcal{L}[q]+(\mathcal{L}[R_3](q)+\hbar\mathcal{L}[R_1](q)).
\eea
We now use \eqref{Relation2} to remove $\mathcal{L}[p_2]$,
\begin{align} &3\mathcal{L}^3[q]+(2P_1'(q)P_1(q)-3P_2'(q))\mathcal{L}^2[q] + (2P_1'(q)^2-3P_2''(q))(\mathcal{L}[q])^2+(\mathcal{L}[2P_1P_1'-3P_2'](q) )\mathcal{L}[q] \cr
&=\left(3R_0'(q;\hbar)\mathcal{L}[q]+3\mathcal{L}[R_0](q;\hbar)-R_0(q;\hbar)(2P_1'(q)P_1(q)-3P_2'(q))\right)\frac{p_2}{\hbar}\cr
&-2R_0(q;\hbar)P_1'(q)\mathcal{L}[q]-R_0(q;\hbar)(3P_3'(q)-P_2(q)P_1'(q)-3P_2'(q)P_1(q)+2P_1'(q)P_1(q)^2)\cr 
&-\hbar R_0(q;\hbar)(2(P_{\infty,1}^{(1)})^2+P_{\infty,1}^{(1)}t_{1,2}+3t_{2,2}t_{3,2})\cr
&+3(R_3'(q)+\hbar R_1'(q))\mathcal{L}[q]+3(\mathcal{L}[R_3](q)+\hbar\mathcal{L}[R_1](q))\,.
\end{align}
Then we use Equation \eqref{Relation1} to remove $p_2$. Thus we get
\begin{align} &3R_0(q;\hbar)\mathcal{L}^3[q]+R_0(q;\hbar)(2P_1'(q)P_1(q)-3P_2'(q))\mathcal{L}^2[q] + R_0(q;\hbar)(2P_1'(q)^2-3P_2''(q))(\mathcal{L}[q])^2\cr
&+R_0(q;\hbar)(\mathcal{L}[2P_1P_1'-3P_2](q) )\mathcal{L}[q] \cr
&=\left(3R_0'(q;\hbar)\mathcal{L}[q]+3\mathcal{L}[R_0](q;\hbar)-R_0(q;\hbar)(2P_1'(q)P_1(q)-3P_2'(q))\right)\cr
&\,\left(\mathcal{L}^2[q]+\frac{1}{3}(2P_1'(q)P_1(q)-3P_2'(q))\mathcal{L}[q]-R_3(q)-\hbar R_1(q)\right) \cr
&-2R_0(q;\hbar)^2P_1'(q)\mathcal{L}[q]-R_0(q;\hbar)^2(3P_3'(q)-P_2(q)P_1'(q)-3P_2'(q)P_1(q)+2P_1'(q)P_1(q)^2)\cr 
&-\hbar R_0(q;\hbar)^2(2(P_{\infty,1}^{(1)})^2+P_{\infty,1}^{(1)}t_{1,2}+3t_{2,2}t_{3,2})\cr
&+3R_0(q;\hbar)(R_3'(q)+\hbar R_1'(q))\mathcal{L}[q]+3R_0(q;\hbar)(\mathcal{L}[R_3](q)+\hbar\mathcal{L}[R_1](q))\,.
\end{align}

It corresponds to an evolution equation of the form
\beq \label{EvolutionEquationGl3}\alpha_0(q;\hbar)\mathcal{L}^3[q]+ \alpha_1(q;\hbar)\mathcal{L}^2[q]\mathcal{L}[q]+\alpha_2(q;\hbar) (\mathcal{L}[q])^2+\alpha_3(q;\hbar) \mathcal{L}^2[q]+ \alpha_4(q;\hbar) \mathcal{L}[q]+\alpha_5(q;\hbar)=0\,,\eeq
 with
\bea \label{CoeffGl33}
\alpha_0(q;\hbar)&=&3R_0(q;\hbar)\cr
\alpha_1(q;\hbar)&=& -3R_0'(q;\hbar)\cr
\alpha_2(q;\hbar)&=&2R_0(q;\hbar)((P_{\infty,1}^{(1)})^2-3P_{\infty,0}^{(2)})-R_0'(q;\hbar)(2P_1'(q)P_1(q)-3P_2'(q))  \cr
\alpha_3(q;\hbar)&=&2R_0(q;\hbar)(2P_1'(q)P_1(q)-3P_2'(q))-3\mathcal{L}[R_0](q;\hbar)\cr
\alpha_4(q;\hbar)&=&R_0(q;\hbar)(\mathcal{L}[2P_1P_1'-3P_2](q) )+3R_0'(q;\hbar)(R_3(q)+\hbar R_1(q))\cr
&&-\frac{1}{3}(2P_1'(q)P_1(q)-3P_2'(q))(3\mathcal{L}[R_0](q;\hbar)-R_0(q;\hbar)(2P_1'(q)P_1(q)-3P_2'(q)))\cr
&&+2R_0(q;\hbar)^2P_1'(q)-3R_0(q;\hbar)(R_3'(q)+\hbar R_1'(q))\cr
\alpha_5(q;\hbar)&=&(3\mathcal{L}[R_0](q;\hbar)-R_0(q;\hbar)(2P_1'(q)P_1(q)-3P_2'(q)))(R_3(q)+\hbar R_1(q))\cr
&&+R_0(q;\hbar)^2(3P_3'(q)-P_2(q)P_1'(q)-3P_2'(q)P_1(q)+2P_1'(q)P_1(q)^2)\cr
&&+\hbar R_0(q;\hbar)^2(2(P_{\infty,1}^{(1)})^2+P_{\infty,1}^{(1)}t_{1,2}+3t_{2,2}t_{3,2})\cr
&&-3R_0(q;\hbar)(\mathcal{L}[R_3](q)+\hbar\mathcal{L}[R_1](q))\,.
\eea
One may easily obtain $\mathcal{L}[R_0]$, $\mathcal{L}[R_1]$, $\mathcal{L}[R_3]$ and $\mathcal{L}[2P_1P_1'-3P_2]$ from \eqref{GL3Evolution1} and \eqref{GL3Evolution2}. 

\medskip

The Lax pair is compatible with formal series expansion given by \eqref{PSIGL3}. In particular, we have $S_{1,0}=\ln Z_{\mathrm{NP}}$. Inserting these expansions into the last line of $L(\lambda,\hbar)$ and the first line of $A(\lambda,\hbar)$ provides
\bea \mathcal{L}[S_{1,0}]&=&t_{1,2}^2 q-t_{1,2}\frac{p_2}{\hbar}-\frac{H}{\hbar}+t_{1,1}t_{1,2}\cr
\mathcal{L}[S_{2,0}]&=&t_{2,2}^2q -t_{2,2}\frac{p_2}{\hbar}-\frac{H}{\hbar}+t_{2,1}(2t_{2,2}-t_{1,2})\cr
\mathcal{L}[S_{3,0}]&=&t_{3,2}^2 q-t_{3,2}\frac{p_2}{\hbar}-\frac{H}{\hbar}+t_{3,1}(2t_{3,2}-t_{1,2}).
\eea
Thus we get
\beq \mathcal{L}[\ln Z_{\mathrm{NP}}]=t_{1,2}^2 q-t_{1,2}\frac{p_2}{\hbar}-\frac{H}{\hbar}+t_{1,1}t_{1,2},\eeq
giving an explicit relation between the partition function and the Darboux coordinates.

\medskip

Finally, we may perform a linear change of variables $(t_{1,1},t_{2,1},t_{3,1}) \leftrightarrow (\tau_1,\tau_2,\tau_3)$ so that $\mathcal{L}$ identifies to $\hbar \partial_{\tau_1}$. Let us define
\beq \begin{pmatrix}\tau_1\\ \tau_2\\ \tau_3\end{pmatrix} =B \begin{pmatrix} t_{1,1}\\ t_{2,1}\\ t_{3,1}\end{pmatrix} \,\Leftrightarrow \,   \begin{pmatrix} t_{1,1}\\ t_{2,1}\\ t_{3,1}\end{pmatrix} =B^{-1}\begin{pmatrix}\tau_1\\ \tau_2\\ \tau_3\end{pmatrix}\,, \eeq
where $B$ is a $3\times 3$ matrix with coefficients expressed in terms of spectral times different from $(t_{i,1})_{1\leq i\leq 3}$. The chain rule implies that
\beq \forall \, i\in \llbracket 1,3\rrbracket\,:\, \partial_{t_{i,1}}=\frac{\partial \tau_1}{\partial t_{i,1}}\partial_{\tau_1}+\frac{\partial \tau_2}{\partial t_{i,1}}\partial_{\tau_2}+\frac{\partial \tau_3}{\partial t_{i,1}}\partial_{\tau_3}
= b_{1,i}\partial_{\tau_1}+b_{2,i}\partial_{\tau_2}+b_{3,i}\partial_{\tau_3}  .
\eeq
In other words,
\beq \begin{pmatrix} \partial_{t_{1,1}}\\ \partial_{t_{2,1}}\\  \partial_{t_{3,1}}\end{pmatrix}=B^t \begin{pmatrix}\partial_{\tau_1}\\ \partial_{\tau_2}\\ \partial_{\tau_3}\end{pmatrix}.\eeq
Since 
\beq \mathcal{L}=(\mathcal{L}[t_{1,1}],\mathcal{L}[t_{2,1}],\mathcal{L}[t_{3,1}]) \begin{pmatrix} \partial_{t_{1,1}}\\ \partial_{t_{2,1}}\\  \partial_{t_{3,1}}\end{pmatrix}=(\mathcal{L}[t_{1,1}],\mathcal{L}[t_{2,1}],\mathcal{L}[t_{3,1}])  B^t \begin{pmatrix}\partial_{\tau_1}\\ \partial_{\tau_2}\\ \partial_{\tau_3}\end{pmatrix},\eeq
$\mathcal{L}=\hbar \partial_{\tau_1}$ if and only if 
\beq ( \mathcal{L}[t_{1,1}],\mathcal{L}[t_{2,1}],\mathcal{L}[t_{3,1}])  B^t=(\hbar,0,0) \,\Leftrightarrow \, B \begin{pmatrix} \mathcal{L}[t_{1,1}]\\\mathcal{L}[t_{2,1}]\\\mathcal{L}[t_{3,1}]\end{pmatrix}= \begin{pmatrix} \hbar\\0\\0\end{pmatrix}\,\Leftrightarrow \begin{pmatrix} \mathcal{L}[t_{1,1}]\\\mathcal{L}[t_{2,1}]\\\mathcal{L}[t_{3,1}]\end{pmatrix}=B^{-1} \begin{pmatrix} \hbar\\0\\0\end{pmatrix}\, .\eeq
In other words
\beq B^{-1}=\begin{pmatrix} \frac{\mathcal{L}[t_{1,1}]}{\hbar}& c_{1,2}&c_{1,3}\\  \frac{\mathcal{L}[t_{2,1}]}{\hbar}&c_{2,2}&c_{2,3}\\  \frac{\mathcal{L}[t_{3,1}]}{\hbar}& c_{3,2}& c_{3,3}\end{pmatrix}.\eeq
We choose coefficients $(c_{i,j})_{1\leq i,\leq 3, 2\leq j\leq 3}$ so that $\det (B^{-1})=-(t_{3,2}-t_{1,2})(t_{3,2}-t_{2,2})(t_{2,2}-t_{1,2})$ which corresponds to the prefactor of the Wronskian. Indeed, we want the change of variables to remain invertible except when the Wronskian vanishes. We find that
\beq B^{-1}=\begin{pmatrix} -(t_{1,2}t_{2,2}+t_{1,2}t_{3,2}+t_{2,2}t_{3,2})& 1&t_{1,2}\\
-(t_{1,2}t_{3,2}+t_{2,2}t_{3,2}+t_{2,2}^2)& 1&t_{2,2}\\
-(t_{1,2}t_{2,2}+t_{2,2}t_{3,2}+t_{3,2}^2)& 1&t_{3,2}
\end{pmatrix}
\eeq
is such that $\det (B^{-1})=-(t_{3,2}-t_{1,2})(t_{3,2}-t_{2,2})(t_{2,2}-t_{1,2})$. It is equivalent to
\begin{align} &-(t_{3,2}-t_{1,2})(t_{3,2}-t_{2,2})(t_{2,2}-t_{1,2})B= \nonumber \\
&\begin{pmatrix}t_{3,2}-t_{2,2}&-(t_{3,2}-t_{1,2})& t_{2,2}-t_{1,2}\\
t_{1,2}(t_{3,2}^2-t_{2,2}^2)& -t_{2,2}(t_{3,2}^2-t_{1,2}^2)& t_{3,2}(t_{2,2}^2-t_{1,2}^2)\\
(t_{3,2}-t_{2,2})(t_{3,2}+t_{2,2}-t_{1,2})& -t_{3,2}(t_{3,2}-t_{1,2})& t_{2,2}(t_{2,2}-t_{1,2})
\end{pmatrix}\,,
\end{align}
which provides a suitable change of variables for which $\mathcal{L}=\hbar \partial_{\tau_1}$.

\bibliographystyle{plain}
\bibliography{BibliQC}

\end{document}